\documentclass[a4paper,onecolumn,11pt,accepted=2025-05-13]{quantumarticle}
\pdfoutput=1
\usepackage[utf8]{inputenc}
\usepackage[english]{babel}
\usepackage[T1]{fontenc}
\usepackage{amsmath}
\usepackage{amssymb}
\usepackage{amsfonts}
\usepackage{hyperref}
\usepackage{bm}

\usepackage[numbers,sort&compress]{natbib}

\usepackage{tikz}
\usepackage{lipsum}

\usepackage{mathtools}
\usepackage{amsthm}
\usepackage{xcolor}
\usepackage{fancybox}
\usepackage{stmaryrd}
\usepackage{ytableau}
\usepackage{dsfont}
\usepackage{multirow}
\usepackage{algorithm,algorithmic}
\usepackage{ascmac}
\usepackage{adjustbox}
\usepackage{physics}

\usepackage{rotating}
\usepackage{pdflscape}

\ytableausetup{boxsize=1em}
\ytableausetup{centertableaux}

\newtheorem{Def}{Definition}
\newtheorem{Thm}[Def]{Theorem}
\newtheorem{Cor}[Def]{Corollary}

\newtheorem{Lem}[Def]{Lemma}

\newtheoremstyle{case}{}{}{}{}{}{:}{ }{}
\theoremstyle{case}

\newcommand{\1}{\mathds{1}}

\newcommand{\dket}[1]{\vert {#1} \rangle\!\rangle}

\newcommand{\dketbra}[1]{\vert {#1} \rangle\!\rangle\!\langle\!\langle {#1} \vert}
\renewcommand{\Im}[0]{{\mathrm{Im}}}
\newcommand{\young}[2]{\mathbb{Y}^{#1}_{#2}}
\newcommand{\myoung}[3]{\mathbb{MY}^{#1}_{#2, #3}}

\newcommand{\mcA}{\mathcal{A}}
\newcommand{\mcB}{\mathcal{B}}
\newcommand{\mcC}{\mathcal{C}}
\newcommand{\mcD}{\mathcal{D}}
\newcommand{\mcE}{\mathcal{E}}
\newcommand{\mcF}{\mathcal{F}}

\newcommand{\mcH}{\mathcal{H}}
\newcommand{\mcI}{\mathcal{I}}

\newcommand{\mcL}{\mathcal{L}}
\newcommand{\mcM}{\mathcal{M}}
\newcommand{\mcN}{\mathcal{N}}
\newcommand{\mcO}{\mathcal{O}}
\newcommand{\mcP}{\mathcal{P}}

\newcommand{\mcR}{\mathcal{R}}
\newcommand{\mcS}{\mathcal{S}}
\newcommand{\mcT}{\mathcal{T}}
\newcommand{\mcU}{\mathcal{U}}

\newcommand{\mcW}{\mathcal{W}}

\newcommand{\CC}{\mathbb{C}}
\newcommand{\RR}{\mathbb{R}}
\newcommand{\mfS}{\mathfrak{S}}

\newcommand{\mfD}{\mathfrak{D}}

\newcommand{\isometry}[2]{\mathbb{V}_{\mathrm{iso}}(#1, #2)}
\newcommand{\U}{\mathbb{U}}
\newcommand{\cone}{\mathrm{Cone}}

\makeatletter
\newcommand{\doublewidetilde}[1]{{
  \mathpalette\double@widetilde{#1}
}}
\newcommand{\double@widetilde}[2]{
  \sbox\z@{$\m@th#1\widetilde{#2}$}
  \ht\z@=.9\ht\z@
  \widetilde{\box\z@}
}
\makeatother

\newcommand{\map}[1]{\mathcal{#1}}
\newcommand{\supermap}[1]{\mathcal{#1}}

\begin{document}

\title{Universal adjointation of isometry operations using conversion of quantum supermaps}
\author{Satoshi Yoshida}
\email{satoshiyoshida.phys@gmail.com}
\affiliation{Department of Physics, Graduate School of Science, The University of Tokyo, Hongo 7-3-1, Bunkyo-ku, Tokyo 113-0033, Japan}
\orcid{0000-0002-0521-5209}
\author{Akihito Soeda}
\email{soeda@nii.ac.jp}
\orcid{0000-0002-7502-5582}
\affiliation{Principles of Informatics Research Division, National Institute of Informatics, 2-1-2 Hitotsubashi, Chiyoda-ku, Tokyo 101-8430, Japan}
\affiliation{Department of Informatics, School of Multidisciplinary Sciences, SOKENDAI (The Graduate University for Advanced Studies), 2-1-2 Hitotsubashi, Chiyoda-ku, Tokyo 101-8430, Japan}
\affiliation{Department of Physics, Graduate School of Science, The University of Tokyo, Hongo 7-3-1, Bunkyo-ku, Tokyo 113-0033, Japan}
\author{Mio Murao}
\email{murao@phys.s.u-tokyo.ac.jp}
\affiliation{Department of Physics, Graduate School of Science, The University of Tokyo, Hongo 7-3-1, Bunkyo-ku, Tokyo 113-0033, Japan}
\affiliation{Trans-scale Quantum Science Institute, The University of Tokyo, Bunkyo-ku, Tokyo 113-0033, Japan}
\orcid{0000-0001-7861-1774}

\maketitle

\begin{abstract}
Identification of possible transformations of quantum objects including quantum states and quantum operations is indispensable in developing quantum algorithms.
Universal transformations, defined as input-independent transformations, appear in various quantum applications.
Such is the case for universal transformations of unitary operations. 
However, extending these transformations to non-unitary operations is nontrivial and largely unresolved. 
Addressing this, we introduce \emph{isometry adjointation} protocols that transform an input isometry operation into its adjoint operation, which include both unitary operation and quantum state transformations. 
The paper details the construction of parallel and sequential isometry adjointation protocols, derived from unitary inversion protocols using quantum combs and the (dual) Clebsch-Gordan transforms, and achieving optimal approximation error.
This error is shown to be independent of the output dimension of the isometry operation.
In particular, we explicitly obtain an asymptotically optimal parallel protocol achieving an approximation error $\epsilon = \Theta(d^2/n)$, where $d$ is the input dimension of the isometry operation and $n$ is the number of calls of the isometry operation.
The research also extends to isometry inversion and universal error detection, employing semidefinite programming to assess optimal performances.
The findings suggest that the optimal performance of general protocols in isometry adjointation and universal error detection is not dependent on the output dimension, and that indefinite causal order protocols offer advantages over sequential ones in isometry inversion and universal error detection.
\end{abstract}

\section{Introduction}
The possibility and impossibility of universal transformation of unknown quantum states have played a major role in quantum information and foundations (e.g., state cloning \cite{wootters1982single}, universal NOT \cite{buvzek1999optimal}, and swap test \cite{buhrman2001quantum} or programmable projective measurement \cite{chabaud2018optimal}).
Recently, transformations of quantum operations, namely, \emph{quantum supermaps}, have been investigated to aim for \emph{higher-order quantum computataion} \cite{bisio2019theoretical}, which is a candidate for the quantum generalization of higher-order functions \cite{ying2016foundations}, as well as its connections to channel resource theory \cite{chitambar2019quantum}, quantum thermodynamics \cite{pollock2018non}, and causal structure \cite{oreshkov2012quantum}.
Since universal transformation of quantum operations can be utilized as an elementary operation in higher-order quantum computation, the full characterization of possible universal transformations of quantum operations is indispensable.
Although universal transformations of unitary operations have been extensively studied \cite{chiribella2016optimal, chiribella2008optimal, bisio2010optimal, sedlak2019optimal, yang2020optimal, sedlak2020probabilistic, bisio2014optimal, dur2015deterministic,  chiribella2015universal, soleimanifar2016no, miyazaki2019complex, ebler2023optimal, araujo2014quantum, bisio2016quantum, dong2019controlled, dong2021success, sardharwalla2016universal, quintino2019probabilistic, quintino2019reversing, quintino2022deterministic, yoshida2023reversing, navascues2018resetting, trillo2020translating, trillo2023universal, chen2024quantum}, their generalization to non-unitary operations have not been well investigated except for a few examples (Ref.~\cite{bisio2011cloning} for POVM measurements and Ref.~\cite{yoshida2023universal} for isometry operations).
In particular, isometry operations represent encoding of quantum information \cite{nielsen2002quantum} or general quantum channel via the Stinespring dilation \cite{stinespring1955positive}, and universal transformation of them would be useful as an elementary building block in higher-order quantum computation.

This work extends one of the most important task of universal transformation of unitary operations, called \emph{unitary inversion} \cite{chiribella2016optimal, sardharwalla2016universal, quintino2019probabilistic, quintino2019reversing, dong2021success, quintino2022deterministic, yoshida2023reversing, navascues2018resetting, trillo2020translating, trillo2023universal, chen2024quantum}, to isometry operations.
Unitary inversion is a task to transform an unknown unitary operation $\map{U}_\mathrm{in}$ into its inverse operation $\map{U}_\mathrm{in}^{-1}$, where $\map{U}_\mathrm{in}$ is given by $\map{U}_\mathrm{in}(\cdot) \coloneqq U_\mathrm{in} \cdot U_\mathrm{in}^\dagger$ for a unitary operator $U_\mathrm{in}$.
Since the inverse of a unitary operator $U_\mathrm{in}$ can be given by the adjoint operator $U_\mathrm{in}^\dagger$, one natural extension of unitary inversion to isometry operations is given by \emph{isometry adjointation}.
We consider an isometry operation $\map{V}_\mathrm{in}(\cdot)\coloneqq V_\mathrm{in} \cdot V_\mathrm{in}^\dagger$, where $V_\mathrm{in}: \CC^d\to \CC^D$ is an isometry operator satisfying $V_\mathrm{in}^\dagger V_\mathrm{in} = \1_d$ for the identity operator $\1_d$.
We denote the set of isometry operators $V_\mathrm{in}:\CC^d\to\CC^D$ by $\isometry{d}{D}$, and the set of $d$-dimensional unitary operators by $\U(d) = \isometry{d}{d}$.
Isometry adjointation is a task to transform an unknown isometry operation $\map{V}_\mathrm{in}$ to a quantum instrument $\{\Phi_a\}_{a\in\{I,O\}}$\footnote{Measurement outcomes $a\in\{I,O\}$ stand for ``in $\Im V_\mathrm{in}$'' and ``out of $\Im V_\mathrm{in}$.''} such that $\Phi_I$ approximates the adjoint operation $\map{V}_\mathrm{in}^{\dagger}$, where $\map{V}_\mathrm{in}^\dagger$ is given by $\map{V}_\mathrm{in}^\dagger(\cdot)\coloneqq V_\mathrm{in}^\dagger \cdot V_\mathrm{in}$.
The adjoint operation $\map{V}_\mathrm{in}^\dagger(\cdot)$ is given as
\begin{align}
    \map{V}_\mathrm{in}^\dagger(\rho_\mathrm{in}) = \map{V}_\mathrm{in}^{-1}(\Pi_{\Im V_\mathrm{in}} \rho_\mathrm{in} \Pi_{\Im V_\mathrm{in}}),
\end{align}
where $\map{V}_\mathrm{in}^{-1}$ is the inverse operation satisfying $\map{V}_\mathrm{in}^{-1}\circ \map{V}_\mathrm{in}(\rho) = \rho$ for all $\rho\in\mcL(\CC^d)$, and $\Pi_{\Im V_\mathrm{in}}$ is the orthogonal projector onto $\Im V_\mathrm{in}$.
It implements the projective measurement $\{\Pi_{\Im V_\mathrm{in}}, \1-\Pi_{\Im V_\mathrm{in}}\}$ and the inverse operation $\map{V}_\mathrm{in}^{-1}$ at the same time (see Fig.~\ref{fig:task}).
Isometry adjointation can be reduced to two relevant tasks called \emph{isometry inversion} \cite{yoshida2023universal} and \emph{universal error detection}, where the former is the task to implement the inverse operation $\map{V}_\mathrm{in}^{-1}$, and the latter is the task to implement the projective merasurement $\{\Pi_{\Im V_\mathrm{in}}, \1-\Pi_{\Im V_\mathrm{in}}\}$, from an unknown isometry operation $\map{V}_\mathrm{in}$.

\begin{figure}
    \centering
    \includegraphics[width=0.7\linewidth]{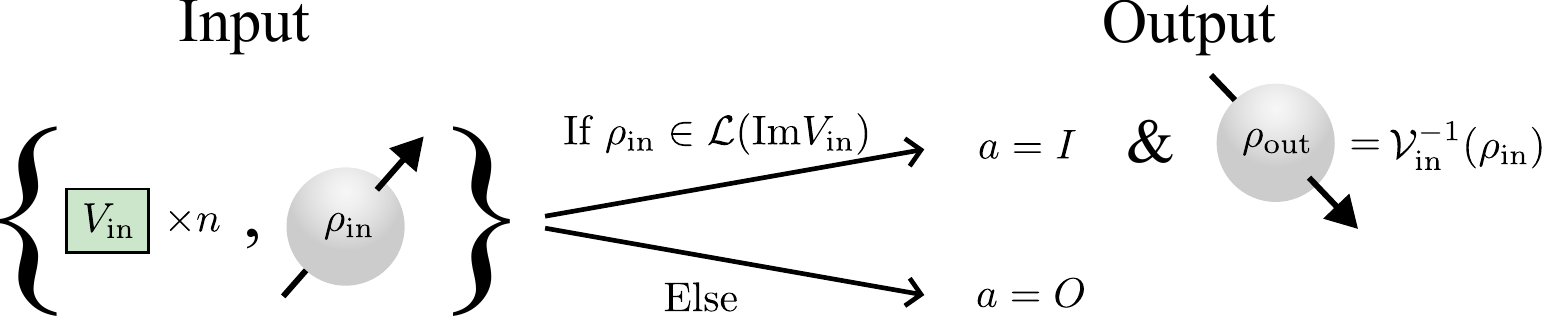}
    \caption{Definition of the task isometry adjointation.}
    \label{fig:task}
\end{figure}

We construct parallel and sequential protocols for isometry adjointation by converting the corresponding unitary inversion protocol (Fig.~\ref{fig:protocol_conversion} and Theorems \ref{thm:parallel_isometry_adjointation} and \ref{thm:sequential_isometry_adjointation}).
Isometry inversion and universal error detection protocols are constructed by discarding the measurement outcome and output quantum state, respectively, from the isometry adjointation protocol (Corollaries \ref{cor:isometry_inversion} and \ref{cor:universal_error_detection}).  We analyze the optimality of the constructed protocols, and we show the following properties:
\begin{itemize}
    \item Our construction of parallel and sequential protocols gives the optimal performances among all parallel or sequential protocols (Theorem \ref{thm:optimal_construction}).
    \item Optimal parallel protocol for universal error detection is explicitly given (Theorem \ref{thm:optimal_parallel_error_detection}).
    \item A parallel protocol for isometry adjointation is given, which achieves an asymptotically optimal approximation error $\epsilon=\Theta(d^2 n^{-1})$  (Theorem \ref{thm:optimal_parallel_isometry_adjointation}).
\end{itemize}
We also give semidefinite programming (SDP) to obtain the optimal performances for these tasks using parallel, sequential, and general protocols including indefinite causal order \cite{hardy2007towards,oreshkov2012quantum,chiribella2013quantum} (Section \ref{subsec:sdp}).

The rest of this work is organized as follows.
Section \ref{sec:problem_setting} defines the tasks isometry adjointation, isometry inversion, and universal error detection, and corresponding figures of merit.
Section \ref{subsec:preliminaries} introduces the technical details to obtain the main result of this paper (Theorem \ref{thm:sequential_isometry_adjointation}) in Section \ref{subsec:construction_isometry_adjointation}, constructing the isometry adjointation protocol.
Section \ref{subsec:reduction} constructs isometry inversion and universal error detection protocols from the isometry adjointation protocol.
Section \ref{sec:unitary_group_symmetry} introduces the Choi operator of the general quantum supermap and shows that the Choi operators corresponding to optimal protocols for isometry adjointation, isometry inversion, and universal error detection have the unitary group symmetry.
Section \ref{subsec:optimal_construction} analyzes the optimal performances using analytical methods based on group theory.
Section \ref{subsec:sdp} numerically investigates the optimal performances using semidefinite programming.
Section \ref{sec:application} discusses potential applications of this work in other fields of quantum information.
Section \ref{sec:conclusion} concludes the work.

\begin{figure}
    \centering
    \includegraphics[width=\linewidth]{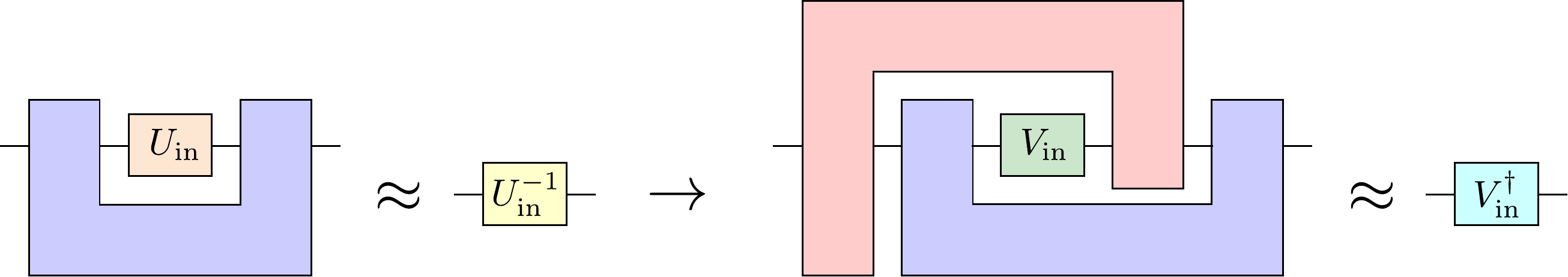}
    \caption{Left panel (unitary inversion): Transformation of a unitary operation $U_\mathrm{in}$ into its inverse operation $U_\mathrm{in}^{-1}$. Right panel (isometry adjointation): Transformation of an isometry operation $V_\mathrm{in}$ into its adjoint operation $V_\mathrm{in}^\dagger$.  The isometry adjointation protocol is constructed by converting the unitary inversion protocol (shown in blue) using a quantum comb shown in red.}
    \label{fig:protocol_conversion}
\end{figure}

\section{Problem setting}
\label{sec:problem_setting}
We consider the following tasks for the given unknown input state $\rho_\mathrm{in}\in\mcL(\CC^D)$ and the unknown isometry operation $\map{V}_\mathrm{in}(\cdot) \coloneqq V_\mathrm{in}\cdot V_\mathrm{in}^\dagger$ for $V_\mathrm{in}\in\isometry{d}{D}$ (see Sections \ref{sec:isometry_inversion_task} -- \ref{sec:isometry_adjointation_task} for the detail):
\begin{itemize}
    \item Probabilistic exact isometry inversion\\
    Promise: The input state $\rho_\mathrm{in}$ is within the code space $\mcL(\Im V_\mathrm{in})$.\\
    Task: Output the quantum state $\map{V}_\mathrm{in}^{-1}(\rho_\mathrm{in})$ with probability $p$ for any input state $\rho_\mathrm{in}\in \mcL(\Im V_\mathrm{in})$ and input operation $\map{V}_\mathrm{in}$.
    \item Deterministic isometry inversion\\
    Promise: The input state $\rho_\mathrm{in}$ is within the code space $\mcL(\Im V)$.\\
    Task: Output a quantum state close to $\map{V}^{-1}(\rho_\mathrm{in})$ for any input state $\rho_\mathrm{in}\in \mcL(\Im V_\mathrm{in})$ and input operation $\map{V}_\mathrm{in}$.
    \item Universal error detection\\
    Task: Output the measurement outcome $a\in\{I, O\}$\footnote{$I$ stands for ``in $\Im V$'' and $O$ stands for ``out of $\Im V$.''} with probability close to $\Tr(\Pi_{\Im V_\mathrm{in}} \rho_\mathrm{in})$ ($a=I$), and $\Tr[(\1_D-\Pi_{\Im V_\mathrm{in}}) \rho_\mathrm{in}]$ ($a=O$) for any input state $\rho_\mathrm{in}\in \mcL(\CC^D)$ and input operation $\map{V}_\mathrm{in}$, where $\Pi_{\Im V_\mathrm{in}}$ is the orthogonal projectors onto the subspace $\Im V_\mathrm{in}$.
    \item Isometry adjointation\\
    Task: Output the quantum state probabilistically with the measurement outcome $a\in\{I, O\}$ with probability close to $\Tr(\Pi_{\Im V_\mathrm{in}} \rho_\mathrm{in})$ ($a=I$), and $\Tr[(\1_D-\Pi_{\Im V_\mathrm{in}}) \rho_\mathrm{in}]$ ($a=O$) for any input state $\rho_\mathrm{in}\in \mcL(\CC^D)$ and input operation $\map{V}_\mathrm{in}$.  When $\rho_\mathrm{in}\in \mcL(\Im V)$, the measurement outcome $a=I$ is obtained with a unit probability with an output quantum state close to $\map{V}^{-1}(\rho_\mathrm{in})$. When the measurement outcome $a=O$ is obtained, an output quantum state is given by a fixed quantum state, e.g., the maximally mixed state.
\end{itemize}

We define the details of the task in the following subsections.
To this, we introduce the notion of quantum supermaps, representing the transformations of quantum channels.  We first consider a deterministic transformation from $n$ quantum channels $\Phi_\mathrm{in}^{(i)}: \mcL(\mcI_i) \to \mcL(\mcO_i)$ for $i\in\{1, \ldots, n\}$ to a quantum channel $\Phi_\mathrm{out}: \mcL(\mcP)\to \mcL(\mcF)$, where $\mcI_i$, $\mcO_i$, $\mcP$, and $\mcF$ are Hilbert spaces\footnote{$\mcP$ stands for ``global past'' and $\mcF$ stands for ``global future.''}.  Such a transformation is represented as a linear supermap
\begin{align}
    \supermap{C}: \bigotimes_{i=1}^{n} [\mcL(\mcI_i)\to \mcL(\mcO_i)] \to [\mcL(\mcP)\to \mcL(\mcF)],
\end{align}
where $[\mcL(\mcH_1)\to\mcL(\mcH_2)]$ is a set of linear maps from $\mcL(\mcH_1)$ to $\mcL(\mcH_2)$.  A probabilistic transformation of quantum channels is represented as a set $\{\supermap{C}_a\}_a$, where $a$ corresponds to the classical outcome.
We consider three classes of transformations, as shown below.
\begin{itemize}
    \item Parallel protocol\\
    We call $\supermap{C}$ is implemented by a parallel protocol when input operations are used in parallel as follows:
    \begin{align}
        \supermap{C}\left[\Phi_\mathrm{in}^{(1)} \otimes \cdots \otimes \Phi_\mathrm{in}^{(n)}\right] = \Lambda^{(2)}\circ \left[\bigotimes_{i=1}^{n} \Phi_\mathrm{in}^{(i)} \otimes \1_{\mcA}\right] \circ \Lambda^{(1)},
    \end{align}
    where $\Lambda^{(1)}: \mcL(\mcP) \to \mcL(\mcI^n \otimes \mcA)$ and $\Lambda^{(2)}: \mcL(\mcO^n \otimes \mcA) \to \mcL(\mcF)$ are quantum channels, $\mcA$ is an auxiliary Hilbert space, and $\mcI^n$ and $\mcO^n$ are joint Hilbert spaces defined by $\mcI^n \coloneqq \bigotimes_{i=1}^{n} \mcI_i$ and $\mcO^n \coloneqq \bigotimes_{i=1}^{n} \mcO_i$, respectively.  Similarly, we consider the implementation of a probabilistic transformation of $\{\supermap{C}_a\}_a$ in a parallel protocol  given by
    \begin{align}
        \supermap{C}_a\left[\Phi_\mathrm{in}^{(1)} \otimes \cdots \otimes \Phi_\mathrm{in}^{(n)}\right] = \Lambda^{(2)}_a\circ \left[\bigotimes_{i=1}^{n} \Phi_\mathrm{in}^{(i)} \otimes \1_{\mcA}\right] \circ \Lambda^{(1)},
    \end{align}
    where $\{\Lambda^{(2)}_a\}_a$ is a quantum instrument\footnote{A quantum instrument is a set of completely-positive maps that sum up to a completely-positive and trace-preserving map \cite{wilde2013quantum}.}.
    We also consider a subclass of parallel protocols called parallel delayed input-state protocols, which are realized as
    \begin{align}
        \supermap{C}\left[\Phi_\mathrm{in}^{(1)} \otimes \cdots \otimes \Phi_\mathrm{in}^{(n)}\right](\rho_\mathrm{in}) = \Lambda \left[\rho_\mathrm{in} \otimes \left(\bigotimes_{i=1}^{n} \Phi_\mathrm{in}^{(i)} \otimes \1_{\mcA}\right)(\phi)\right],\label{eq:def_parallel_delayed_input}
    \end{align}
    where $\rho_\mathrm{in} \in \mcL(\mcP)$ is an input quantum state and $\phi\in\mcL(\mcI^n\otimes \mcA)$ is a quantum state.  Similarly, we consider the parallel delayed input-state protocol of a probabilistic transformation given by
    \begin{align}
        \supermap{C}_a\left[\Phi_\mathrm{in}^{(1)} \otimes \cdots \otimes \Phi_\mathrm{in}^{(n)}\right](\rho_\mathrm{in}) = \Lambda_a \left[\rho_\mathrm{in} \otimes \left(\bigotimes_{i=1}^{n} \Phi_\mathrm{in}^{(i)} \otimes \1_{\mcA}\right)(\phi)\right].
    \end{align}
    \item Sequential protocol (quantum comb \cite{chiribella2008transforming})\\
    We call $\supermap{C}$ is implemented by a sequential protocol or a quantum comb \cite{chiribella2008transforming} when input operations are used sequentially as follows:
    \begin{align}
        \supermap{C}\left[\Phi_\mathrm{in}^{(1)} \otimes \cdots \otimes \Phi_\mathrm{in}^{(n)}\right] = \Lambda^{(n+1)} \circ \left[\Phi_\mathrm{in}^{(n)}\otimes \1_{\mcA_n}\right] \circ \Lambda^{(n)} \circ \cdots \circ \Lambda^{(2)} \circ \left[\Phi_\mathrm{in}^{(1)}\otimes \1_{\mcA_1}\right] \circ \Lambda^{(1)},\label{eq:def_quantum_comb}
    \end{align}
    where $\Lambda^{(1)}: \mcL(\mcP) \to \mcL(\mcI_1\otimes \mcA_1), \Lambda^{(2)}: \mcL(\mcO_1\otimes \mcA_1) \to \mcL(\mcI_2\otimes \mcA_2), \ldots, \Lambda^{(n)}: \mcL(\mcO_{n-1} \otimes \mcA_{n-1}) \to \mcL(\mcI_n \otimes \mcA_n), \Lambda^{(n+1)}: \mcL(\mcO_n \otimes \mcA_n) \to \mcL(\mcF)$ are quantum channels and $\mcA_1, \ldots, \mcA_n$ are auxiliary Hilbert spaces. Similarly, we consider the implementation of a probabilistic transformation $\{\supermap{C}_a\}_a$ in a sequential protocol given by
    \begin{align}
        \supermap{C}_a\left[\Phi_\mathrm{in}^{(1)} \otimes \cdots \otimes \Phi_\mathrm{in}^{(n)}\right] = \Lambda^{(n+1)}_a \circ \left[\Phi_\mathrm{in}^{(n)}\otimes \1_{\mcA_n}\right] \circ \Lambda^{(n)} \circ \cdots \circ \Lambda^{(2)} \circ \left[\Phi_\mathrm{in}^{(1)}\otimes \1_{\mcA_1}\right] \circ \Lambda^{(1)},\label{eq:def_quantum_comb_probabilistic}
    \end{align}
    where $\{\Lambda^{(n+1)}_a\}_a$ is a quantum instrument.
    \item General protocol\\
    We consider the most general case, where $\supermap{C}$ satisfies the following properties:
    \begin{enumerate}
        \item Completely CP preserving: $(\supermap{C} \otimes \1)\left[\Phi_\mathrm{in}^{(1)} \otimes \cdots \otimes \Phi_\mathrm{in}^{(n)}\right]$ is completely positive (CP) for all CP maps $\Phi_\mathrm{in}^{(1)}: \mcL(\mcI_1 \otimes \mcA_1) \to \mcL(\mcO_1 \otimes \mcB_1), \ldots, \Phi_\mathrm{in}^{(n)}: \mcL(\mcI_n \otimes \mcA_n) \to \mcL(\mcO_n \otimes \mcB_n)$, where $\mcA_1, \ldots, \mcA_n$ and $\mcB_1, \ldots, \mcB_n$ are auxiliary Hilbert spaces and $\1$ is the identity supermap defined by $\supermap{\1}(\Phi) = \Phi$ for all $\Phi: \mcL\left[\bigotimes_{i=1}^{n}\mcA_i\right] \to \mcL\left[\bigotimes_{i=1}^{n}\mcB_i\right]$.
        \item TP preserving: $\supermap{C}\left[\Phi_\mathrm{in}^{(1)} \otimes \cdots \otimes \Phi_\mathrm{in}^{(n)}\right]$ is trace preserving (TP) for all TP maps $\Phi_\mathrm{in}^{(1)}: \mcL(\mcI_1) \to \mcL(\mcO_1), \ldots, \Phi_\mathrm{in}^{(n)}: \mcL(\mcI_n) \to \mcL(\mcO_n)$.
    \end{enumerate}
    We say that $\supermap{C}$ is a \emph{quantum superchannel} \cite{chiribella2008transforming} if it satisfies the above conditions. For a probabilistic transformation, we consider the case when $\{\supermap{C}_a\}_a$ is given by a \emph{quantum superinstrument}, namely, $\supermap{C}_a$ is completely CP preserving and $\sum_a \supermap{C}_a$ is TP preserving. The classes of quantum superchannel and superinstrument include indefinite causal order \cite{hardy2007towards,oreshkov2012quantum,chiribella2013quantum}, which is not realizable in a conventional quantum circuit model.
\end{itemize}

\subsection{Isometry inversion}
\label{sec:isometry_inversion_task}
We consider a probabilistic exact or deterministic implementation of isometry inversion.  For a probabilistic exact implementation, we require that the output state is $\map{V}^{-1}(\rho_\mathrm{in})$ for all $\rho_\mathrm{in}\in \mcL(\Im V)$ with probability $p$.  For a deterministic implementation, we require that the output state is obtained deterministically for all $\rho_\mathrm{in}\in \mcL(\Im V)$.  We consider the worst-case channel fidelity defined by
\begin{align}
    F_\mathrm{worst} = \inf_{V\in\isometry{d}{D}} F_\mathrm{ch}[\supermap{C}[\map{V}^{\otimes n}] \circ \map{V}, \1_d],\label{eq:def_average_case_channel_fidelity}
\end{align}
where $F_\text{ch}$ is the channel fidelity \cite{raginsky2001fidelity} and $\1_d$ is the identity operation.
The channel fidelity between a quantum channel $\Lambda$ on a $d$-dimensional system and a $d$-dimensional unitary operation $\map{U}(\cdot) = U \cdot U^\dagger$ is given by
\begin{align}
    F_{\mathrm{ch}}(\Lambda, \map{U}) = {1\over d^2} \Tr[J_{\Lambda}\dketbra{U}],\label{eq:def_channel_fidelity}
\end{align}
where $J_{\Lambda}$ and $\dket{U}$ are the Choi operator of $\Lambda$ and the Choi vector of $U$ \cite{choi1975completely,jamiolkowski1972linear}, respectively, defined by 
\begin{align}
    J_{\Lambda} &\coloneqq \sum_{i,j} \ketbra{i}{j} \otimes \Lambda(\ketbra{i}{j}),\\
    \dket{U} &\coloneqq \sum_i \ket{i}\otimes U\ket{i}
\end{align}
using the computational basis $\{\ket{i}\}$ of the input system (see Section \ref{subsec:choi} for the detail of the Choi representation). Therefore, $F_{\mathrm{ch}}(\Lambda, \map{U})$ is linear with respect to $\Lambda$.  The channel fidelity is invariant under the action of unitary operations, i.e.,
\begin{align}
    F_{\mathrm{ch}}(\map{U}' \circ \Lambda^{(1)} \circ \map{U}, \map{U}' \circ \Lambda^{(2)} \circ \map{U}) = F_{\mathrm{ch}}(\Lambda^{(1)}, \Lambda^{(2)})
\end{align}
holds for any quantum channels $\Lambda^{(1)}$ and $\Lambda^{(2)}$ and unitary operations $\map{U}$ and $\map{U}'$. Therefore, when $D=d$, the definition (\ref{eq:def_average_case_channel_fidelity}) of the figure of merit for deterministic isometry inversion corresponds to the worst-case channel fidelity for deterministic unitary inversion introduced in Ref.~\cite{quintino2022deterministic} given by
\begin{align}
    F_\mathrm{worst} = \inf_{U\in \U(d)} F_{\mathrm{ch}}[\supermap{C}(\map{U}^{\otimes n}), \map{U}^{-1}].\label{eq:unitary_inversion_fidelity}
\end{align}
We denote the optimal success probability of isometry inversion in parallel, sequential, and general protocols by $p_\mathrm{opt}^{(x)}(d,D,n)$ for $x=\mathrm{PAR}$ (parallel), $x=\mathrm{SEQ}$ (sequential), and $x=\mathrm{GEN}$ (general), respectively. Similarly, we denote the optimal fidelity of isometry inversion by $F_\mathrm{opt}^{(x)}(d,D,n)$.

\subsection{Universal error detection}
\label{sec:universal_error_detection_task}

We define the task called universal error detection to implement the POVM $\{\Pi_I, \Pi_O\}$ satisfying the one-sided error condition, i.e.,
\begin{align}
\label{eq:one-sided_error_condition}
    \Tr(\Pi_I \rho_\mathrm{in}) = 1 \quad \forall \rho\in \mcL(\Im V_\mathrm{in}).
\end{align}
We define the figure of merit of universal error detection by the worst-case operational distance given by
\begin{align}
    \delta_\mathrm{worst} \coloneqq \inf_{V_\mathrm{in}\in\isometry{d}{D}} D_\mathrm{op}(\{\Pi_I, \Pi_O\}, \{\Pi_{\Im V_\mathrm{in}}, \1_D-\Pi_{\Im V_\mathrm{in}}\}),
\end{align}
where $D_\mathrm{op}$ is the operational distance \cite{navascues2014energy, puchala2018strategies, puchala2021multiple} of two POVMs defined by
\begin{align}
\label{eq:def_operational_distance}
    D_\mathrm{op}(\{\Pi_i\}_i, \{\Pi'_i\}_i)\coloneqq  \sup_{\rho\geq 0, \Tr\rho = 1} {1\over 2} \sum_i \abs{\Tr(\Pi_i \rho) - \Tr(\Pi'_i \rho)},
\end{align}
which satisfies the convexity and the unitary invariance:
\begin{align}
\label{eq:POVM_distance_convexity}
    &D_\mathrm{op}(\{\sum_j p_j \Pi_i^{(j)}\}_i, \{\sum_j p_j \Pi_i^{\prime (j)}\}_i) \leq \sum_j p_j D_\mathrm{op}(\{ \Pi_i^{(j)}\}_i, \{\Pi_i^{\prime (j)}\}_i) \quad \forall p_j\geq 0,\\
\label{eq:POVM_distance_unitary_invariance}
    &D_\mathrm{op}(\{U \Pi_i U^\dagger\}_i, \{U \Pi_i^{\prime} U^\dagger\}_i) = D_\mathrm{op}(\{\Pi_i\}_i, \{\Pi_i^{\prime}\}_i) \quad \forall U\in\U(D).
\end{align}
For any distance measure satisfying \eqref{eq:POVM_distance_convexity} and \eqref{eq:POVM_distance_unitary_invariance}, we show that the optimal figure of merit is achieved by the protocol outputting the following POVM (see Theorems~\ref{thm:unitary_group_symmetry} and \ref{thm:white_noise_is_enough}):
\begin{align}
\label{eq:universal_error_detection_white_noise}
\begin{split}
    \Pi_I &= \Pi_{\Im V_\mathrm{in}} + \alpha (\1_D-\Pi_{\Im V_\mathrm{in}}),\\
    \Pi_O &= (1-\alpha) (\1_D-\Pi_{\Im V_\mathrm{in}}),
\end{split}
\end{align}
and the worst-case operational distance is given by $\delta_\mathrm{worst} = \alpha$.
We call $\alpha$ the approximation error of the universal error detection, and use $\alpha$ as the figure of merit in the rest of this paper.
We denote the minimal value of $\alpha$ in parallel, sequential, and general protocols as $\alpha_\mathrm{opt}^{(x)}(d,D,n)$ for $x=\mathrm{PAR}$ (parallel), $x=\mathrm{SEQ}$ (sequential), and $x=\mathrm{GEN}$ (general), respectively.

\subsection{Isometry adjointation}
\label{sec:isometry_adjointation_task}
We demand the quantum superinstrument $\{\supermap{C}_I, \supermap{C}_O\}$ satisfies
\begin{align}
\label{eq:one-sided_error_condition_isometry_adjointation}
    \Tr \supermap{C}_I[\map{V}_\mathrm{in}^{\otimes n}](\rho_\mathrm{in}) &= 1 \quad \forall \rho_\mathrm{in} \in \mcD(\Im V),
\end{align}
which corresponds to the one-sided error condition in universal error detection.  We define the quantum channels corresponding to the output quantum instrument and adjoint operation by
\begin{align}
    \supermap{C}[\map{V}_\mathrm{in}^{\otimes n}](\cdot) &\coloneqq \supermap{C}_I [\map{V}_\mathrm{in}^{\otimes n}](\cdot) \otimes \ketbra{0}{0} + \supermap{C}_O [\map{V}_\mathrm{in}^{\otimes n}](\cdot) \otimes \ketbra{1}{1},\label{eq:def_Ctotal}\\
    \map{V}_\mathrm{adjoint}(\cdot)&\coloneqq V_\mathrm{in}^\dagger \cdot V_\mathrm{in} \otimes \ketbra{0}{0} + \Tr[(\1_D-\Pi_{\Im V_\mathrm{in}}) \cdot] {\1_d\over d} \otimes \ketbra{1}{1}, \label{eq:def_Vadjoint}
\end{align}
and define the figure of merit by the worst-case diamond-distance error:
\begin{align}
    \epsilon = \sup_{V_\mathrm{in}\in\isometry{d}{D}} {1\over 2}\Big\| \supermap{C}[\map{V}_\mathrm{in}^{\otimes n}] - \map{V}_\mathrm{adjoint}\Big\|_{\diamond},
\end{align}
where $\|\cdot \|_{\diamond}$ is the diamond norm \cite{kitaev1997quantum}.
We can show that the minimum value of $\epsilon$ is achieved by the protocol outputting the following quantum instrument (see Theorems~\ref{thm:unitary_group_symmetry} and \ref{thm:white_noise_is_enough}):
\begin{align}
\label{eq:isometry_adjointation_detection_decode_error}
\begin{split}
    \supermap{C}_I[\map{V}_\mathrm{in}^{\otimes n}](\cdot) &= (1-\eta) \map{V}_\mathrm{in}^\dagger(\cdot) + \eta {\1_d\over d} \Tr[\Pi_{\Im V_\mathrm{in}} \cdot] + \alpha {\1_d \over d} \Tr[(\1_D-\Pi_{\Im V_\mathrm{in}}) \cdot],\\
    \supermap{C}_O[\map{V}_\mathrm{in}^{\otimes n}](\cdot) &= (1-\alpha){\1_d \over d} \Tr[(\1_D-\Pi_{\Im V_\mathrm{in}}) \cdot],
\end{split}
\end{align}
using $\alpha$ and $\eta$ satisfying $0\leq \alpha, \eta \leq 1$.
We call $\alpha$ and $\eta$ the detection error and the decode error, respectively.
Using the detection error and the decode error, the diamond-distance error is given by (see Appendix~\ref{appendix_subsec:evaluation_diamond-norm})
\begin{align}
    \epsilon = \max\{\alpha, \eta\}.
\end{align}
We denote the minimal value of $\epsilon$ in parallel, sequential, and general protocols by $\epsilon_\mathrm{opt}^{(x)}(d,D,n)$ for $x=\mathrm{PAR}$ (parallel), $x=\mathrm{SEQ}$ (sequential), and $x=\mathrm{GEN}$ (general), respectively.

\section{Construction of isometry adjointation protocols and reduction to isometry inversion and universal error detection}
\label{sec:construction}
In Section \ref{subsec:preliminaries}, we introduce the Choi representation to represent quantum channels and quantum supermaps.  We also introduce the link product to represent their compositions.
Then, we introduce the group theoretic technique called the Schur-Weyl duality.
In Section \ref{subsec:construction_isometry_adjointation}, we construct isometry adjointation protocols by converting unitary inversion protocols using the quantum comb derived from the Schur-Weyl duality.
We also derive isometry inversion and universal error detection protocols from the isometry adjointation protocol.

In the following discussions, we suppose $\mcI_1 = \cdots \mcI_n = \mcF = \mcO'_1 = \cdots = \mcO'^n = \mcP' = \CC^d$, $\mcO_1 = \cdots = \mcO_n = \mcP = \CC^D$, and denote the joint Hilbert spaces by $\mcI^i\coloneqq \bigotimes_{j=1}^{i} \mcI_j$, $\mcO'^i\coloneqq \bigotimes_{j=1}^{i}\mcO'_j$, and $\mcO^i = \bigotimes_{j=1}^{i}\mcO_j$ for $i\in\{1, \ldots, n\}$.
To illustrate the dimensions of the Hilbert spaces corresponding to the wires in the quantum circuits shown in this Section, we utilize the following color code of wires: a red wire corresponds to a $d$-dimensional Hilbert space, a blue wire corresponds to a $D$-dimensional Hilbert space, and a black wire corresponds to a Hilbert space with an arbitrary dimension.
The dual wires in the quantum circuits represent classical information transmissions.

\subsection{Preliminaries}
\label{subsec:preliminaries}
\subsubsection{Choi representation and link product}
\label{subsec:choi}
Any quantum channel $\Lambda: \mcL(\mcI) \to \mcL(\mcO)$ can be represented by the Choi operator $J_{\Lambda} \in \mcL(\mcI\otimes \mcO)$ defined by \cite{choi1975completely,jamiolkowski1972linear}
\begin{align}
    J_{\Lambda} \coloneqq \sum_{i,j} \ketbra{i}{j}_{\mcI} \otimes \Lambda(\ketbra{i}{j})_{\mcO},
\end{align}
where $\{\ket{i}\}$ is the computational basis of $\mcI$. In particular, the Choi operator $J_{\map{V}}$ of an isometry operation $\map{V}(\cdot)\coloneqq V\cdot V^\dagger$ for an isometry operator $V: \mcI \to \mcO$ is represented as
\begin{align}
    J_{\map{V}} = \dketbra{V},
\end{align}
where $\dket{V}\in\mcI\otimes \mcO$ is the Choi vector of $V$ defined by
\begin{align}
    \dket{V} \coloneqq \sum_i \ket{i}_\mcI \otimes (V\ket{i})_\mcO.
\end{align}
A composition of two quantum channels $\Lambda^{(1)}: \mcL(\mcI) \to \mcL(\mcO_1)$ and $\Lambda^{(2)}: \mcL(\mcO_1) \to \mcL(\mcO_2)$ can be represented in terms of the corresponding Choi operators as
\begin{align}
    J_{\Lambda^{(2)}\circ \Lambda^{(1)}} = J_{\Lambda^{(2)}} \star J_{\Lambda^{(1)}},
\end{align}
where $\star$ is the link product \cite{chiribella2009theoretical} defined by
\begin{align}
    A\star B \coloneqq \Tr_{\mcB}[(A^{T_{\mcB}}\otimes \1_{\mcC})(\1_{\mcA} \otimes B)]
\end{align}
for $A\in \mcL(\mcA \otimes \mcB)$ and $B\in \mcL(\mcB\otimes \mcC)$, and $A^{T_{\mcB}}$ is the partial transpose of $A$ with respect to the subsystem $\mcB$. The link product satisfies the commutativity and associativity relations by definition, given as
\begin{align}
    A\star B &= B\star A,\label{eq:commutativity}\\
    (A\star B) \star C &= A\star (B\star C)\label{eq:associativity}
\end{align}
for $A\in \mcL(\mcA \otimes \mcB)$, $B\in \mcL(\mcB\otimes \mcC)$, and $C\in \mcL(\mcC\otimes \mcD)$.  If the two operators $A$ and $B$ does not have an overlap subsystem, i.e., $A\in \mcL(\mcA)$ and $B\in\mcL(\mcB)$ for $\mcA\neq \mcB$, the link product of $A$ and $B$ becomes the tensor product:
\begin{align}
    A\star B = A\otimes B.\label{eq:tensor_product}
\end{align}

Using the Choi operator and the link product, we can represent the quantum combs defined in Eq.~(\ref{eq:def_quantum_comb}).
The Choi operator of the quantum channel shown in the right-hand side of Eq.~(\ref{eq:def_quantum_comb}) is given by
\begin{align}
    J_{\Lambda^{(n+1)}} \star J_{\Phi_\mathrm{in}^{(n)}} \star J_{\Lambda^{(n)}} \star \cdots \star J_{\Lambda^{(2)}} \star J_{\Phi_\mathrm{in}^{(1)}} \star J_{\Lambda^{(1)}}.
\end{align}
Using Eqs.~(\ref{eq:commutativity})-(\ref{eq:tensor_product}), this can be rewritten as
\begin{align}
    C\star \bigotimes_{i=1}^{n}J_{\Phi_\mathrm{in}^{(i)}},
\end{align}
where $C\in \mcL(\mcI^n \otimes \mcO^n \otimes \mcP \otimes \mcF)$ is the Choi operator of the quantum comb $\supermap{C}$ defined by
\begin{align}
    C\coloneqq J_{\Lambda^{(n+1)}} \star J_{\Lambda^{(n)}} \star \cdots \star J_{\Lambda^{(1)}}.\label{eq:choi_comb}
\end{align}
Therefore, the action of a quantum comb $\supermap{C}$ on input quantum operations $\{\Phi_\mathrm{in}^{(1)}, \ldots, \Phi_\mathrm{in}^{(n)}\}$ is given in the Choi representation by
\begin{align}
    J_{\supermap{C}[\Phi_\mathrm{in}^{(1)}\otimes \cdots \otimes \Phi_\mathrm{in}^{(n)}]} = C \star \bigotimes_{i=1}^{n}J_{\Phi_\mathrm{in}^{(i)}}.\label{eq:superchannel_choi}
\end{align}
The following Theorem characterizes the set of Choi operators of quantum combs.
\begin{Thm}
\label{thm:comb_characterization}
\emph{\cite{chiribella2008quantum,chiribella2009theoretical}}
    Suppose $\mcP, \mcF, \mcI_i, \mcO_i$ for $i\in\{1, \ldots, n\}$ are Hilbert spaces and define the joint Hilbert spaces by $\mcI^n \coloneqq \bigotimes_{i=1}^{n} \mcI_i$ and $\mcO^n \coloneqq \bigotimes_{i=1}^{n} \mcO_i$.  The operator $C\in \mcL(\mcI^n \otimes \mcO^n \otimes \mcP \otimes \mcF)$ can be written as Eq.~(\ref{eq:choi_comb}) using quantum channels $\Lambda^{(1)}: \mcL(\mcP) \to \mcL(\mcI_1\otimes \mcA_1), \Lambda^{(2)}: \mcL(\mcO_1\otimes \mcA_1) \to \mcL(\mcI_2\otimes \mcA_2), \ldots, \Lambda^{(n)}: \mcL(\mcO_{n-1} \otimes \mcA_{n-1}) \to \mcL(\mcI_n \otimes \mcA_n), \Lambda^{(n+1)}: \mcL(\mcO_n \otimes \mcA_n) \to \mcL(\mcF)$ and auxiliary Hilbert spaces $\mcA_1, \ldots, \mcA_n$ if and only if $C$ satisfies the following equations:
    \begin{gather}
        C\geq 0,\\
        \Tr_{\mcI_i} C^{(i)} = C^{(i-1)} \otimes \1_{\mcO_{i-1}} \quad \forall i\in \{1, \ldots, n+1\},
    \end{gather}
    where $\mcO_0$ and $\mcI_{n+1}$ are defined by $\mcO_0\coloneqq \mcP$ and $\mcI_{n+1}\coloneqq \mcF$, and $C^{(i)}$ for $i\in \{0, \ldots, n+1\}$ are defined by $C^{(n+1)}\coloneqq C$, $C^{(i-1)} \coloneqq \Tr_{\mcO_{i-1}\mcI_i} C^{(i)}/\dim \mcO_{i-1}$ and $C^{(0)}\coloneqq 1$.
\end{Thm}
The probabilistic transformation $\{\supermap{C}_a\}_a$ in a sequential protocol given by Eq.~(\ref{eq:def_quantum_comb_probabilistic}) can be similarly represented by
\begin{align}
    J_{\supermap{C}_a[\Phi_\mathrm{in}^{(1)}\otimes \cdots \otimes \Phi_\mathrm{in}^{(n)}]} = C_a \star \bigotimes_{i=1}^{n}J_{\Phi_\mathrm{in}^{(i)}},\label{eq:superinstrument_choi}
\end{align}
where $C_a$ is the Choi operator of the probabilistic transformation $\{\supermap{C}_a\}_a$ defined by
\begin{align}
    C_a\coloneqq J_{\Lambda^{(n+1)}_a} \star J_{\Lambda^{(n)}} \star \cdots \star J_{\Lambda^{(1)}}.
\end{align}
The following Theorem characterizes the set of Choi operators of probabilistic transformations implemented in sequential protocols.
\begin{Thm}
\label{thm:probabilistic_comb_characterization}
\emph{\cite{chiribella2008quantum,chiribella2009theoretical}}
    Suppose $\mcP, \mcF, \mcI_i, \mcO_i$ for $i\in\{1, \ldots, n\}$ are Hilbert spaces and define the joint Hilbert spaces by $\mcI^n \coloneqq \bigotimes_{i=1}^{n} \mcI_i$ and $\mcO^n \coloneqq \bigotimes_{i=1}^{n} \mcO_i$.  The set of operators $\{C_a\}_a\subset \mcL(\mcI^n \otimes \mcO^n \otimes \mcP \otimes \mcF)$ can be written as Eq.~(\ref{eq:choi_comb}) using quantum channels $\Lambda^{(1)}: \mcL(\mcP) \to \mcL(\mcI_1\otimes \mcA_1), \Lambda^{(2)}: \mcL(\mcO_1\otimes \mcA_1) \to \mcL(\mcI_2\otimes \mcA_2), \ldots, \Lambda^{(n)}: \mcL(\mcO_{n-1} \otimes \mcA_{n-1}) \to \mcL(\mcI_n \otimes \mcA_n)$, a quantum instrument $\{\Lambda^{(n+1)}_a\}_a: \mcL(\mcO_n \otimes \mcA_n) \to \mcL(\mcF)$ and auxiliary Hilbert spaces $\mcA_1, \ldots, \mcA_n$ if and only if $C$ satisfies the following equations:
    \begin{gather}
        C_a\geq 0,\\
        \Tr_{\mcI_i} C^{(i)} = C^{(i-1)} \otimes \1_{\mcO_{i-1}} \quad \forall i\in \{1, \ldots, n+1\},
    \end{gather}
    where $\mcO_0$ and $\mcI_{n+1}$ are defined by $\mcO_0\coloneqq \mcP$ and $\mcI_{n+1}\coloneqq \mcF$, and $C^{(i)}$ for $i\in \{0, \ldots, n+1\}$ are defined by $C^{(n+1)}\coloneqq \sum_a C_a$, $C^{(i-1)} \coloneqq \Tr_{\mcO_{i-1}\mcI_i} C^{(i)}/\dim \mcO_{i-1}$ and $C^{(0)}\coloneqq 1$.
\end{Thm}

The link product also represents the composition of two quantum combs.
We consider two quantum combs $\supermap{C}'$ and $\supermap{T}$ defined by
\begin{align}
    \supermap{C}'\left[\Phi_{\text{in}}'^{(1)}\otimes \cdots \otimes \Phi_{\text{in}}'^{(n)}\right]&\coloneqq \Lambda'^{(n+1)} \circ (\Phi_{\text{in}}'^{(n)} \otimes \1_{\mcA_n}) \circ \cdots \circ \Lambda'^{(2)} \circ (\Phi_{\text{in}}'^{(1)} \otimes \1_{\mcA_1}) \circ \Lambda'^{(1)},\\
    \supermap{T}\left[\Phi_{\text{in}}''^{(1)} \otimes \cdots \otimes \Phi_{\text{in}}''^{(n)}\right]&\coloneqq \Gamma^{(n+1)} \circ (\Phi_{\text{in}}''^{(n)} \otimes \1_{\mcB_n}) \circ \cdots \circ \Gamma^{(2)} \circ (\Phi_{\text{in}}''^{(1)} \otimes \1_{\mcB_1}) \circ \Gamma^{(1)},
\end{align}
where $\mcA_i$ and $\mcB_i$ for $i\in\{1, \ldots, n\}$ are auxiliary Hilbert spaces and $\Lambda'^{(1)}: \mcL(\mcP) \to \mcL(\mcI_1\otimes \mcA_1)$, $\Lambda'^{(i)}: \mcL(\mcO'_{i-1}\otimes \mcA_{i-1}) \to \mcL(\mcI_i\otimes \mcA_i)$ for $i\in\{2, \ldots, n\}$, $\Lambda^{(n+1)}: \mcL(\mcO'_n \otimes \mcA_n) \to \mcL(\mcF)$, $\Gamma^{(1)}: \mcL(\mcP) \to \mcL(\mcP'\otimes \mcB_1)$, $\Gamma^{(i)}: \mcL(\mcO_{i-1} \otimes \mcB_{i-1}) \to \mcL(\mcO'_{i-1} \otimes \mcB_i)$ for $i\in\{2, \ldots, n\}$, and $\Gamma^{(n+1)}: \mcL(\mcO_n \otimes \mcB_n) \to \mcL(\mcO'_n)$ are quantum channels.
We define the composed quantum comb $\supermap{C}$ by
\begin{align}
    \supermap{C}\left[\Phi_{\text{in}}^{(1)}, \ldots, \Phi_{\text{in}}^{(n)}\right]&\coloneqq \Lambda^{(n+1)} \circ (\Phi_{\text{in}}^{(n)} \otimes \1_{\mcA_n \mcB_n}) \circ \cdots \circ \Lambda^{(2)} \circ (\Phi_{\text{in}}^{(1)} \otimes \1_{\mcA_1 \mcB_1}) \circ \Lambda^{(1)},\\
    \Lambda^{(i)} &\coloneqq (\Lambda'^{(i)} \otimes \1_{\mcB_i}) \circ (\Gamma^{(i)} \otimes \1_{\mcA_{i-1}})\quad \forall i\in \{1, \ldots, n+1\}.\label{eq:protocol_conversion_sequential}
\end{align}
In terms of the Choi operator, this composition can be written as
\begin{align}
    C = C'\star T,\label{eq:protocol_conversion_sequential_choi}
\end{align}
where $C'$, $T$, and $C$ are Choi operators of $\supermap{C}'$, $\supermap{T}$, and $\supermap{C}$, respectively.

\subsubsection{Schur-Weyl duality}
\label{sec:sw_duality}
We introduce the Schur-Weyl duality as follows.  We consider representations of the special unitary group $\U(d)$ and the permutation group $\mfS_n$ on a $n$-fold Hilbert space $(\CC^d)^{\otimes n}$ defined by
\begin{align}
    \U(d) \ni U &\mapsto U^{\otimes n} \in \mcL(\CC^d)^{\otimes n},\\
    \mfS_n \ni \pi &\mapsto P_\pi \in \mcL(\CC^d)^{\otimes n},
\end{align}
where $P_\pi$ is a permutation operator on $(\CC^d)^{\otimes n}$ defined by
\begin{align}
    P_\pi (\ket{\psi_1}\otimes \cdots \otimes \ket{\psi_n}) \coloneqq \ket{\psi_{\pi^{-1}(1)}}\otimes \cdots \otimes\ket{\psi_{\pi^{-1}(n)}} \quad \forall \ket{\psi_1}, \ldots, \ket{\psi_n}\in \CC^d.
\end{align}
The Schur-Weyl duality asserts a simultaneous irreducible decomposition of the two representations $U^{\otimes n}$ and $P_\pi$ given by
\begin{align}
    (\CC^d)^{\otimes n} &= \bigoplus_{\mu\in\young{d}{n}} \mcU_\mu^{(d)} \otimes \mcS_\mu,\label{eq:schur_weyl_decomposition_space}\\
    U^{\otimes n} &= \bigoplus_{\mu\in\young{d}{n}} (U_\mu)_{\mcU_\mu^{(d)}} \otimes \1_{\mcS_\mu},\\
    P_\pi &= \bigoplus_{\mu\in\young{d}{n}} \1_{\mcU_\mu^{(d)}} \otimes (\pi_\mu)_{\mcS_\mu},\label{eq:def_pi_mu}
\end{align}
where $\mu$ runs in the set of Young diagrams with $n$ boxes whose depth is less than or equal to $d$, denoted by $\young{d}{n}$, and $U_\mu$ and $\pi_\mu$ are irreducible representations of $\U(d)$ and $\mfS_n$ on the linear spaces $\mcU_\mu^{(d)}$ and $\mcS_\mu$, respectively.\footnote{The superscript $d$ is put on $\mcU_\mu^{(d)}$ since the properties of the representation space $\mcU_\mu^{(d)}$, e.g., the dimension, depends on the local dimension $d$.  On the other hand, we do not put the superscript $d$ on $\mcS_\mu$ like $\mcS_\mu^{(d)}$ since the representation space $\mcS_\mu^{(d)}$ is automorphic to $\mcS_\mu^{(d')}$ for an arbitrary local dimension $d'$.}
We denote the dimensions of the irreducible representation spaces $\mcU_\mu^{(d)}$ and $\mcS_\mu$ by $m_\mu^{(d)}$ and $d_\mu$, respectively.\footnote{The dimension of $\mcU_{\mu}^{(d)}$ is denoted by $m_{\mu}^{(d)}$ since $\mcU_{\mu}^{(d)}$ is the multiplicity space corresponding to the irreducible representation space $\mcS_\mu$ [see Eq.~(\ref{eq:schur_weyl_decomposition_space})].} They are given by \cite{ceccherini2010representation}
\begin{align}
    d_\mu &= {n! \over \mathrm{hook}(\mu)},\label{eq:dim_hook}\\
    m_\mu^{(d)} &= {\prod_{(i,j)\in\mu}(d+j-i) \over \text{hook}(\mu)},\label{eq:mult_hook}
\end{align}
where $\mathrm{hook}(\mu)$ for $\mu\in\young{d}{n}$ is defined by
\begin{align}
    \mathrm{hook}(\mu) = \prod_{(i,j)\in\mu} (\mu_i + \mu'_j -i-j+1),\label{eq:def_hook}
\end{align}
and $(i,j)$ is the coordinate of a box in the Young diagram $\mu$ such that $i$ represents the row number running from bottom to top and $j$ represents the column number running from left to right. The numbers $\mu_i$ and $\mu'_j$ are the numbers of boxes in the $i$-th row and the $j$-th column, respectively.
The $n$-fold isometry operator $V^{\otimes n}$ for $V\in\isometry{d}{D}$ can also be decomposed as
\begin{align}
    V^{\otimes n} = \bigoplus_{\mu\in\young{d}{n}} (V_\mu)_{\mcU_\mu^{(d)}\to\mcU_\mu^{(D)}} \otimes \1_{\mcS_\mu},\label{eq:decomposition_isometry}
\end{align}
where $V_\mu: \mcU_\mu^{(d)}\to\mcU_\mu^{(D)}$ is an isometry operator, as shown in Ref.~\cite{yoshida2023universal}.

Due to Schur's lemma, any operator commuting with $U^{\otimes n}$ for all $U\in\U(d)$ can be written as a linear combination of the operators $E^{\mu, d}_{ij}$ defined by 
\begin{align}
    E^{\mu, d}_{ij}\coloneqq \1_{\mcU_\mu^{(d)}} \otimes \ketbra{\mu, i}{\mu, j}_{\mcS_{\mu}}\label{eq:def_E}
\end{align}
for $i,j\in \{1, \ldots, d_{\mu}\}$, where $\{\ket{\mu, i}\}$ is an orthonormal basis of $\mcS_{\mu}$. Then, the following relation holds:
\begin{align}
    \Tr E^{\mu, d}_{ij} = m_\mu^{(d)} \delta_{ij},\;\;\;E^{\mu, d}_{ij}E^{\nu, d}_{kl} = \delta_{\mu\nu}\delta_{jk} E^{\mu, d}_{il},
\end{align}
where $\delta_{ij}$ is Kronecker's delta defined by $\delta_{ii}=1$ and $\delta_{ij}=0$ for $i\neq j$.
In particular, $\{m_\mu^{(d)-1/2} E^{\mu,d}_{ij}\}$ forms an orthonormal basis of the set of operators commuting with $U^{\otimes n}$ for all $U\in\U(d)$ under the Hilbert-Schmidt inner product $\langle X, Y \rangle \coloneqq \Tr(X^\dagger Y)$. Thus, any operator $\rho$ commuting with $U^{\otimes n}$ for all $U\in\U(d)$ can be represented as
\begin{align}
    \rho = \sum_{\mu\in\young{d}{n}}\sum_{i,j=1}^{d_\mu} {\Tr(E^{\mu,d}_{ji} \rho) \over m_{\mu}^{(d)}} E^{\mu, d}_{ij}.\label{eq:decomposition_unitary_group_symmetric_state}
\end{align}
Also, we define the Young projector $\Pi_\mu$ by
\begin{align}
    \Pi_\mu^{(d)}\coloneqq \sum_{i=1}^{d_\mu} E^{\mu, d}_{ii},\label{eq:young_projector}
\end{align}
which is an orthonormal projector onto the subspace $\mcU_\mu^{(d)} \otimes \mcS_\mu$.

In particular, we consider the Schur basis of $(\CC^d)^{\otimes n}$ defined by
\begin{align}
    \ket{\mu, u, i}\coloneqq \ket{\mu,u}_{\mcU_\mu^{(d)}} \otimes \ket{\mu,i}_{\mcS_\mu},
\end{align}
where $\{\ket{\mu, u}\}$ is the Gelfand-Zetlin basis of $\mcU_{\mu}^{(d)}$ and $\{\ket{\mu,i}\}$ is the Young-Yamanouchi basis of $\mcS_{\mu}$.
The change of the basis from the computational basis to the Schur basis is called the quantum Schur transform \cite{bacon2006efficient, bacon2007quantum, krovi2019efficient, kirby2018practical, pearce2022multigraph, wills2023generalised}, denoted by $\mathrm{Sch}_{n,d}$.
In the quantum circuit, the Schur basis is represented in three registers corresponding to Hilbert spaces $\mcT_n$, $\mcU_n^{(d)}$ and $\mcS_n$, as
\begin{align}
    \mathrm{Sch}_{n,d} \ket{\mu,u,i} = \ket{\mu}_{\mcT_n} \otimes \ket{u}_{\mcU_n^{(d)}} \otimes \ket{i}_{\mcS_n},
\end{align}
such that the spaces $\mcU_\mu^{(d)}$ and $\mcS_\mu$ can be embedded into $\mcU_n^{(d)}$ and $\mcS_n$, respectively.
We call $\mcT_n$, $\mcU_n^{(d)}$ and $\mcS_n$ the Young diagram register, the unitary group register, and the symmetric group register, respectively.
The standard tableaux with frame $\mu$ is indexed by $i\in\{1, \ldots, d_{\mu}\}$ and the $i$-th standard tableau is denoted by $s^{\mu}_i$. Each element in the Young-Yamanouchi basis $\{\ket{\mu, i}\}$ is associated with the standard tableaux $s^{\mu}_i$.  We also define the set of operators $\{E^{\lambda, d}_{ab}\}$ on $(\CC^d)^{\otimes n-1}$ for $\lambda \in \young{d}{n-1}$ and $a,b\in\{1, \ldots, d_\lambda\}$.
To express the relation between $\{E^{\mu,d}_{ij}\}$ and $\{E^{\lambda,d}_{ab}\}$, we introduce the following notations on the Young diagrams. We denote the set of Young diagrams obtained by adding (removing) a box to $\lambda$ by $\lambda+\square$ ($\lambda-\square$), and define the set $\lambda+_d \square$ by
\begin{align}
    \lambda+_d \square \coloneqq \{\mu\in\lambda+\square \mid (\text{depth of } \mu) \leq d \}.
\end{align}
We also denote the index of the standard tableau $s^{\mu}_{a_{\mu}^{\lambda}}$ obtained by adding a box \fbox{$n$} to a standard tableau $s^{\lambda}_a$ by $a_{\mu}^{\lambda}$.
For instance, for $\lambda$ given by
\begin{align}
    \lambda = \ydiagram[]{2,1},
\end{align}
the sets $\lambda+\square, \lambda+_2 \square, \lambda-\square$ are given by
\begin{align}
    \lambda+\square = \left\{\ydiagram[]{3,1}, \ydiagram[]{2,2}, \ydiagram[]{2,1,1}\right\}, \quad \lambda+_2\square  = (\lambda+\square)\setminus \left\{\ydiagram[]{2,1,1}\right\}, \quad \lambda-\square = \left\{\ydiagram[]{1,1}, \ydiagram[]{2,0}\right\}.
\end{align}
For $\mu\in\lambda+\square$ and $s_a^\lambda$ given by
\begin{align}
    \mu = \ydiagram[]{3,1}, \quad s_a^\lambda = \ytableaushort{1 3,2},
\end{align}
$s^\mu_{a_\mu^\lambda}$ is given by
\begin{align}
    s^\mu_{a_\mu^\lambda} = \ytableaushort{1 3 4,2}.
\end{align}
Then, the following Lemma holds.

\begin{Lem}
    \label{lem:yy}
    \emph{\cite{ram1992matrix, mozrzymas2018simplified, studzinski2022efficient, ebler2023optimal, yoshida2023reversing}}
    The tensor product $E^{\lambda, d}_{ab} \otimes \1_d$ and the partial trace of $E^{\mu, d}_{ij}$ in the last system for $\lambda\in\young{d}{n-1}$ and $\mu\in\young{d}{n}$ are given by 
    \begin{align}
        E^{\lambda, d}_{ab} \otimes \1_d = \sum_{\mu\in\lambda+_d \square} E^{\mu, d}_{a_\mu^\lambda b_\mu^\lambda},\;\;\;\Tr_{n} E^{\mu, d}_{a_\mu^\lambda b_\mu^\kappa} = \delta_{\lambda \kappa}\frac{m_\mu^{(d)}}{m_\lambda^{(d)}} E^{\lambda, d}_{ab}.
    \end{align}
\end{Lem}

\subsection{Construction of isometry adjointation protocol}
\label{subsec:construction_isometry_adjointation}
\subsubsection{Conversion from unitary inversion to isometry adjointation}
\label{subsubsec:protocol_conversion}
In this section, we derive a probabilistic quantum comb $\{\supermap{T}_I, \supermap{T}_O\}$ converting a unitary inversion protocol to an isometry adjointation protocol as
\begin{align}
    &C'\star \dketbra{U_\mathrm{in}}^{\otimes n}_{\mcI^n\mcO'^n} \approx \dketbra{U_\mathrm{in}^{-1}}_{\mcP'\mcF} \quad \forall U_\mathrm{in}\in\U(d)\nonumber\\
    &\Longrightarrow T_I\star C' \star \dketbra{V_\mathrm{in}}^{\otimes n}_{\mcI^n\mcO^n} \approx \dketbra{V_\mathrm{in}^{\dagger}}_{\mcP\mcF} \quad \forall V_\mathrm{in}\in\isometry{d}{D},\label{eq:transformation_unitary_inversion_to_isometry adjointation}
\end{align}
where $C'$ is the Choi operator corresponding to a quantum comb implementing unitary inversion, and $T_I$ is the Choi operator corresponding to $\supermap{T}_I$.
The action of $T_I$ is to ``compress'' the sequentially given quantum states in $\mcP\otimes \mcO^n = (\CC^D)^{\otimes n+1}$ into $\mcP'\otimes \mcO^{\prime n} = (\CC^d)^{\otimes n+1}$ [see also the red quantum comb in Fig.~\ref{fig:deterministic_isometry_adjointation_sequential}~(b)].

\begin{figure}[tbp]
    \centering
    \begin{itembox}[l]{(a) Unitary inversion}
        \begin{minipage}{0.52\linewidth}
            [1] Deterministic protocol
            \begin{center}
            \begin{adjustbox}{width=\linewidth}
                \includegraphics{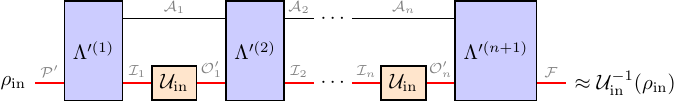}
            \end{adjustbox}
            \end{center}
        \end{minipage}
        \vrule
        \begin{minipage}{0.47\linewidth}
            [2] Probabilistic exact protocol
            \begin{center}
            \begin{adjustbox}{width=\linewidth}
                \includegraphics{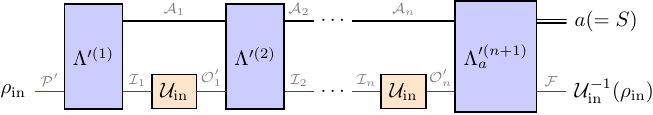}
            \end{adjustbox}
            \end{center}
        \end{minipage}
    \end{itembox}
    \begin{align*}
        \rightarrow \quad
        \begin{minipage}{0.94\linewidth}
            \begin{itembox}[l]{(b) Isometry adjointation}
                \centering
                \begin{adjustbox}{width={\linewidth}}
                    \includegraphics{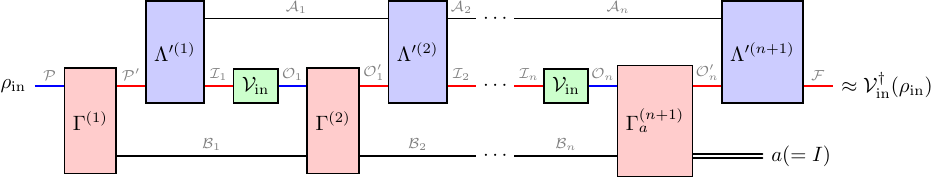}
                \end{adjustbox}
            \end{itembox}
        \end{minipage}
        \end{align*}
        \begin{align*}
        &\rightarrow
        \begin{cases}
            \begin{minipage}{0.94\linewidth}
            \begin{itembox}[l]{(c) Isometry inversion}
                [1] Deterministic protocol\\
                \begin{adjustbox}{width={\linewidth}}
                    \includegraphics{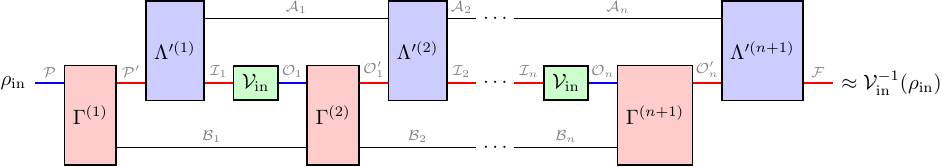}
                \end{adjustbox}\\
                \hrule
                \quad\\
                [2] Probabilistic exact protocol\\
                    \begin{adjustbox}{width=\linewidth}
                        \includegraphics{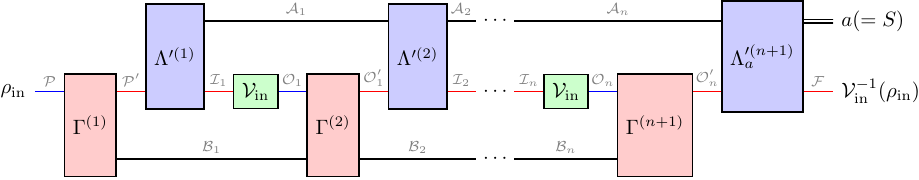}
                    \end{adjustbox}
                \end{itembox}
            \end{minipage}\\
            \quad\\
            \begin{minipage}{0.94\linewidth}
            \begin{itembox}[l]{(d) Universal error detection}
                \centering
                \begin{adjustbox}{width={\linewidth}}
                    \includegraphics{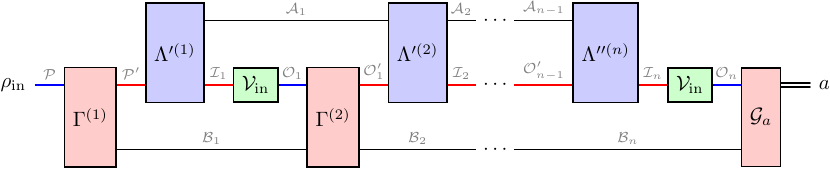}
                \end{adjustbox}
            \end{itembox}
            \end{minipage}
        \end{cases}
    \end{align*}
    \caption{(a) Sequential protocols for deterministic and probabilistic exact unitary inversion. (b) A sequential isometry adjointation protocol is constructed by converting the unitary inversion protocol using a probabilistic quantum comb. (c, d) Reduction to isometry inversion and universal error detection by discarding the measurement outcome and the output state of the isometry adjointation protocol, respectively.}
    \label{fig:deterministic_isometry_adjointation_sequential}
\end{figure}

We define $T_I \in \mcL(\mcP'\otimes \mcO'^n\otimes \mcP\otimes \mcO^n)$ using the operators $\{E^{\mu, d}_{ij}\}$ introduced in Section \ref{sec:sw_duality} by
\begin{align}
    T_I &\coloneqq \sum_{\mu\in \young{d}{n+1}} \sum_{i,j=1}^{d_{\mu}} {(E^{\mu, d}_{ij})_{\mcP'\mcO'^{n}} \otimes (E^{\mu, D}_{ij})_{\mcP\mcO^{n}} \over {m_{\mu}^{(d)}}}.\label{eq:def_TI}
\end{align}
To understand the action of $T_I$, we consider the link product
\begin{align}
    \psi_\mathrm{out}\coloneqq T_I \star \ketbra{\psi_\mathrm{in}}
\end{align}
for $\ket{\psi_\mathrm{in}} \in \mcP\otimes \mcO^n$.
Suppose $\ket{\psi_\mathrm{in}}$ is expressed in the Schur basis as
\begin{align}
    \ket{\psi_\mathrm{in}} = \bigoplus_{\mu\in\young{D}{n+1}} \sum_{u,i} a_{\mu,u,i} \ket{\mu,u}_{\mcU_\mu^{(D)}} \otimes \ket{\mu,i}_{\mcS_\mu}.
\end{align}
Then, $\psi_\mathrm{out}$ is given by
\begin{align}
\label{eq:psi_out}
    \psi_\mathrm{out} = \sum_{\mu\in\young{d}{n+1}} \sum_{u,i,i'} a_{\mu,u,i} a_{\mu,u,i'}^* {\1_{\mcU_\mu^{(d)}} \over m_\mu^{(d)}} \otimes \ketbra{\mu, i}{\mu,i'}_{\mcS_\mu}.
\end{align}
From this observation, the action of $T_I$ is understood to be the following two actions:
\begin{itemize}
    \item Detect whether the state in $\mcP\otimes \mcO^n = (\CC^D)^{\otimes n+1}$ is inside $\bigoplus_{V\in\isometry{d}{D}} \Im V^{\otimes n+1} = \bigoplus_{\mu\in\young{d}{n+1}} \mcU_\mu^{(D)}\otimes \mcS_\mu$ by checking the depth of $\mu$.
    \item Discard the state in $\mcU_\mu^{(D)}$ to compress the state in $\mcP\otimes \mcO^n = (\CC^D)^{\otimes n+1}$ into $\mcP'\otimes \mcO^{\prime n} = (\CC^d)^{\otimes n+1}$.
\end{itemize}
We can use these two properties to implement isometry adjointation, which is discussed more formally in Theorem~\ref{thm:sequential_isometry_adjointation}.

As shown in Fig.~\ref{fig:probabilistic_quantum_comb_construction}, we construct a probabilistic quantum comb $\{\supermap{T}_I, \supermap{T}_O\}$ such that the Choi operator of $\supermap{T}_I$ is given by Eq.~\eqref{eq:def_TI}, i.e., we construct quantum channels $\Gamma^{(1)}: \mcL(\mcP) \to \mcL(\mcI_1\otimes \mcB_1), \Gamma^{(2)}: \mcL(\mcO_1\otimes \mcB_1) \to \mcL(\mcO'_2\otimes \mcB_2), \ldots, \Gamma^{(n)}: \mcL(\mcO_{n-1} \otimes \mcB_{n-1}) \to \mcL(\mcO'_n \otimes \mcB_n)$, a quantum instrument $\{\Gamma^{(n+1)}_I, \Gamma^{(n+1)}_O\}: \mcL(\mcO_n \otimes \mcB_n) \to \mcL(\mcF)$ and auxiliary Hilbert spaces $\mcB_1, \ldots, \mcB_n$ such that
\begin{align}
    T_a = J_{\Gamma^{(n+1)}_a} \star J_{\Gamma^{(n)}} \star \cdots \star J_{\Gamma^{(1)}} \quad \forall a\in\{I, O\},\label{eq:decomp_Ta}
\end{align}
where $J_{\Gamma}$ is the Choi operator of $\Gamma \in \{\Gamma^{(1)}, \ldots, \Gamma^{(n)}, \Gamma_a^{(n+1)}\}$.
The construction of the probabilistic quantum comb $\{\supermap{T}_I, \supermap{T}_O\}$ is based on the (dual) Clebsch-Gordan (CG) transforms (see Appendix~\ref{appendix_subsec:mixed_schur_transform} for the definition).

\begin{landscape}
\begin{figure}[p]
    \centering
    \includegraphics[width=\linewidth]{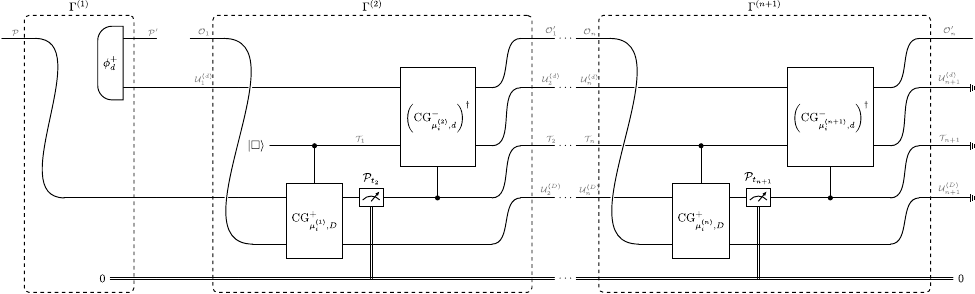}
    \caption{The construction of the probabilistic quantum comb $\{\supermap{T}_I, \supermap{T}_O\}$. The open wires $\mcP, \mcP', \mcO_1, \mcO'_1, \ldots, \mcO_n, \mcO'_n$ correspond to the blue and red wires shown in Fig.~\ref{fig:deterministic_isometry_adjointation_sequential}~(b), and the double wire corresponds to the classical 1-bit register $\mcC$. This figure corresponds to $\supermap{T}_I$. If the measurement outcome of the quantum instrument $\mcP_{t_i}$ is 1 for $i\in\{2,\ldots, n+1\}$, we output the maximally mixed state $\1_{\mcO'_{i-1}}/d$, and the rest of the quantum channels $\Gamma^{(i+1)}, \ldots, \Gamma^{(n+1)}$ becomes trace-and-replace channels. Otherwise, we apply the quantum channel $\Gamma^{(i+1)}$ shown in this figure [see also Eq.~\eqref{eq:Gamma_construction}].}
    \label{fig:probabilistic_quantum_comb_construction}
\end{figure}
\end{landscape}

We define the Hilbert spaces $\mcB_i$ for $i\in\{1,\ldots, n+1\}$ by
\begin{align}
    \mcB_1&\coloneqq \mcU_1^{(d)}\otimes \mcU_1^{(D)}\otimes \mcC,\\
    \mcB_i &\coloneqq \mcU_i^{(d)}\otimes \mcT_i \otimes \mcU_i^{(D)}\otimes \mcC \quad \forall i\geq 2,
\end{align}
where $\mcC$ is a 1-bit classical register, i.e., $\mcC = \mathrm{span}\{\ketbra{0}, \ketbra{1}\}$.
The quantum channel $\Gamma^{(1)}: \mcL(\mcP) \to \mcL(\mcP'\otimes \mcB_1)$ is given by
\begin{align}
    \Gamma^{(1)}(\cdot) \coloneqq (\phi^+_d)_{\mcP'\mcU_1^{(d)}} \otimes (\cdot)_{\mcU_1^{(D)}} \otimes \ketbra{0}_{\mcC}, 
\end{align}
where $\phi^+_d$ is defined by $\phi^+_d = \ketbra{\phi^+_d}$.
The quantum channel $\Gamma^{(i)}: \mcL(\mcO_{i-1} \otimes \mcB_{i-1}) \to \mcL(\mcO'_{i-1}\otimes \mcB_i)$ for $i\in\{2,\ldots, n+1\}$ is given by
\begin{align}
\label{eq:Gamma_construction}
    &\Gamma^{(i)}(\rho_{\mcO_{i-1}\mcU_{i-1}^{(d)}\mcT_{i-1}\mcU_i^{(D)}} \otimes \ketbra{s}_{\mcC})\nonumber\\
    &\coloneqq
    \begin{cases}
        \sum_{t_i=0}^{1}(\mathcal{D}^{(t_i)})_{\mcU_{i-1}^{(d)}\mcT_{i-1}\mcT_i\mcU_i^{(D)} \to \mcO'_{i-1}\mcU_i^{(d)}\mcT_i\mcU_i^{(D)}} \circ [(\mcP_{t_i})_{\mcT_i} \otimes \1_{\mcU_{i-1}^{(d)}\mcT_{i-1}\mcU_i^{(D)}}]\\
        \hspace{30pt}\circ [(\map{CG}^+_{i,D})_{\mcT_{i-1}\mcU_{i-1}^{(D)}\mcO_{i-1} \to \mcT_{i-1}\mcT_i \mcU_i^{(D)}} \otimes \1_{\mcU_{i-1}^{(d)}}](\rho) \otimes \ketbra{t_i}_{\mcC}& (s=0)\\
        {\1_{\mcO'_i} \over D} \Tr(\rho) \otimes \mathrm{garbage}_{\mcU_i^{(d)}\mcT_i\mcU_i^{(D)}} \otimes \ketbra{1}_{\mcC} & (s=1)
    \end{cases},
\end{align}
where $\map{CG}^+_{i,D}$ is the isometry operation defined by
\begin{align}
    \map{CG}^+_{i,D}(\cdot)\coloneqq \mathrm{CG}^+_{i,D}\cdot \left(\mathrm{CG}^+_{i,D}\right)^\dagger
\end{align}
using the CG transform $\mathrm{CG}^+_{i,D}$ defined in Eq.~\eqref{eq:cg_transform}, $\{\mathcal{P}_t\}_{t=0}^{1}$ is the quantum instrument defined by
\begin{align}
    \mathcal{P}_t \coloneqq P_t\cdot P_t, \quad P_0 \coloneqq \sum_{\mu\in \young{d}{i}} \ketbra{\mu}_{\mcT_i}, \quad 
    P_1 \coloneqq \1_{\mcT_i}-P_0,
\end{align}
$\mathcal{D}^{(t)}$ is the quantum channel conditioned on the measurement outcome $t$ of the quantum instrument $\{\mathcal{P}_t\}_{t=0}^{1}$ defined by
\begin{align}
    \mathcal{D}^{(t)}(\cdot)\coloneqq
    \begin{cases}
        [(\map{CG}^-_{i,d})^\dagger_{\mcU_{i-1}^{(d)}\mcT_{i-1}\mcT_i \to \mcO'_{i-1}\mcU_i^{(d)}\mcT_i} \otimes \1_{\mcU_i^{(D)}}](\cdot) & (t=0)\\
        {\1_{\mcO'_i} \over D} \Tr(\cdot) \otimes \mathrm{garbage}_{\mcU_i^{(d)}\mcT_i\mcU_i^{(D)}} & (t=1)
    \end{cases},
\end{align}
$\map{CG}^-_{i,d}$ is the isometry operation defined by
\begin{align}
    \map{CG}^-_{i,d}(\cdot) \coloneqq \mathrm{CG}^-_{i,d} \cdot \left(\mathrm{CG}^-_{i,d}\right)^\dagger
\end{align}
using the dual CG transform $\mathrm{CG}^-_{i,d}$ defined in Eq.~\eqref{eq:dual_cg_transform}, and $\mathrm{garbage}$ is an arbitrary quantum state.
In the quantum circuit shown in Fig.~\ref{fig:probabilistic_quantum_comb_construction}, the (dual) CG transforms are represented as controlled operations since they can be written as controlled isometry operations as shown in Eqs.~\eqref{eq:cg_transform_ctrl} and \eqref{eq:dual_cg_transform_ctrl}.
Note that $\ketbra{\square}_{\mcT_1}$ is omitted in Eq.~\eqref{eq:Gamma_construction} for $i=2$ since $\mcT_1=\CC$ holds.
The quantum instrument $\{\Gamma^{(n+1)}_I, \Gamma^{(n+1)}_O\}: \mcL(\mcO_n\otimes \mcB_n) \to \mcL(\mcO'_n)$ is defined by
\begin{align}
    \Gamma^{(n+1)}_{a_s}(\cdot) \coloneqq \Tr_{\mcU_{n+1}^{(d)}\mcT_{n+1}\mcU_{n+1}^{(D)}\mcC}[\Gamma^{(n+1)}(\cdot) (\1_{\mcO'_n \mcU_{n+1}^{(d)}\mcT_{n+1}\mcU_{n+1}^{(D)}} \otimes \ketbra{s}_{\mcC})] \quad \forall s\in\{0,1\}
\end{align}
for $a_0=I$ and $a_1=O$.
In Appendix~\ref{appendix_subsec:construction_of_red_comb}, we show that this construction gives the Choi operator $T_I$ given in Eq.~\eqref{eq:def_TI}.

The probabilistic quantum comb $\{\supermap{T}_I, \supermap{T}_O\}$ can be efficiently implementable by using efficient implementations of (dual) CG transforms \cite{bacon2006efficient,bacon2007quantum, nguyen2023mixed, grinko2023gelfand, fei2023efficient}.
The bottleneck of the quantum circuit shown in Fig.~\ref{fig:probabilistic_quantum_comb_construction} is a sequence of (dual) CG transforms, which can be implemented with the precision $\epsilon'$ with the circuit complexity $O(\mathrm{poly}(n,D,\log \epsilon'^{-1}))$.

We compose the probabilistic quantum comb $\{\supermap{T}_I, \supermap{T}_O\}$ with the unitary inversion protocol to implement isometry adjointation.
Suppose a quantum comb $\supermap{C}': \bigotimes_{i=1}^{n} [\mcL(\mcI_i) \to \mcL(\mcO'_i)] \to [\mcL(\mcP') \to \mcL(\mcF)]$ given by [see Fig.~\ref{fig:deterministic_isometry_adjointation_sequential} (a-1)]
\begin{align}
    \supermap{C}'\left[\Phi_{\text{in}}^{(1)}, \ldots, \Phi_{\text{in}}^{(n)}\right]&\coloneqq \Lambda'^{(n+1)} \circ (\Phi_{\text{in}}^{(n)} \otimes \1_{\mcA_n}) \circ \cdots \circ \Lambda'^{(2)} \circ (\Phi_{\text{in}}^{(1)} \otimes \1_{\mcA_1}) \circ \Lambda'^{(1)}\label{eq:decomposition_unitary_inversion_comb}
\end{align}
implements deterministic unitary inversion approximately:
\begin{align}
    \supermap{C}'(\map{U}_\mathrm{in}^{\otimes n}) \approx \map{U}_\mathrm{in}^{-1} \quad \forall U_\mathrm{in}\in\U(d).\label{eq:deterministic_sequential_unitary_inversion_protocol}
\end{align}
Reference \cite{quintino2022deterministic} shows that the optimal worst-case channel fidelity of unitary inversion is achieved with the protocol whose Choi operator $C'$ satisfies the $\U(d)\times \U(d)$ symmetry given by
\begin{align}
    [C',U'^{\otimes n+1}_{\mcI^n \mcF} \otimes U''^{\otimes n+1}_{\mcP'\mcO'^n}]=0 \quad \forall U', U''\in\U(d).\label{eq:untiary_inversion_sudsuDsymmetry}
\end{align}
From the sequential unitary inversion protocol satisfying the $\U(d)\times \U(d)$ symmetry (\ref{eq:untiary_inversion_sudsuDsymmetry}), we define a probabilistic transformation as follows [see Fig.~\ref{fig:deterministic_isometry_adjointation_sequential} (b)]:
\begin{align}
    \supermap{C}_a\left[\Phi_{\text{in}}^{(1)}, \ldots, \Phi_{\text{in}}^{(n)}\right]&\coloneqq \Lambda^{(n+1)}_a \circ (\Phi_{\text{in}}^{(n)} \otimes \1_{\mcA_n \mcB_n}) \circ \cdots \circ \Lambda^{(2)} \circ (\Phi_{\text{in}}^{(1)} \otimes \1_{\mcA_1 \mcB_1}) \circ \Lambda^{(1)},\label{eq:concatenation_det}\\
    \Lambda^{(i)} &\coloneqq (\Lambda'^{(i)} \otimes \1_{\mcB_i}) \circ (\Gamma^{(i)} \otimes \1_{\mcA_{i-1}})\quad \forall i\in \{1, \ldots, n\},\\
    \Lambda^{(n+1)}_a &\coloneqq (\Lambda'^{(n+1)} \otimes \1_{\mcB_{n+1}}) \circ (\Gamma^{(n+1)}_a \otimes \1_{\mcA_{n}}),
\end{align}
where $\Gamma^{(1)}, \ldots, \Gamma^{(n)}$ and $\{\Gamma^{(n+1)}_a\}_a$ are quantum channels and a quantum instrument composing the probabilistic quantum comb $\{\supermap{T}_a\}_a$ and $\Lambda'^{(1)}, \ldots, \Lambda'^{(n+1)}$ are quantum channels composing the unitary inversion quantum comb $\supermap{C}'$ [see Eqs.~(\ref{eq:decomp_Ta}) and (\ref{eq:decomposition_unitary_inversion_comb})].
The supermap $\{\supermap{C}_a\}_a$ implements isometry adjointation as shown in the following Theorem.

\begin{Thm}
\label{thm:sequential_isometry_adjointation}
    The sequential protocol shown in Fig.~\ref{fig:deterministic_isometry_adjointation_sequential} (b) implements an isometry adjointation the detection error $\alpha = \alpha_{C'}$ and the decode error $\eta = 1-F_\mathrm{UI}$ in Eq.~\eqref{eq:isometry_adjointation_detection_decode_error}, where $\alpha_{C'}$ is defined by
    \begin{align}
    \alpha_{C'} \coloneqq \Tr[\Tr_\mcF(C') \Sigma],\label{eq:def_alpha_C'}
    \end{align}
    $\Sigma$ is defined by
    \begin{align}
        \Sigma \coloneqq \sum_{\lambda\in\young{d}{n}} \sum_{\mu\in\lambda+_d\square}\sum_{a,b=1}^{d_\lambda}{\mathrm{hook}(\lambda) \over \mathrm{hook}(\mu)}(E^{\lambda, d}_{ab})_{\mcI^n} \otimes {(E^{\mu,d}_{a^\lambda_\mu b^\lambda_\mu})_{\mcO'^n\mcP'} \over m_\mu^{(d)}},\label{eq:def_Sigma}
    \end{align}
    and $F_\mathrm{UI}$ is the worst-case channel fidelity of the unitary inversion.
    Thus, the worst-case diamond-norm error is given by
    \begin{align}
        \epsilon = \max \{ \alpha_{C'}, 1-F_\mathrm{UI}\}.
    \end{align}
\end{Thm}
\begin{proof}[Proof sketch]
We show that the probabilistic quantum comb $\{\supermap{T}_I, \supermap{T}_O\}$ derived in this section satisfies
\begin{align}
    T_I\star \dketbra{V_\mathrm{in}}^{\otimes n}_{\mcI^n\mcO^n} \approx \int_{\U(d)} \dd U \dketbra{UV_\mathrm{in}^\dagger}_{\mcP\mcP'} \otimes \dketbra{U}^{\otimes n}_{\mcI^n\mcO'^n},\label{eq:isometry_to_unitary_approximate}
\end{align}
where $\dd U$ is the Haar measure of $\U(d)$.
Then, Eq.~(\ref{eq:transformation_unitary_inversion_to_isometry adjointation}) holds since
\begin{align}
    T_I\star C' \star \dketbra{V_\mathrm{in}}^{\otimes n}_{\mcI^n\mcO^n}
    &\approx C'\star \int_{\U(d)} \dd U \dketbra{UV_\mathrm{in}^\dagger}_{\mcP\mcP'} \otimes \dketbra{U}^{\otimes n}_{\mcI^n\mcO'^n}\\
    &\approx \int_{\U(d)} \dd U \dketbra{UV_\mathrm{in}^\dagger}_{\mcP\mcP'} \star \dketbra{U^{-1}}_{\mcP'\mcF}\\
    &= \dketbra{V_\mathrm{in}^{\dagger}}_{\mcP\mcF}.
\end{align}
See Appendix \ref{appendix_sec:sequential_isometry_adjointation} for the detail.
\end{proof}

\subsubsection{Asymptotically optimal parallel protocol for isometry adjointation}
\label{subsubsec:parallel_isometry_adjointation_protocol}
We construct an asymptotically optimal parallel protocol for isometry adjointation by converting the parallel protocol for unitary inversion.
The optimal parallel protocol for deterministic unitary inversion is investigated in Ref.~\cite{quintino2022deterministic}, which shows that the estimation-based protocol achieves the optimal worst-case channel fidelity among all parallel protocols.
In the estimation-based unitary inversion protocol, one first estimates the input unitary operation $U_\mathrm{in}$ by applying $U_\mathrm{in}$ in parallel to a quantum state $\phi\in\mcL(\mcI^n\otimes \mcA)$, and measure the output state by a POVM $\{M_i\}\subset \mcL(\mcO'^n\otimes \mcA)$, where $\mcA$ is an auxiliary Hilbert space. We define the measurement channel $\{\map{M}_i\}_i$ corresponding to the POVM $\{M_i\}$ defined by $\map{M}_i(\cdot)\coloneqq \Tr(M_i \cdot)$. Then, one calculates the inverse operation $\map{R}_i$ of the estimated unitary operation and applies $\map{R}_i$ on the input quantum state $\rho_\mathrm{in}$.  This protocol can be expressed as [see Fig.~\ref{fig:deterministic_isometry_adjointation_parallel} (a-1)]
\begin{align}
    \sum_i (\map{R}_i\otimes \map{M}_i)[\rho_\mathrm{in} \otimes  (\map{U}_\mathrm{in}^{\otimes n} \otimes \1_{\mcA})(\phi)] \approx \map{U}_\mathrm{in}^{-1}(\rho_\mathrm{in})\label{eq:deterministic_unitary_inversion_parallel}
\end{align}
for all $U_\mathrm{in}\in\U(d)$ and $\rho_\mathrm{in}\in \mcL(\CC^d)$.
The worst-case channel fidelity of the unitary inversion is the same as the entanglement fidelity of the unitary estimation protocol given by
\begin{align}
    F_\mathrm{est} \coloneqq \inf_{U_\mathrm{in}\in\U(d)} F_{\mathrm{ch}}[\map{E}_{U_\mathrm{in}}, \map{U}_\mathrm{in}].
\end{align}
Here, $\map{E}_{U_\mathrm{in}}$ is the measure-and-prepare channel defined by
\begin{align}
    \map{E}_{U_\mathrm{in}}\coloneqq \sum_i p(\hat{U}_i | U_\mathrm{in}) \hat{\map{U}}_i,\label{eq:def_map_channel}
\end{align}
where $p(\hat{U}_i|U_\mathrm{in})$ is the probability to estimate the input unitary operation $U_\mathrm{in}$ as $\hat{U}_i$.
The optimal estimation is shown to be done with the covariant protocol \cite{chiribella2005optimal}, satisfying 
\begin{align}
    p(U'\hat{U}_i U''|U'U_\mathrm{in}U'') = p(\hat{U}_i|U_\mathrm{in}) \quad \forall U',U''\in\U(d) .
\end{align}

By converting the estimation-based protocol (\ref{eq:deterministic_unitary_inversion_parallel}) using the covariant unitary estimation as shown in Eq.~(\ref{eq:transformation_unitary_inversion_to_isometry adjointation}), we obtain an isometry adjointation protocol given by [see Fig.~\ref{fig:deterministic_isometry_adjointation_parallel} (b)]
\begin{align}
    \supermap{C}_a[\map{V}_\mathrm{in}^{\otimes n}](\rho_\mathrm{in}) \coloneqq \sum_i (\map{R}_i \otimes \map{M}_i) \circ (\Psi_a \otimes \1_{\mcA})[\rho_\mathrm{in} \otimes (\map{V}_\mathrm{in}^{\otimes n} \otimes \1_{\mcA})(\phi)] \label{eq:deterministic_isometry_adjointation_from_unitary_inversion_parallel}
\end{align}
for all $a\in\{I, O\}$, $V_\mathrm{in}\in\isometry{d}{D}$ and $\rho_\mathrm{in}\in \mcL(\CC^D)$, where $\Psi_a: \mcL(\mcO^n\otimes \mcP) \to \mcL(\mcO'^n\otimes \mcP')$ is a quantum instrument implemented by Algorithm \ref{alg:Psi_implementation}.
Note that the quantum instrument $\{\Psi_a\}_a$ is efficiently implementable using an efficient implementation of the quantum Schur transform \cite{bacon2006efficient, bacon2007quantum, krovi2019efficient, kirby2018practical, wills2023generalised}.
The bottleneck of Algorithm \ref{alg:Psi_implementation} is the quantum Schur transform $\mathrm{Sch}_{n+1,D}$, which can be implemented with the precision $\epsilon'$ with the circuit complexity $O(\mathrm{poly}(n, D, \log \epsilon'^{-1}))$ \cite{bacon2006efficient,bacon2007quantum} or $O(\mathrm{poly}(n,\log D, \log \epsilon'^{-1}))$ \cite{krovi2019efficient}\footnote{Note that the result by Ref.~\cite{krovi2019efficient} is challenged by several researchers \cite{qiptalk, grinkoprivate}.}.
The approximation error $\epsilon$ of isometry adjointation is given in the following Theorem.
\begin{Thm}
\label{thm:parallel_isometry_adjointation}
    The parallel protocol shown in Fig.~\ref{fig:deterministic_isometry_adjointation_parallel} (b) implements an isometry adjointation with the detection error $\alpha = \alpha_\phi$ and the decode error $\eta = 1-F_\mathrm{est}$ in Eq.~\eqref{eq:isometry_adjointation_detection_decode_error}, where $\alpha_\phi$ is defined by
    \begin{align}
        \alpha_\phi 
        &\coloneqq
        \sum_{\lambda\in\young{d}{n}} \Tr[\Tr_{\mcA}(\phi) \Pi_\lambda^{(d)}] \sum_{\mu\in\lambda+_d\square} {\mathrm{hook}(\lambda)\over \mathrm{hook}(\mu)}\\
        &=\sum_{\lambda\in\young{d}{n}} \Tr[\Tr_{\mcA}(\phi) \Pi_\lambda^{(d)}] \left[1-\sum_{\mu\in\lambda+\square \setminus \young{d}{n+1}} {\mathrm{hook}(\lambda)\over \mathrm{hook}(\mu)}\right],
    \end{align}
    $\mathrm{hook}(\mu)$ is defined by Eq.~(\ref{eq:def_hook}), and $F_\mathrm{est}$ is the entanglement fidelity of the covariant unitary estimation.
    Thus, the worse-case diamond-norm error is given by
    \begin{align}
        \epsilon
        &= \max \{ \alpha_\phi, 1-F_\mathrm{est}\}.
    \end{align}
\end{Thm}
\begin{proof}
See Appendix \ref{appendix_sec:parallel_isometry_adjointation} for the proof.
\end{proof}

\begin{algorithm}[t]
\caption{Implementation of the quantum instrument $\{\Psi_a\}_{a}$ utilized in the parallel protocol (\ref{eq:deterministic_isometry_adjointation_from_unitary_inversion_parallel}) for isometry adjointation.}
\label{alg:Psi_implementation}
\begin{algorithmic}[1]
\renewcommand{\algorithmicrequire}{\textbf{Input:}}
\renewcommand{\algorithmicensure}{\textbf{Output:}}
\REQUIRE Quantum state $\psi_\mathrm{in}\in\mcL(\CC^D)^{\otimes n+1}$
\ENSURE  Quantum state $\psi_\mathrm{out}\in\mcL(\CC^d)^{\otimes n+1}$ with a measurement outcome $a\in\{I, O\}$
\STATE Apply the quantum Schur transform $\mathrm{Sch}_{n+1,D}$ on the input quantum state $\psi_\mathrm{in}$ to obtain the quantum state in $\mcT_{n+1}\otimes \mcU_{n+1}^{(D)}\otimes \mcS_{n+1}$.
\STATE Measure the Young diagram register $\mcT_{n+1}$ to obtain the measurement outcome $\mu$.
\STATE Trace out the unitary group register $\mcU_{n+1}^{(D)}$.
\STATE Let $\tau_\mu \in\mcL(\mcS_\mu)$ be the quantum state in the symmetric group register $\mcS_{n+1}$.
\IF {$\mu\in\young{d}{n+1}$}
\STATE $a\gets I$.
\STATE Prepare a quantum state $\ketbra{\mu}{\mu}\otimes \1_{\mcU_\mu^{(d)}}/m_\mu^{(d)}$.
\STATE Apply the inverse quantum Schur transform $\mathrm{Sch}_{n+1,d}^\dagger$ on the joint state $\ketbra{\mu}{\mu}\otimes \1_{\mcU_\mu^{(d)}}/m_\mu^{(d)} \otimes \tau_\mu$ to obtain the output quantum state $\psi_\mathrm{out}$.
\ELSE
\STATE $a\gets O$.
\STATE Trace out the quantum state $\tau_\mu \in\mcL(\mcS_\mu)$.
\STATE $\psi_\mathrm{out} \gets \1_d^{\otimes n+1}/d^{n+1}$.
\ENDIF
\RETURN $\psi_\mathrm{out}, a$ 
\end{algorithmic} 
\end{algorithm}

For $d=2$, one can utilize the maximum-likelihood qubit-unitary estimation presented in Refs.~\cite{bagan2004entanglement,chiribella2004efficient,chiribella2005optimal} to achieve
\begin{align}
    \epsilon
    &= \frac{6.2287}{n}+O(n^{-2}).\label{eq:isometry_adjoint_q}
\end{align}
For a higher dimension $d>2$, we can utilize the unitary estimation presented in Ref.~\cite{yang2020optimal} to achieve the following scaling:
\begin{align}
    \epsilon &= {3\ln 2 \over 2}{d^2 \over n}+ O(d^4 n^{-2}, dn^{-1})\\
    &= 1.0397 {d^2 \over n} + O(d^4 n^{-2}, dn^{-1}). 
\end{align}
See Appendix \ref{appendix_sec:optimal_parallel_isometry_adjointation} for the details.
As shown later (Theorem \ref{thm:optimal_parallel_isometry_adjointation}), these protocols achieve the asymptotically optimal worse-case diamond-norm error $\epsilon = \Theta(d^2/n)$.

\begin{figure}[tbp]
    \begin{minipage}{0.6\linewidth}
    \begin{itembox}[l]{(a) Unitary inversion}
        \begin{minipage}{0.49\linewidth}
            {\scriptsize [1] Deterministic protocol}\\
            \begin{adjustbox}{width=\linewidth}
                \includegraphics[]{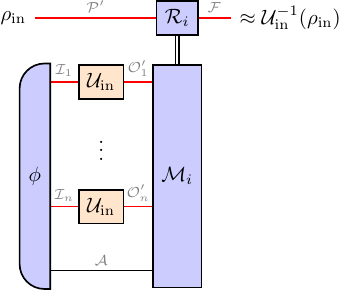}
            \end{adjustbox}
        \end{minipage}
        \vrule
        \begin{minipage}{0.49\linewidth}
            {\scriptsize [2] Probabilistic exact protocol}\\
            \begin{adjustbox}{width=\linewidth}
            \includegraphics{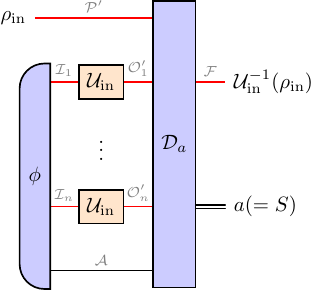}
            \end{adjustbox}
        \end{minipage}
    \end{itembox}
    \end{minipage}
    $\rightarrow$
    \begin{minipage}{0.35\linewidth}
    \begin{itembox}[l]{(b) Isometry adjointation}
    \begin{adjustbox}{width=\linewidth}
    \includegraphics{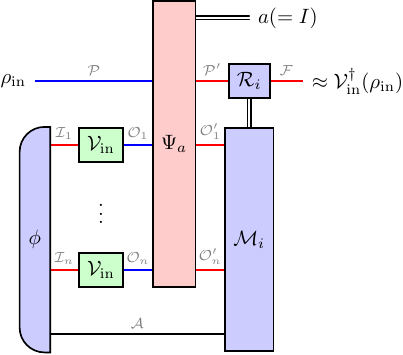}
    \end{adjustbox}
    \end{itembox}
    \end{minipage}
    \begin{align*}
        \rightarrow
        \begin{cases}
            \begin{minipage}{0.94\linewidth}
                \begin{itembox}[l]{(c) Isometry inversion}
                \begin{minipage}{0.49\linewidth}
                    {\scriptsize [1] Deterministic protocol}\\
                    \begin{adjustbox}{width=\linewidth}
                    \includegraphics{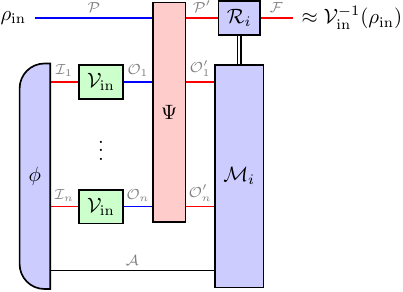}
                    \end{adjustbox}
                \end{minipage}
                \vrule
                \begin{minipage}{0.49\linewidth}
                    {\scriptsize [2] Probabilistic exact protocol}\\
                    \begin{adjustbox}{width=\linewidth}
                    \includegraphics{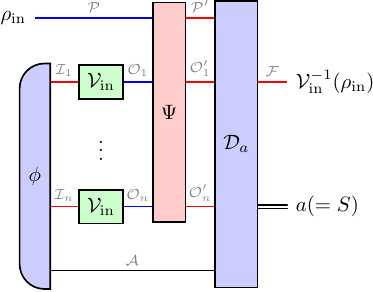}
                    \end{adjustbox}
                \end{minipage}
                \end{itembox}
            \end{minipage}\\
            \quad\\
            \begin{minipage}{0.5\linewidth}
                \begin{itembox}[l]{(d) Universal error detection}
                \centering
                \begin{adjustbox}{width=0.6\linewidth}
                    \includegraphics{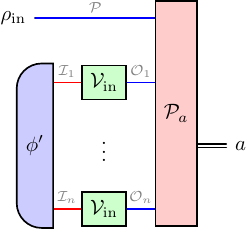}
                \end{adjustbox}
                \end{itembox}
            \end{minipage}
        \end{cases}
    \end{align*}
    \caption{(a) [1] The estimation-based protocol using covariant unitary estimation achieves the optimal protocol for deterministic unitary inversion among all parallel protocols.  [2] The delayed input-state protocol achieves the optimal protocol for probabilistic exact unitary inversion among parallel protocols. (b) A parallel isometry adjointation protocol is constructed by converting the unitary inversion protocol using a quantum instrument. (c, d) Reduction to isometry inversion and universal error detection by discarding the measurement outcome and the output state of the isometry adjointation protocol, respectively.}
    \label{fig:deterministic_isometry_adjointation_parallel}
\end{figure}

\subsection{Reduction to isometry inversion and universal error detection}
\label{subsec:reduction}
By discarding the measurement outcome from the isometry adjointation protocols, we obtain deterministic isometry inversion protocols shown in Figs.~\ref{fig:deterministic_isometry_adjointation_sequential}~(c-1) and \ref{fig:deterministic_isometry_adjointation_parallel}~(c-1), which are given by replacing $\Gamma_a^{(n+1)}$ and $\Psi_a$ in the original isometry adjointation protocols with $\Gamma^{(n+1)}\coloneqq \sum_a \Gamma_a^{(n+1)}$ and $\Psi\coloneqq \sum_a \Psi_a$, respectively.
The worst-case channel fidelity of the derived isometry inversion protocol is shown to be the same as the original unitary inversion protocol (see Appendix \ref{appendix_subsec:reduction_deterministic_isometry_inversion}).
By replacing the original unitary inversion protocol with the probabilistic exact one, we obtain the probabilistic exact isometry inversion protocols as shown in Figs.~\ref{fig:deterministic_isometry_adjointation_sequential}~(c-2), \ref{fig:deterministic_isometry_adjointation_parallel}~(c-2).
The parallel protocol for probabilistic exact isometry inversion is the same as that shown in Ref.~\cite{yoshida2023universal}.
The derived protocols achieve the same success probability as the untiary inversion protocol (see Appendix \ref{appendix_subsec:reduction_probabilistic_isometry_inversion}).
In conclusion, we obtain the following Corollary.

\begin{Cor}
\label{cor:isometry_inversion}
    Suppose there exists a parallel or sequential protocol for probabilistic exact (deterministic) $d$-dimensional unitary inversion achieving success probability $p_\mathrm{UI}$ (average-case channel fidelity $F_\mathrm{UI}$) using $n$ calls of $U_\mathrm{in}\in\U(d)$. Then, we can construct a parallel protocol for probabilistic exact (deterministic) isometry inversion for $V_\mathrm{in}\in\isometry{d}{D}$ achieving success probability $p=p_\mathrm{UI}$ (worst-case channel fidelity $F=F_\mathrm{UI}$) using $n$ calls of $V_\mathrm{in}\in\isometry{d}{D}$.
\end{Cor}

Reference \cite{yoshida2023reversing} shows a deterministic exact sequential protocol for qubit-unitary inversion using four calls of the input qubit-unitary operation $U_\mathrm{in}\in\U(2)$.  Combining this protocol with Theorem \ref{cor:isometry_inversion}, we can construct a deterministic exact sequential protocol for qubit-encoding isometry inversion using four calls of the input qubit-encoding isometry operation $V_\mathrm{in}\in\isometry{2}{D}$ for any $D\geq 2$.

By discarding the output state from the isometry adjointation protocols, we obtain universal error detection protocols shown in Figs.~\ref{fig:deterministic_isometry_adjointation_sequential}~(d) and \ref{fig:deterministic_isometry_adjointation_parallel}~(d).
The sequential protocol shown in Fig.~\ref{fig:deterministic_isometry_adjointation_sequential}~(d) is obtained by discarding $\Lambda''^{(n+1)}$ in the original isometry adjointation protocol and replacing $\Gamma_a^{(n+1)}$ and $\Lambda'^{(n)}$ with the POVM measurement $\map{G}_a\coloneqq \Tr\circ \Gamma_a^{(n+1)}$ and the quantum channel $\Lambda''^{(n)}\coloneqq \Tr_{\mcA_n} \circ \Lambda'^{(n)}$, respectively.
The parallel protocol shown in Fig.~\ref{fig:deterministic_isometry_adjointation_parallel}~(d) is obtained by discarding $\map{R}_i$ and $M_i$ in the original isometry adjointation protocol, and replacing $\phi$ and $\Psi_a$ with the quantum state $\phi'\coloneqq \Tr_{\mcA} (\phi)$ and the POVM measurement $\map{P}_a\coloneqq \Tr\circ \Psi_a$, respectively.
The approximation errors of the derived protocols are shown in the following Corollary (see Appendix \ref{appendix_subsec:reduction_universal_error_detection}).

\begin{Cor}
\label{cor:universal_error_detection}
    The sequential and parallel protocols shown in Figs.~\ref{fig:deterministic_isometry_adjointation_sequential} (d) and \ref{fig:deterministic_isometry_adjointation_parallel} (d) implement a universal error detection with the approximation errors $\alpha^{(x)}$ for $x=\mathrm{SEQ}$ (sequential protocol) and $x=\mathrm{PAR}$ (parallel protocol) given by
    \begin{align}
        \alpha^{(\mathrm{SEQ})} &= \Tr(C''\Sigma'),\\
        \alpha^{(\mathrm{PAR})}
        &= \sum_{\lambda\in\young{d}{n}} \Tr(\phi' \Pi_\lambda^{(d)}) \sum_{\mu\in\lambda+_d\square} {\mathrm{hook}(\lambda)\over \mathrm{hook}(\mu)}\\
        &= \sum_{\lambda\in\young{d}{n}} \Tr(\phi' \Pi_\lambda^{(d)}) \left[1-\sum_{\mu\in\lambda+\square \setminus \young{d}{n+1}} {\mathrm{hook}(\lambda)\over \mathrm{hook}(\mu)}\right],
    \end{align}
    where $C''$ is the Choi operator of the quantum comb given by $C''\coloneqq J_{\Lambda'^{(1)}} \star J_{\Lambda'^{(2)}} \star \cdots \star J_{\Lambda''^{(n)}}$, $\Sigma'$ is defined by $\Sigma'\coloneqq \Tr_{\mcO'_n} \Sigma$ using $\Sigma$ defined in Eq.~(\ref{eq:def_Sigma}), and $\phi'$ is a quantum state shown in the protocol \ref{fig:deterministic_isometry_adjointation_parallel} (d).
\end{Cor}

\subsection{Relationship to programmable projective measurement}
Reference \cite{chabaud2018optimal} considers a task to construct a projective measurement $\{\ketbra{\psi}{\psi}, \1-\ketbra{\psi}{\psi}\}$ from $n$ copies of an unknown quantum state $\ket{\psi}\in \CC^D$. The task is to construct a measurement $\{\Pi_I, \Pi_O\}$ such that
\begin{align}
    \Tr(\Pi_I \ketbra{\psi}{\psi})&=1,\\
    \Tr(\Pi_I \ketbra{\psi^{\perp}}{\psi^{\perp}})&=\alpha\;\;\;(\forall \ket{\psi^{\perp}}\perp \ket{\psi}).
\end{align}
This task can be considered as a special case ($d=1$) of universal error detection. Reference \cite{chabaud2018optimal} shows that the optimal failure probability is given by
\begin{align}
    \alpha=\frac{1}{n+1}.
\end{align}
The optimal success probability is achieved by a protocol shown in Fig.~\ref{fig:projective_measurement}, where $\mcM = \{M_I, M_O\}$ is a POVM defined by $M_I = \Pi_{\mathrm{sym}}$ and $M_O = \1-\Pi_{\mathrm{sym}}$, where $\Pi_{\mathrm{sym}}$ is an orthonormal projector onto the totally symmetric subspace of $(\CC^D)^{\otimes n+1}$.  This protocol corresponds to the universal error detection protocol for $d=1$.
The universal error detection protocol is a generalization of programmable projective measurement for rank-$d$ (destructive) projective measurement given by $\{\Pi_{\Im V_\mathrm{in}}, \1-\Pi_{\Im V_\mathrm{in}}\}$ using an isometry operator $V_\mathrm{in}:\CC^d\to \CC^D$.  In particular, the parallel protocol shown in Fig.~\ref{fig:deterministic_isometry_adjointation_parallel} (d) can be regarded as an implementation of the rank-$d$ projective measurement using a program state $\map{V}_\mathrm{in}^{\otimes n}(\phi')$ (see also Section~\ref{subsec:universal_programming}).

\begin{figure}[th]
    \centering
    \includegraphics{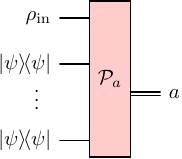}
    \caption{Optimal protocol for programmable projective measurement shown in Ref.~\cite{chabaud2018optimal}.}
    \label{fig:projective_measurement}
\end{figure}

\section{Analysis of the optimal protocols}
\label{sec:analysis}

In Section \ref{sec:unitary_group_symmetry}, we reintroduce the Choi operator of the quantum supermap to describe the general superinstrument that is not implementable by the quantum circuit.
We show that the optimal protocols for isometry adjointation, isometry inversion, and universal error detection can be found in the Choi operators satisfying the unitary group symmetry.
We utilize the unitary group symmetry to investigate the optimal performances analytically in Section \ref{subsec:optimal_construction} and numerically in Section \ref{subsec:sdp}.

\subsection{Choi representation of general superinstruments and $\U(d)\times \U(D)$ symmetry of the tasks}
\label{sec:unitary_group_symmetry}
Similarly to Eq.~(\ref{eq:superchannel_choi}) for a sequential protocol, a general superchannel $\supermap{C}: \bigotimes_{i=1}^{n} [\mcL(\mcI_i) \to \mcL(\mcO_i)] \to [\mcL(\mcP) \to \mcL(\mcF)]$ can be represented in the Choi operator $C$ satisfying
\begin{align}
    J_{\supermap{C}[\Phi_\mathrm{in}^{(1)} \otimes \cdots \otimes \Phi_\mathrm{in}^{(n)}]} = C \star \bigotimes_{i=1}^{n} J_{\Phi_\mathrm{in}^{(i)}},
\end{align}
where $J_{\Phi}$ is the Choi operator of a quantum channel $\Phi$, and $C\in\mcL(\mcI^n \otimes \mcO^n \otimes \mcP \otimes \mcF)$ is the Choi operator of $\supermap{C}$.  The set of superchannels implemented by parallel ($x=\mathrm{PAR}$) and sequential ($x=\mathrm{SEQ}$) protocols, and the set of general superchannel ($x=\mathrm{GEN}$) can be characterized by the positivity and linear conditions on $C$ as
\begin{align}
    C&\geq 0,\\
    C&\in\mcW^{(x)},
\end{align}
where $\mcW^{(x)}$ is a linear subspace of $\mcL(\mcI^n \otimes \mcO^n \otimes \mcP \otimes \mcF)$ [see Appendix \ref{appendix_sec:general_superinstrument_choi} for the definition of $\mcW^{(x)}$].  The case $x=\mathrm{SEQ}$ corresponds to Theorem \ref{thm:comb_characterization}.
A general superinstrument $\supermap{C}_a: \bigotimes_{i=1}^{n} [\mcL(\mcI_i) \to \mcL(\mcO_i)] \to [\mcL(\mcP) \to \mcL(\mcF)]$ can also be represented in the Choi operator $C_a$ satisfying
\begin{align}
    J_{\supermap{C}_a[\Phi_\mathrm{in}^{(1)} \otimes \cdots \otimes \Phi_\mathrm{in}^{(n)}]} = C_a \star \bigotimes_{i=1}^{n} J_{\Phi_\mathrm{in}^{(i)}}.
\end{align}
The set of superinstruments implemented by parallel ($x=\mathrm{PAR}$) and sequential ($x=\mathrm{SEQ}$) protocols, and the set of general superinstrument ($x=\mathrm{GEN}$) can be characterized by the positivity and linear conditions on $C\coloneqq \sum_a C_a$ as
\begin{align}
    C_a&\geq 0,\\
    C\coloneqq \sum_a C_a &\in\mcW^{(x)}.
\end{align}
The case $x=\mathrm{SEQ}$ corresponds to Theorem \ref{thm:probabilistic_comb_characterization}.

The protocols for isometry inversion, universal error detection, and isometry adjointation can be represented by the Choi operators of the corresponding superchannel or superinstrument given by
\begin{align}
    \begin{cases}
        C\in\mcL(\mcI^n \otimes \mcO^n \otimes \mcP \otimes \mcF) & (\mathrm{deterministic \; isometry \; inversion})\\
        \{C_S, C_F\}\subset\mcL(\mcI^n \otimes \mcO^n \otimes \mcP \otimes \mcF)& (\mathrm{probabilistic \; exact \; isometry \; inversion})\\
        \{C_I, C_O\} \subset\mcL(\mcI^n \otimes \mcO^n \otimes \mcP) & (\mathrm{universal \; error \; detection})\\
        \{C_I, C_O\} \subset\mcL(\mcI^n \otimes \mcO^n \otimes \mcP \otimes \mcF) & (\mathrm{isometry \; adjointation})
    \end{cases}.\label{eq:choi_each_task}
\end{align}
The optimization of the Choi operators can be done under the $\U(d)\times \U(D)$ symmetry as shown in the following Theorem.

\begin{Thm}
\label{thm:unitary_group_symmetry}
    The optimal performances of isometry inversion, universal error detection, and isometry adjointation can be searched within Choi operators satisfying
    \begin{align}
    \label{eq:unitary_group_symmetry}
    \begin{cases}
        [C, U^{\prime\otimes n+1}_{\mcI^n \mcF} \otimes U^{\prime\prime \otimes n+1}_{\mcP\mcO^n}] = 0 & (\mathrm{deterministic \; isometry \; inversion})\\
        [C_a, U^{\prime\otimes n+1}_{\mcI^n \mcF} \otimes U^{\prime\prime \otimes n+1}_{\mcP\mcO^n}] = 0   \quad \forall a\in\{S,F\}& (\mathrm{probabilistic \; exact \; isometry \; inversion})\\
        [C_a, U^{\prime\otimes n}_{\mcI^n} \otimes U^{\prime\prime \otimes n+1}_{\mcP\mcO^n}] = 0 \quad \forall a\in\{I,O\} & (\mathrm{universal \; error \; detection})\\
        [C_a, U^{\prime\otimes n+1}_{\mcI^n \mcF} \otimes U^{\prime\prime \otimes n+1}_{\mcP\mcO^n}] = 0 \quad \forall a\in\{I,O\} & (\mathrm{isometry \; adjointation})
    \end{cases}
    \end{align}
    for all $U'\in\U(d)$ and $U''\in\U(D)$.
\end{Thm}
\begin{proof}
    See Appendix \ref{appendix_sec:proof_unitary_group_symmetry} for the proof.
\end{proof}

The universal error detection having the $\U(d)\times \U(D)$ symmetry is shown to output the POVM in the form of Eq.~\eqref{eq:universal_error_detection_white_noise}.
\begin{Thm}
\label{thm:white_noise_is_enough}
    The universal error detection protocol with the Choi operator satisfying the $\U(d)\times \U(D)$ symmetry shown in Theorem~\ref{thm:unitary_group_symmetry} outputs the POVMs given by Eq.~\eqref{eq:universal_error_detection_white_noise}.
    Similarly, the isometry adjointation protocol with the Choi operator satisfying the $\U(d)\times \U(D)$ symmetry shown in Theorem~\ref{thm:unitary_group_symmetry} outputs the quantum instrument given by Eq.~\eqref{eq:isometry_adjointation_detection_decode_error}.
\end{Thm}
\begin{proof}
As shown in Lemma~\ref{lem:deconposition_of_fV} in Appendix~\ref{appendix_sec:unitary_group_symmetry_conditions}, if the Choi operator of the universal error detection protocol satisfies the $\U(d)\times \U(D)$ symmetry \eqref{eq:unitary_group_symmetry}, its output POVM $\{\Pi_a\}$ is given by
\begin{align}
    \Pi_a = v_a \Pi_{\Im V_\mathrm{in}} + w_a (\1_D-\Pi_{\Im V_\mathrm{in}}) \quad \forall a\in\{I,O\}, V_\mathrm{in}\in\isometry{d}{D}
\end{align}
using $v_I, v_O, w_I, w_O\in\CC$.
From the one-sided error condition \eqref{eq:one-sided_error_condition}, the output POVM is given by Eq.~\eqref{eq:universal_error_detection_white_noise}.
Similarly, if the Choi operator of the isometry adjointation protocol satisfies the $\U(d)\times \U(D)$ symmetry \eqref{eq:unitary_group_symmetry}, its output instrument is given by
\begin{align}
    \supermap{C}_a(\map{V}_\mathrm{in}^{\otimes n})(\rho_\mathrm{in}) = x_a V_\mathrm{in}^\dagger \rho_\mathrm{in} V_\mathrm{in} + {\1_\mcF \over d} \Tr[\rho_\mathrm{in} (y_a \Pi_{\Im V_\mathrm{in}} + z_a (\1_D-\Pi_{\Im V_\mathrm{in}}))] \nonumber\\
    \forall a\in\{I,O\}, V_\mathrm{in}\in\isometry{d}{D}
\end{align}
using $x_I, x_O, y_I, y_O, z_I, z_O\in\CC$.
From the one-sided error condition \eqref{eq:one-sided_error_condition_isometry_adjointation}, the output instrument is given by Eq.~\eqref{eq:isometry_adjointation_detection_decode_error}.
\end{proof}

\subsection{Optimal construction of isometry inversion, universal error detection, and isometry adjointation protocols}
\label{subsec:optimal_construction}
We show that the construction of parallel or sequential protocols of isometry adjointation, isometry inversion, and universal error detection are the optimal, as shown in the following Theorem.

\begin{Thm}
\label{thm:optimal_construction}
    The parallel or sequential protocols for probabilistic exact isometry inversion, deterministic isometry inversion, universal error detection, and isometry adjointation shown in Theorems \ref{thm:sequential_isometry_adjointation} and \ref{thm:parallel_isometry_adjointation} and Corollaries \ref{cor:isometry_inversion} and \ref{cor:universal_error_detection} achieve the optimal performances among all parallel or sequential protocols, respectively.
\end{Thm}
\begin{proof}
    See Appendix \ref{appendix_sec:proof_optimal_construction} for the proof.
\end{proof}

Since the figure of merits shown in Theorems \ref{thm:sequential_isometry_adjointation} and \ref{thm:parallel_isometry_adjointation} and Corollaries \ref{cor:isometry_inversion} and \ref{cor:universal_error_detection} do not depend on $D$, we can show that the optimal performances do not depend on $D$.

\begin{Cor}
\label{cor:independence}
    For $D>d$ and $x\in\{\mathrm{PAR}, \mathrm{SEQ}\}$, the following relations hold:
    \begin{align}
        \epsilon_{\mathrm{opt}}^{(x)}(d,D,n) &= \epsilon_{\mathrm{opt}}^{(x)}(d,d+1,n),\\
        F_{\mathrm{opt}}^{(x)}(d,D,n) &= F_{\mathrm{opt}}^{(x)}(d,d,n),\\
        p_{\mathrm{opt}}^{(x)}(d,D,n) &= p_{\mathrm{opt}}^{(x)}(d,d,n),\\
        \alpha_\mathrm{opt}^{(x)}(d,D,n) &= \alpha_\mathrm{opt}^{(x)}(d,d+1,n).
    \end{align}
\end{Cor}

Using Theorem \ref{thm:optimal_construction}, we analyze the optimal protocols for isometry inversion, universal error detection, and isometry adjointation in the following subsections.

\subsubsection{Isometry inversion}
\label{subsec:optimal_isometry_inversion}
The optimal success probability and fidelity of isometry inversion are given by those of unitary inversion when we use parallel or sequential protocols, as shown in Corollary \ref{cor:independence}.  It is already shown in Ref.~\cite{yoshida2023universal} for the parallel protocol, but the sequential protocol case is newly shown in this work, which is conjectured in Ref.~\cite{yoshida2023universal}.

To investigate the generalization of Corollary \ref{cor:independence} for general protocols including indefinite causal order, we calculate the optimal probability or worst-case channel fidelity numerically (see Section \ref{subsec:sdp} for the detail).
Numerical results show that a similar equation does not hold for general protocols including indefinite causal order, i.e., $p_\mathrm{opt}^{(\mathrm{GEN})}(d,D,n) < p_\mathrm{opt}^{(\mathrm{GEN})}(d,d,n)$ or $F_\mathrm{opt}^{(\mathrm{GEN})}(d,D,n) < F_\mathrm{opt}^{(\mathrm{GEN})}(d,d,n)$ hold for some cases (see Fig.~\ref{fig:isometry_inversion_general_optimal}).
This behavior is compatible with the fact that the composition of a general quantum supermap with a quantum comb does not yield a valid quantum supermap in general \cite{guerin2019composition}, so the construction of isometry inversion protocols shown in Fig.~\ref{fig:deterministic_isometry_adjointation_sequential}~(c) cannot be applied for general protocols.

\begin{figure}
    \centering
    \includegraphics[width=\linewidth]{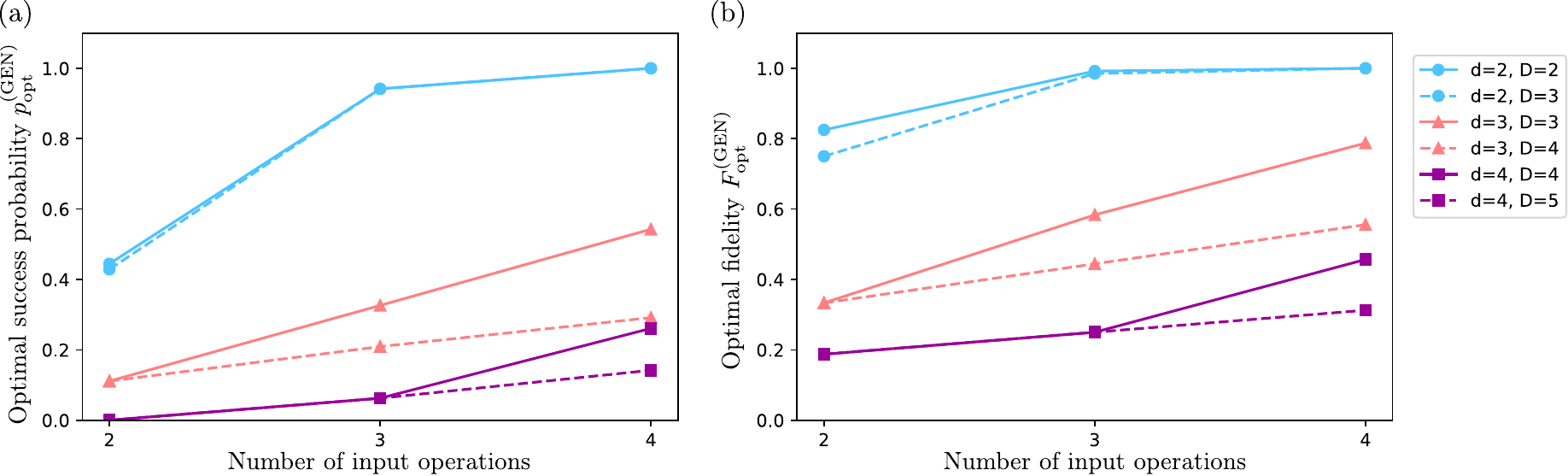}
    \caption{(a) Optimal success probability $p_\mathrm{opt}^{(\mathrm{GEN})}$ and (b) optimal worst-case channel fidelity $F_\mathrm{opt}^{(\mathrm{GEN})}$ of isometry inversion among general protocols including indefinite causal order. Solid lines represent the case $D=d$ (unitary inversion), and dashed lines represent the case $D=d+1$ for $d=2$ ({\color[rgb]{0.30196,0.76862,1} $\bullet$}), $d=3$ ({\color[rgb]{1,0.50196,0.50980} $\blacktriangle$}), and $d=4$ ({\color[rgb]{0.6,0,0.6} $\blacksquare$}).  The optimal values shown here are obtained by numerical calculations, and numerical values are shown in Appendix \ref{appendix_sec:numerical_results}.}
    \label{fig:isometry_inversion_general_optimal}
\end{figure}

\subsubsection{Universal error detection}
The optimal performance of parallel protocol $\alpha_\mathrm{opt}^{(\mathrm{PAR})}(d,D,n)$ is given as follows.

\begin{Thm}
\label{thm:optimal_parallel_error_detection}
    \begin{align}
        \alpha_\mathrm{opt}^{(\mathrm{PAR})}(d,D,n) = {1 \over d+k+1}\left(d+{d-l \over d+k-l}\right) = {d^2 \over n} + O(d^2 n^{-2}),
    \end{align}
    where $k$ and $l$ are given by $n=kd+l$ ($k\in \mathbb{Z}, 0\leq l<d$).
\end{Thm}
\begin{proof}
    See Appendix \ref{appendix_sec:optimal_parallel_error_detection} for the proof.
\end{proof}
From Theorem \ref{thm:optimal_parallel_error_detection} and the result in Ref.~\cite{chabaud2018optimal}, we can show the following property on the scaling of $\alpha_\mathrm{opt}^{(x)}(d,D,n)$ with respect to $n$.

\begin{Cor}
    \begin{align}
        \alpha_\mathrm{opt}^{(x)}(d,D,n)= \Theta(n^{-1}) \quad \forall D>d,  \forall x\in\{\mathrm{PAR}, \mathrm{SEQ}, \mathrm{GEN}\}
    \end{align}
    holds for an arbitrary fixed value of $d$.
\end{Cor}
\begin{proof}
    From Theorem \ref{thm:optimal_parallel_error_detection}, we obtain
    \begin{align}
        \alpha_\mathrm{opt}^{(x)}(d,D,n) \leq \alpha_\mathrm{opt}^{(\mathrm{PAR})}(d,D,n) = O(n^{-1}).
    \end{align}
    The Hilbert space $\CC^D$ can be embedded onto $\isometry{d}{D}$ by identifying $\ket{\psi}\in\CC^D$ with $V = \ketbra{\psi}{0} + \sum_{i=1}^{d-1} \ketbra{\psi_i}{i} \in \isometry{d}{D}$ for a set of orthonormal vectors $\{\ket{\psi_i}\}_{i=1}^{d-1}$ satisfying $\braket{\psi_i}{\psi}=0$ for all $i\in\{1, \ldots, d-1\}$.  Therefore, a universal error detection protocol for $V\in\isometry{d}{D}$ can simulate a programmable projective measurement for $\ket{\psi}\in\CC^D$, which leads to
    \begin{align}
        \alpha_\mathrm{opt}^{(x)}(d,D,n)\geq {1\over n+1} = \Omega(n^{-1}).
    \end{align}
    Thus, we obtain $\alpha_\mathrm{opt}^{(x)}(d,D,n) = \Theta(n^{-1}).$
\end{proof}

\subsubsection{Isometry adjointation}
The optimal scaling of the approximation error of parallel isometry adjointation is given as follows.
\begin{Thm}
\label{thm:optimal_parallel_isometry_adjointation}
    \begin{align}
        \epsilon_\mathrm{opt}^{(\mathrm{PAR})}(d,D,n) = \Theta(d^2 n^{-1}).
    \end{align}
\end{Thm}
\begin{proof}
    First, $\epsilon_\mathrm{opt}^{(\mathrm{PAR})}(d,D,n)\geq \alpha_\mathrm{opt}^{(\mathrm{PAR})}$ holds since any isometry adjointation protocol can be converted to an error detection protocol by discarding the output state of an isometry adjointation protocol.  Thus, we obtain $\epsilon_\mathrm{opt}^{(\mathrm{PAR})}(d,D,n) = \Omega(d^2 n^{-1})$ from Theorem \ref{thm:optimal_parallel_error_detection}.  On the other hand, from the construction shown in Section \ref{subsubsec:parallel_isometry_adjointation_protocol}, we see that $\epsilon_\mathrm{opt}^{(\mathrm{PAR})}(d,D,n) = O(d^2 n^{-1})$ (see Appendix \ref{appendix_sec:optimal_parallel_isometry_adjointation}).  Thus, we obtain $\epsilon_\mathrm{opt}^{(\mathrm{PAR})}(d,D,n) = \Theta(d^2 n^{-1})$.
\end{proof}

\subsection{Numerical results for the optimal protocols}
\label{subsec:sdp}
From numerical calculations, we investigate the optimal protocols for isometry inversion, universal error detection, and isometry adjointation. We represent the protocols implementing isometry inversion, universal error detection, and isometry adjointation by their Choi operators and formulate the optimization of the figure of merits within the possible protocols as SDP.  We utilize a $\U(d) \times \U(D)$ symmetry of the Choi operator (Theorem \ref{thm:unitary_group_symmetry}) to simplify the SDP, which is a similar technique presented in Ref.~\cite{yoshida2023reversing} (see also Ref.~\cite{grinko2024linear}).
See Appendix \ref{appendix_sec:sdp} for the details of the derivation of the SDP.
We calculate the derived SDP in \textsc{MATLAB} \cite{matlab} using the interpreter \textsc{CVX} \cite{cvx, gb08} with the solvers \textsc{SDPT3} \cite{sdpt3,toh1999sdpt3,tutuncu2003solving}, \textsc{SeDuMi} \cite{sedumi} and \textsc{MOSEK} \cite{mosek}.   Group-theoretic calculations to write down the SDP are done with \textsc{SageMath} \cite{sagemath}. 
See Appendix \ref{appendix_sec:numerical_results} for the numerical results. 
All codes are available at Ref.~\cite{github} under the MIT license \cite{mit_license}.

The optimal performances for parallel or sequential protocols are calculated for the case of $D=d+1$ since they do not depend on $D$ (see Corollary \ref{cor:independence}), which is shown by the construction of protocols shown in Figs.~\ref{fig:deterministic_isometry_adjointation_sequential} and \ref{fig:deterministic_isometry_adjointation_parallel}.
Although the same construction is impossible for general protocols including indefinite causal order, we also see that the optimal performances for general protocols do not depend on $D$ as long as $D\geq d+1$ holds by checking the values for $D=d+1, \ldots, d+10$.
However, the maximum success probability (channel fidelity) of probabilistic exact (deterministic) isometry inversion for the case of $D=d$ (unitary inversion) is different from those for the case of $D=d+1$ (isometry inversion), as shown in Tables \ref{tab:comparison_probabilistic_isometry_inversion} and \ref{tab:comparison_deterministic_isometry_inversion} (corresponding to Fig.~\ref{fig:isometry_inversion_general_optimal} in Section \ref{subsec:optimal_isometry_inversion}).
We also see the advantage of indefinite causal order over sequential protocols in isometry inversion and universal error detection, but the advantage disappears in the case of isometry adjointation.

\section{Discussions on the potential applications}
\label{sec:application}
Finally, we discuss the potential applications of this work.
\subsection{Petz recovery map}
Petz recovery map is an almost optimal decoding channel from a noise channel and is widely used in approximate quantum error correction \cite{barnum2002reversing, ng2010simple}, quantum communication \cite{beigi2016decoding}, and information scrambling \cite{hayden2007black, nakayama2023petz, utsumi2024explicit}.
For a quantum channel $\mcN: \mcL(\mcA)\to \mcL(\mcB)$ and a quantum state $\sigma\in\mcL(\mcA)$, the Petz recovery map $\mcP^{\sigma, \mcN}: \mcL(\mcB)\to \mcL(\mcA)$ is defined by \cite{hayden2004structure}
\begin{align}
    \mcP^{\sigma, \mcN}(\cdot)\coloneqq \sigma^{1/2}\mcN^\dagger(\mcN(\sigma)^{-1/2} \cdot \mcN(\sigma)^{-1/2})\sigma^{1/2}.
\end{align}
An implementation of the Petz recovery map is shown in Ref.~\cite{gilyen2022quantum}, which requires
\begin{itemize}
    \item The quantum circuit of the block-encoding unitaries $U^{\sigma}$ and $U^{\mcN(\sigma)}$ of $\sigma$ and $\mcN(\sigma)$, respectively.
    \item The quantum circuit $U^{\mcN}: \mcE'\otimes \mcA \to \mcE\otimes \mcB$ such that $U^{\mcN}\ket{0}_{\mcE'} = V^{\mcN}$ is a Stinespring dilation of $\mcN$, i.e., $\Tr_{\mcE}[V^{\mcN} \cdot V^{\mcN\dagger}] = \mcN(\cdot)$.
\end{itemize}
The implementation shown in Ref.~\cite{gilyen2022quantum} composes the following three maps:
\begin{align}
\label{eq:sigmasigma}
    (\cdot) &\mapsto \mcN(\sigma)^{-1/2} (\cdot) \mcN(\sigma)^{-1/2},\\
\label{eq:Ndagger}
    (\cdot) &\mapsto \mcN^{\dagger}(\cdot),\\
\label{eq:NsigmaNsigma}
    (\cdot) &\mapsto \sigma^{1/2} (\cdot) \sigma^{1/2}.
\end{align}
The maps \eqref{eq:sigmasigma} and \eqref{eq:NsigmaNsigma} are implemented by using the quantum singular value transformation applied on $U^{\sigma}$ and $U^{\mcN(\sigma)}$.
The map \eqref{eq:Ndagger} is implemented by using the following equation:
\begin{align}
    \mcN^\dagger(\cdot) \propto \Tr_{\overline{\mcE}}\left[\bra{0}_{\mcE'} U^{\mcN\dagger} \left(\ketbra{\phi^+}_{\mcE\overline{\mcE}} \otimes \cdot\right) U^{\mcN} \ket{0}_{\mcE'}\right],
\end{align}
or equivalently,
\begin{align}
    \mcN^\dagger(\cdot) \propto \Tr_{\overline{\mcE}}\left[ V^{\mcN\dagger} \left(\ketbra{\phi^+}_{\mcE\overline{\mcE}} \otimes \cdot\right) V^{\mcN}\right].
\end{align}
Thus, using the inverse of $U^{\mcN}$, one can implement the map \eqref{eq:Ndagger}.
However, it is a too strong assumption that the quantum circuit of $U^{\mcN}$ (or the inverse of $U^{\mcN}$) is given in several situations.
For instance, the Petz recovery map can be used for the information recovery from an information scrambling dynamics followed by an erasure error \cite{hayden2007black, nakayama2023petz, utsumi2024explicit}, and $U^{\mcN}$ corresponds to the information scrambling dynamics in this case.
To implement the inverse of $U^{\mcN}$ in this case, it is required to implement the backward time evolution, which is only possible in few-body systems in the current experimental settings \cite{garttner2017measuring, li2017measuring, joshi2020quantum, blok2021quantum}.
One can also utilize $D$-dimensional unitary inversion for $D = \dim (\mcE \otimes \mcB)$, but it requires $\Theta(D^2)$ calls of $U^{\mcN}$ \cite{yoshida2023reversing, chen2024quantum, odake2024analytical}.
If we use isometry adjointation instead, one can implement $V^{\mcN \dagger}$ with approximation error $\epsilon$ by using $O(d^2/\epsilon)$ calls of $V^{\mcN}$ for $d=\dim \mcA$ (see Theorem~\ref{thm:parallel_isometry_adjointation}).
Therefore, for $D\gg d$, our protocol for isometry adjointation works better than unitary inversion.
\subsection{Universal programming of two-outcome projective measurement}
\label{subsec:universal_programming}

A universal programmable quantum processor implements quantum operations specified by a quantum state called the program state \cite{nielsen1997programmable}.
In general, for a set of quantum operations $\mathbb{S}\coloneqq \{\Lambda\}_{\Lambda}$, we define the task to implement a set of the program states $\{\phi_\Lambda\}_{\Lambda\in\mathbb{S}}$ and the quantum channel $\mcD$ such that
\begin{align}
    \mcD(\rho_\mathrm{in} \otimes \phi_\Lambda) = \Lambda(\rho_\mathrm{in})
\end{align}
for any input state $\rho_\mathrm{in}$ and any $\Lambda\in\mathbb{S}$.
For the set of probabilistic operations $\mathbb{S} = \{\mcM = \{\Lambda_a\}\}_{\mcM}$, we replace $\mcD$ with the probabilistic operation $\{\mcD_a\}$ and the task is to implement
\begin{align}
    \mcD_a(\rho_\mathrm{in} \otimes \phi_{\mcM}) = \Lambda_a(\rho_\mathrm{in}).
\end{align}
It is shown that an exact universal programmable quantum processor is impossible for unitary operations \cite{nielsen1997programmable} and for projective measurements \cite{duvsek2002quantum, fiuravsek2002universal}, and approximate or probabilistic implementations are investigated (e.g., Refs.~\cite{nielsen1997programmable, duvsek2002quantum, fiuravsek2002universal, fiuravsek2004probabilistic, dariano2005efficient, lewandowska2022storage} for projective or POVM measurements).
In the approximate implementation, the tradeoff between the size of the program state called the program cost, and the approximation error is investigated (e.g., Refs.~\cite{dariano2005efficient, perez2006optimality} for POVM measurements).
The program cost is defined by
\begin{align}
    C_\mathbb{S}\coloneqq \log \left[\dim \bigoplus_{\Lambda\in\mathbb{S}} \mathrm{Supp} (\phi_\Lambda)\right],
\end{align}
where $\mathrm{Supp} (\phi_\Lambda)$ is the support of $\phi_\Lambda$ defined by $\mathrm{Supp} (\phi_\Lambda)\coloneqq \mathrm{span}\{\ket{\psi} \mid \phi_\Lambda \ket{\psi}\neq 0\}$.
For the universal programming of POVM measurements, we consider the worst-case operational distance as the approximation error defined by \cite{dariano2005efficient, perez2006optimality}
\begin{align}
    \delta \coloneqq \inf_{\mcM\in\mathbb{S}} D_\mathrm{op}(\mcM', \mcM),
\end{align}
where $\mcM' = \{P'_a\}$ is the POVM defined by $\Tr[P'_a \cdot]\coloneqq \mcD_a(\cdot\otimes \phi_{\mcM})$ and $D_\mathrm{op}$ is the operational distance defined in Eq.~\eqref{eq:def_operational_distance}.

The parallel universal error detection protocol shown in Fig.~\ref{fig:deterministic_isometry_adjointation_parallel}~(d) can implement universal programming of rank-$d$ projective measurement, where $\mathbb{S}$ is given by
\begin{align}
    \mathbb{S} = \mathbb{S}_\mathrm{PVM}^{(d, D)}\coloneqq \{\mcM\coloneqq \{\Pi, \1_D-\Pi\} \mid \Pi \text{ is an orthogonal projector with rank } d\},
\end{align}
by using the program state $\phi_\mcM = \map{V}_\mathrm{in}^{\otimes n}(\phi')$ for $V_\mathrm{in}\in\isometry{d}{D}$ such that $\Pi = \Pi_{\Im V_\mathrm{in}}$.
Since any two-outcome POVM measurement can be realized as a probabilistic mixture of two-outcome projective measurements \cite{davies1976quantum, masanes2005extremal}, we can similarly implement universal programming of two-outcome POVM measurement, where $\mathbb{S}$ is given by
\begin{align}
    \mathbb{S} = \mathbb{S}_\mathrm{POVM}^{(D)}\coloneqq \{\mcM\coloneqq \{\Pi, \1_D-\Pi\} \mid 0\leq \Pi\leq \1_D\}.
\end{align}
We show the following Corollary from Theorem~\ref{thm:optimal_parallel_error_detection}:
\begin{Cor}
\label{cor:universal_programming}
    There exists a protocol for universal programming of rank-$d$ projective measurement with the approximation error $\delta$ and the program cost given by
    \begin{align}
    \label{eq:universal_programming_PVM}
        C_{\mathbb{S}_\mathrm{PVM}^{(d, D)}}\leq d(D-d)\log {\Theta(Dd)\over \delta}.
    \end{align}
    Similarly, there exists a protocol for universal programming of two-outcome POVM measurement with the approximation error $\delta$ and the program cost given by
    \begin{align}
        C_{\mathbb{S}_\mathrm{POVM}^{(D)}}\leq
        \begin{cases}
            {D^2\over 4} \log {\Theta(D^2)\over \delta} & (D \mathrm{\;is\;even})\\
            {D^2-1\over 4} \log {\Theta(D^2)\over \delta} & (D \mathrm{\;is\;odd})\\
        \end{cases}.
    \end{align}
\end{Cor}
\begin{proof}
    See Appendix~\ref{appendix_sec:universal_programming} for the proof.
\end{proof}
Note that Ref.~\cite{dariano2005efficient} proposes universal programming of $D$-dimensional $D$-outcome projective measurements, where $\mathbb{S}$ is given by
\begin{align}
    \mathbb{S} = \mathbb{S}_\mathrm{PVM}^{(D)}\coloneqq \{\mcM\coloneqq \{\ketbra{\psi_1}, \ldots, \ketbra{\psi_D}\} \mid \braket{\psi_i}{\psi_j} = \delta_{ij} \quad \forall i,j\},
\end{align}
with the program cost $C_{\mathbb{S}_\mathrm{PVM}^{(D)}} \leq D(D-1)\log {\Theta(D^{5\over 2})\over \delta}$.
We obtain a tighter bound of the program cost in Eq.~\eqref{eq:universal_programming_PVM} by restricting the measurement to two-outcome projective measurements.
We conjecture the program cost~\eqref{eq:universal_programming_PVM} to be optimal from the following argument based on the parameter counting.
In Ref.~\cite{yang2020optimal}, the optimal program cost of unitary operations with $\nu$ real parameters is conjecture to be
\begin{align}
    C = {\nu \over 2} \log {C_{\nu, d}\over \epsilon},
\end{align}
where $\epsilon$ is the diamond-norm approximation error, and $C_{\nu,d}$ is a parameter independent of $\epsilon$.
By assuming this conjecture holds for the projective measurements by substituting $\epsilon$ with $\delta$, we can show that the program cost~\eqref{eq:universal_programming_PVM} is optimal.
The rank-$d$ projective measurement can be uniquely specified by a $d$-dimensional subspace of $\CC^d$, which can be specified by $2d(D-d)$ real parameters\footnote{The $d$-dimensional subspace of $\CC^d$ can be represented as $\mathrm{span}\{\ket{\psi_1}, \ldots, \ket{\psi_d}\}$ by using an orthonormal basis $\{\ket{\psi_1}, \ldots, \ket{\psi_d}\}\subset \CC^D$, i.e., the set satisfying $\braket{\psi_i}{\psi_j} = 1$ for all $i,j$.
This representation is unique up to the following action of $\U(d)$: $\ket{\psi_i}\mapsto \sum_{j} u_{ij} \ket{\psi_j}$ for $(u_{ij})_{i,j=1}^{d} \in \U(d)$.
The number of real parameters to specify an orthonormal basis $\{\ket{\psi_1}, \ldots, \ket{\psi_d}\}\subset \CC^D$ is given by $2Dd-d^2$, and the number of real parameters to specify an element of $\U(d)$ is given by $d^2$.
Therefore, the number of real parameters to specify a $d$-dimensional subspace of $\CC^d$ can be given as $2Dd-d^2-d^2 = 2d(D-d)$.}.
Therefore, by assuming the conjecture, we can show that the optimal program cost of rank-$d$ projective measurement is given by
\begin{align}
    C_{\mathbb{S}_\mathrm{PVM}^{(d,D)}} = d(D-d) \log {C_{d,D} \over \delta},
\end{align}
where $C_{d,D}$ is a parameter independent of $\delta$.

\section{Conclusion}
\label{sec:conclusion}
In this work, we investigate the universal transformation of isometry operations to explore the possible transformation in higher-order quantum computation. We define the task called isometry adjointation, which is to transform the isometry operation into its adjoint operation.
The isometry adjointation protocol is constructed by converting the unitary inversion protocol using the probabilistic quantum comb $\{\supermap{T}_I, \supermap{T}_O\}$. 
Using the idea of composition of quantum combs, the problem of designing isometry adjointation protocol reduces to designing a unitary inversion protocol, which is extensively studied in previous works \cite{chiribella2016optimal, sardharwalla2016universal, quintino2019probabilistic, quintino2019reversing, dong2021success, quintino2022deterministic, yoshida2023reversing, navascues2018resetting, trillo2020translating, trillo2023universal}.
In special cases, the isometry adjointation reduces to unitary inversion (transformation of unitary operations) and programmable projective measurement (transformation of pure states).
Due to such reducibility, isometry adjointation is a useful task to understand the difference between higher-order quantum computation and ``lower-order'' quantum computation (i.e., transformation of states), which are exhibited in several examples such as distinguishability \cite{acin2001statistical} and superreplication \cite{chiribella2013quantum2, dur2015deterministic, chiribella2015universal}, under a unified formulation.
We also construct isometry inversion and universal error detection protocols by discarding measurement outcome and output quantum state, respectively.

We show that our construction gives the optimal performances among all parallel or sequential protocols, which implies that the optimal performances among parallel or sequential protocols do not depend on the output dimension $D$ of the isometry operation.
We analyze the optimal performances of isometry adjointation and show that the optimal approximation error $\epsilon$ of the parallel isometry adjointation protocol is given by $\epsilon = \Theta(d^2 n^{-1})$, where $d$ is the input dimension of the isometry operation, and $n$ is the number of calls of the input operation.

To investigate the general protocols including indefinite causal order, we also provide the numerical results for the optimal performances of isometry adjointation, isometry inversion, and universal error detection using the semidefinite programming combined with the unitary group symmetry.
The optimal performances of probabilistic exact (deterministic) unitary inversion ($D=d$) are different from those of isometry inversion ($D\geq d+1$), which is compatible with the impossibility of the composition of a general supermap with a quantum comb.
However, numerical results also show that the optimal performances of isometry adjointation and universal error detection using the general protocol do not depend on $D$.
We also see the advantage of indefinite causal order protocols over sequential protocols in isometry inversion and universal error detection, but the advantage disappears in isometry adjointation.
We show the potential application of this work for the implementation of the Petz recovery map and universal programming of two-outcome projective measurement.

This work identifies that unitary inversion protocol can be converted to isometry adjointation with $D$-independent approximation error.
However, Ref.~\cite{yoshida2023universal} shows that isometry complex conjugation is impossible while unitary complex conjugation is possible \cite{miyazaki2019complex}, and conjectures that the optimal success probability of isometry transposition depends on $D$.
From this observation, it is not trivial whether a given universal transformation of unitary operations can be converted to the corresponding transformation of isometry operations.
We leave it an open problem to investigate whether such a conversion is possible for other universal transformations of unitary operations.

\begin{acknowledgments}
    We acknowledge M.~Studzi\'{n}ski, T.~M\l{}ynik, M.~T.~Quintino, M.~Ozols, D.~Grinko, P.~Taranto, and M.~Koashi for valuable discussions.
    We also thank T.~J. Volkoff for pointing out a typographical error in the previous manuscript.
    This work was supported by MEXT Quantum Leap Flagship Program (MEXT QLEAP) JPMXS0118069605, JPMXS0120351339, Japan Society for the Promotion of Science (JSPS) KAKENHI Grants No.~18K13467, 21H03394, 23KJ0734, FoPM, WINGS Program, the University of Tokyo, DAIKIN Fellowship Program, the University of Tokyo, and IBM-UTokyo lab.
    This work was also supported by JST Moonshot R\&D Grant Number JPMJMS226C.
\end{acknowledgments}

\appendix

\section{Construction of isometry adjointation protocols}

\subsection{Evaluation of the diamond-norm error of isometry adjointation}
\label{appendix_subsec:evaluation_diamond-norm}
We evaluate the diamond-norm error of isometry adjointation protocol satisfying Eq.~\eqref{eq:isometry_adjointation_detection_decode_error}.
First, $\supermap{C}[\map{V}_\mathrm{in}^{\otimes n}] - \map{V}_\mathrm{adjoint}$ decomposes into two completely positive maps $\Phi_1$ and $\Phi_2$ as
\begin{align}
    \supermap{C}[\map{V}_\mathrm{in}^{\otimes n}] - \map{V}_\mathrm{adjoint} &= \Phi_1 + \Phi_2,\\
    \Phi_1(\rho_\mathrm{in})&\coloneqq -{d^2\eta \over d^2-1} [\map{V}_\mathrm{in}^\dagger (\rho_\mathrm{in}) -\Tr(\Pi_{\Im V_\mathrm{in}}\rho_\mathrm{in})]\otimes \ketbra{0},\\
    \Phi_2(\rho_\mathrm{in})&\coloneqq \alpha \Tr[(\1_D-\Pi_{\Im V_\mathrm{in}})\rho_\mathrm{in}] \left({\1_d\over d}\otimes \ketbra{0} - {\1_d \over d} \otimes \ketbra{1}\right),
\end{align}
whose kernels do not intersect, i.e., $\ker \Phi_1 \cap \ker \Phi_2 = \emptyset$. Since the diamond norm $\|\Phi_1 + \Phi_2\|_\diamond$ is defined by
\begin{align}
    \|\Phi_1 + \Phi_2\|_\diamond\coloneqq \sup_{\psi\in \mcL(\mcP \otimes \mcA), \|\psi\|_1\leq 1} \|(\Phi_1 + \Phi_2)\otimes \1_\mcA(\psi)\|_1
\end{align}
using an auxiliary Hilbert space $\mcA$ and the trace norm $\|\cdot\|_1$, it is given by
\begin{align}
    &\|\Phi_1 + \Phi_2\|_\diamond\nonumber\\
    &=\sup_{\psi_i\in \ker\Phi_i \otimes \mcL(\mcA), \|\psi_1+\psi_2\|_1\leq 1} \|(\Phi_1 + \Phi_2)\otimes \1_\mcA(\psi_1+\psi_2)\|_1\\
    &= \max_{0\leq x\leq 1} [x\sup_{\substack{\psi_1\in \ker \Phi_1 \otimes \mcL(\mcA)\\ \|\psi_1\|_1\leq 1}} \|\Phi_1\otimes \1_\mcA(\psi_1)\|_1 + (1-x)\sup_{\substack{\psi_2\in \ker \Phi_2 \otimes \mcL(\mcA) \\ \|\psi_2\|_1\leq 1}} \|\Phi_1\otimes \1_\mcA(\psi_2)\|_1]\\
    &= \max\{\|\Phi_1\|_\diamond, \|\Phi_2\|_\diamond\}.
\end{align}
The diamond norm $\|\Phi_1\|_\diamond$ is given by
\begin{align}
    \|\Phi_1\|_\diamond
    &= \|\Phi_1\circ \map{V}_\mathrm{in}\|_\diamond\\
    &= \|\1_d - \map{D}_q\|_\diamond\\
    &= 2\eta,
\end{align}
where $\map{D}_q$ is the depolarizing channel
\begin{align}
    \map{D}_q(\rho)\coloneqq (1-q)\rho + q {\1_d\over d} \Tr(\rho),
\end{align}
$q$ is given by $q\coloneqq {d^2\over d^2-1}\eta$, and we utilize the fact that \cite{matsumoto2012input}
\begin{align}
    \|\1_d - \map{D}_q\|_\diamond = {2(d^2-1) \over d^2}q.
\end{align}
The diamond norm $\|\Phi_2\|_\diamond$ is given by
\begin{align}
    &\|\Phi_2\|_\diamond \nonumber\\
    &= \sup_{\psi\in \mcL(\mcP \otimes \mcA), \|\psi\|_1\leq 1} \|\alpha \Tr_{\mcP}[(\1_D-\Pi_{\Im V_\mathrm{in}}) \psi] \otimes ({\1_d\over d}\otimes \ketbra{0} - {\1_d \over d} \otimes \ketbra{1})\|_1 \\
    &= 2\alpha.
\end{align}
Thus, we obtain
\begin{align}
    \epsilon &\coloneqq {1\over 2} \sup_{V_\mathrm{in}\in\isometry{d}{D}} \|\supermap{C}[\map{V}_\mathrm{in}^{\otimes n}] - \map{V}_\mathrm{adjoint}\|_\diamond\\
    &= \max \{ \alpha, \eta\}.
\end{align}

\subsection{Clebsch-Gordan transform and dual Clebsch-Gordan transform}
\label{appendix_subsec:mixed_schur_transform}
In the next subsection, we construct the probabilistic quantum comb $\{\supermap{T}_I, \supermap{T}_O\}$ used for the construction of isometry adjointation protocols.
To this end, we first review the construction of the quantum Schur transform based on the concatenation of Clebsch-Gordan (CG) transforms.
Then, we show the dual CG transform, which was introduced to construct the mixed Schur transform.
Finally, we show a lemma on the dual CG transform, which will be used for the construction of the probabilistic quantum comb $\{\supermap{T}_I, \supermap{T}_O\}$ in the next subsection.

The quantum Schur transform can be implemented by using a series of CG transforms \cite{bacon2006efficient, bacon2007quantum}.
For the irreducible representation $\mcU_\lambda^{(d)}$ for $\lambda\in\young{d}{n-1}$, the tensor product space $\mcU_\lambda^{(d)} \otimes \CC^d$ can be decomposed as follows:
\begin{align}
    \mcU_\lambda^{(d)} \otimes \CC^d = \bigoplus_{\mu\in \lambda+_d \square} \mcU_\mu^{(d)},
\end{align}
where the vector $\ket{\lambda,u}_{\mcU_\lambda^{(d)}} \otimes \ket{j}_{\CC^d} \in \mcU_\lambda^{(d)} \otimes \CC^d$ is given by
\begin{align}
    \ket{\lambda,u}_{\mcU_\lambda^{(d)}} \otimes \ket{j}_{\CC^d} = \bigoplus_{\mu\in\lambda+_d \square} \sum_{u'} C^{\lambda, \mu}_{u,j,u'} \ket{\mu,u'}_{\mcU_\mu^{(d)}}
\end{align}
using the coefficients $C^{\lambda, \mu}_{u,j,u'}\in\CC$ called the CG coefficients.
The CG transform $\mathrm{CG}^+_{n,d}: \mcT_{n-1}\otimes \mcU_{n-1}\otimes \CC^d \to \mcT_{n-1}\otimes \mcT_n \otimes \mcU_n$ is defined by
\begin{align}
\label{eq:cg_transform}
    \mathrm{CG}^+_{n,d}(\ket{\lambda}_{\mcT_{n-1}}\otimes \ket{u}_{\mcU_{n-1}} \otimes \ket{j}_{\CC^d}) \coloneqq \sum_{\mu\in\lambda+_d\square} \sum_{u'} C^{\lambda, \mu}_{u,j,u'} \ket{\lambda}_{\mcT_{n-1}} \otimes \ket{\mu}_{\mcT_n}\otimes \ket{u'}_{\mcU_n},
\end{align}
which can be represented as a controlled isometry as
\begin{align}
\label{eq:cg_transform_ctrl}
    \mathrm{CG}^+_{n,d} = \sum_{\lambda\in\young{d}{n-1}} \ketbra{\lambda}_{\mcT_{n-1}} \otimes \mathrm{CG}^+_{\lambda, d},
\end{align}
using $\mathrm{CG}^+_{\lambda, d}: \mcU_{n-1}\otimes \CC^d \to \mcT_n \otimes \mcU_n$.
By using the CG transform, we can construct the quantum Schur transform as follows (see Fig.~\ref{fig:implementation_quantum_schur_transform}):
\begin{align}
\label{eq:quantum_schur_transform}
    &\mathrm{Sch}_{n,d}(\ket{j_1}_{\CC^d}\otimes \cdots \otimes \ket{j_n}_{\CC^d})\nonumber\\
    &= (\1_{\mcT_2\cdots \mcT_{n-1}}\otimes \mathrm{CG}^+_{n,d})  [\cdots [(\1_{\mcT_2} \otimes \mathrm{CG}^+_{3,d})[\mathrm{CG}^+_{2,d}(\ket{\square}\otimes \ket{j_1}\otimes \ket{j_2})\otimes \ket{j_3}] \cdots ] \otimes \ket{j_n}]\\
\label{eq:schur_cg_coefficients}
    &= \sum_{\mu\in\young{d}{n}} \sum_{u,i} C^{\mu, u, i}_{j_1\cdots j_n} \ket{\mu_i^{(1)}\cdots \mu_i^{(n)}}_{\mcT_1\cdots \mcT_n} \otimes \ket{u}_{\mcU_n},
\end{align}
using coefficients $C^{\mu, u, i}_{j_1\cdots j_n}\in \CC$,
where the sequence of Young diagrams $\mu_i^{(1)}\cdots \mu_i^{(n)}$ satisfying $\mu_i^{(j+1)} \in \mu_i^{(j)}+_d\square$ for $j\in\{1, \ldots, n-1\}$ is related to the standard Young tableau $s_i^\mu$ such that $\mu_i^{(j)}$ is the Young diagram obtained by removing the boxes \fbox{$j+1$}, \ldots, \fbox{$n$} from the Young diagram $\mu$ and $\mu_i^{(n)} = \mu$.
For instance, the standard tableau
\begin{align}
    \ytableaushort{1 2 3, 4}
\end{align}
is related to the sequence of the Young diagrams
\begin{align}
    \mu_i^{(1)} = \ydiagram[]{1}, \quad \mu_i^{(2)} = \ydiagram[]{2}, \quad \mu_i^{(3)} = \ydiagram[]{3}, \quad \mu_i^{(4)} = \ydiagram[]{3,1}.
\end{align}
In this way, the Schur basis is represented as follows:
\begin{align}
    \mathrm{Sch}_{n,d}(\ket{\mu, u}_{\mcU_\mu^{(d)}}\otimes \ket{\mu, i}_{\mcS_\mu}) = \ket{u}_{\mcU_n} \otimes \ket{\mu_i^{(1)}\cdots \mu_i^{(n)}}_{\mcT_1\cdots \mcT_n},
\end{align}

\begin{figure}
    \centering
    \includegraphics{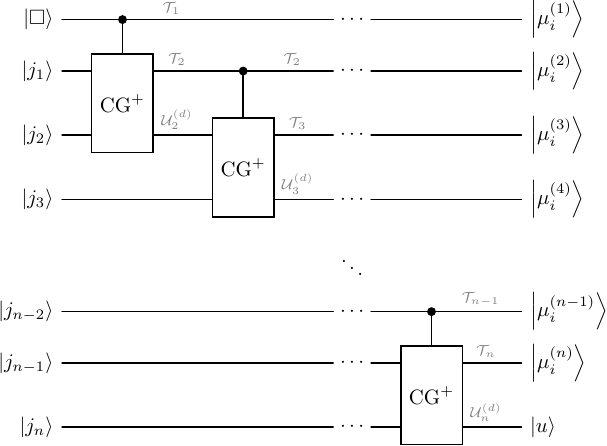}
    \caption{Implementation of the quantum Schur transform using the CG transforms [see Eq.~\eqref{eq:quantum_schur_transform}]. The subscripts of $\mathrm{CG}^+$ are omitted for simplicity.}
    \label{fig:implementation_quantum_schur_transform}
\end{figure}

We then introduce the mixed Schur transform and the dual CG transform \cite{nguyen2023mixed, grinko2023gelfand, fei2023efficient}.
We refer Ref.~\cite{grinko2025thesis} for the review of the mixed Schur-Weyl duality.
We first consider the following representation of $\U(d)$:
\begin{align}
    \U(d) \ni U \mapsto U^{\otimes n}\otimes U^{*\otimes m} \in \mcL(\CC^d)^{\otimes (n+m)}.
\end{align}
Since $[P_\sigma, U^{\otimes n} \otimes U^{\otimes m}] = 0$ for all $\sigma\in\mfS_{n+m}, U\in\U(d)$, we obtain the following commutation relation:
\begin{align}
    [P_\sigma^{T_{n+1\cdots n+m}}, U^{\otimes n} \otimes U^{*\otimes m}] = 0 \quad \forall \sigma\in\mfS_{n+m}, U\in\U(d),
\end{align}
where $T_{n+1\cdots n+m}$ represents the partial trace over $n+1, \ldots, n+m$-th systems.
More strongly, we can show the following relation called the mixed Schur-Weyl duality:
\begin{align}
    [X, U^{\otimes n}\otimes U^{*\otimes m}] = 0 \quad \forall U\in\U(d) &\Rightarrow X\in\mathrm{span}\{P_\sigma^{T_{n+1\cdots n+m}} \mid \sigma\in\mfS_{n+m}\},\\
    [X, P_\sigma^{T_{n+1\cdots n+m}}] = 0 \quad \forall \sigma\in\mfS_{n+m} &\Rightarrow X\in\mathrm{span}\{U^{\otimes n}\otimes U^{*\otimes m} \mid U\in\U(d)\}.
\end{align}
From this relationship, we can introduce the decomposition of the Hilbert space given by
\begin{align}
    (\CC^d)^{\otimes (n+m)} &= \bigoplus_{(\lambda_l, \lambda_r)\in\myoung{d}{n}{m}} \mcU_{(\lambda_l, \lambda_r)}^{(d)}\otimes \mcS_{(\lambda_l, \lambda_r)},\\
    U^{\otimes n}\otimes U^{*\otimes m} &= \bigoplus_{(\lambda_l, \lambda_r)\in\myoung{d}{n}{m}} U_{(\lambda_l, \lambda_r)} \otimes \1_{\mcS_{(\lambda_l, \lambda_r)}},\\
    P_\sigma^{T_{n+1\cdots n+m}} &= \bigoplus_{(\lambda_l, \lambda_r)\in\myoung{d}{n}{m}} \1_{\mcU_{(\lambda_l, \lambda_r)}^{(d)}} \otimes \sigma_{(\lambda_l, \lambda_r)},
\end{align}
where $\myoung{d}{n}{m}$ is the set of mixed Young diagrams given by
\begin{align}
    \myoung{d}{n}{m}\coloneqq \{(\lambda_l, \lambda_r) \mid \lambda_l\in\young{d}{k}, \lambda_r\in\young{d}{n-m+k}, k\in\{0,\ldots, m\}\},
\end{align}
and $U_{(\lambda_l, \lambda_r)}\in\mcL(\mcU_{(\lambda_l, \lambda_r)}^{(d)})$ and $\sigma_{(\lambda_l, \lambda_r)}\in\mcL(\mcS_{(\lambda_l, \lambda_r)})$.
The Hilbert space $\mcU_{(\lambda_l, \lambda_r)}^{(d)}$ for $(\lambda_l, \lambda_r)\in\myoung{d}{n}{m}$ is isomorphic to $\mcU_{\hat{\lambda}}^{(d)}$ for $\hat{\lambda}=(\hat{\lambda}^{(1)},\ldots, \hat{\lambda}^{(d)})$ defined by
\begin{align}
    \hat{\lambda}_{i} \coloneqq (\lambda_l)_{i}+(\lambda_r)_{1}-(\lambda_r)_{d+1-i},
\end{align}
where $\mu_{i}$ represents the number of boxes in the $i$-th row in the Young diagram $\mu$.
For instance, for $d=3$ and $\lambda_l, \lambda_r$ given by
\begin{align}
    \lambda_l = \ydiagram[]{3,2,1}, \quad \lambda_r = \ydiagram[]{2,1},
\end{align}
$\hat{\lambda}$ is given by
\begin{align}
    \hat{\lambda} = \ydiagram[]{5,3,1}.
\end{align}
Therefore, we can use the Gelfand-Zetlin basis $\{\ket{\hat{\lambda}, u}\}_u$ of $\mcU_{\hat{\lambda}}^{(d)}$ as the basis of $\mcU_{(\lambda_l, \lambda_r)}^{(d)}$.
We can also introduce the Gelfand-Zetlin basis $\{\ket{(\lambda_l, \lambda_r), i}\}_i$ of $\mcS_{(\lambda_l, \lambda_r)}$, which is labelled by the mixed standard tableau $s_i$, which is a sequence of mixed Young diagrams $(\lambda_l^{(1)}, \lambda_r^{(1)}) \in \myoung{d}{1}{0}, \ldots, (\lambda_l^{(n)}, \lambda_r^{(n)}) \in \myoung{d}{n}{0}, \ldots, (\lambda_l^{(n+m)}, \lambda_r^{(n+m)}) \in \myoung{d}{n}{m}$ such that
\begin{align}
    (\lambda_l^{(i+1)}, \lambda_r^{(i+1)})\in (\lambda_l^{(i)}, \lambda_r^{(i)})+_d\square \quad &\forall i\in \{1,\ldots,n-1\},\\
    (\lambda_l^{(i+1)}, \lambda_r^{(i+1)})\in (\lambda_l^{(i)}, \lambda_r^{(i)})-_d\square \quad &\forall i\in\{n, \ldots, n+m-1\},
\end{align}
where $(\lambda_l, \lambda_r)\pm_d\square$ are the sets defined by
\begin{align}
    (\lambda_l, \lambda_r)+_d\square &\coloneqq \{(\mu_l, \lambda_r) \mid \mu_l\in \lambda_l+_d\square\},\\
    (\lambda_l, \lambda_r)-_d\square &\coloneqq \{(\mu_l, \lambda_r) \mid \mu_l\in \lambda_l-\square\} \cup \{(\lambda_l, \mu_r) \mid \mu_r\in \lambda_r+_d\square\}.
\end{align}
Using the basis of $\mcU_{(\lambda_l,\lambda_r)}^{(d)}$ and $\mcS_{(\lambda_l,\lambda_r)}$, we define the mixed Schur basis of $(\CC^d)^{\otimes (n+m)}$ defined by
\begin{align}
    \ket{(\lambda_l,\lambda_r), u, i}\coloneqq \ket{\hat{\lambda}, u}_{\mcU_{(\lambda_l,\lambda_r)}^{(d)}} \otimes \ket{(\lambda_l,\lambda_r), i}_{\mcS_{(\lambda_l,\lambda_r)}}.
\end{align}
We define the mixed Schur transform $\mathrm{Sch}_{n,m}^{(d)}$ by the isometry operator from the computational basis to the mixed Schur basis as
\begin{align}
\label{eq:mixed_schur_transform}
    \mathrm{Sch}_{n,m}^{(d)}(\ket{(\lambda_l,\lambda_r), u, i}) \coloneqq \ket{u}_{\mcU_{n,m}^{(d)}} \otimes \ket{(\lambda_l^{(1)},\lambda_r^{(1)})\cdots (\lambda_l^{(n+m)},\lambda_r^{(n+m)})}_{\mcT_{1,0}\cdots\mcT_{n,0}\cdots\mcT_{n,m}},
\end{align}
where $\mcU_{n,m}^{(d)}$ is the Hilbert space storing the Gelfand-Zetlin basis and $\mcT_{i,j}$ for $(i,j)\in\{(1,0),\ldots,(n,0),\ldots,(n,m)\}$ are the Hilbert spaces stoing mixed Young diagrams.
The mixed Schur transform can be decomposed into the CG transforms and the dual CG transforms.
For the irreducible representation $\hat{\lambda}$ corresponding to $(\lambda_l, \lambda_r) \in \myoung{d}{n}{m}$ and the dual representation $\overline{\square}$ corresponding to $\U(d)\ni U \mapsto U^* \in \mcL(\CC^d)$, the tensor product of the representation $\hat{\lambda}$ and $\overline{\square}$ is decomposed as
\begin{align}
    \mcU_{\hat{\lambda}}^{(d)}\otimes \mcU_{\overline{\square}}^{(d)} = \bigoplus_{(\mu_l, \mu_r)\in(\lambda_l, \lambda_r)-_d\square} \mcU_{\hat{\mu}}^{(d)},
\end{align}
where the vector $\ket{\hat{\lambda}, u}\otimes \ket{j}_{\CC^d} \in \mcU_{\hat{\lambda}}^{(d)}\otimes \mcU_{\overline{\square}}^{(d)}$ is given by
\begin{align}
    \ket{\hat{\lambda}, u}\otimes \ket{j}_{\CC^d} = \bigoplus_{(\mu_l, \mu_r)\in(\lambda_l, \lambda_r)-_d\square} \sum_{u'} \tilde{C}^{\lambda_l, \lambda_r, \mu_l, \mu_r}_{u,j,u'} \ket{\hat{\mu},u'}
\end{align}
using coefficients $\tilde{C}^{\lambda_l, \lambda_r, \mu_l, \mu_r}_{u,j,u'}\in \CC$ called dual CG coefficients.
The dual CG transform is defined by
\begin{align}
\label{eq:dual_cg_transform}
    \mathrm{CG}^-_{n,m,d}(\ket{\lambda_l, \lambda_r} \otimes \ket{u} \otimes \ket{j}_{\CC^d}) = \sum_{(\mu_l,\mu_r)\in(\lambda_l,\lambda_r)-_d\square} \tilde{C}^{\lambda_l, \lambda_r, \mu_l, \mu_r}_{u,j,u'}  \ket{\lambda_l, \lambda_r} \otimes \ket{\mu_l, \mu_r} \otimes \ket{u'},
\end{align}
which can be represented as a controlled isometry as
\begin{align}
\label{eq:dual_cg_transform_ctrl}
    \mathrm{CG}^-_{n,m,d} = \sum_{(\lambda_l, \lambda_r)\in \myoung{d}{n}{m-1}} \ketbra{\lambda_l, \lambda_r}_{\mcT_{n,m-1}}\otimes \mathrm{CG}^-_{\lambda_l, \lambda_r, d}
\end{align}
using $\mathrm{CG}^-_{\lambda_l, \lambda_r, d}: \mcU_{n,m-1}^{(d)} \otimes \CC^d \to \mcT_{n,m}\otimes \mcU_{n,m}^{(d)}$.
By using the CG transform and the dual CG transform, we can construct the mixed Schur transform as follows:
\begin{align}
\label{eq:mixed_schur_transform}
    &\mathrm{Sch}_{n,m,d}(\ket{j_1}_{\CC^d}\otimes \cdots \otimes \ket{j_{n+m}}_{\CC^d})\nonumber\\
    &= (\1_{\mcT_{1,0}\cdots\mcT_{n,0}\cdots \mcT_{n,m-1}}\otimes \mathrm{CG}^-_{n,m,d})[\cdots[(\1_{\mcT_{1,0}\cdots\mcT_{n,0}}\otimes \mathrm{CG}^-_{n,1,d})\nonumber\\
    &\hspace{12pt} \times[\mathrm{Sch}_{n,d}(\ket{j_1}\otimes \cdots \otimes \ket{j_n})\otimes \ket{j_{n+1}}]\cdots ] \otimes \ket{j_{n+m}}],
\end{align}
where we identify the quantum Schur transform $\mathrm{Sch}_{n,d}$ as the mixed Schur transform $\mathrm{Sch}_{n,0,d}$ by identifying a Young diagram $\lambda$ as a mixed Young diagram $(\lambda, \emptyset)$.

The dual CG transform satisfies the following Lemma, which will be used in the construction of the probabilistic quantum comb $\{\supermap{T}_I, \supermap{T}_O\}$ in the next subsection.
\begin{Lem}
\label{lem:dual_schur_transform}
    For any $n> m$, $\lambda\in\young{m}{d}$ and $\mu\in\lambda+_d\square$,
    \begin{align}
        [\1_d^{\otimes m} \otimes (\mathrm{CG}^-_{n,n-m,d})^\dagger](\ket{\phi^+_{\lambda, a}} \otimes \ket{\mu,\emptyset}) = \ket{\phi^+_{\mu, a^\lambda_\mu}}
    \end{align}
    holds, where we define $\ket{\phi^+_{\lambda, a}}\in (\CC^d)^{\otimes m} \otimes \mcU_{m}^{(d)}\otimes \mcT_{n,n-m}$ and $\ket{\phi^+_{\mu, a^{\lambda}_{\mu}}}\in (\CC^d)^{\otimes m+1} \otimes \mcU_{m+1}^{(d)}\otimes \mcT_{n,n-m-1}$ by
    \begin{align}
        \ket{\phi^+_{\lambda}}&\coloneqq \sqrt{1\over m_\lambda^{(d)}} \sum_{u}  (\ket{\lambda,u}_{\mcU_\lambda^{(d)}} \otimes \ket{\lambda,a}_{\mcS_\lambda}) \otimes \ket{u}_{\mcU_{m}^{(d)}} \otimes \ket{\lambda,\emptyset}_{\mcT_{n,n-m}},\\
        \ket{\phi^+_{\mu}}&\coloneqq \sqrt{1\over m_\mu^{(d)}} \sum_{u'}  (\ket{\mu,u'}_{\mcU_\mu^{(d)}} \otimes \ket{\mu,a^{\lambda}_{\mu}}_{\mcS_\mu}) \otimes \ket{u'}_{\mcU_{m+1}^{(d)}} \otimes \ket{\mu,\emptyset}_{\mcT_{n,n-m-1}}.
    \end{align}
\end{Lem}
\begin{proof}
    Since the dual CG coefficient satisfies [see Eq.~(10), p.289 of Ref.~\cite{klimyk1995representations}]
    \begin{align}
        \tilde{C}^{\mu, \emptyset, \lambda, \emptyset}_{u',j,u} = \sqrt{m_\lambda^{(d)}\over m_\mu^{(d)}} \left(C^{\lambda, \mu}_{u,j,u'}\right)^*,
    \end{align}
    we can show that
    \begin{align}
        &[\1_d^{\otimes m} \otimes (\mathrm{CG}^-_{n,n-m,d})^\dagger](\ket{\phi^+_{\lambda, a}}\nonumber\\
        &= \sqrt{1\over m_\lambda^{(d)}} \sum_{u}  (\ket{\lambda,u}_{\mcU_\lambda^{(d)}} \otimes \ket{\lambda,a}_{\mcS_\lambda}) \otimes (\mathrm{CG}^-_{n,n-m,d})^\dagger(\ket{\mu, \emptyset}_{\mcT_{n,n-m-1}} \otimes \ket{\lambda,\emptyset}_{\mcT_{n,n-m}} \otimes \ket{u}_{\mcU_{m}^{(d)}})\\
        &= \sqrt{1\over m_\lambda^{(d)}} \sum_{u, u', j}  (\ket{\lambda,u}_{\mcU_\lambda^{(d)}} \otimes \ket{\lambda,a}_{\mcS_\lambda}) \otimes \left(\tilde{C}^{\mu,\emptyset,\lambda,\emptyset}_{u',j,u}\right)^* \ket{\mu, \emptyset}_{\mcT_{n,n-m-1}} \otimes \ket{u'}_{\mcU_{m+1}^{(d)}} \otimes \ket{j}_{\CC^d}\\
        &= \sqrt{1\over m_\mu^{(d)}} \sum_{u, u', j}  C^{\lambda, \mu}_{u,j,u'} (\ket{\lambda,u}_{\mcU_\lambda^{(d)}} \otimes \ket{j}_{\CC^d} \otimes \ket{\lambda,a}_{\mcS_\lambda}) \otimes \ket{\mu, \emptyset}_{\mcT_{n,n-m-1}} \otimes \ket{u'}_{\mcU_{m+1}^{(d)}}\\
        &= \sqrt{1\over m_\mu^{(d)}} \sum_{u'} (\ket{\mu,u'}_{\mcU_\mu^{(d)}} \otimes \ket{\mu,a^\lambda_\mu}_{\mcS_\mu}) \otimes \ket{\mu, \emptyset}_{\mcT_{n,n-m-1}} \otimes \ket{u'}_{\mcU_{m+1}^{(d)}}\\
        &= \ket{\phi^+_{\mu, a^\lambda_\mu}}.
    \end{align}
\end{proof}

\subsection{Construction of the probabilistic quantum comb $\{\supermap{T}_I, \supermap{T}_O\}$}
\label{appendix_subsec:construction_of_red_comb}
\begin{figure}
    \centering
    \includegraphics[width=.9\linewidth]{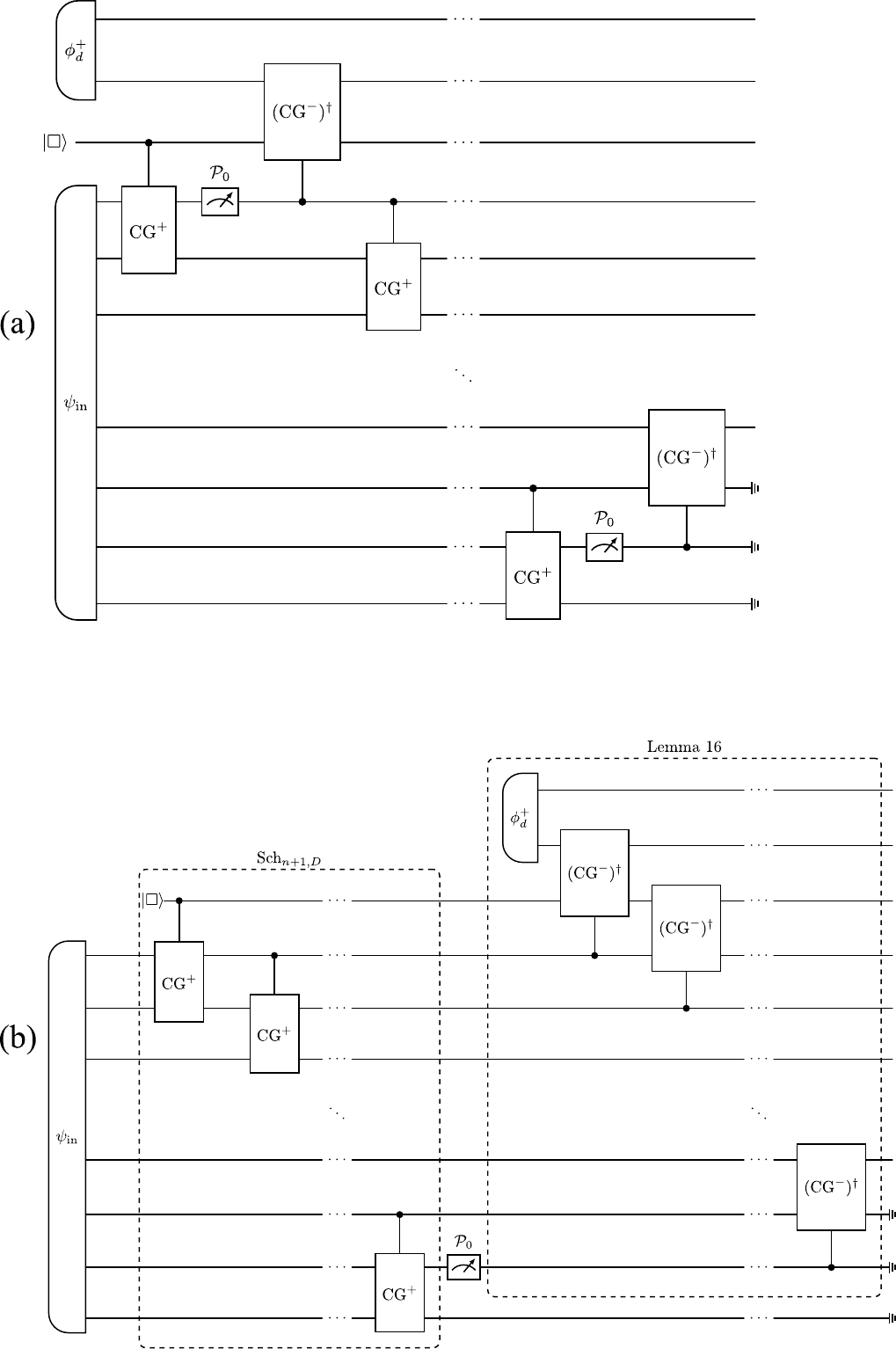}
    \caption{The supplementary figure to explain the calculation of $\psi_\mathrm{out}$ in Eqs.~\eqref{eq:calculation_psi_out_begin} -- \eqref{eq:calculation_psi_out_end}.
    The subscripts on $\mathrm{CG}^{\pm}$ are omitted for simplicity.
    The operator $\psi_\mathrm{out}$ is represented by the diagram (a).
    By commuting the CG transforms $\mathrm{CG}^+$ with the inverse of the dual CG transforms $(\mathrm{CG}^-)^\dagger$, we obtain (b), where the sequence of the CG transforms corresponds to the quantum Schur transform $\mathrm{Sch}_{n+1,D}$, and we can apply Lemma~\ref{lem:dual_schur_transform} for the sequence of the inverse of the dual CG transforms.}
    \label{fig:psi_out_calculation}
\end{figure}

We show that this construction gives the Choi operator $T_I$ given in Eq.~\eqref{eq:def_TI}.
To this end, it is sufficient to check that
\begin{align}
    \psi_\mathrm{out}\coloneqq J_{\Gamma^{(n+1)}_I} \star J_{\Gamma^{(n)}} \star \cdots \star J_{\Gamma^{(1)}} \star \psi_\mathrm{in}
\end{align}
for any pure state $\psi_\mathrm{in} = \ketbra{\psi_\mathrm{in}}$ for $\ket{\psi_\mathrm{in}}\in \mcP\otimes \mcO^n$
satisfies
\begin{align}
    \psi_\mathrm{out}= T_I \star \psi_\mathrm{in}.
\end{align}
Suppose $\ket{\psi_\mathrm{in}}$ is given in the Schur basis as
\begin{align}
    \ket{\psi_\mathrm{in}} = \bigoplus_{\mu\in\young{D}{n+1}} \sum_{u, i} a_{\mu,u,i} \ket{\mu,u}_{\mcU_\mu^{(D)}} \otimes \ket{\mu,i}_{\mcS_\mu}.
\end{align}
The operator $\psi_\mathrm{out}$ is evaluated as
\begin{align}
\label{eq:calculation_psi_out_begin}
    \psi_\mathrm{out} = \Tr_{\mcU_{n+1}^{(d)}\mcT_{n+1}\mcU_{n+1}^{(D)}}[\ketbra{\psi_\mathrm{out}'}]
\end{align}
using the vector $\ket{\psi_\mathrm{out}'}\in\mcP'\otimes \mcO^{\prime n}\otimes \mcU_{n+1}^{(d)}\otimes \mcT_n\otimes \mcU_{n+1}^{(D)}$ given by (see also Fig.~\ref{fig:psi_out_calculation}):
\begin{align}
    &\ket{\psi_\mathrm{out}'}\nonumber\\
    &= \prod_{i=1}^{n}[((\mathrm{CG}^-_{n,n-i,d})^\dagger\otimes \1_{\mcU_{i+1}^{(D)}})(P_{t_{i+1}}\otimes \1_{\mcU_{i}^{(d)}\mcT_{i}\mcU_{i+1}^{(D)}})(\1_{\mcU_{i}^{(d)}}\otimes \mathrm{CG}^+_{i+1,D})\otimes \1_{\mcP'\mcO'_1\cdots \mcO'_{i}\mcO_{i+2}\cdots \mcO_n}] \nonumber\\
    &\times (\ket{\psi_\mathrm{in}}\otimes \ket{\phi^+_d}_{\mcP'\mcU_1^{(d)}})\\
    &= \prod_{i=1}^{n} [(\mathrm{CG}^-_{n,n-i,d})^\dagger \otimes \1_{\mcP'\mcO'_1\cdots \mcO'_{i-1} \mcT_{i+2}\cdots \mcT_{n+1} \mcU_{n+1}^{(D)}}] \nonumber\\
    &\times [(P_0)_{\mcT_{n+1}} \otimes \1_{\mcT_1\cdots\mcT_{n}\mcU_{n+1}^{(D)} \mcP'\mcU_1^{(d)}\mcU_{n+1}^{(D)}}][(\mathrm{Sch}_{n+1,D})_{\mcP\mcO^n\to \mcT_1\cdots \mcT_{n+1}\mcU_{n+1}^{(D)}}\otimes \1_{\mcP'\mcU_1^{(d)}}](\ket{\psi_\mathrm{in}}\otimes \ket{\phi^+_d}_{\mcP'\mcU_1^{(d)}})\\
    &= \prod_{i=1}^{n} [(\mathrm{CG}^-_{n,n-i,d})^\dagger \otimes \1_{\mcP'\mcO'_1\cdots \mcO'_{i-1} \mcT_{i+2}\cdots \mcT_{n+1} \mcU_{n+1}^{(D)}}] \nonumber\\
    &\times \sum_{\mu\in\young{d}{n+1}}\sum_{u,i} a_{\mu,u,i} \ket{u}_{\mcU_{n+1}^{(D)}}\otimes \ket{\mu_i^{(1)}\cdots \mu_i^{(n+1)}}_{\mcT_1\cdots \mcT_{n+1}} \otimes \ket{\phi^+_d}_{\mcP'\mcU_1^{(d)}}\\
    &= \sum_{\mu\in\young{d}{n+1}}\sum_{u,i} {a_{\mu,u,i}\over \sqrt{m_\mu^{(d)}}} \sum_{u'}(\ket{\mu,u'}_{\mcU_\mu^{(d)}}\otimes \ket{\mu,i}_{\mcS_\mu})_{\mcP'\mcO^{\prime n}} \otimes \ket{u'}_{\mcU_{n+1}^{(d)}}\otimes \ket{\mu}_{\mcT_{n+1}} \otimes \ket{u}_{\mcU_{n+1}^{(D)}},
    \label{eq:calculation_psi_out_end}
\end{align}
where $\prod_{i=1}^{n} X_i$ for operators $X_1, \ldots, X_{n}$ represents $X_{n}\cdots X_1$, and we recursively use Lemma~\ref{lem:dual_schur_transform}.
Therefore, the operator $\psi_\mathrm{out}$ is evaluated as
\begin{align}
    \psi_\mathrm{out}
    &= \sum_{\mu\in\young{d}{n+1}} \sum_{u,i,i'} a_{\mu,u,i} a_{\mu,u,i'}^* {\1_{\mcU_\mu^{(d)}} \over m_\mu^{(d)}} \otimes \ketbra{\mu, i}{\mu,i'}_{\mcS_\mu}\\
    &= T_I \star \psi_\mathrm{in},
\end{align}
where we use Eq.~\eqref{eq:psi_out}.
We can also check that the Choi operator $T_O\coloneqq J_{\Gamma^{(n+1)}_O} \star J_{\Gamma^{(n)}} \star \cdots \star J_{\Gamma^{(1)}}$ satisfies the $\U(d)\times \U(D)$ symmetry:
\begin{align}
    [T_O, U^{\prime \otimes n+1}_{\mcP'\mcO^{\prime n}} \otimes U^{\prime\prime \otimes n+1}_{\mcP\mcO^n}] = 0 \quad \forall U'\in\U(d), U''\in\U(D).
\end{align}

\subsection{Proof of Theorem \ref{thm:sequential_isometry_adjointation}: Construction of a sequential isometry adjointation protocol}
\label{appendix_sec:sequential_isometry_adjointation}
First, we show that the probabilistic quantum comb $\{\supermap{T}_I, \supermap{T}_O\}$ satisfies
\begin{align}
    T_I \star \dketbra{V}^{\otimes n}_{\mcI^n \mcO^n} = \int_{\U(d)} \dd U \dketbra{UV^\dagger}_{\mcP\mcP'} \otimes \dketbra{U}^{\otimes n}_{\mcI^n\mcO'^n} + (\1_D-\Pi_{\Im V})^T_{\mcP} \otimes \Sigma \label{eq:TI_V_composition}
\end{align}
for all $V\in\isometry{d}{D}$, where $\dd U$ is the Haar measure on $\U(d)$ and $\Sigma$ is defined in Eq.~(\ref{eq:def_Sigma}).
By definition (\ref{eq:def_TI}) of $T_I$, it satisfies
\begin{align}
    [T_I, \1_{\mcP'\mcO'^n} \otimes U'^{\otimes n+1}_{\mcP\mcO^n}] = 0 \quad \forall U'\in\U(D).
\end{align}
Therefore,
\begin{align}
     T_I \star \dketbra{V}^{\otimes n}_{\mcI^n \mcO^n}
     &= (\1_{\mcP'\mcO'^n} \otimes \map{U}'^{\otimes n+1}_{\mcP\mcO^n})(T_I) \star \dketbra{V}^{\otimes n}_{\mcI^n \mcO^n}\\
     &= (\1_{\mcP'\mcO'^n\mcI^n}\otimes \map{U}'_{\mcP})(T_I\star \dketbra{U'^T V}^{\otimes n}_{\mcI^n \mcO^n})
     \label{eq:T_IstarV_decomposition}
\end{align}
Since $U'^T V = V$ holds for $U'^T = \1_{\Im V} \oplus U''_{(\Im V)^\perp}$ using $U''\in\U(D-d)$, we obtain
\begin{align}
    [T_I \star \dketbra{V}^{\otimes n}_{\mcI^n \mcO^n}, \1_{\mcP'\mcO'^n\mcI^n} \otimes (\1_{\Im V} \oplus U''_{(\Im V)^{\perp}})^T_{\mcP}] = 0 \quad \forall U''\in\U(D-d).
\end{align}
Therefore, $T_I \star \dketbra{V}^{\otimes n}_{\mcI^n \mcO^n}$ satisfies
\begin{align}
    T_I \star \dketbra{V}^{\otimes n}_{\mcI^n \mcO^n}
    =& [\1_{\mcP'\mcO'^n\mcI^n} \otimes (\Pi_{\Im V}^T)_{\mcP}] (T_I \star \dketbra{V}^{\otimes n}_{\mcI^n \mcO^n}) [\1_{\mcP'\mcO'^n\mcI^n} \otimes (\Pi_{\Im V}^T)_{\mcP}]\nonumber\\
    &+ (\1_D-\Pi_{\Im V})^T_{\mcP} \otimes \Sigma,
\end{align}
where $\Sigma \in\mcL(\mcP'\mcO'^n\mcI^n)$ is given by
\begin{align}
    \Sigma = {1\over D-d} T_I \star \dketbra{V}^{\otimes n}_{\mcI^n \mcO^n} \star (\1_D-\Pi_{\Im V})_{\mcP}.\label{eq:def_Sigma_another}
\end{align}
Since $\Pi_{\Im V} = V V^\dagger$ holds, the first term in Eq.~(\ref{eq:T_IstarV_decomposition}) can be evaluated as
\begin{align}
    &[\1_{\mcP'\mcO'^n\mcI^n} \otimes (\Pi_{\Im V}^T)_{\mcP}] (T_I \star \dketbra{V}^{\otimes n}_{\mcI^n \mcO^n}) [\1_{\mcP'\mcO'^n\mcI^n} \otimes (\Pi_{\Im V}^T)_{\mcP}]\nonumber\\
    &= \sum_{\mu\in \young{d}{n+1}} \sum_{i,j=1}^{d_{\mu}} {(E^{\mu, d}_{ij})_{\mcP'\mcO'^{n}} \over {m_{\mu}^{(d)}}}\otimes (\1_{\mcO^n} \otimes \map{V}^*_{\mcP''\to \mcP} \circ \map{V}^T_{\mcP\to \mcP''})(E^{\mu, D}_{ij})_{\mcP\mcO^{n}}\star \dketbra{V}^{\otimes n}_{\mcI^n \mcO^n}\\
    &= \sum_{\mu\in \young{d}{n+1}} \sum_{i,j=1}^{d_{\mu}} {(E^{\mu, d}_{ij})_{\mcP'\mcO'^{n}} \over {m_{\mu}^{(d)}}}\otimes (\1_{\mcI^n} \otimes \map{V}^*_{\mcP''\to \mcP})\circ \map{V}^{T\otimes n+1}_{\mcP\mcO^n\to \mcP''\mcI^n}(E^{\mu, D}_{ij})_{\mcP\mcO^{n}}\\
    &= \sum_{\mu\in \young{d}{n+1}} \sum_{i,j=1}^{d_{\mu}} {(E^{\mu, d}_{ij})_{\mcP'\mcO'^{n}} \over {m_{\mu}^{(d)}}}\otimes (\1_{\mcI^n} \otimes \map{V}^*_{\mcP''\to \mcP})(E^{\mu, d}_{ij})_{\mcP''\mcI^{n}}\\
    &= \int_{\U(d)}\dd U (\1_{\mcP'}\otimes \map{V}^*_{\mcP''\to\mcP})(\dketbra{U}_{\mcP''\mcP'}) \otimes \dketbra{U}_{\mcI^n \mcO'^n}^{\otimes n}\\
    &= \int_{\U(d)}\dd U \dketbra{UV^\dagger}_{\mcP\mcP'} \otimes \dketbra{U}_{\mcI^n \mcO'^n}^{\otimes n},
\end{align}
where $\mcP''$ is a Hilbert spaces given by $\mcP'' = \CC^d$, $\dd U$ is the Haar measure on $\U(d)$, and we have utilized the following relations:
\begin{align}
    V^{\dagger \otimes n+1} E^{\mu,D}_{ij} V^{\otimes n+1}&=
    \begin{cases}
        E^{\mu, d}_{ij} & (\mu\in\young{d}{n+1})\\
        0 & (\mu\notin\young{d}{n+1})
    \end{cases},\label{eq:action_of_isometry_on_E}\\
    \int_{\U(d)}\dd U \dketbra{U}_{\mcP''\mcI^n, \mcP'\mcO'^n}^{\otimes n+1} &= \sum_{\mu\in \young{d}{n+1}} \sum_{i,j=1}^{d_{\mu}} {(E^{\mu, d}_{ij})_{\mcP'\mcO'^{n}} \otimes (E^{\mu, d}_{ij})_{\mcP''\mcI^{n}} \over {m_{\mu}^{(d)}}},
\end{align}
the former of which follows from the decomposition of $V^{\otimes n+1}$ shown in Eq.~(\ref{eq:decomposition_isometry}) and the definition (\ref{eq:def_E}) of $E^{\mu,d}_{ij}$, and the latter of which is shown in Ref.~\cite{quintino2022deterministic}.
The operator $\Sigma$ defined in Eq.~(\ref{eq:def_Sigma_another}) is shown to be the same as Eq.~(\ref{eq:def_Sigma}) by the following calculation:
\begin{align}
    \Sigma
    &={1\over D-d} T_I \star \dketbra{V}^{\otimes n}_{\mcI^n \mcO^n} \star (\1_D-\Pi_{\Im V})_{\mcP}\nonumber\\
    &= {1\over D-d} \left[(\1_{\mcP'\mcO'^n\mcP} \otimes \map{V}^{T\otimes n}_{\mcO^n\to\mcI^n})(\Tr_\mcP T_I) - \Tr_{\mcP''} \circ (\1_{\mcP'\mcO'^n} \otimes \map{V}^{T\otimes n+1}_{\mcP\mcO^n\to \mcP''\mcI^n})(T_I) \right]\\
    &= {1\over D-d}\sum_{\lambda\in \young{d}{n}} \sum_{\mu\in\lambda+_d\square} \sum_{a,b=1}^{d_{\lambda}} \left[{m_\mu^{(D)}\over m_\lambda^{(D)}} - {m_\mu^{(d)}\over m_\lambda^{(d)}}\right] (E^{\lambda,d}_{ab})_{\mcI^n} \otimes {(E^{\mu,d}_{a^\lambda_\mu b^\lambda_\mu})_{\mcO'^n \mcP'} \over m_\mu^{(d)}},
\end{align}
where we have utilized Lemma \ref{lem:yy} and Eq.~(\ref{eq:action_of_isometry_on_E}).
To calculate it further, we employ the dimension formula of the irreducible representation $\mcU_\mu^{(D)}$ of $\U(D)$ given by Eq.~(\ref{eq:mult_hook}). 
In particular, for $\mu\in\lambda+\square$, the ratio of $m_\mu^{(d)}$ and $m_\lambda^{(d)}$ is given by
\begin{align}
    {m_\mu^{(d)}\over m_\lambda^{(d)}} = (d+j-i) {\mathrm{hook}(\lambda) \over \mathrm{hook}(\mu)},
\end{align}
where $(i,j)$ is a coordinate of the box added to obtain $\mu$ from $\lambda$, i.e.,
\begin{align}
    {m_\mu^{(D)}\over m_\lambda^{(D)}} - {m_\mu^{(d)}\over m_\lambda^{(d)}}
    &= (D-d) {\mathrm{hook}(\lambda) \over \mathrm{hook}(\mu)}
\end{align}
holds.
Thus, we obtain Eq.~(\ref{eq:TI_V_composition}).

We evaluate the Choi operator of $\supermap{C}_a[\map{V}_\mathrm{in}^{\otimes n}]$ given by
\begin{align}
    J_{\supermap{C}_a[\map{V}_\mathrm{in}^{\otimes n}]}
    &= C' \star T_a \star \dketbra{V_\mathrm{in}}^{\otimes n}_{\mcI^n\mcO^n}
\end{align}
to show Theorem \ref{thm:sequential_isometry_adjointation}.  From Eq.~(\ref{eq:TI_V_composition}), it is given by
\begin{align}
    J_{\supermap{C}_I[\map{V}_\mathrm{in}^{\otimes n}]}
    &= C' \star \int_{\U(d)} \dd U \dketbra{UV_\mathrm{in}^\dagger}_{\mcP\mcP'} \otimes \dketbra{U}^{\otimes n}_{\mcI^n \mcO^n} + (\1_D-\Pi_{\Im V_\mathrm{in}})^T_{\mcP} \otimes C'\star \Sigma.\label{eq:evaluation_sequential_isometry_adjointation}
\end{align}
Since $C'$ is the Choi operator of the covariant unitary inversion protocol, we obtain \cite{quintino2022deterministic}
\begin{align}
    C'\star \dketbra{U}^{\otimes n}_{\mcI^n\mcO^n} = {d^2F_\mathrm{UI}-1 \over d^2-1} \dketbra{U^{-1}}_{\mcP'\mcF} + {d^2(1-F_\mathrm{UI}) \over d^2-1} {\1_{\mcF} \over d} \otimes \1_{\mcP'}.
\end{align}
Thus, the first term in Eq.~(\ref{eq:evaluation_sequential_isometry_adjointation}) is given by
\begin{align}
    &C' \star \int_{\U(d)} \dd U \dketbra{UV_\mathrm{in}^\dagger}_{\mcP\mcP'} \otimes \dketbra{U}^{\otimes n}_{\mcI^n \mcO^n}\nonumber\\
    &= \int_{\U(d)} \dd U \dketbra{UV_\mathrm{in}^\dagger}_{\mcP\mcP'} \star \left[{d^2F_\mathrm{UI}-1 \over d^2-1} \dketbra{U^{-1}}_{\mcP'\mcF} + {d^2(1-F_\mathrm{UI}) \over d^2-1} {\1_{\mcF} \over d} \otimes \1_{\mcP'}\right]\\
    &={d^2F_\mathrm{UI}-1 \over d^2-1} \dketbra{V_\mathrm{in}^\dagger}_{\mcP\mcF} + {d^2(1-F_\mathrm{UI}) \over d^2-1} {\1_{\mcF} \over d} \otimes (\Pi_{\Im V_\mathrm{in}}^T)_\mcP.\label{eq:evaluation_sequential_isometry_adjointation1}
\end{align}
Since $C'$ and $\Sigma$ satisfy the $\U(d)$ symmetry, namely,
\begin{align}
    [C',U'^{\otimes n+1}_{\mcI^n\mcF} \otimes \1_{\mcP'\mcO'^n}]&=0,\\
    [\Sigma, U'^{\otimes n+1}_{\mcI^n} \otimes \1_{\mcP'\mcO'^n}]&=0
\end{align}
for all $U'\in\U(d)$, $C'\star \Sigma$ satisfies
\begin{align}
    [C'\star \Sigma, U'_{\mcF}]=0
\end{align}
for all $U'\in\U(d)$, i.e., it is proportional to the identity operator.  Thus, the second term in Eq.~(\ref{eq:evaluation_sequential_isometry_adjointation}) is given by
\begin{align}
    (\1_D-\Pi_{\Im V_\mathrm{in}})^T_{\mcP} \otimes C'\star \Sigma
    &= (\1_D-\Pi_{\Im V_\mathrm{in}})^T_{\mcP} \otimes \Tr[C'\star \Sigma] {\1_\mcF \over d}\\
    &= (\1_D-\Pi_{\Im V_\mathrm{in}})^T_{\mcP} \otimes \Tr[\Tr_\mcF(C') \Sigma] {\1_\mcF \over d},\label{eq:evaluation_sequential_isometry_adjointation2}
\end{align}
where we utilize the fact that $\Sigma=\Sigma^T$. Combining Eqs.~(\ref{eq:evaluation_sequential_isometry_adjointation}), (\ref{eq:evaluation_sequential_isometry_adjointation1}), and (\ref{eq:evaluation_sequential_isometry_adjointation2}), we obtain
\begin{align}
    J_{\supermap{C}_I[\map{V}_\mathrm{in}^{\otimes n}]}
    &= {d^2F_\mathrm{UI}-1 \over d^2-1} \dketbra{V_\mathrm{in}^\dagger}_{\mcP\mcF} + {\1_{\mcF} \over d}\otimes \left[{d^2(1-F_\mathrm{UI}) \over d^2-1}  (\Pi_{\Im V_\mathrm{in}}^T)_\mcP + \alpha_{C'} (\1_D-\Pi_{\Im V_\mathrm{in}})^T_{\mcP}\right],
\end{align}
where $\alpha_{C'}$ is defined by
\begin{align}
    \alpha_{C'}\coloneqq \Tr[\Tr_\mcF(C') \Sigma].
\end{align}
Therefore, we obtain
\begin{align}
    \supermap{C}_I[V_\mathrm{in}^{\otimes n}](\rho_\mathrm{in}) &= {d^2F_\mathrm{UI}-1 \over d^2-1} \map{V}_\mathrm{in}^\dagger (\rho_\mathrm{in}) + {\1_d \over d}\Tr\left\{\left[{d^2(1-F_\mathrm{UI}) \over d^2-1}\Pi_{\Im V_\mathrm{in}} + \alpha_{C'} (\1_D-\Pi_{\Im V_\mathrm{in}})\right]\rho_\mathrm{in}\right\}.\label{eq:CI_V_composition}
\end{align}
Since $\supermap{C}_I[V_\mathrm{in}^{\otimes n}] + \supermap{C}_O[V_\mathrm{in}^{\otimes n}]$ is a CPTP map and
\begin{align}
    \Tr \supermap{C}_I[V_\mathrm{in}^{\otimes n}](\rho_\mathrm{in}) = \Tr\left\{\left[\Pi_{\Im V_\mathrm{in}} + \alpha_{C'} (\1_D - \Pi_{\Im V_\mathrm{in}})\right] \rho_\mathrm{in} \right\}
\end{align}
holds, $\Tr \supermap{C}_O[V_\mathrm{in}^{\otimes n}](\rho_\mathrm{in}) = (1-\alpha_{C'})\Tr[(\1_D-\Pi_{\Im V_\mathrm{in}}) \rho_\mathrm{in}]$ holds.  Since $\supermap{C}_O[V_\mathrm{in}^{\otimes n}]$ is a CP map, we obtain
\begin{align}
    \supermap{C}_O[V_\mathrm{in}^{\otimes n}](\rho_\mathrm{in}) = \supermap{C}_O[V_\mathrm{in}^{\otimes n}](\Pi_{(\Im V_\mathrm{in})^\perp}\rho_\mathrm{in}\Pi_{(\Im V_\mathrm{in})^\perp}).\label{eq:CO_independency}
\end{align}
Due to the $\U(d)\times \U(D)$ symmetries of $C'$ and $T_O$,
\begin{align}
    [C'\star T_O, U_{\mcI^n\mcF}'^{\otimes n+1} \otimes U_{\mcP\mcO^n}''^{\otimes n+1}] = 0 \quad \forall U'\in\U(d), U''\in\U(D)
\end{align}
holds.  Thus, we obtain
\begin{align}
    \supermap{C}_O[V_\mathrm{in}^{\otimes n}](\rho_\mathrm{in}) = \map{U}'\circ \supermap{C}_O[(U'' V_\mathrm{in} U')^{\otimes n}]\circ \map{U}''(\rho_\mathrm{in}) \quad \forall U'\in\U(d), U''\in\U(D).\label{eq:CO_symmetry}
\end{align}
In particular, for all $U'\in \U(d)$, we take $U'' = U'''_{\Im V_\mathrm{in}} \oplus \1_{(\Im V_\mathrm{in})^\perp}$ for $U'''\in\U(d)$ such that $U'' V_\mathrm{in} U' = V_\mathrm{in}$. From Eqs.~(\ref{eq:CO_independency}) and (\ref{eq:CO_symmetry}), we obtain
\begin{align}
    \supermap{C}_O[V_\mathrm{in}^{\otimes n}](\rho_\mathrm{in}) = \map{U}'\circ \supermap{C}_O[V_\mathrm{in}^{\otimes n}](\rho_\mathrm{in}) \quad \forall U'\in\U(d),\label{eq:CO_V_composition}
\end{align}
thus, $\supermap{C}_O[V_\mathrm{in}^{\otimes n}](\rho_\mathrm{in}) \propto {\1_d \over d}$.  From Eq.~(\ref{eq:CO_independency}), we obtain
\begin{align}
    \supermap{C}_O[V_\mathrm{in}^{\otimes n}](\rho_\mathrm{in}) = {\1_d \over d} (1-\alpha_{C'}) \Tr[(\1_D-\Pi_{\Im V_\mathrm{in}})\rho_\mathrm{in}].
\end{align}
Thus, we obtain
\begin{align}
    \epsilon = \max \{ \alpha_{C'}, 1-F_\mathrm{UI}\}.
\end{align}

\subsection{Proof of Theorem \ref{thm:parallel_isometry_adjointation}: Construction of a parallel isometry adjointation protocol}
\label{appendix_sec:parallel_isometry_adjointation}

The parallel protocol shown in Theorem \ref{thm:parallel_isometry_adjointation} corresponds to the special case of the sequential protocol such that the Choi operator of the original untiary inversion protocol is given by
\begin{align}
    C' = \phi_{\mcI^n\mcA} \star \sum_i (M_i^T)_{\mcO'^n\mcA} \otimes (J_{\mcR_i})_{\mcP'\mcF}.
\end{align}
Since the Choi operator of a CPTP map $\mcR_i$ satisfies $\Tr_\mcF (J_{\mcR_i})_{\mcP'\mcF} = \1_{\mcP'}$, we obtain
\begin{align}
    \Tr_\mcF(C')
    &= \phi_{\mcI^n\mcA} \star \sum_i (M_i^T)_{\mcO'^n\mcA} \otimes \1_{\mcP'}\\
    &= \phi_{\mcI^n\mcA} \star \1_{\mcO'^n\mcA} \otimes \1_{\mcP'}\\
    &= \Tr_\mcA(\phi) \otimes \1_{\mcO'^n\mcP'}.
\end{align}
Then, $\alpha_{C'}$ shown in Eq.~(\ref{eq:def_alpha_C'}) reduces to $\alpha_\phi$ given by
\begin{align}
    \alpha_{\phi}
    &\coloneqq \Tr[\Tr_\mcA(\phi) \Tr_{\mcO'^n\mcP'}(\Sigma)]\\
    &= \sum_{\lambda\in\young{d}{n}} \Tr[\Tr_{\mcA}(\phi) \Pi_\lambda^{(d)}] \sum_{\mu\in\lambda+_d\square} {\mathrm{hook}(\lambda)\over \mathrm{hook}(\mu)}\\
    &= \sum_{\lambda\in\young{d}{n}} \Tr[\Tr_{\mcA}(\phi) \Pi_\lambda^{(d)}] \left[1- \sum_{\mu\in\lambda+\square \setminus \young{d}{n+1}} {\mathrm{hook}(\lambda) \over \mathrm{hook}(\mu)}\right],
\end{align}
where we have utilized Lemma \ref{lem:yy}, Eq.~(\ref{eq:def_Sigma}), and the equality
\begin{align}
    \sum_{\nu\in\lambda+\square} {\mathrm{hook}(\lambda)\over \mathrm{hook}(\nu)} =1,\label{eq:hook_equality}
\end{align}
shown as follows. From Lemma \ref{lem:yy},
\begin{align}
    E_{aa}^{\lambda, D} \otimes \1_D = \sum_{\mu\in\lambda +_D\square} E_{a^\lambda_\mu a^\lambda_\mu}^{\mu,D}
\end{align}
holds for $a\in \{1, \ldots, d_\lambda\}$. By taking the trace of the both-hand side, we obtain
\begin{align}
    D m_\lambda^{(D)} = \sum_{\mu\in\lambda+_D\square} m_\mu^{(D)}.\label{eq:mult_bunki}
\end{align}
By using Eqs.~(\ref{eq:mult_hook}) and (\ref{eq:mult_bunki}), we show Eq.~(\ref{eq:hook_equality}) as
\begin{align}
    \sum_{\nu\in\lambda+\square} {\mathrm{hook}(\lambda)\over \mathrm{hook}(\nu)}
    &= \lim_{D\to\infty} {\sum_{\nu\in\lambda+_D\square} m_\nu^{(D)} \over Dm_\lambda^{(D)}}\\
    &=1.
\end{align}
Thus, the approximation error of isometry adjointation is given by
\begin{align}
    \epsilon = \max\{\alpha_\phi, 1-F_\mathrm{est}\}.
\end{align}

\subsection{Reduction to isometry inversion and universal error detection}
\subsubsection{Deterministic isometry inversion}
\label{appendix_subsec:reduction_deterministic_isometry_inversion}
By discarding the measurement outcome of the isometry adjointation protocol shown in Eqs.~(\ref{eq:CI_V_composition}) and (\ref{eq:CO_V_composition}), we obtain the following protocol:
\begin{align}
    \supermap{C}[\map{V}_\mathrm{in}^{\otimes n}] = \supermap{C}_I[\map{V}_\mathrm{in}^{\otimes n}] + \supermap{C}_O[\map{V}_\mathrm{in}^{\otimes n}]
\end{align}
This protocol satisfies
\begin{align}
    \supermap{C}[\map{V}_\mathrm{in}^{\otimes n}]\circ \map{V}_\mathrm{in}(\rho_\mathrm{in}) = {d^2F_\mathrm{UI}-1 \over d^2-1} \rho_\mathrm{in} + {\1_d \over d} \left(1-{d^2F_\mathrm{UI}-1 \over d^2-1}\right) \Tr(\rho_\mathrm{in}).
\end{align}
Thus, the worst-case channel fidelity of isometry inversion is given by
\begin{align}
    F_\mathrm{ch}(\supermap{C}[\map{V}_\mathrm{in}^{\otimes n}]\circ \map{V}_\mathrm{in}, \1_d) = F_\mathrm{UI}.
\end{align}

\subsubsection{Probabilistic exact isometry inversion}
\label{appendix_subsec:reduction_probabilistic_isometry_inversion}
The construction of deterministic exact unitary inversion protocol can be rewritten in terms of the Choi operator as
\begin{align}
    &C'\star \dketbra{U_\mathrm{in}}^{\otimes n}_{\mcI^n\mcO'^n} \approx \dketbra{U_\mathrm{in}^{-1}}_{\mcP'\mcF} \quad \forall U_\mathrm{in}\in\U(d)\nonumber\\
    &\Longrightarrow T\star C' \star \dketbra{V_\mathrm{in}}^{\otimes n}_{\mcI^n\mcO^n} \star \dketbra{V_\mathrm{in}}_{\mcP''\mcP} \nonumber\\
    &\hspace{24pt}= {d^2F_\mathrm{UI}-1 \over d^2-1} \dketbra{\1_d}_{\mcP''\mcF} + {\1_{\mcP''\mcP} \over d} \left(1-{d^2F_\mathrm{UI}-1 \over d^2-1}\right) \quad \forall V_\mathrm{in}\in\isometry{d}{D},
\end{align}
where $\mcP''$ is a Hilbert space given by $\mcP''=\CC^d$.
By replacing $C'$ with the probabilistic exact unitary inversion comb $\{C'_S, C'_F\}$ with the success probability $p_\mathrm{UI}$ satisfying
\begin{align}
    C'_S\star \dketbra{U_\mathrm{in}}^{\otimes n}_{\mcI^n\mcO'^n}  = p_\mathrm{UI} \dketbra{U_\mathrm{in}^{-1}}_{\mcP'\mcF} \quad \forall U_\mathrm{in}\in\U(d),
\end{align}
we obtain the probabilistic exact isometry inversion protocol given by
\begin{align}
    T\star C'_S \star \dketbra{V_\mathrm{in}}^{\otimes n}_{\mcI^n\mcO^n} \star \dketbra{V_\mathrm{in}}_{\mcP''\mcP}= p_\mathrm{UI} \dketbra{\1_d}_{\mcP''\mcF} \nonumber\\
    \forall V_\mathrm{in}\in\isometry{d}{D}.
\end{align}

\subsubsection{Universal error detection}
\label{appendix_subsec:reduction_universal_error_detection}
By discarding the output state of the isometry adjointation protocol shown in Eqs.~(\ref{eq:CI_V_composition}) and (\ref{eq:CO_V_composition}), we obtain the following protocol:
\begin{align}
    \Tr(\Pi_a \rho_\mathrm{in}) = \Tr\circ \supermap{C}_a[\map{V}_\mathrm{in}^{\otimes n}](\rho_\mathrm{in}),
\end{align}
This protocol satisfies
\begin{align}
    \Pi_I &= \Pi_{\Im V_\mathrm{in}} + \alpha_{C'} (\1_D-\Pi_{\Im V_\mathrm{in}}),\\
    \Pi_O &= (1- \alpha_{C'}) (\1_D-\Pi_{\Im V_\mathrm{in}}).
\end{align}
Since $C'$ corresponds to $C' = C''\otimes {\1_{\mcO'_n \mcF} \over d}$ using $C''$ shown in Corollary \ref{cor:universal_error_detection}, $\alpha_{C'}$ is given by
\begin{align}
    \alpha_C' = \Tr(C''\Sigma').
\end{align}

\section{Analysis on the optimal protocols}
\subsection{Proof of Theorem \ref{thm:unitary_group_symmetry}: $\U(d)\times \U(D)$ symmetry of the tasks}
\label{appendix_sec:proof_unitary_group_symmetry}
We show Theorem \ref{thm:unitary_group_symmetry} using the idea of a twirling map \cite{horodecki1999reduction} similarly to Refs.~\cite{quintino2019probabilistic, quintino2022deterministic}.  Assume that the Choi operators (\ref{eq:choi_each_task})
achieve the optimal performances in parallel, sequential or general protocols including indefinite causal order.  We define the $\U(d)\times\U(D)$-twirled
Choi operators by
\begin{align}
\begin{cases}
    C'\coloneqq \int_{\U(d)} \dd U' \int_{\U(D)} \dd U'' \map{U}'^{\otimes n+1}_{\mcI^n\mcF} \otimes \map{U}''^{\otimes n+1}_{\mcP\mcO^n} (C) & (\mathrm{deterministic \; isometry \; inversion})\\
    C'_a\coloneqq \int_{\U(d)} \dd U' \int_{\U(D)} \dd U'' \map{U}'^{\otimes n+1}_{\mcI^n\mcF} \otimes \map{U}''^{\otimes n+1}_{\mcP\mcO^n} (C_a)& (\mathrm{probabilistic \; exact \; isometry \; inversion})\\
    C'_a\coloneqq \int_{\U(d)} \dd U' \int_{\U(D)} \dd U'' \map{U}'^{\otimes n}_{\mcI^n} \otimes \map{U}''^{\otimes n+1}_{\mcP\mcO^n} (C_a) & (\mathrm{universal \; error \; detection})\\
    C'_a\coloneqq \int_{\U(d)} \dd U' \int_{\U(D)} \dd U'' \map{U}'^{\otimes n+1}_{\mcI^n\mcF} \otimes \map{U}''^{\otimes n+1}_{\mcP\mcO^n} (C_a)& (\mathrm{isometry \; adjointation})
\end{cases},
\end{align}
where $\dd U'$ and $\dd U''$ are the Haar measures of $\U(d)$ and $\U(D)$, respectively.
Then, these operators satisfy $C'_a\geq 0$ and $C'=\sum_a C'_a\in\mcW^{(x)}$, and the $\U(d)\times \U(D)$ symmetry given by
\begin{align}
\begin{cases}
    [C', U^{\prime\otimes n+1}_{\mcI^n \mcF} \otimes U^{\prime\prime \otimes n+1}_{\mcP\mcO^n}] = 0 & (\mathrm{deterministic \; isometry \; inversion})\\
    [C'_a, U^{\prime\otimes n+1}_{\mcI^n \mcF} \otimes U^{\prime\prime \otimes n+1}_{\mcP\mcO^n}] = 0   \quad \forall a\in\{S,F\}& (\mathrm{probabilistic \; exact \; isometry \; inversion})\\
    [C'_a, U^{\prime\otimes n}_{\mcI^n} \otimes U^{\prime\prime \otimes n+1}_{\mcP\mcO^n}] = 0 \quad \forall a\in\{I,O\} & (\mathrm{universal \; error \; detection})\\
    [C'_a, U^{\prime\otimes n+1}_{\mcI^n \mcF} \otimes U^{\prime\prime \otimes n+1}_{\mcP\mcO^n}] = 0 \quad \forall a\in\{I,O\} & (\mathrm{isometry \; adjointation})
\end{cases}
\end{align}
for all $U'\in\U(d)$ and $U''\in\U(D)$.  Then, we can show Theorem \ref{thm:unitary_group_symmetry} by showing that the $\U(d)\times\U(D)$-twirling does not decrease performances of each task.  We show this in the following subsections.

\subsubsection{Deterministic isometry inversion}
The worst-case fidelity $F_\mathrm{worst}$ is given using the Choi operator $C$ of the quantum superchannel by
\begin{align}
    F_\mathrm{worst}
    &= \inf_{V_\mathrm{in}\in\isometry{d}{D}} F_\mathrm{ch}[\supermap{C}(\map{V}_\mathrm{in}^{\otimes n})\circ \map{V}_\mathrm{in}, \1_d]\\
    &= {1\over d^2} \inf_{V_\mathrm{in}\in\isometry{d}{D}} \Tr[C\star (\dketbra{V_\mathrm{in}}_{\mcP''\mcP} \otimes \dketbra{V_\mathrm{in}}^{\otimes n}_{\mcI^n\mcO^n}) \dketbra{\1_d}_{\mcP''\mcF}]\\
    &= {1\over d^2} \inf_{V_\mathrm{in}\in\isometry{d}{D}} \Tr[C (\dketbra{V_\mathrm{in}^*}_{\mcF\mcP} \otimes \dketbra{V_\mathrm{in}^*}^{\otimes n}_{\mcI^n\mcO^n})].
\end{align}
Then, the worst-case channel fidelity of the $\U(d)\times\U(D)$-twirled Choi operator is given by
\begin{align}
    F'_\mathrm{worst}
    &={1\over d^2} \inf_{V_\mathrm{in}\in\isometry{d}{D}} \int_{\U(d)} \dd U' \int_{\U(D)} \dd U'' \nonumber\\
    &\hspace{30pt}\Tr[\map{U}'^{\otimes n+1}_{\mcI^n\mcF} \otimes \map{U}''^{\otimes n+1}_{\mcP\mcO^n}(C) (\dketbra{V_\mathrm{in}^*}_{\mcF\mcP} \otimes \dketbra{V_\mathrm{in}^*}^{\otimes n}_{\mcI^n\mcO^n})]\\
    &\geq {1\over d^2} \int_{\U(d)} \dd U' \int_{\U(D)} \dd U''\inf_{V_\mathrm{in}\in\isometry{d}{D}} \nonumber\\
    &\hspace{30pt} \Tr[C \map{U}'^{\dagger\otimes n+1}_{\mcI^n\mcF} \otimes \map{U}''^{\dagger\otimes n+1}_{\mcP\mcO^n}(\dketbra{V_\mathrm{in}^*}_{\mcF\mcP} \otimes \dketbra{V_\mathrm{in}^*}^{\otimes n}_{\mcI^n\mcO^n})]\\
    &={1\over d^2} \int_{\U(d)} \dd U' \int_{\U(D)} \dd U''\inf_{V_\mathrm{in}\in\isometry{d}{D}} \nonumber\\
    &\hspace{30pt} \Tr[C (\dketbra{U''^\dagger V_\mathrm{in}^* U'^*}_{\mcF\mcP} \otimes \dketbra{U''^\dagger V_\mathrm{in}^* U'^*}^{\otimes n}_{\mcI^n\mcO^n})]\\
    &= \int_{\U(d)} \dd U' \int_{\U(D)} \dd U'' F_\mathrm{worst}\\
    &= F_\mathrm{worst}.
\end{align}

\subsubsection{Probabilistic exact isometry inversion}
Suppose a Choi operator $C_S$ satisfies
\begin{align}
    C_S \star (\dketbra{V_\mathrm{in}}_{\mcP''\mcP} \otimes \dketbra{V_\mathrm{in}}^{\otimes n}_{\mcI^n\mcO^n}) = p \dketbra{\1_d}_{\mcP''\mcF} \quad \forall V_\mathrm{in}\in\isometry{d}{D}.
\end{align}
Then, the $\U(d)\times\U(D)$-twirled Choi operator also satisfies
\begin{align}
    &C'_S \star (\dketbra{V_\mathrm{in}}_{\mcP''\mcP} \otimes \dketbra{V_\mathrm{in}}^{\otimes n}_{\mcI^n\mcO^n})\nonumber\\
    &= \int_{\U(d)} \dd U' \int_{\U(D)} \dd U'' \map{U}'^{\otimes n+1}_{\mcI^n\mcF} \otimes \map{U}''^{\otimes n+1}_{\mcP\mcO^n} (C_S) \star (\dketbra{V_\mathrm{in}}_{\mcP''\mcP} \otimes \dketbra{V_\mathrm{in}}^{\otimes n}_{\mcI^n\mcO^n})\\
    &= \int_{\U(d)} \dd U' \int_{\U(D)} \dd U'' C_S \star \map{U}'^{T\otimes n+1}_{\mcI^n\mcF} \otimes \map{U}''^{T\otimes n+1}_{\mcP\mcO^n}(\dketbra{V_\mathrm{in}}_{\mcP''\mcP} \otimes \dketbra{V_\mathrm{in}}^{\otimes n}_{\mcI^n\mcO^n})\\
    &= \int_{\U(d)} \dd U' \int_{\U(D)} \dd U'' C_S \star (\dketbra{U''^T V_\mathrm{in} U'}_{\mcP''\mcP} \otimes \dketbra{U''^T V_\mathrm{in} U'}^{\otimes n}_{\mcI^n\mcO^n})\\
    &= \int_{\U(d)} \dd U' \int_{\U(D)} \dd U'' p \dketbra{\1_d}_{\mcP''\mcF}\\
    &= p \dketbra{\1_d}_{\mcP''\mcF},
\end{align}
i.e., it implements isometry inversion with success probability $p$.

\subsubsection{Universal error detection}
Suppose a supermap $\supermap{C}_I$ satisfies the one-sided error condition given by
\begin{align}
    \Tr[\supermap{C}_I(\map{V}_\mathrm{in}^{\otimes n}) \rho_\mathrm{in}] = 1 \quad \forall \rho\in \mcL(\Im V_\mathrm{in}), V_\mathrm{in}\in\isometry{d}{D}.
\end{align}
In terms of the Choi operator, this relation is given by
\begin{align}
    C_I \star \dketbra{V_\mathrm{in}}^{\otimes n}_{\mcI^n\mcO^n} \star (\rho_\mathrm{in})_\mcP = 1 \quad \forall \rho\in \mcL(\Im V_\mathrm{in}), V_\mathrm{in}\in\isometry{d}{D}.
\end{align}
Then, the $\U(d)\times\U(D)$-twirlied Choi operator also satisfies the one-sided error condition, i.e.,
\begin{align}
    &C'_I \star \dketbra{V_\mathrm{in}}^{\otimes n}_{\mcI^n\mcO^n}\star (\rho_\mathrm{in})_\mcP\nonumber\\
    &= \int_{\U(d)} \dd U' \int_{\U(D)} \dd U'' \map{U}'^{\otimes n}_{\mcI^n} \otimes \map{U}''^{\otimes n+1}_{\mcP\mcO^n} (C_I) \star \dketbra{V_\mathrm{in}}^{\otimes n}_{\mcI^n\mcO^n}\star (\rho_\mathrm{in})_\mcP\\
    &= \int_{\U(d)} \dd U' \int_{\U(D)} \dd U'' C_I \star \map{U}'^{T\otimes n}_{\mcI^n}\otimes \map{U}''^{T\otimes n}_{\mcO^n}(\dketbra{V_\mathrm{in}}^{\otimes n}_{\mcI^n\mcO^n}) \star \map{U}^{\prime\prime T}(\rho_\mathrm{in})_{\mcP}\\
    &= \int_{\U(d)} \dd U' \int_{\U(D)} \dd U'' C_I \star (\dketbra{U''^T V_\mathrm{in} U'}^{\otimes n}_{\mcI^n\mcO^n}) \star \map{U}^{\prime\prime T}(\rho_\mathrm{in})_{\mcP}\\
    &= 1
\end{align}
holds for all $\rho_\mathrm{in}\in\mcL(\Im V_\mathrm{in})$ and $V_\mathrm{in}\isometry{d}{D}$.
The worst-case operational distance of the original protocol is given by
\begin{align}
    \delta_\mathrm{worst} = \inf_{V_\mathrm{in}\in\isometry{d}{D}} D_\mathrm{op}(\{(C_a\star \dketbra{V_\mathrm{in}})^T\}_{a}, \{\Pi_{\Im V_\mathrm{in}}, \1_D-\Pi_{\Im V_\mathrm{in}}\}).
\end{align}
Then, the worst-case operation distance of the $\U(d)\times\U(D)$-twirlied Choi operator is given by
\begin{align}
    &D'_\mathrm{worst}\nonumber\\
    &= \inf_{V_\mathrm{in}\in\isometry{d}{D}} D_\mathrm{op}\Bigg(\left\{\int_{\U(d)} \dd U' \int_{\U(D)} \dd U'' \mcU''[C_a \star (\dketbra{U''^T V_\mathrm{in} U'}^{\otimes n}]^T\right\}_{a}, \nonumber\\
    &\hspace{120pt}\{\Pi_{\Im V_\mathrm{in}}, \1_D-\Pi_{\Im V_\mathrm{in}}\}\Bigg)\\
    &\leq \int_{\U(d)} \dd U' \int_{\U(D)} \dd U'' \inf_{V_\mathrm{in}\in\isometry{d}{D}}  D_\mathrm{op} \Big([C_a \star (\dketbra{U''^T V_\mathrm{in} U'}^{\otimes n}]^T, \nonumber\\
    &\hspace{120pt} \{\mcU^{\prime\prime *}(\Pi_{\Im V_\mathrm{in}}), \mcU^{\prime\prime *}(\1_D-\Pi_{\Im V_\mathrm{in}})\}\Big)\\
    & =  \int_{\U(d)} \dd U' \int_{\U(D)} \dd U'' \inf_{V_\mathrm{in}\in\isometry{d}{D}}  D_\mathrm{op} \Big([C_a \star (\dketbra{U''^T V_\mathrm{in} U'}^{\otimes n}]^T, \nonumber\\
    &\hspace{120pt} \{\Pi_{\Im U''^T V_\mathrm{in} U'}), \1_D-\Pi_{\Im U''^T V_\mathrm{in} U'}\}\Big)\\
    &\leq \int_{\U(d)} \dd U' \int_{\U(D)} \dd U'' \delta_\mathrm{worst}\\
    &= \delta_\mathrm{worst}
\end{align}
holds, where we use the convexity \eqref{eq:POVM_distance_convexity} and the unitary group invariance \eqref{eq:POVM_distance_unitary_invariance} of the operational distance.

\subsubsection{Isometry adjointation}
We introduce the notation to represent the diamond norm of a quantum channel $\Phi$ in terms of its Choi operator $J_{\Phi}$, i.e., we define $\mfD[J_{\Phi}]$ by
\begin{align}
    \mfD[J_{\Phi}] \coloneqq \|\Phi\|_{\diamond}.
\end{align}
Then, $\mfD[J_{\Phi}]$ satisfies the following properties:
\begin{align}
    \mfD[(U_\mcI \otimes U'_\mcO)J_{\Phi} (U_\mcI \otimes U'_\mcO)^\dagger] = \mfD[J_{\Phi}] \quad &\forall U\in\U(\dim \mcI), U'\in\U(\dim \mcO),\\
    \mfD[a J_{\Lambda^{(1)}} + b J_{\Lambda^{(2)}}] \leq a \mfD[J_{\Lambda^{(1)}}]+ b \mfD[J_{\Lambda^{(2)}}] \quad &\forall a,b\geq 0,
\end{align}
which corresponds to the following properties of the diamond norm:
\begin{align}
    \|\map{U}' \circ \Phi \circ \map{U}\|_{\diamond} = \|\Phi\|_{\diamond} \quad &\forall U\in\U(\dim \mcI), U'\in\U(\dim \mcO),\\
    \|a\Lambda^{(1)} + b \Lambda^{(2)}\|_{\diamond} \leq a \|\Lambda^{(1)}\|_{\diamond} + b \|\Lambda^{(2)}\|_{\diamond} \quad &\forall a,b\geq 0,
\end{align}
The worst-case diamond-norm error $\epsilon$ is given using the Choi operator $C_I$ by
\begin{align}
    \epsilon
    = \inf_{V_\mathrm{in}\in\isometry{d}{D}}  \mfD[C\star \dketbra{V_\mathrm{in}}_{\mcI^n \mcO^n}^{\otimes n} \otimes \ketbra{0}- \dketbra{V_\mathrm{in}^\dagger}_{\mcP\mcF} \otimes \ketbra{0}_{\mcA} - {\1_\mcP\otimes \1_\mcF \over d} \otimes \ketbra{1}_{\mcA}],
\end{align}
where $C'\coloneqq C_I\otimes \ketbra{0}_\mcA + C_O\otimes \ketbra{1}_\mcA$ and $\mcA = \CC^2$.
Thus, the worst-case diamond-norm error of the $\U(d)\times\U(D)$-twirled Choi operator is given by
\begin{align}
    \epsilon'
    &= \inf_{V_\mathrm{in}\in\isometry{d}{D}} \mfD[\int_{\U(d)} \dd U' \int_{\U(D)} \dd U'' \map{U}'^{\otimes n+1}_{\mcI^n\mcF} \otimes \map{U}''^{\otimes n+1}_{\mcP\mcO^n} \otimes \1_\mcA (C)\star \dketbra{V_\mathrm{in}}_{\mcI^n \mcO^n}^{\otimes n} \nonumber\\
    &\hspace{90pt}- \dketbra{V_\mathrm{in}^\dagger}_{\mcP\mcF} \otimes \ketbra{0}_{\mcA} - {\1_\mcP\otimes \1_\mcF \over d} \otimes \ketbra{1}_{\mcA}]\\
    &\leq \int_{\U(d)} \dd U' \int_{\U(D)} \dd U'' \inf_{V_\mathrm{in}\in\isometry{d}{D}} \mfD[\map{U}'^{\otimes n+1}_{\mcI^n\mcF} \otimes \map{U}''^{\otimes n+1}_{\mcP\mcO^n} \otimes \1_\mcA (C)\star \dketbra{V_\mathrm{in}}_{\mcI^n \mcO^n}^{\otimes n} \nonumber\\
    &\hspace{150pt} - \dketbra{V_\mathrm{in}^\dagger}_{\mcP\mcF} \otimes \ketbra{0}_{\mcA} - {\1_\mcP\otimes \1_\mcF \over d} \otimes \ketbra{1}_{\mcA}]\\
    &= \int_{\U(d)} \dd U' \int_{\U(D)} \dd U'' \inf_{V_\mathrm{in}\in\isometry{d}{D}} \mfD[C\star \map{U}'^{T\otimes n}_{\mcI^n} \otimes \map{U}''^{T\otimes n}_{\mcO^n}(\dketbra{V_\mathrm{in}}_{\mcI^n \mcO^n}^{\otimes n}) \nonumber\\
    &\hspace{60pt}- \map{U}'^\dagger_{\mcF} \otimes \map{U}''^\dagger_{\mcP}\otimes \1_\mcA(\dketbra{V_\mathrm{in}^\dagger}_{\mcP\mcF}\otimes \ketbra{0}_\mcA + {\1_\mcP\otimes \1_\mcF \over d} \otimes \ketbra{1}_{\mcA})]\\
    &= \int_{\U(d)} \dd U' \int_{\U(D)} \dd U'' \inf_{V_\mathrm{in}\in\isometry{d}{D}}\mfD[C\star (\dketbra{U''^T V_\mathrm{in}U'}_{\mcI^n \mcO^n}^{\otimes n}) \nonumber\\
    &\hspace{90pt}- \dketbra{U'^\dagger V_\mathrm{in}^\dagger U''^*}_{\mcP\mcF}\otimes \ketbra{0}_\mcA - {\1_\mcP\otimes \1_\mcF \over d} \otimes \ketbra{1}_{\mcA}]\\
    &= \int_{\U(d)} \dd U' \int_{\U(D)} \dd U'' \epsilon\\
    &= \epsilon,
\end{align}
i.e., it implements isometry adjointation with approximation error $\epsilon'\leq \epsilon$.

\subsection{Proof of Theorem \ref{thm:optimal_parallel_error_detection}: Optimal parallel protocol for universal error detection}
\label{appendix_sec:optimal_parallel_error_detection}
Due to Theorem \ref{thm:optimal_construction}, the optimal performance $\alpha$ of parallel error detection is given by minimizing $\alpha_\phi$ shown in Theorem \ref{cor:universal_error_detection}:
\begin{align}
    \alpha_\mathrm{opt}^{(\mathrm{PAR})} = \min_{\Tr(\phi) = 1, \phi \geq 0} \sum_{\lambda\in\young{d}{n}} \Tr(\phi \Pi_\lambda) \left[1-\sum_{\mu\in\lambda+\square \setminus \young{d}{n+1}} {\mathrm{hook}(\lambda)\over \mathrm{hook}(\mu)}\right].
\end{align}
The right-hand side has the minimum value
\begin{align}
    \alpha_\mathrm{opt}^{(\mathrm{PAR})} = 1-\max_{\lambda\in\young{d}{n}} \sum_{\mu\in\lambda+\square \setminus \young{d}{n+1}} {\mathrm{hook}(\lambda)\over \mathrm{hook}(\mu)}
\end{align}
at $\phi = {\1_{\mcU_\lambda^{(d)}} \over m_\lambda^{(d)}} \otimes \mathrm{arb}_{\mcS_\lambda}$ for an arbitrary state $\mathrm{arb}\in\mcL(\mcS_\lambda)$, where $\lambda$ is given by
\begin{align}
    \lambda = \arg\max_{\lambda\in\young{d}{n}} \sum_{\mu\in\lambda+\square \setminus \young{d}{n+1}} {\mathrm{hook}(\lambda)\over \mathrm{hook}(\mu)}.
\end{align}
We define $f(\lambda)$ for $\lambda\in\young{d}{n}$ by
\begin{align}
    f(\lambda) \coloneqq \sum_{\mu\in\lambda+\square \setminus \young{d}{n+1}} {\mathrm{hook}(\lambda)\over \mathrm{hook}(\mu)},
\end{align}
and derive the maximum value of $f(\lambda)$. We denote the number of boxes in the $i$-th row of $\lambda$ and $\mu$ by $\lambda_i$ and $\mu_i$, respectively. By definition of Young diagrams, $\lambda_i$ satisfies $\lambda_1\geq \cdots \lambda_d \geq 0$ and $\sum_i \lambda_i = n$. If $\lambda_d = 0$, any Young diagram in the set $\lambda+\square$ has depth smaller than or equal to $d$, i.e., $f(\lambda) = 0$. If $\lambda_1\geq \cdots \lambda_d \geq 1$ holds, the set $\lambda+\square\setminus \young{d}{n+1}$ is a one-point set whose element $\mu$ is given by $\mu_i = \lambda_i$ for $i\in\{1, \ldots, d\}$ and $\mu_{d+1} = 1$.
Therefore, $f(\lambda)$ is given by
\begin{align}
    f(\lambda) = {\text{hook}(\lambda) \over \text{hook}(\mu)} = \prod_{i=1}^{d} {\lambda_i+d-i \over \lambda_i+d+1-i}.
\end{align}
We derive the maximum value of $f(\lambda)$ for $\lambda_1, \ldots, \lambda_d$ such that $\lambda_1\geq \cdots \geq \lambda_d\geq 1$ and $\sum_i \lambda_i=n$. We show that $\lambda$ giving the maximum value of $f(\lambda)$ should satisfy $|\lambda_{i_1}-\lambda_{i_2}|\leq 1$ for any $1\leq i_1 < i_2\leq d$ by contradiction. If there exists $1\leq i_1 < i_2 \leq d$ such that $\lambda_{i_1}\geq \lambda_{i_2} +2$, at least one of the following conditions holds:
\begin{enumerate}
    \item There exists $i'_1 \in \{1, \ldots, d\}$ such that $\lambda_{i'_1}\geq \lambda_{i'_1+1}+2$\label{eq:young_diagram_cond1}
    \item There exists $i'_1, i'_2\in \{1, \ldots, d\}$ such that $i'_1+1\leq i'_2-1$, $\lambda_{i'_1}\geq \lambda_{i'_1+1}+1$ and $\lambda_{i'_2-1}\geq \lambda_{i'_2}+1$
\end{enumerate}
If the first condition holds, we define $i'_2\coloneqq i'_1+1$. Then, $\kappa_1, \ldots, \kappa_d$ defined by
\begin{align}
    \kappa_i=
    \begin{cases}
        \kappa_i & (i\neq i'_1, i'_2)\\
        \kappa_{i'_1}-1 & (i=i'_1)\\
        \kappa_{i'_2}+1 & (i=i'_2)
    \end{cases}
\end{align}
satisfies $\kappa_1\geq \cdots \geq \kappa_d$, $\sum_i \kappa_i = n$, and
\begin{align}
    \frac{f(\kappa)}{f(\lambda)}=\frac{\lambda_{i'_1}+d-i'_i-1}{\lambda_{i'_1}+d-i'_1+1}\frac{\lambda_{i'_2}+d-i'_2+2}{\lambda_{i'_2}+d-i'_2}>1,
\end{align}
i.e., $\lambda$ cannot give the maximum value of $f(\lambda)$.
Therefore, the Young diagram $\lambda$ giving the maximal value of $f(\lambda)$ should satisfy $|\lambda_{i_1}-\lambda_{i_2}|\leq 1$ for any $1\leq i_1 < i_2\leq d$. Such a Young diagram is uniquely determined as
\begin{align}
\label{eq:def_lambda}
    \lambda_i=
    \begin{cases}
        k+1 & (i\in\{1, \ldots, l\})\\
        k & (i\in \{l+1, \ldots, d\})
    \end{cases},
\end{align}
where $k\in \mathbb{Z}$ and $l\in \{0, \ldots, d-1\}$ are defined by $n=kd+l$. Then, we obtain
\begin{align}
    \alpha_\mathrm{opt}^{(\mathrm{PAR})} = 1-\max_{\lambda\in\young{d}{n}} f(\lambda) = {1 \over d+k+1}\left(d+{d-l \over d+k-l}\right).
\end{align}

\subsection{Proof of Theorem \ref{thm:optimal_parallel_isometry_adjointation}: Asymptotically optimal parallel protocol for isometry adjointation}
\label{appendix_sec:optimal_parallel_isometry_adjointation}
For $d=2$, Refs.~\cite{chiribella2004efficient,chiribella2005optimal} present the maximum-likelihood qubit-unitary estimation protocol achieving the entanglement fidelity $F_\mathrm{est} = 1- O(n^{-2})$ using $n$ calls of input unitary operation.  Although Refs.~\cite{chiribella2004efficient,chiribella2005optimal} do not present the explicit form of the probe state, one can utilize the resource state for entanglement-assisted alignment of the reference frames presented in Ref.~\cite{bagan2004entanglement} to achieve the same asymptotic scaling of the entanglement fidelity $F_\mathrm{est} = 1- O(n^{-2})$. The resource state presented in Ref.~\cite{bagan2004entanglement} is given by
\begin{align}
    \ket{\phi} = {2\over \sqrt{n+3}} \sum_{j=0(1/2)}^{n/2} {1\over \sqrt{2j+1}}\sin{(2j+1)\pi \over n+3} \sum_{m=-j}^{j}\ket{jm}_{\mcI^n} \ket{jm}_{\mcA},
\end{align}
where $j$ and $m$ represents total angular momentum and $z$-component of total angular momentum of a $n$-qubit system. The summation of $j$ starts from $0$ if $n$ is even and $1/2$ if $n$ is odd.  These values correspond to the Schur basis, where $j$ corresponds to the Young diagram $\lambda$ whose number of boxes in $i$-th row, denoted by $\lambda_i$ ($i=1,2$), is determined by $j=(\lambda_1-\lambda_2)/2$, and $m$ corresponds to an element in $\mcU_\lambda^{(2)}$.
Then, the value $\alpha_\phi$ shown in Theorem \ref{thm:parallel_isometry_adjointation} is calculated as
\begin{align}
    \alpha_\phi
    &= {4 \over n+3}\sum_{j=0(1/2)}^{n/2} \sin^2{(2j+1)\pi \over n+3} \left[1-{({n\over 2} + j +1)({n\over 2}-j) \over ({n\over 2}+j+2)({n\over 2}-j+1)}\right].\label{eq:alpha_qubit}
\end{align}
To evaluate $\alpha_\phi$ in the asymptotic limit $n\to \infty$, we introduce the variable $x={2j\over n}$ and approximate the sum in Eq.~(\ref{eq:alpha_qubit}) by the integral over $x$ as
\begin{align}
    \alpha_\phi
    &= {4\over n+3} \int_{0}^{1} \dd x {n\over 2} \sin^2 {(x+{1\over n})\pi \over 1+{3\over n}} \left[1-{(1+x+{2\over n})(1-x) \over (1+x+{4\over n})(1-x+{2\over n})}\right] +O(n^{-2})\\
    &= {8\over n} \int_{0}^{1} \dd x {\sin^2(x\pi) \over 1-x^2} + O(n^{-2})\\
    &\approx \frac{6.2287}{n}+O(n^{-2}),
\end{align}
where the integral is evaluated by \textsc{Mathematica} \cite{mathematica}. Thus, utilizing this estimation method to construct the parallel isometry adjointation protocol, one can achieve
\begin{align}
    \epsilon = \max\{\alpha_\phi, 1-F_\mathrm{est}\} = {6.2287 \over n} + O(n^{-2}).
\end{align}

For a higher-dimension $d>2$, Ref.~\cite{yang2020optimal} presents an asymptotically optimal unitary estimation protocol achieving the entanglement fidelity $F_\mathrm{est} = 1-O(d^4 n^{-2})$.  This protocol utilizes the probe state given by
\begin{align}
    \ket{\phi} = \sum_{\lambda\in \mathrm{S}_\mathrm{young}} \sqrt{{q_\lambda \over d_\lambda m_\lambda^{(d)}}} \sum_{i=1}^{d_\lambda} \sum_{u=1}^{m_\lambda^{(d)}} (\ket{\lambda,i}_{\mcS_\lambda} \otimes \ket{\lambda,u}_{\mcU_\lambda^{(d)}})_{\mcI^n} \otimes (\ket{\lambda,i}_{\mcS_\lambda} \otimes \ket{\lambda,u}_{\mcU_\lambda^{(d)}})_{\mcA},
\end{align}
where $\{q_\lambda\}$ is a probability distribution over the set $\mathrm{S}_{\mathrm{young}} \subset \young{d}{n}$ defined by
\begin{multline}
    \mathrm{S}_{\mathrm{young}} \coloneqq \{\lambda\in\young{d}{n}| \lambda_i = \mu_{0,i} + N(2d-3)+1-(N+1)(i-1)+\tilde{\lambda}_i \\
    \forall i\leq d-1, \exists \tilde{\lambda}\in\{0,\cdots,N-1\}^{d-1}\}.
\end{multline}
Here, $\lambda_i$ represents the number of boxes in the $i$-th row of $\lambda$, $N$ is defined by $N\coloneqq \lfloor {1\over 3d-2}({2n\over d-1}+d-2)\rfloor$, $\mu_{0,i}$ is defined by $\mu_{0,i}\coloneqq \lfloor {n_0 \over d}\rfloor + 1$ for $i\in \{1, \ldots, n_0-\lfloor {n_0 \over d}\rfloor d\}$ and $\mu_{0,i}\coloneqq \lfloor {n_0 \over d}\rfloor$ for $i\in \{n_0-\lfloor {n_0 \over d}\rfloor d+1, \ldots, d\}$, and $n_0$ is defined by $n_0\coloneqq n-{((3d-2)N-d+2)(d-1) \over 2}$. Then, the value $\alpha_\phi$ shown in Theorem \ref{thm:parallel_isometry_adjointation} is evaluated as
\begin{align}
    \alpha_\phi
    &= \sum_{\lambda\in\mathrm{S}_\mathrm{young}} q_\lambda \left[1-\sum_{\mu\in\lambda+\square\setminus\young{d}{n+1}} {\mathrm{hook}(\lambda)\over\mathrm{hook}(\mu)}\right]\\
    &= \sum_{\lambda\in\mathrm{S}_\mathrm{young}} q_\lambda \left[1-\prod_{i=1}^{d} {\lambda_i+d-i \over \lambda_i+d-i+1}\right].
\end{align}
We evaluate $\alpha_\phi$ in the asymptotic limit $n\to\infty$.  In this region, $\lambda_i$ is given by
\begin{align}
    \lambda_i = {4n\over 3d} - {2n\over 3d^2}(i-1) + O(d, d^{-2}n).
\end{align}
Thus, $\alpha_\phi$ is evaluated as follows.
\begin{align}
    \alpha_\phi
    &= 1- \prod_{i=1}^{d}\left[1-{3d\over 2n}{1+O(d^2n^{-1}, d^{-1})\over 2 - {i-1\over d}}\right]\\
    &= {3d\over 2n}\sum_{i=1}^{d} {1\over 2-{i-1\over d}} +O(d^4n^{-2}, dn^{-1}).\\
    &= {3d^2\over 2n}\int_{0}^{1}{\dd x \over 2-x} + O(d^4n^{-2}, dn^{-1})\\
    &= {3\ln 2\over 2} {d^2\over n} + O(d^4n^{-2}, dn^{-1}).
\end{align}
Thus, utilizing this estimation method to construct the parallel isometry adjointation protocol, one can achieve
\begin{align}
    \epsilon = \max\{\alpha_\phi, 1-F_\mathrm{est}\} = {3\ln 2 \over 2}{d^2\over n} + O(d^4n^{-2}, dn^{-1}).
\end{align}

\section{Proof of Theorem \ref{thm:optimal_construction}: Optimal construction of isometry inversion, universal error detection, and isometry adjointation protocols}
\label{appendix_sec:proof_optimal_construction}
We show Theorem \ref{thm:optimal_construction} by constructing the parallel or sequential protocol for isometry inversion, universal error detection, and isometry adjointation in the forms shown in Theorems \ref{thm:sequential_isometry_adjointation} and \ref{thm:parallel_isometry_adjointation} and Corollaries \ref{cor:isometry_inversion} and \ref{cor:universal_error_detection} achieving the optimal performances among parallel or sequential protocols.  To this end, we utilize Theorem \ref{thm:unitary_group_symmetry} to write down the $\U(d)\times \U(D)$ symmetric Choi operators of parallel or sequential protocols achieving the optimal performances.  Then we construct the protocols in the forms shown in Theorems \ref{thm:sequential_isometry_adjointation} and \ref{thm:parallel_isometry_adjointation} and Corollaries \ref{cor:isometry_inversion} and \ref{cor:universal_error_detection} by constructing the corresponding Choi operators.

Suppose $\mcI_1 = \cdots = \mcI_n = \mcF = \CC^d$, $\mcO_1 = \cdots = \mcO_n = \mcP = \CC^D$, $\mcO'_1 = \cdots = \mcO'_n = \mcP'= \CC^d$, and we define the joint Hilbert spaces $\mcI^n\coloneqq \bigotimes_{i=1}^{n} \mcI_i$, $\mcO^n\coloneqq \bigotimes_{i=1}^{n} \mcO_i$, and $\mcO'^n\coloneqq \bigotimes_{i=1}^{n} \mcO'_i$.  We assume that 
\begin{align}
    C\in
    \begin{cases}
        \mcL(\mcI^n\otimes \mcO^n\otimes \mcP\otimes\mcF) & (\mathrm{isometry\;inversion}, \mathrm{isometry\;adjointation})\\
        \mcL(\mcI^n\otimes \mcO^n\otimes \mcP) & (\mathrm{universal\;error\;detection})
    \end{cases}
\end{align}
is the Choi operator satisfying the $\U(d)\times\U(D)$ symmetry (\ref{eq:C_sudsuDsymmetry}), which can be written as Eqs.~(\ref{eq:c_yy_basis1}) or (\ref{eq:c_yy_basis2}) using $\{C_{\mu\nu}\}$ or $\{C_{\lambda\nu}\}$.  If $C$ is the Choi operator of parallel or sequential protocol, i.e., $C\in\mcW^{(x)}$ for $x\in\{\mathrm{PAR}, \mathrm{SEQ}\}$, the Choi operator
\begin{align}
    C'\in
    \begin{cases}
        \mcL(\mcI^n\otimes \mcO'^n\otimes \mcP'\otimes\mcF) & (\mathrm{isometry\;inversion}, \mathrm{isometry\;adjointation})\\
        \mcL(\mcI^n\otimes \mcO'^n\otimes \mcP') & (\mathrm{universal\;error\;detection})
    \end{cases}
\end{align}
defined by
\begin{align}
    C' \coloneqq \sum_{\mu\in \young{d}{n+1}} \sum_{\nu\in \young{d}{n+1}}\sum_{i,j=1}^{d_\mu} \sum_{k,l=1}^{d_\nu} \frac{[C_{\mu\nu}]_{ik, jl}}{m_\mu^{(d)} m_\nu^{(D)}} (E_{ij}^{\mu, d})_{\mcI^n\mcF} \otimes (E_{kl}^{\nu, d})_{\mcP'\mcO'^n}\label{eq:c'_yybasis1}
\end{align}
for isometry inversion or isometry adjointation, and
\begin{align}
    C' \coloneqq \sum_{\lambda\in \young{d}{n}} \sum_{\nu\in \young{d}{n+1}}\sum_{a,b=1}^{d_\lambda} \sum_{k,l=1}^{d_\nu} \frac{[C_{\lambda\nu}]_{ak,bl}}{m_\lambda^{(d)} m_\nu^{(D)}} (E_{ab}^{\lambda, d})_{\mcI^n} \otimes (E_{kl}^{\nu, d})_{\mcP'\mcO'^n}\label{eq:c'_yybasis2}
\end{align}
for universal error detection, satisfies $C'\in \mcW^{(x)}$ [see Eqs.~(\ref{eq:symmetric_parallel_condition}) and (\ref{eq:symmetric_sequential_condition})].  We utilize this fact to show the construction of the protocols in the forms shown in Theorems \ref{thm:sequential_isometry_adjointation} and \ref{thm:parallel_isometry_adjointation} and Corollaries \ref{cor:isometry_inversion} and \ref{cor:universal_error_detection}.

\subsection{Probabilistic exact isometry inversion}
Assume that $\{C_S, C_F\}\subset \mcL(\mcI^n\otimes\mcO^n\otimes\mcP\otimes\mcF)$ is the $\U(d)\times\U(D)$ symmetric Choi operator of parallel ($x=\mathrm{PAR}$) or sequential ($x=\mathrm{SEQ}$) protocol achieving the optimal success probability $p_\mathrm{opt}^{(x)}(d,D,n)$ of isometry inversion, which can be written as Eqs.~(\ref{eq:cs_yybasis}) and (\ref{eq:cf_yybasis}).  Defining $\{C'_S, C'_F\} \subset \mcL(\mcI^n\otimes\mcO'^n\otimes\mcP'\otimes\mcF)$ by
\begin{align}
    C'_S &= \sum_{\mu\in \young{d}{n+1}} \sum_{\nu\in\young{d}{n+1}}\sum_{i,j=1}^{d_\mu} \sum_{k,l=1}^{d_\nu} \frac{[S_{\mu\nu}]_{ik, jl}}{m_\mu^{(d)} m_\nu^{(D)}} (E_{ij}^{\mu, d})_{\mcI^n\mcF} \otimes (E_{kl}^{\nu, d})_{\mcP'\mcO'^n},\\
    C'_F &= \sum_{\mu\in \young{d}{n+1}} \sum_{\nu\in\young{d}{n+1}}\sum_{i,j=1}^{d_\mu} \sum_{k,l=1}^{d_\nu} \frac{[F_{\mu\nu}]_{ik, jl}}{m_\mu^{(d)} m_\nu^{(D)}} (E_{ij}^{\mu, d})_{\mcI^n\mcF} \otimes (E_{kl}^{\nu, d})_{\mcP'\mcO'^n},
\end{align}
$C'_S, C'_F\geq 0$ and $C'\coloneqq C'_S+C'_F\in\mcW^{(x)}$ holds. Thus, the corresponding supermap $\{\supermap{C}_S, \supermap{C}_F\}$ can be implemented in a parallel ($x=\mathrm{PAR}$) or sequential ($x=\mathrm{SEQ}$) protocol. In particular for the case $x=\mathrm{PAR}$, since its Choi operator satisfies the $\U(d)\times\U(d)$ symmetry, it can be implemented using a delayed input-state protocol \cite{quintino2019reversing,quintino2019probabilistic}. Since $\{C_S, C_F\}$ achieves the optimal success probability $p_\mathrm{opt}^{(x)}(d,D,n)$, it satisfies [see Eq.~(\ref{eq:sdp_probabilistic_isometry_inversion_comp_basis})]
\begin{align}
    \Tr(C_S \Omega) &= p_\mathrm{opt}^{(x)}(d,D,n),\\
    \Tr(C_S \Omega) &= \Tr[C_S(\Xi\otimes \1_\mcF)].
\end{align}
Defining $\Omega'$ and $\Xi'$ by replacing $\mcP\mcO^n$ and $D$ in Eqs.~(\ref{eq:Omega_yybasis}) and (\ref{eq:Xi_yybasis}) with $\mcP'\mcO'^n$ and $d$, respectively, $\{C'_S, C'_F\}$ satisfies
\begin{align}
    \Tr(C'_S \Omega') &= p_\mathrm{opt}^{(x)}(d,D,n),\\
    \Tr(C'_S \Omega') &= \Tr[C'_S(\Xi'\otimes \1_\mcF)],
\end{align}
which implies $\{C'_S, C'_F\}$ implements probabilistic exact $d$-dimensional unitary inversion with success probability $p_\mathrm{opt}^{(x)}(d,D,n)$.  Therefore, we can construct a parallel or sequential protocol achieving the optimal success probability $p_\mathrm{opt}^{(x)}(d,D,n)$ of isometry inversion using the construction shown in Theorem \ref{cor:isometry_inversion}.

\subsection{Deterministic isometry inversion}
\label{appendix_sec:proof_optimal_construction_deterministic_isometry_inversion}
Assume that $C\in \mcL(\mcI^n\otimes\mcO^n\otimes\mcP\otimes\mcF)$ is the $\U(d)\times\U(D)$ symmetric Choi operator of parallel ($x=\mathrm{PAR}$) or sequential ($x=\mathrm{SEQ}$) protocol achieving the optimal worst-case channel fidelity $F_\mathrm{opt}^{(x)}(d,D,n)$, which can be written as Eq.~(\ref{eq:c_yy_basis1}).  Defining $C'\in\mcL(\mcI^n\otimes\mcO'^n\otimes\mcP'\otimes\mcF)$ by Eq.~(\ref{eq:c'_yybasis1}), $C'\geq 0$ and $C'\in\mcW^{(x)}$ holds. Thus, the corresponding supermap $\supermap{C}$ can be implemented in a parallel or sequential protocol. In particular for the case $x=\mathrm{PAR}$, since its Choi operator satisfies the $\U(d)\times\U(d)$ symmetry, it can be implemented using an estimation-based protocol \cite{quintino2022deterministic}. Since $C$ achieves the optimal worst-case channel fidelity $F_\mathrm{opt}^{(x)}(d,D,n)$, it satisfies [see Eq.~(\ref{eq:sdp_deterministic_isometry_inversion_comp_basis})]
\begin{align}
    \Tr(C \Omega) &= F_\mathrm{opt}^{(x)}(d,D,n).
\end{align}
Defining $\Omega'$ by replacing $\mcP\mcO^n$ and $D$ in Eq.~(\ref{eq:Omega_yybasis}) with $\mcP'\mcO'^n$ and $d$, respectively, $C'$ satisfies
\begin{align}
    \Tr(C' \Omega') &= F_\mathrm{opt}^{(x)}(d,D,n),
\end{align}
which implies $C'$ implements deterministic $d$-dimensional unitary inversion with worst-case channel fidelity $p_\mathrm{opt}^{(x)}(d,D,n)$.  Therefore, we can construct a parallel or sequential protocol achieving the optimal worst-case channel fidelity $F_\mathrm{opt}^{(x)}(d,D,n)$ of isometry inversion using the construction shown in Theorem \ref{cor:isometry_inversion}.

\subsection{Universal error detection}
\label{appendix_sec:proof_optimal_construction_universal_error_detection}
Assume that $\{C_I, C_O\} \subset\mcL(\mcI^n \otimes \mcO^n \otimes \mcP)$ is the $\U(d)\times\U(D)$ symmetric Choi operator of a parallel ($x=\mathrm{PAR}$) or sequential ($x=\mathrm{SEQ}$) protocol achieving the optimal error $\alpha_\mathrm{opt}^{(x)}(d,D,n)$, which can be written as Eqs.~(\ref{eq:ci_yybasis}) and (\ref{eq:co_yybasis}). Defining $C'' \in\mcL(\mcI^n \otimes \mcO'^n \otimes \mcP')$ by Eq.~(\ref{eq:c'_yybasis2}) for $C_{\lambda\nu}\coloneqq I_{\lambda\nu}+O_{\lambda\nu}$ for $\lambda\in\young{d}{n}$ and $\nu\in\young{d}{n+1}$, $C''\geq 0$ and $C''\in\mcW^{(x)}$ hold.  Characterization of $\mcW^{(x)}$ shown in Eqs.~(\ref{eq:parallel_condition}) and (\ref{eq:sequential_condition}) for the case $\mcF=\CC$ (no global future) and the $\U(d)\times\U(d)$ symmetry of $C''$ implies that $C''$ can be written as
\begin{align}
    C'' =
    \begin{cases}
        \phi_{\mcI^n}\otimes \1_{\mcP'\mcO'^n} & (x=\mathrm{PAR}),\\
        C' \otimes \1_{\mcO_n} & (x=\mathrm{SEQ}),
    \end{cases}\label{eq:c''_decomp}
\end{align}
where $\phi\in\mcL(\mcI^n)$ is a quantum state and $C'\in\mcL(\mcI^n \otimes \mcO'^{n-1} \otimes \mcP')$ is the Choi operator of a $(n-1)$-slot sequential protocol.  Since $\{C_I,C_O\}$ achieves the optimal error $\alpha_\mathrm{opt}^{(x)}(d,D,n)$, it satisfies [see Eq.(\ref{eq:sdp_universal_error_detection_comp_basis})]
\begin{align}
    \Tr(C_I\Sigma) &= \alpha_\mathrm{opt}^{(x)}(d,D,n),\\
    \Tr(C_I\Xi) &= 1.\label{eq:cond_cixi}
\end{align}
Since Eq.~(\ref{eq:cond_cixi}) corresponds to the condition 
\begin{align}
    \Tr\;\supermap{C}_I (\map{V}_\mathrm{in}^{\otimes n})(\rho_\mathrm{in}) = 1 \quad \forall \rho_\mathrm{in}\in\mcL(\Im V_\mathrm{in}),
\end{align}
it can be replaced with the equivalent condition
\begin{align}
    \Tr\;\supermap{C}_O (\map{V}_\mathrm{in}^{\otimes n})(\rho_\mathrm{in}) = 0 \quad \forall \rho_\mathrm{in}\in\mcL(\Im V_\mathrm{in}),
\end{align}
which can be represented as
\begin{align}
    \Tr(C_O \Xi) = 0.
\end{align}
Therefore, we obtain [see Eq.~(\ref{eq:Xi_yybasis})]
\begin{align}
    0 &= \sum_{\lambda\in\young{d}{n}}\sum_{\nu\in\young{d}{n+1}} \Tr(O_{\lambda\nu} \Xi_{\lambda\nu})\\
    &=\sum_{\nu\in\young{d}{n+1}}\sum_{\lambda\in\nu-\square}\sum_{a,b=1}^{d_\lambda} \sum_{k,l=1}^{d_\nu} {m_\nu^{(d)}\over d m_\nu^{(D)} m_\lambda^{(d)}}  [\pi_\nu]_{a_\nu^{\lambda} k}^* [\pi_\nu]_{lb_\nu^\lambda} [O_{\lambda\nu}]_{ba, lk},
\end{align}
which leads to
\begin{align}
    \sum_{a,b=1}^{d_\lambda} \sum_{k,l=1}^{d_\nu} [\pi_\nu]_{a_\nu^{\lambda} k}^* [\pi_\nu]_{lb_\nu^\lambda} [O_{\lambda\nu}]_{ba, lk} &=0 \quad \forall \nu\in\young{d}{n+1}, \lambda\in\nu-\square.
\end{align}
Thus, we obtain [see Eq.~(\ref{eq:Sigma_yybasis})]
\begin{align}
    \Tr(O_{\lambda\nu} \Sigma_{\lambda\nu}) = 0 \quad \forall \nu\in\young{d}{n+1}, \lambda\in\nu-\square.
\end{align}
Since Eq.~(\ref{eq:Sigma_yybasis}) reduces to $\Sigma$ defined in Eq.~(\ref{eq:def_Sigma}) by replacing $\mcP\mcO^n$ and $D$ with $\mcP'\mcO'^n$ and $d$,
\begin{align}
    \Tr(C''\Sigma)
    &= \sum_{\nu\in\young{d}{n+1}}\sum_{\lambda\in\nu-\square} \Tr(C_{\lambda\nu}\Sigma_{\lambda\nu})\\
    &= \sum_{\nu\in\young{d}{n+1}}\sum_{\lambda\in\nu-\square} \Tr(I_{\lambda\nu}\Sigma_{\lambda\nu})\\
    &\leq \Tr(C_I \Sigma)\\
    &= \alpha_\mathrm{opt}^{(x)}\label{eq:c''sigma''}
\end{align}
holds.
Substituting Eq.~(\ref{eq:c''_decomp}) for $x=\mathrm{PAR}, \mathrm{SEQ}$ to Eq.~(\ref{eq:c''sigma''}), we obtain
\begin{align}
    \alpha_\mathrm{opt}^{(\mathrm{PAR})}
    &\geq \Tr(\phi \Tr_{\mcP'\mcO'^n}(\Sigma))\\
    &= \sum_{\lambda\in\young{d}{n}}\sum_{\nu\in\lambda+_d\square}\Tr(\phi \Pi_\lambda^{(d)}) {\mathrm{hook}(\lambda) \over \mathrm{hook}(\nu)}\\
    &=\alpha_\phi,\\
    \alpha_\mathrm{opt}^{(\mathrm{SEQ})}
    &\geq \Tr(C' \Tr_{\mcO'_n}(\Sigma))\\
    &= \Tr(C'\Sigma')\\
    &= \alpha_{C'},
\end{align}
where $\alpha_\phi$ and $\alpha_{C'}$ are defined in Theorem \ref{cor:universal_error_detection}, respectively.
Thus, the parallel and sequential protocols constructed in Theorem \ref{cor:universal_error_detection} achieves the optimal error $\alpha_\mathrm{opt}^{(x)}(d,D,n)$.

\subsection{Isometry adjointation}
Assume that $\{C_I, C_O\}\subset \mcL(\mcI^n\otimes\mcO^n\otimes\mcP\otimes\mcF)$ is the $\U(d)\times\U(D)$ symmetric Choi operator of parallel ($x=\mathrm{PAR}$) or sequential ($x=\mathrm{SEQ}$) protocol achieving the optimal worst-case diamond-norm error $\epsilon_\mathrm{opt}^{(x)}(d,D,n)$ of isometry adjointation, which can be written as Eqs.~(\ref{eq:ci_yybasis2}) and (\ref{eq:co_yybasis2}).  Defining $C' \in \mcL(\mcI^n\otimes\mcO'^n\otimes\mcP'\otimes\mcF)$ by Eq.~(\ref{eq:c'_yybasis1}) for $C_{\mu\nu}\coloneqq I_{\mu\nu}+O_{\mu\nu}$ for $\mu, \nu\in\young{d}{n+1}$, $C'\geq 0$ and $C'\in\mcW^{(x)}$ hold.  Thus, the corresponding supermap $\supermap{C}''$ can be implemented by a parallel ($x=\mathrm{PAR}$) or sequential ($x=\mathrm{SEQ}$) protocol. In particular, for the case $x=\mathrm{PAR}$, since its Choi operator satisfies the $\U(d)\times\U(d)$ symmetry, it can be implemented using a covariant-estimation-based protocol as shown in Fig.~\ref{fig:deterministic_isometry_adjointation_parallel} (a-1) \cite{yang2020optimal,quintino2022deterministic}. Since $\{C_I, C_O\}$ achieves the optimal worst-case diamond-norm error $\epsilon_\mathrm{opt}^{(x)}(d,D,n)$ of isometry adjointation, it satisfies [see Eq.~(\ref{eq:sdp_isometry_adjointation_comp_basis})]
\begin{align}
    1-\Tr(C_I\Omega)&\leq \epsilon_\mathrm{opt}^{(x)}(d,D,n),\\
    \Tr[C_I(\Sigma\otimes \1_{\mcF})] &\leq \epsilon_\mathrm{opt}^{(x)}(d,D,n),\\
    \Tr[C_I(\Xi\otimes\1_{\mcF})]&=1.
\end{align}
Similarly to Sections \ref{appendix_sec:proof_optimal_construction_deterministic_isometry_inversion} and \ref{appendix_sec:proof_optimal_construction_universal_error_detection}, defining $\Omega'$ by replacing $\mcP\mcO^n$ and $D$ in Eq.~(\ref{eq:Omega_yybasis}) with $\mcP'\mcO'^n$ and $d$, we obtain
\begin{align}
    1-\Tr(C'\Omega') &\leq  \epsilon_\mathrm{opt}^{(x)}(d,D,n),\\
    \Tr[C'(\Sigma\otimes\1_{\mcF})] &\leq \epsilon_\mathrm{opt}^{(x)}(d,D,n).
\end{align}

For the case $x=\mathrm{PAR}$, $C'$ can be implemented by a covariant unitary-estimation protocol achieving the average fidelity $F_\mathrm{est} = \Tr(C'\Omega')$ \cite{yang2020optimal,quintino2022deterministic}.  $\Tr[C'(\Sigma\otimes\1_{\mcF})]$ can be evaluated by the probe state $\phi\in\mcL(\mcI^n\otimes \mcA)$ of the unitary estimation protocol as
\begin{align}
    \Tr[C'(\Sigma\otimes\1_{\mcF})]
    &= \Tr[\Tr_\mcF(C) \Sigma]\\
    &= \Tr(\Tr_\mcA(\phi) \Tr_{\mcP'\mcO'^n}(\Sigma))\\
    &= \sum_{\lambda\in\young{d}{n}} \Tr[\Tr_\mcA(\phi) \Pi_\lambda^{(d)}]\left[1-\sum_{\nu\in\lambda+\square\setminus\young{d}{n+1}}{\mathrm{hook}(\lambda) \over \mathrm{hook}(\nu)}\right]\\
    &=\alpha_\phi,
\end{align}
where $\alpha_\phi$ is defined in Theorem \ref{thm:parallel_isometry_adjointation}.
Thus, the parallel protocol constructed in Theorem \ref{thm:parallel_isometry_adjointation} achieves the optimal worst-case diamond-norm error $\epsilon_\mathrm{opt}^{(\mathrm{PAR})}(d,D,n)$.

For the case $x=\mathrm{SEQ}$, $\Tr(C'\Omega')$ represents the worst-case channel fidelity of $d$-dimensional unitary inversion, and $\Tr[C'(\Sigma\otimes\1_{\mcF})]$ can be evaluated as [see Eq.~(\ref{eq:sequential_condition})]
\begin{align}
    \Tr[C'(\Sigma\otimes\1_{\mcF})]
    &= \Tr[\Tr_\mcF(C')\Sigma]\\
    &= \alpha_{C'},
\end{align}
where $\alpha_{C'}$ is defined in Theorem \ref{thm:sequential_isometry_adjointation}.  Thus, the sequential protocol constructed in Theorem \ref{thm:sequential_isometry_adjointation} achieves the optimal worst-case diamond-norm error $\epsilon_\mathrm{opt}^{(\mathrm{SEQ})}(d,D,n)$.

\section{Numerical results}
\label{appendix_sec:numerical_results}
We show the numerical results on the optimal performances of probabilistic exact isometry inversion, deterministic isometry inversion, universal error detection, and isometry adjointation using $n$ calls of an input isometry operation $V_\mathrm{in}\in\isometry{d}{D}$ with parallel, sequential, or general protocols including indefinite causal order in Tables \ref{tab:probabilistic_isometry_inversion}, \ref{tab:deterministic_isometry_inversion}, \ref{tab:universal_error_detection} and \ref{tab:isometry_adjointation}.  
The obtained values are compatible with the previous works \cite{quintino2019probabilistic, quintino2022deterministic, yoshida2023universal, yoshida2023reversing}, where Ref.~\cite{quintino2019probabilistic} shows the maximum success probability of unitary inversion for the cases of $d=2, n\leq 3$ and $d=3, n\leq 2$, Ref.~\cite{quintino2022deterministic} shows the maximum channel fidelity of unitary inversion for the cases of $d=2, n\leq 3$ and $d=3, n\leq 2$, Ref.~\cite{quintino2019probabilistic} shows the maximum success probability of isometry inversion for the cases of $d=2, n\leq 3$ and $d=3, n\leq 2$, and Ref.~\cite{yoshida2023reversing} shows the maximum channel fidelity of unitary inversion for parallel and sequential protocols for the cases of $d\leq 6$ and $n\leq 5$.

\begin{table}[htb]
\centering
    \begin{tabular}{c|ccc|ccc|ccc}\hline\hline
        \multirow{2}{*}{$p_\mathrm{opt}^{(x)}$} & \multicolumn{3}{c|}{Parallel ($x=\mathrm{PAR}$)} & \multicolumn{3}{c|}{Sequential ($x = \mathrm{SEQ}$)} & \multicolumn{3}{c}{General ($x = \mathrm{GEN}$)}\\
         & $d=2$ & $d=3$ & $d=4$ & $d=2$ & $d=3$ & $d=4$ & $d=2$ & $d=3$ & $d=4$\\\hline
       $n=1$ & $0.2500$ & $0.0000$ & $0.0000$ & $0.2500$ & $0.0000$ & $0.0000$ & $0.2500$ & $0.0000$ & $0.0000$ \\
       $n=2$ & $0.4000$ & $0.1111$ & $0.0000$ & $0.4286$ & $0.1111$ & $0.0000$ & $0.4286$ & $0.1111$ & $0.0000$\\
       $n=3$ & $0.5000$ & $0.1385$ & $0.0625$ & $0.7500$  & $0.1861$ & $0.0625$ & $0.9415$ & $0.2093$ & $0.0625$ \\
       $n=4$ & $0.5715$ & $0.2000$ & $0.0708$ & $1.0000$ & $0.2674$ & $0.1064$ & $1.0000$ & $0.2915$ & $0.1419$\\
       $n=5$ & $0.6250$ & $0.2408$ & $0.0865$ & $1.0000$ & $0.4662$ & $0.1447$ & - & - & -\\\hline\hline
    \end{tabular}
\caption{The maximum success probability of isometry inversion using $n$ calls of an input isometry operation $V_\mathrm{in}\in\isometry{d}{D}$ ($D\geq d+1$) in parallel, sequential, and general protocols.}
\label{tab:probabilistic_isometry_inversion}
\end{table}

\begin{table}[htb]
\centering
    \begin{tabular}{c|ccc|ccc|ccc}\hline\hline
        \multirow{2}{*}{$F_\mathrm{opt}^{(x)}$} & \multicolumn{3}{c|}{Parallel ($x=\mathrm{PAR}$)} & \multicolumn{3}{c|}{Sequential ($x = \mathrm{SEQ}$)} & \multicolumn{3}{c}{General ($x = \mathrm{GEN}$)}\\
         & $d=2$ & $d=3$ & $d=4$ & $d=2$ & $d=3$ & $d=4$ & $d=2$ & $d=3$ & $d=4$\\\hline
       $n=1$ & $0.5000$ & $0.2222$ & $0.1250$ & $0.5000$ & $0.2222$ & $0.1250$ & $0.5000$ & $0.2222$ & $0.1250$ \\
       $n=2$ & $0.6545$ & $0.3333$ & $0.1875$ & $0.7500$ & $0.3333$ & $0.1875$ & $0.7500$ & $0.3333$ & $0.1875$\\
       $n=3$ & $0.7500$ & $0.4310$ & $0.2500$ & $0.9330$  & $0.4444$ & $0.2500$ & $0.9851$ & $0.4444$ & $0.2500$ \\
       $n=4$ & $0.8117$ & $0.5131$ & $0.3105$ & $1.0000$ & $0.5556$ & $0.3125$ & $1.0000$ & $0.5556$ & $0.3125$\\
       $n=5$ & $0.8536$ & $0.5810$ & $0.3675$ & $1.0000$ & $0.6667$ & $0.3750$ & - & - & -\\\hline\hline
    \end{tabular}
\caption{The maximum worst-case channel fidelity of isometry inversion using $n$ calls of an input isometry operation $V_\mathrm{in}\in\isometry{d}{D}$ ($D\geq d+1$) in parallel, sequential, and general protocols.}
\label{tab:deterministic_isometry_inversion}
\end{table}

\begin{table}[htb]
\centering
    \begin{tabular}{c|ccc|ccc|ccc}\hline\hline
        \multirow{2}{*}{$\alpha_\mathrm{opt}^{(x)}$} & \multicolumn{3}{c|}{Parallel ($x=\mathrm{PAR}$)} & \multicolumn{3}{c|}{Sequential ($x = \mathrm{SEQ}$)} & \multicolumn{3}{c}{General ($x = \mathrm{GEN}$)}\\
         & $d=2$ & $d=3$ & $d=4$ & $d=2$ & $d=3$ & $d=4$ & $d=2$ & $d=3$ & $d=4$\\\hline
       $n=1$ & \bm{$1$} & \bm{$1$} & \bm{$1$} & $1.0000$ & $1.0000$ & $1.0000$ & $1.0000$ & $1.0000$ & $1.0000$\\
       $n=2$ & \bm{$2/3$} & \bm{$1$} & \bm{$1$} & $0.6667$ & $1.0000$ & $1.0000$ & $0.6667$ & $1.0000$ & $1.0000$\\
       $n=3$ & \bm{$5/8$} & \bm{$3/4$} & \bm{$1$} & $0.5000$  & $0.7500$ & $1.0000$ & $0.4375$ & $0.7500$ & $1.0000$\\
       $n=4$ & \bm{$1/2$} & \bm{$11/15$} & \bm{$4/5$} & $0.4000$ & $0.6000$ & $0.8000$ & $0.3600$  & $0.4667$ & $0.8000$\\
       $n=5$ & \bm{$7/15$} & \bm{$7/10$} & \bm{$19/24$} & $0.3333$ & $0.5000$ & $0.6667$ & - & - & - \\\hline\hline
    \end{tabular}
\caption{The minimum approximation error of universal error detection using $n$ calls of an input isometry operation $V_\mathrm{in}\in\isometry{d}{D}$ ($D\geq d+1$) in parallel, sequential, and general protocols. Bold values are obtained analytically.}
\label{tab:universal_error_detection}
\end{table}

\begin{table}[htb]
\centering
    \begin{tabular}{c|ccc|ccc|ccc}\hline\hline
        \multirow{2}{*}{$\epsilon_\mathrm{opt}^{(x)}$} & \multicolumn{3}{c|}{Parallel ($x=\mathrm{PAR}$)} & \multicolumn{3}{c|}{Sequential ($x = \mathrm{SEQ}$)} & \multicolumn{3}{c}{General ($x = \mathrm{GEN}$)}\\
         & $d=2$ & $d=3$ & $d=4$ & $d=2$ & $d=3$ & $d=4$ & $d=2$ & $d=3$ & $d=4$\\\hline
       $n=1$ & $1.0000$ & $1.0000$ & $1.0000$ & $1.0000$ & $1.0000$ & $1.0000$ & $1.0000$ & $1.0000$ & $1.0000$ \\
       $n=2$ & $0.6736$ & $1.0000$ & $1.0000$ & $0.6667$ & $1.0000$ & $1.0000$ & $0.6667$ & $1.0000$ & $1.0000$\\
       $n=3$ & $0.6250$ & $0.7822$ & $1.0000$ & $0.5000$  & $0.7500$ & $1.0000$ & $0.5000$ & $0.7500$ & $1.0000$ \\
       $n=4$ & $0.5169$ & $0.7373$ & $0.8448$ & $0.4444$ & $0.6429$ & $0.8000$ & $0.4444$ & $0.6429$ & $0.8000$\\\hline\hline
    \end{tabular}
\caption{The minimum approximation error of isometry adjointation using $n$ calls of an input isometry operation $V_\mathrm{in}\in\isometry{d}{D}$ ($D\geq d+1$) in parallel, sequential, and general protocols.}
\label{tab:isometry_adjointation}
\end{table}

\begin{table}[htb]
\centering
    \begin{tabular}{c|cc|cc|cc}\hline\hline
        \multirow{2}{*}{$p_\mathrm{opt}^{(\mathrm{GEN})}$} & \multicolumn{2}{c|}{$d=2$} & \multicolumn{2}{c|}{$d=3$} & \multicolumn{2}{c}{$d=4$}\\
         & $D=2$ & $D=3$& $D=3$ & $D=4$ & $D=4$ & $D=5$\\\hline
         $n=2$ & $0.4444$ & $0.4286$ & $0.1111$ & $0.1111$ & $0.0000$ & $0.0000$\\
         $n=3$ & $0.9415$ & $0.9415$ & $0.3262$ & $0.2093$ & $0.0625$ & $0.0625$\\
         $n=4$ & $1.0000$ & $1.0000$ & $0.5427$ & $0.2915$ & $0.2609$ & $0.1419$ \\\hline\hline
    \end{tabular}
\caption{Comparison with the maximum success probabilities of unitary inversion ($D=d$) and isometry inversion ($D=d+1$) in general protocols.}
\label{tab:comparison_probabilistic_isometry_inversion}
\end{table}

\begin{table}[htb]
\centering
    \begin{tabular}{c|cc|cc|cc}\hline\hline
        \multirow{2}{*}{$F_\mathrm{opt}^{(\mathrm{GEN})}$} & \multicolumn{2}{c|}{$d=2$} & \multicolumn{2}{c|}{$d=3$} & \multicolumn{2}{c}{$d=4$}\\
         & $D=2$ & $D=3$& $D=3$ & $D=4$ & $D=4$ & $D=5$\\\hline
         $n=2$ & $0.8249$ & $0.7500$ & $0.3333$ & $0.3333$ & $0.1875$ & $0.1875$\\
         $n=3$ & $0.9921$ & $0.9851$ & $0.5835$ & $0.4444$ & $0.2500$ & $0.2500$\\
         $n=4$ & $1.0000$ & $1.0000$ & $0.7874$ & $0.5556$ & $0.4567$ & $0.3125$\\\hline\hline
    \end{tabular}
\caption{Comparison with the maximum worst-case channel fidelities of unitary inversion ($D=d$) and isometry inversion ($D=d+1$) in general protocols.}
\label{tab:comparison_deterministic_isometry_inversion}
\end{table}

\clearpage

\section{Semidefinite programming to obtain the optimal transformations of isometry operations with the $\U(d)\times \U(D)$ symmetry}
\label{appendix_sec:sdp}
As shown in Theorem \ref{thm:unitary_group_symmetry}, the optimal protocols for isometry inversion, universal error detection and isometry adjointation can be searched within the Choi operators having the $\U(d)\times \U(D)$ symmetry.  To utilize this symmetry for a numerical search of the optimal protocols, we derive the characterization of the quantum superchannels and the conditions for isometry inversion, universal error detection, and isometry adjointation under the $\U(d)\times \U(D)$ symmetry.  Then, we derive the SDPs giving the optimal performances of these tasks, which are shown below.

\subsection{Choi representation of general superchannels and $\U(d)\times \U(D)$ symmetry}
\label{appendix_sec:general_superinstrument_choi}
As shown in Section \ref{sec:unitary_group_symmetry}, the Choi operator $C$ of a quantum superchannel $\supermap{C}: \bigotimes_{i=1}^{n} [\mcL(\mcI_i) \to \mcL(\mcO_i)] \to [\mcL(\mcP) \to \mcL(\mcF)]$ implemented by parallel ($x=\mathrm{PAR}$) and sequential ($x=\mathrm{SEQ}$) protocols, and general superchannels ($x=\mathrm{GEN}$) can be represented by a Choi operator $C\in\mcL(\mcI^n \otimes \mcO^n \otimes \mcP \otimes \mcF)$ satisfying
\begin{align}
    C&\geq 0,\\
    C&\in\mcW^{(x)}.
\end{align}
The set $\mcW^{(x)}$ for $x\in\{\mathrm{PAR}, \mathrm{SEQ}\}$ are given by \cite{chiribella2009theoretical, quintino2019probabilistic}
\begin{align}
    C\in \mcW^{(\mathrm{PAR})} &\Longleftrightarrow
    \begin{cases}
        \Tr_{\mcF} C = \Tr_{\mcO^n} C \otimes \1_{\mcO^n} / \dim \mcO^n\\
        \Tr_{\mcI^n \mcO^n} C = \dim \mcO^n \1_{\mcP} 
    \end{cases},\label{eq:parallel_condition}\\
    C\in \mcW^{(\mathrm{SEQ})} &\Longleftrightarrow
    \Tr_{\mcI_i} C^{(i)} = C^{(i-1)} \otimes \1_{\mcO_{i-1}} \quad \forall i\in\{1, \ldots, n+1\},\label{eq:sequential_condition}
\end{align}
where $\mcO_0$ and $\mcI_{n+1}$ are defined by $\mcO_0\coloneqq \mcP$ and $\mcI_{n+1}\coloneqq \mcF$, and $C^{(i)}$ for $i\in \{0, \ldots, n+1\}$ are defined by $C^{(n+1)}\coloneqq C$, $C^{(i-1)} \coloneqq \Tr_{\mcO_{i-1}\mcI_i} C^{(i)}/\dim \mcO_{i-1}$ and $C^{(0)}\coloneqq 1$.  The characterization of the set $\mcW^{(\mathrm{GEN})}$ is shown in Ref.~\cite{araujo2015witnessing}.

We consider the case $\mcI_1 = \cdots = \mcI_n = \mcF = \CC^d, \mcP = \mcO_1 = \cdots = \mcO_n = \CC^D$ (isometry inversion, isometry adjointation) and $\mcI_1 = \cdots = \mcI_n = \CC^d, \mcP = \mcO_1 = \cdots = \mcO_n = \CC^D, \mcF = \CC$ (universal error detection), and characterize the Choi operator $C\in\mcL(\mcI^n \otimes \mcO^n \otimes \mcP \otimes \mcF)$ under the $\U(d)\times \U(D)$ symmetry given by
\begin{align}
    \begin{cases}
        [C, U^{\otimes n+1}_{\mcI^n\mcF} \otimes U^{\prime\otimes n+1}_{\mcP\mcO^n}] = 0 & (\mathrm{isometry\;inversion}, \mathrm{isometry\;adjointation})\\
        [C, U^{\otimes n}_{\mcI^n} \otimes U^{\prime\otimes n+1}_{\mcP\mcO^n}] = 0 & (\mathrm{universal\;error\;detection})
    \end{cases}\label{eq:C_sudsuDsymmetry}
\end{align}
for all $U\in\U(d)$ and $U'\in\U(D)$ (see Theorem \ref{thm:unitary_group_symmetry}). Due to this symmetry, the Choi operator $C$ can be written using the operator $E^{\mu,d}_{ij}$ introduced in Eq.~(\ref{eq:def_E}) as
\begin{align}
    C = \sum_{\mu\in \young{d}{n+1}} \sum_{\nu\in \young{D}{n+1}}\sum_{i,j=1}^{d_\mu} \sum_{k,l=1}^{d_\nu} \frac{[C_{\mu\nu}]_{ik, jl}}{m_\mu^{(d)} m_\nu^{(D)}} (E_{ij}^{\mu, d})_{\mcI^n\mcF} \otimes (E_{kl}^{\nu, D})_{\mcP\mcO^n}\label{eq:c_yy_basis1}
\end{align}
for isometry inversion or isometry adjointation, and
\begin{align}
    C = \sum_{\lambda\in \young{d}{n}} \sum_{\nu\in \young{D}{n+1}}\sum_{a,b=1}^{d_\lambda} \sum_{k,l=1}^{d_\nu} \frac{[C_{\lambda\nu}]_{ak,bl}}{m_\lambda^{(d)} m_\nu^{(D)}} (E_{ab}^{\lambda, d})_{\mcI^n} \otimes (E_{kl}^{\nu, D})_{\mcP\mcO^n}\label{eq:c_yy_basis2}
\end{align}
for universal error detection using a $d_\mu d_\nu$ ($d_\lambda d_\nu$)-dimensional square matrix $C_{\mu\nu}$ ($C_{\lambda\nu}$), where $ik$ ($ak$) and $jl$ ($bl$) are the indices for row and column numbers, respectively. The characterization of quantum superchannels is rewritten in terms of $C_{\mu\nu}$ or $C_{\lambda\nu}$ as follows.  By definition of $E^{\mu,d}_{ij}$ in Eq.~(\ref{eq:def_E}), the positivity of $C$ is written as
\begin{align}
\begin{cases}
    C_{\mu\nu}\geq 0 \quad \forall \mu\in\young{d}{n+1}, \nu\in\young{D}{n+1} & (\mathrm{isometry\;inversion}, \mathrm{isometry\;adjointation})\\
    C_{\lambda\nu}\geq 0 \quad \forall \lambda\in\young{d}{n}, \nu\in\young{D}{n+1} & (\mathrm{universal\;error\;detection})
\end{cases}.
\end{align}
Using Lemma \ref{lem:yy}, the condition $C\in\mcW^{(x)}$ for $x\in\{\mathrm{PAR}, \mathrm{SEQ}\}$ is given by
\begin{align}
    C\in \mcW^{(x)} \Longleftrightarrow
    \begin{cases}
        \{C_{\mu\nu}\} \in \mcW_{\mathrm{sym}}^{(x)} & (\mathrm{isometry\;inversion}, \mathrm{isometry\;adjointation})\\
        \{C_{\lambda\nu}\} \in \mcW_{\mathrm{sym}}^{(x)} & (\mathrm{universal\;error\;detection})
    \end{cases},
\end{align}
where $\mcW_{\mathrm{sym}}^{(x)}$ is given by
\begin{align}
    &\{C_{\mu\nu}\} \in \mcW_{\mathrm{sym}}^{(\mathrm{PAR})} \Longleftrightarrow \nonumber\\
    &
    \begin{cases}
        \begin{aligned}
            &\sum_{\mu\in\lambda+_d\square} (X^\lambda_\mu \otimes \1_{d_\nu}){C_{\mu\nu} \over m_\nu^{(D)}}(X^\lambda_\mu \otimes \1_{d_\nu})^\dagger = D_\lambda \otimes {\1_{d_\nu} \over D^{n+1}} \quad \forall \lambda\in\young{d}{n}, \nu\in\young{D}{n+1}\\
            &\sum_{\mu\in\young{d}{n+1}}\sum_{\nu\in\young{D}{n+1}}\Tr(C_{\mu\nu}) = D^{n+1}
        \end{aligned}
    \end{cases},\\
    &\{C_{\mu\nu}\} \in \mcW_{\mathrm{sym}}^{(\mathrm{SEQ})} \Longleftrightarrow \nonumber\\
    &
    \begin{cases}
        \begin{aligned}
            &\begin{aligned}
                \sum_{\lambda\in\gamma+_d\square}(X^\gamma_\lambda\otimes \1_{d_\kappa}){C_{\lambda\kappa}^{(i)} \over m_\kappa^{(D)}} (X^\gamma_\lambda\otimes \1_{d_\kappa})^\dagger = \sum_{\delta\in\kappa-\square} (\1_{d_\gamma} \otimes X^\delta_\kappa)^\dagger {C_{\gamma\delta}^{(i-1)} \over m_\delta^{(D)}}(\1_{d_\gamma} \otimes X^\delta_\kappa)\\
                \forall i\in\{1,\cdots,n+1\},\gamma\in\young{d}{i-1},\kappa\in\young{D}{i}
            \end{aligned}\\
            &C_{\emptyset \emptyset}^{(0)} = 1
        \end{aligned}
    \end{cases}
\end{align}
for isometry inversion and isometry adjointation, where $D_\lambda$ for $\lambda\in\young{d}{n}$ are defined by
\begin{align}
    D_\lambda\coloneqq \sum_{\mu\in\lambda+_d\square}\sum_{\nu\in\young{D}{n+1}} \Tr_{\nu}[(X^\lambda_\mu \otimes \1_{d_\nu}) C_{\mu\nu} (X^\lambda_\mu \otimes \1_{d_\nu})^\dagger]
\end{align}
$C^{(i)}_{\lambda\kappa}$ for $\lambda\in\young{d}{i}, \kappa\in\young{D}{i}$ are defined by
\begin{align}
    C^{(i)}_{\lambda\kappa}\coloneqq
    \begin{cases}
        C_{\lambda\kappa} & (i=n+1)\\
        {1\over D} \sum_{\mu\in\lambda+_d\square, \nu\in\kappa+_D\square} (X^\lambda_\mu\otimes X^\kappa_\nu) C_{\mu\nu}^{(i+1)} (X^\lambda_\mu\otimes X^\kappa_\nu)^\dagger & (0\leq i\leq n)
    \end{cases},
\end{align}
$X^{\gamma}_\lambda$ for $\lambda\in\gamma+\square, \gamma\in\young{d}{i-1}$ are $d_\gamma \times d_\lambda$ matrices defined by
\begin{align}
    [X^{\gamma}_{\lambda}]_{c,a} \coloneqq \delta_{c^\gamma_\lambda, a},\label{eq:def_X}
\end{align}
$c^\gamma_\lambda$ is the index of the standard tableau $s^\lambda_{c^\gamma_\lambda}$ obtained by adding a box \fbox{$i$} to the standard tableau $s^\gamma_c$, and $\emptyset$ represents the Young tableau with zero boxes.
The set $\mcW_{\mathrm{sym}}^{(x)}$ is given by
\begin{align}
    &\{C_{\lambda\nu}\} \in \mcW_{\mathrm{sym}}^{(\mathrm{PAR})} \Longleftrightarrow \nonumber\\
    &
    \begin{cases}
        \begin{aligned}
            &{C_{\lambda\nu} \over m_\nu^{(D)}} = D_\lambda \otimes {\1_{d_\nu} \over D^{n+1}} \quad \forall \lambda\in\young{d}{n}, \nu\in\young{D}{n+1}\\
            &\sum_{\lambda\in\young{d}{n+1}}\sum_{\nu\in\young{D}{n+1}}\Tr(C_{\lambda\nu}) = D^{n}
        \end{aligned}
    \end{cases},\label{eq:symmetric_parallel_condition}\\
    &\{C_{\lambda\nu}\} \in \mcW_{\mathrm{sym}}^{(\mathrm{SEQ})} \Longleftrightarrow \nonumber\\
    &
    \begin{cases}
        \begin{aligned}
            &{C_{\lambda\nu} \over m_\nu^{(D)}}= \sum_{\kappa\in\nu-\square} (\1_{d_\lambda} \otimes X^\kappa_\nu)^\dagger {C_{\lambda\kappa}^{(n)} \over m_\kappa^{(D)}}(\1_{d_\lambda} \otimes X^\kappa_\nu) \quad \forall \lambda\in\young{d}{n}, \nu\in\young{D}{n+1}\\
            &\begin{aligned}
                \sum_{\lambda\in\gamma+_d\square}(X^\gamma_\lambda\otimes \1_{d_\kappa}){C_{\lambda\kappa}^{(i)} \over m_\kappa^{(D)}} (X^\gamma_\lambda\otimes \1_{d_\kappa})^\dagger = \sum_{\delta\in\kappa-\square} (\1_{d_\nu} \otimes X^\delta_\kappa)^\dagger {C_{\gamma\delta}^{(i-1)} \over m_\delta^{(D)}}(\1_{d_\nu} \otimes X^\delta_\kappa)\\
                \forall i\in\{1,\cdots,n\},\gamma\in\young{d}{i-1},\kappa\in\young{D}{i}
            \end{aligned}\\
            &C_{\emptyset \emptyset}^{(0)} = 1
        \end{aligned}
    \end{cases}\label{eq:symmetric_sequential_condition}
\end{align}
for universal error detection, where $D_\lambda$ for $\lambda\in\young{d}{n}$ are defined by
\begin{align}
    D_\lambda\coloneqq \sum_{\nu\in\young{D}{n+1}} \Tr_{\nu}(C_{\lambda\nu}),
\end{align}
$C^{(i)}_{\lambda\kappa}$ for $\lambda\in\young{d}{i}, \kappa\in\young{D}{i}$ are defined by
\begin{align}
    C^{(i)}_{\lambda\kappa}\coloneqq
    \begin{cases}
        {1\over D} \sum_{\nu\in\kappa+_D\square} (\1_{d_\lambda} \otimes X^\kappa_\nu) C_{\lambda\nu} (\1_{d_\lambda} \otimes X^\kappa_\nu)^\dagger & (i=n)\\
        {1\over D} \sum_{\mu\in\lambda+_d\square, \nu\in\kappa+_D\square} (X^\lambda_\mu\otimes X^\kappa_\nu) C_{\mu\nu}^{(i+1)} (X^\lambda_\mu\otimes X^\kappa_\nu)^\dagger & (0\leq i\leq n-1)
    \end{cases},
\end{align}
and $X^{\gamma}_\lambda$ are defined in Eq.~(\ref{eq:def_X}).  

\subsection{Conditions for isometry inversion, universal error detection, and isometry adjointation under the $\U(d) \times \U(D)$ symmetry}
\label{appendix_sec:unitary_group_symmetry_conditions}
We consider the action of a supermap $\supermap{C}$ on $n$ calls of an isometry operation $V_\mathrm{in}\in\isometry{d}{D}$ when its Choi operator $C$ satisfies the $\U(d)\times \U(D)$ symmetry (\ref{eq:C_sudsuDsymmetry}). Then, $\supermap{C}(\map{V}_\mathrm{in}^{\otimes n})$ is given in the following form.
\begin{Lem}
\label{lem:deconposition_of_fV}
    If the Choi operator of a quantum supermap $\supermap{C}$, denoted by $C$, satisfies the $\U(d)\times \U(D)$ symmetry (\ref{eq:C_sudsuDsymmetry}), then $\supermap{C}(\map{V}_\mathrm{in}^{\otimes n})$ for $V_\mathrm{in}\in\isometry{d}{D}$ is given by
    \begin{align}
        \supermap{C}(\map{V}_\mathrm{in}^{\otimes n})(\rho_\mathrm{in}) =
        \begin{cases}
            \begin{aligned}
                &x V_\mathrm{in}^\dagger \rho_\mathrm{in} V_\mathrm{in} + {\1_\mcF \over d} \Tr[\rho_\mathrm{in} (y \Pi_{\Im V_\mathrm{in}} + z (\1_D-\Pi_{\Im V_\mathrm{in}}))]\\
                &\hspace{90pt}(\mathrm{isometry\;inversion}, \mathrm{isometry\;adjointation})
            \end{aligned}\\
            \Tr[\rho_\mathrm{in} (v \Pi_{\Im V_\mathrm{in}} + w (\1_D-\Pi_{\Im V_\mathrm{in}}))] \hspace{24pt} (\mathrm{universal\;error\;detection})
        \end{cases}
    \end{align}
    for all $V_\mathrm{in}\in \isometry{d}{D}$ and $\rho_\mathrm{in} \in \mcL(\CC^d)$, where $\Pi_{\Im V_\mathrm{in}}$ and $(\1_D-\Pi_{\Im V_\mathrm{in}})$ are orthogonal projectors onto the image $\Im V_\mathrm{in}$ of $V_\mathrm{in}$ and its complement $(\Im V_\mathrm{in})^\perp$, and $x,y,z,v,w \in\CC$ are constant numbers given by
    \begin{align}
        x &\coloneqq {1\over d^2-1}\Tr[C(d^2 \Omega-\Xi\otimes \1_\mcF)],\\
        y &\coloneqq {d^2 \over d^2-1} \Tr[C(\Xi\otimes \1_\mcF-\Omega)],\\
        z &\coloneqq \Tr[C(\Sigma\otimes \1_\mcF)],
    \end{align}
    for isometry inversion and isometry adjointation, and
    \begin{align}
        v &\coloneqq \Tr(C\Xi),\\
        w &\coloneqq \Tr(C\Sigma),
    \end{align}
    for universal error detection, $\Omega, \Xi, \Sigma$ are defined by
    \begin{align}
        \Omega &\coloneqq \sum_{\mu\in\young{d}{n+1}} \sum_{i,j,k,l=1}^{d_\mu}[\Omega_\mu]_{ik,jl} (E^{\mu,d}_{ij})_{\mcI^n \mcF} \otimes (E^{\mu,D}_{kl})_{\mcP\mcO^n},\\
        \Xi &\coloneqq \sum_{\lambda\in\young{d}{n}}\sum_{\nu\in\lambda+_d\square} \sum_{a,b=1}^{d_\lambda} \sum_{k,l=1}^{d_\nu} [\Xi_{\lambda\nu}]_{ak,bl} (E^{\lambda,d}_{ab})_{\mcI^n} \otimes (E^{\nu,D}_{kl})_{\mcP\mcO^n}\\
        \Sigma &\coloneqq \sum_{\lambda\in\young{d}{n}}\sum_{\nu\in\lambda+_D\square} \sum_{a,b=1}^{d_\lambda} \sum_{k,l=1}^{d_\nu} [\Sigma_{\lambda\nu}]_{ak,bl} (E^{\lambda,d}_{ab})_{\mcI^n} \otimes (E^{\nu,D}_{kl})_{\mcP\mcO^n},
    \end{align}
    and $\Omega_\mu \in \mcL(\CC^{d_\mu} \otimes \CC^{d_\mu})$ and $\Xi_{\lambda\nu}, \Sigma_{\lambda\nu} \in\mcL(\CC^{d_\lambda} \otimes \CC^{d_\nu})$ are defined by
    \begin{align}
        [\Omega_\mu]_{ik,jl} &\coloneqq {[\pi_\mu]_{ik}^*[\pi_\mu]_{lj} \over d^2 m_\mu^{(D)}},\\
        [\Xi_{\lambda\nu}]_{ak,bl} &\coloneqq {m_\nu^{(d)}\over d m_\nu^{(D)} m_\lambda^{(d)}} [\pi_\nu]^*_{a^\lambda_\nu k}[\pi_\nu]_{l b_\nu^\lambda},\\
        [\Sigma_{\lambda\nu}]_{ak,bl} &\coloneqq {1\over D-d}\left[ {1\over m_\lambda^{(D)}} - \delta_{\nu\in\young{D}{n+1}}{m_\nu^{(d)} \over m_\nu^{(D)} m_\lambda^{(d)}} \right] [\pi_\nu]^*_{a^\lambda_\nu k}[\pi_\nu]_{l b_\nu^\lambda},
    \end{align}
    where $[\pi_\mu]_{ij}$ are matrix elements of the irreducible representation $\pi_\mu$ for $\pi\coloneqq (12\cdots n+1)\in\mfS_{n+1}$ shown in Eq.~(\ref{eq:def_pi_mu}) defined by $[\pi_\mu]_{ij}\coloneqq \bra{\mu,i}\pi_\mu\ket{\mu,j}$, $\delta_{\nu\in\young{D}{n+1}}$ is defined by $\delta_{\nu\in\young{D}{n+1}} =1$ for $\nu\in\young{D}{n+1}$ and $\delta_{\nu\in\young{D}{n+1}} =0$ for $\nu\notin\young{d}{n+1}$, and $a^\lambda_\nu$ is the index of the standard tableau $s_{a^\lambda_\nu}^\nu$ obtained by adding a box \fbox{$n+1$} to the standard tableau $s^\lambda_a$.
\end{Lem}
\begin{proof}
First, we consider the Choi operator $C$ satisfying
\begin{align}
    [C, U^{\otimes n+1}_{\mcI^n \mcF} \otimes U^{\prime\otimes n+1}_{\mcP\mcO^n}] = 0
\end{align}
for all $U\in \U(d)$ and $U'\in\U(D)$. Then, $C\star \dketbra{V_\mathrm{in}}^{\otimes n}_{\mcI^n\mcO^n}$ for $V_\mathrm{in}\in\isometry{d}{D}$ satisfies
\begin{align}
    &C\star \dketbra{U'V_\mathrm{in}U}^{\otimes n}_{\mcI^n\mcO^n}\nonumber\\
    &= C\star (U^{T\otimes n}_{\mcI^n} \otimes U^{\prime \otimes n}_{\mcO^n})\dketbra{V_\mathrm{in}}^{\otimes n}_{\mcI^n\mcO^n} (U^{T\otimes n}_{\mcI^n} \otimes U^{\prime \otimes n}_{\mcO^n})^\dagger\\
    &= (U^{T\otimes n}_{\mcI^n} \otimes \1_{\mcF} \otimes \1_{\mcP} \otimes U^{\prime \otimes n}_{\mcO^n})^T C (U^{T\otimes n}_{\mcI^n} \otimes \1_{\mcF} \otimes \1_{\mcP} \otimes U^{\prime \otimes n}_{\mcO^n})^* \star \dketbra{V_\mathrm{in}}^{\otimes n}_{\mcI^n\mcO^n}\\
    &= (\1_{\mcI^n} \otimes U^{\dagger}_{\mcF} \otimes U^{\prime *}_{\mcP} \otimes \1_{\mcO^n}) C (\1_{\mcI^n} \otimes U^{\dagger}_{\mcF} \otimes U^{\prime *}_{\mcP} \otimes \1_{\mcO^n})^\dagger \star \dketbra{V_\mathrm{in}}^{\otimes n}_{\mcI^n\mcO^n}\\
    &=(U^{\prime *}_{\mcP} \otimes U^{\dagger}_{\mcF}) [C \star \dketbra{V_\mathrm{in}}^{\otimes n}_{\mcI^n\mcO^n}] (U^{\prime *}_{\mcP} \otimes U^{\dagger}_{\mcF})^\dagger\label{eq:choi_fV_symmetry}
\end{align}
for all $U\in\U(d)$ and $U'\in\U(D)$. For $U\in\U(d)$ and $U''\in\U[(\Im V_\mathrm{in})^\perp]$, $U' \coloneqq V_\mathrm{in} U V_\mathrm{in}^\dagger + U''$ is a unitary operator and $U'V_\mathrm{in}U = V_\mathrm{in}$ holds.  By substituting $U$ and $U' = V_\mathrm{in} U V_\mathrm{in}^\dagger + U''$ to Eq.~(\ref{eq:choi_fV_symmetry}), we obtain
\begin{align}
    [(V_\mathrm{in} U V_\mathrm{in}^\dagger + U'')^*_{\mcP} \otimes U^\dagger_{\mcF}] [C \star \dketbra{V_\mathrm{in}}^{\otimes n}_{\mcI^n\mcO^n}] [(V_\mathrm{in} U V_\mathrm{in}^\dagger + U'')^*_{\mcP} \otimes U^\dagger_{\mcF}]^\dagger = C \star \dketbra{V_\mathrm{in}}^{\otimes n}_{\mcI^n\mcO^n}.\label{eq:choi_fV_commutator}
\end{align}
Decomposing $C \star \dketbra{V_\mathrm{in}}^{\otimes n}_{\mcI^n\mcO^n}$ as
\begin{align}
    C \star \dketbra{V_\mathrm{in}}^{\otimes n}_{\mcI^n\mcO^n} = (V_\mathrm{in}^*\otimes \1_\mcF) P (V_\mathrm{in}^*\otimes \1_\mcF)^\dagger + (V_\mathrm{in}^*\otimes \1_\mcF) Q + R (V_\mathrm{in}^*\otimes \1_\mcF)^\dagger + S
\end{align}
using linear operators $P\in \mcL(\CC^d \otimes \mcF)$, $Q:(\Im V_\mathrm{in})^\perp \otimes \mcF \to \CC^d \otimes \mcF$, $R: \CC^d \otimes \mcF \to (\Im V_\mathrm{in})^\perp \otimes \mcF$, and $S\in\mcL((\Im V_\mathrm{in})^\perp \otimes \mcF)$, Eq.~(\ref{eq:choi_fV_commutator}) is written as
\begin{align}
    (U^*\otimes U)A(U^*\otimes U)^\dagger &= A,\\
    (U^*\otimes U)B(U''^T\otimes U)^\dagger &= B,\\
    (U''^T \otimes U)C(U^*\otimes U)^\dagger &= C,\\
    (U''^T \otimes U)D(U''^T \otimes U)^\dagger &= D.
\end{align}
Therefore, $P$ is written as a linear combination of $\dketbra{\1_d}$ and $\1_d\otimes \1_d$, $Q=0$, $R=0$, and $S$ is proportional to $\1_{(\Im V_\mathrm{in})^\perp}\otimes \1_\mcF$. Therefore, $C \star \dketbra{V_\mathrm{in}}^{\otimes n}_{\mcI^n\mcO^n}$ can be expressed by three parameters $x_{V_\mathrm{in}}$, $y_{V_\mathrm{in}}$, and $z_{V_\mathrm{in}}$ as
\begin{align}
    C \star \dketbra{V_\mathrm{in}}^{\otimes n}_{\mcI^n\mcO^n}
    &= x_{V_\mathrm{in}} \dketbra{V_\mathrm{in}^\dagger}_{\mcP\mcF} + y_{V_\mathrm{in}} (\Pi^T_{\Im V_\mathrm{in}})_{\mcP} \otimes {\1_\mcF \over d} + z_{V_\mathrm{in}}  (\Pi^T_{(\Im V_\mathrm{in})^\perp})_{\mcP} \otimes {\1_\mcF \over d}.
\end{align}
From Eq.~(\ref{eq:choi_fV_symmetry}), we obtain
\begin{align}
    &x_{U'V_\mathrm{in}U} \dketbra{(U'V_\mathrm{in}U)^\dagger}_{\mcP\mcF} + y_{U'V_\mathrm{in}U} (\Pi^T_{\Im (U'V_\mathrm{in}U)})_{\mcP} \otimes {\1_\mcF \over d} + z_{U'V_\mathrm{in}U} (\Pi^T_{(\Im (U'V_\mathrm{in}U))^\perp})_{\mcP} \otimes {\1_\mcF \over d}\nonumber\\
    &= (U^{\prime *}_{\mcP} \otimes U^{\dagger}_{\mcF})[x_{V_\mathrm{in}} \dketbra{V_\mathrm{in}^\dagger}_{\mcP\mcF} + y_{V_\mathrm{in}} (\Pi^T_{\Im V_\mathrm{in}})_{\mcP} \otimes {\1_\mcF \over d} + z_{V_\mathrm{in}}  (\Pi^T_{(\Im V_\mathrm{in})^\perp})_{\mcP} \otimes {\1_\mcF \over d}](U^{\prime *}_{\mcP} \otimes U^{\dagger}_{\mcF})^\dagger\\
    &= x_{V_\mathrm{in}} \dketbra{(U'V_\mathrm{in}U)^\dagger}_{\mcP\mcF} + y_{V_\mathrm{in}} (\Pi^T_{\Im (U'V_\mathrm{in}U)})_{\mcP} \otimes {\1_\mcF \over d} + z_{V_\mathrm{in}} (\Pi^T_{(\Im (U'V_\mathrm{in}U))^\perp})_{\mcP} \otimes {\1_\mcF \over d},
\end{align}
thus, $x_{V_\mathrm{in}}$, $y_{V_\mathrm{in}}$ and $z_{V_\mathrm{in}}$ does not depend on $V_\mathrm{in}$. By rewriting $x_{V_\mathrm{in}} = x$, $y_{V_\mathrm{in}} = y$ and $z_{V_\mathrm{in}} = z$, we obtain
\begin{align}
    C \star \dketbra{V_\mathrm{in}}^{\otimes n}_{\mcI^n\mcO^n}
    &= x \dketbra{V_\mathrm{in}^\dagger}_{\mcP\mcF} + y (\Pi^T_{\Im V_\mathrm{in}})_{\mcP} \otimes {\1_\mcF \over d} + z  (\Pi^T_{(\Im V_\mathrm{in})^\perp})_{\mcP} \otimes {\1_\mcF \over d}.\label{eq:choi_CV}
\end{align}
This Choi operator corresponds to the map
\begin{align}
    \supermap{C}(\map{V}_\mathrm{in}^{\otimes n})(\rho_\mathrm{in}) = x V_\mathrm{in}^\dagger \rho_\mathrm{in} V_\mathrm{in} + {\1_\mcF \over d} \Tr[\rho_\mathrm{in} (y \Pi_{\Im V_\mathrm{in}} + z (\1_D-\Pi_{\Im V_\mathrm{in}}))].
\end{align}
We calculate $x$, $y$ and $z$ as follows. Equation (\ref{eq:choi_CV}) is written as
\begin{align}
    &\Tr_{\mcI^n\mcO^n}[C (\1_{\mcP\mcF}\otimes \dketbra{V_\mathrm{in}^*}^{\otimes n}_{\mcI^n\mcO^n})]\nonumber\\
    &= x \dketbra{V_\mathrm{in}^\dagger}_{\mcP\mcF} + y (\Pi^T_{\Im V_\mathrm{in}})_{\mcP} \otimes {\1_\mcF \over d} + z  (\Pi^T_{(\Im V_\mathrm{in})^\perp})_{\mcP} \otimes {\1_\mcF \over d}.
\end{align}
Taking the Hilbert-Schmidt inner product with $\dketbra{V_\mathrm{in}^*}_{\mcF\mcP}$, $(\Pi_{\Im V_\mathrm{in}})^*_{\mcP} \otimes \1_{\mcF}$ and $((\1_D-\Pi_{\Im V_\mathrm{in}}))^*_{\mcP} \otimes \1_\mcF$, we obtain
\begin{align}
    \Tr[C \dketbra{V_\mathrm{in}^*}^{\otimes n}_{\mcI^n\mcO^n} \otimes \dketbra{V_\mathrm{in}^*}_{\mcF\mcP}] &= d^2 x + y,\\
    \Tr[C \dketbra{V_\mathrm{in}^*}^{\otimes n}_{\mcI^n\mcO^n} \otimes (\Pi_{\Im V_\mathrm{in}})^*_{\mcP} \otimes \1_{\mcF}] &= d(x+y),\\
    \Tr[C \dketbra{V_\mathrm{in}^*}^{\otimes n}_{\mcI^n\mcO^n} \otimes ((\1_D-\Pi_{\Im V_\mathrm{in}}))^*_{\mcP} \otimes \1_{\mcF}] &= (D-d)z.
\end{align}
Taking the Haar integral $\dd V_\mathrm{in}$ on $\isometry{d}{D}$, we obtain
\begin{align}
    \Tr(C\Omega) &= x + {y\over d^2},\\
    \Tr[C (\Xi\otimes \1_\mcF)] &= x+y,\\
    \Tr[C (\Sigma\otimes \1_\mcF)] &= z,
\end{align}
where $\Omega$, $\Xi$ and $\Sigma$ are defined by
\begin{align}
    \Omega &\coloneqq {1\over d^2} \int_{\isometry{d}{D}} \dd V \dketbra{V}^{\otimes n}_{\mcI^n \mcO^n} \otimes \dketbra{V}_{\mcF\mcP},\\
    \Xi &\coloneqq {1 \over d} \int_{\isometry{d}{D}} \dd V \dketbra{V}^{\otimes n}_{\mcI^n \mcO^n} \otimes (\Pi_{\Im V})_{\mcP},\\
    \Sigma &\coloneqq {1\over D-d} \int_{\isometry{d}{D}} \dd V \dketbra{V}^{\otimes n}_{\mcI^n \mcO^n} \otimes ((\1_D-\Pi_{\Im V}))_{\mcP}.
\end{align}
Therefore, $x$, $y$ and $z$ are given by
\begin{align}
    x &= {1\over d^2-1}\Tr[C(d^2 \Omega-\Xi\otimes \1_\mcF)],\\
    y &= {d^2 \over d^2-1} \Tr[C(\Xi\otimes \1_\mcF-\Omega)],\\
    z &= \Tr[C(\Sigma\otimes \1_\mcF)].
\end{align}

Next, we consider the Choi operator $C$ satisfying
\begin{align}
    [C, U^{\otimes n}_{\mcI^n} \otimes U^{\prime\otimes n+1}_{\mcP\mcO^n}] = 0
\end{align}
for all $U\in\U(d)$ and $U'\in\U(D)$.  Defining $C' \coloneqq C\otimes {\1_{\mcF} \over d}$, $C'$ satisfies
\begin{align}
    [C', U^{\otimes n+1}_{\mcI^n\mcF} \otimes U^{\prime\otimes n+1}_{\mcP\mcO^n}] = 0
\end{align}
for all $U\in\U(d)$ and $U'\in\U(D)$.  Therefore, we can show that there exist constant numbers $u$, $v$, and $w$ such that
\begin{align}
    &C\star \dketbra{V_\mathrm{in}}^{\otimes n}_{\mcI^n \mcO^n} \otimes {\1_{\mcF} \over d}\nonumber\\
    &= C'\star \dketbra{V_\mathrm{in}}^{\otimes n}_{\mcI^n \mcO^n}\\
    &=u \dketbra{V_\mathrm{in}^\dagger}_{\mcP\mcF} + v (\Pi^T_{\Im V_\mathrm{in}})_{\mcP} \otimes {\1_\mcF \over d} + w (\Pi^T_{(\Im V_\mathrm{in})^\perp})_{\mcP} \otimes {\1_\mcF \over d},
\end{align}
i.e., $u=0$ and
\begin{align}
    C\star \dketbra{V_\mathrm{in}}^{\otimes n}_{\mcI^n \mcO^n} = v (\Pi^T_{\Im V_\mathrm{in}})_{\mcP} + w (\Pi^T_{(\Im V_\mathrm{in})^\perp})_{\mcP}\label{eq:choi_CV'}
\end{align}
holds. We calculate $v$ and $w$ as follows.  Equation (\ref{eq:choi_CV'}) is written as
\begin{align}
    \Tr_{\mcI^n\mcO^n}[C (\1_{\mcP}\otimes \dketbra{V_\mathrm{in}^*}^{\otimes n}_{\mcI^n\mcO^n})]= v (\Pi^T_{\Im V_\mathrm{in}})_{\mcP} +w (\Pi^T_{(\Im V_\mathrm{in})^\perp})_{\mcP}.
\end{align}
Taking the Hilbert-Schmidt inner product with $\Pi_{\Im V_\mathrm{in}}^*$ and $(\1_D-\Pi_{\Im V_\mathrm{in}})^*$, we obtain
\begin{align}
    \Tr[C \dketbra{V_\mathrm{in}^*}^{\otimes n}_{\mcI^n\mcO^n} \otimes (\Pi_{\Im V_\mathrm{in}})^*_{\mcP}] &= dv,\\
    \Tr[C \dketbra{V_\mathrm{in}^*}^{\otimes n}_{\mcI^n\mcO^n} \otimes ((\1_D-\Pi_{\Im V_\mathrm{in}}))^*_{\mcP}] &= (D-d)w.
\end{align}
Taking the Haar integral $\dd V_\mathrm{in}$ on $\isometry{d}{D}$, we obtain
\begin{align}
    \Tr(C \Xi) &= v,\\
    \Tr(C \Sigma) &= w.
\end{align}

We calculate $\Omega$, $\Xi$ and $\Sigma$ as follows.  First, due to the left- and right-invariance of the Haar measure $\dd V$ given by $\dd V = \dd (U'VU)$ for all $U\in\U(d)$ and $U'\in\U(D)$, $\Omega$ and $\Xi$ satisfies the $\U(d)\times \U(D)$ symmetry given by
\begin{align}
    [\Omega, U^{\otimes n+1}_{\mcI^n \mcF} \otimes U^{\prime\otimes n+1}_{\mcP\mcO^n}] &= 0
\end{align}
for all $U\in\U(d)$ and $U'\in\U(D)$.
Thus, they can be written as
\begin{align}
    \Omega &= \sum_{\mu\in\young{d}{n+1}}\sum_{\nu\in\young{D}{n+1}} \sum_{i,j=1}^{d_\mu}\sum_{k,l=1}^{d_\nu}\Omega^{\mu\nu}_{ijkl} (E^{\mu,d}_{ij})_{\mcI^n \mcF} \otimes (E^{\nu,D}_{kl})_{\mcP\mcO^n},
\end{align}
using complex coefficients $\Omega^{\mu\nu}_{ijkl}\in\CC$. The coefficients can be calculated as
\begin{align}
    &\Omega^{\mu\nu}_{ijkl}\nonumber\\
    &= {1\over m_\mu^{(d)}m_\nu^{(D)}}\Tr[\Omega (E^{\mu,d}_{ji})_{\mcI^n\mcF} \otimes (E^{\nu,D}_{lk})_{\mcP\mcO^n}]\\
    &= {1\over m_\mu^{(d)}m_\nu^{(D)}}\Tr[\Omega (E^{\mu,d}_{ji})_{\mcI^n\mcF} \otimes (P_\pi^\dagger E^{\nu,D}_{lk} P_\pi)_{\mcO^n\mcP}]\\
    &= {1\over m_\mu^{(d)}m_\nu^{(D)}}\sum_{k',l'=1}^{d_\nu}{1\over d^2} \int_{\isometry{d}{D}} \dd V \Tr[\dketbra{V}^{\otimes n+1}_{\mcI^n \mcF, \mcO^n \mcP} (E^{\mu,d}_{ji})_{\mcI^n\mcF} \otimes [\pi_\nu]_{ll'}(E^{\nu,D}_{l'k'})_{\mcO^n\mcP} [\pi_\nu]_{k'k}^*]\\
    &= {1\over m_\mu^{(d)}m_\nu^{(D)}}\sum_{k',l'=1}^{d_\nu}{1\over d^2} \int_{\isometry{d}{D}} \dd V \Tr[\dketbra{\1_d}^{\otimes n+1} (E^{\mu,d}_{ji}) \otimes [\pi_\nu]_{ll'}V^{\dagger \otimes n+1}(E^{\nu,D}_{l'k'}) V^{\otimes n+1}[\pi_\nu]_{k'k}^*],
\end{align}
where $P_\pi$ is the permutation of Hilbert spaces defined in Eq.~(\ref{eq:def_pi_mu}) for $\pi=(1 2 \cdots n+1)\in\mfS_{n+1}$, $[\pi_\mu]_{ij}$ are matrix elements of the irreducible representation $\pi_\mu$ defined by $[\pi_\mu]_{ij}\coloneqq \bra{\mu,i}\pi_\mu\ket{\mu,j}$.  As shown in Eq.~(\ref{eq:action_of_isometry_on_E}) in Appendix \ref{appendix_sec:sequential_isometry_adjointation}, $V^{\dagger \otimes n+1}(E^{\nu,D}_{l'k'}) V^{\otimes n+1}$ is given by $(E^{\nu,d}_{l'k'}) \delta_{\nu\in\young{D}{n+1}}$.  We consider the Schur basis introduced in Section \ref{sec:sw_duality}.  Since the quantum Schur transform is a real matrix, the maximally entangled state in the Schur basis is the same as that in the Schur basis, i.e.,
\begin{align}
    \dket{\1_d}^{\otimes n+1} = \sum_{\mu\in\young{d}{n+1}}\sum_{u=1}^{m_\mu^{(d)}}\sum_{i=1}^{d_\mu} \ket{\mu,u}_{\mcU_\mu^{(d)}} \otimes \ket{\mu,i}_{\mcS_\mu} \otimes \ket{\mu,u}_{\mcU_\mu^{(d)}} \otimes \ket{\mu,i}_{\mcS_\mu}.
\end{align}
Thus, $\Omega^{\mu\nu}_{ijkl}$ is further calculated as
\begin{align}
    \Omega^{\mu\nu}_{ijkl}&= {\delta_{\nu\in\young{D}{n+1}}\over m_\mu^{(d)}m_\nu^{(D)}}\sum_{k',l'=1}^{d_\nu}{1\over d^2} \Tr[\dketbra{\1_d}^{\otimes n+1} (E^{\mu,d}_{ji}) \otimes [\pi_\nu]_{ll'}(E^{\nu,d}_{l'k'}) [\pi_\nu]_{k'k}^*]\\
    &= {\delta_{\mu\nu} [\pi_\mu]_{ik}^* [\pi_\mu]_{lj} \over d^2 m_\mu^{(D)}},
\end{align}
Therefore, $\Omega$ is given by
\begin{align}
    \Omega &= \sum_{\mu\in\young{d}{n+1}} \sum_{i,j,k,l=1}^{d_\mu}{ [\pi_\mu]_{ik}^* [\pi_\mu]_{lj} \over d^2 m_\mu^{(D)}} (E^{\mu,d}_{ij})_{\mcI^n \mcF} \otimes (E^{\mu,D}_{kl})_{\mcP\mcO^n}.\label{eq:Omega_yybasis}
\end{align}
Since $(\Pi_{\Im V})_{\mcP} = \Tr_{\mcF}\dketbra{V}_{\mcF\mcP}$ holds, $\Xi$ is calculated using Lemma \ref{lem:yy} as
\begin{align}
    \Xi &= d\Tr_{\mcF} \Omega\\
    &= \sum_{\nu\in\young{d}{n+1}} \sum_{\lambda\in\nu-\square} \sum_{a,b=1}^{d_\lambda} \sum_{k,l=1}^{d_\nu} {m_\nu^{(d)}\over d m_\nu^{(D)} m_\lambda^{(d)}}  [\pi_\nu]_{a_\nu^{\lambda} k}^* [\pi_\nu]_{lb_\nu^\lambda}(E^{\lambda,d}_{ab})_{\mcI^n} \otimes (E^{\nu,D}_{kl})_{\mcP\mcO^n}.\label{eq:Xi_yybasis}
\end{align}
Since $((\1_D-\Pi_{\Im V}))_{\mcP} = \1_{\mcP} - (\Pi_{\Im V})_{\mcP}$ holds, $\Sigma$ is calculated using Lemma \ref{lem:yy} as
\begin{align}
    \Sigma
    &= {1\over D-d} \Bigg[\int \dd V \dketbra{V}_{\mcI^n \mcO^n}^{\otimes n} \otimes \1_{\mcP} - d\Xi\Bigg]\\
    &= {1\over D-d} \sum_{\lambda\in\young{d}{n}}\sum_{a,b=1}^{d_\lambda} \Bigg[{(E_{ab}^{\lambda, d})_{\mcI^n} \otimes (E_{ab}^{\lambda, D})_{\mcO^n} \over m_\lambda^{(D)}} \otimes \1_{\mcP} \nonumber\\
    &\hspace{48pt}- \sum_{\nu\in\lambda+_d\square} \sum_{k,l=1}^{d_\nu} {m_\nu^{(d)}\over m_\nu^{(D)} m_\lambda^{(d)}}  [\pi_\mu]_{a_\nu^{\lambda} k}^* [\pi_\mu]_{lb_\nu^\lambda}(E^{\lambda,d}_{ab})_{\mcI^n} \otimes (E^{\nu,D}_{kl})_{\mcP\mcO^n}\Bigg]\\
    &= {1\over D-d} \sum_{\lambda\in\young{d}{n}}\sum_{a,b=1}^{d_\lambda}\sum_{\nu\in\lambda+_D\square}\Bigg[{(E_{a b}^{\lambda,d})_{\mcI^n} \otimes (P_\pi^\dagger E_{a^\lambda_\nu b^\lambda_\nu }^{\nu, D}P_\pi)_{\mcP\mcO^n} \over m_\lambda^{(D)}} \nonumber\\
    &\hspace{48pt}- \delta_{\nu\in\young{d}{n+1}} \sum_{k,l=1}^{d_\nu} {m_\nu^{(d)}\over m_\nu^{(D)} m_\lambda^{(d)}}  [\pi_\nu]_{a_\nu^{\lambda} k}^* [\pi_\nu]_{lb_\nu^\lambda}(E^{\lambda,d}_{ab})_{\mcI^n} \otimes (E^{\nu,D}_{kl})_{\mcP\mcO^n}\Bigg]\\
    &= {1\over D-d} \sum_{\lambda\in\young{d}{n}}\sum_{\nu\in\lambda+_D\square}\sum_{a,b=1}^{d_\lambda}\sum_{k,l=1}^{d_\nu}\Bigg[ {1\over m_\lambda^{(D)}} - \delta_{\nu\in\young{d}{n+1}}  {m_\nu^{(d)}\over m_\nu^{(D)} m_\lambda^{(d)}}\Bigg][\pi_\nu]_{a_\nu^{\lambda} k}^* [\pi_\nu]_{lb_\nu^\lambda}(E^{\lambda,d}_{ab})_{\mcI^n} \otimes (E^{\nu,D}_{kl})_{\mcP\mcO^n}\\
    &= \sum_{\lambda\in\young{d}{n}}\sum_{\nu\in\lambda+_D\square}\sum_{a,b=1}^{d_\lambda}\sum_{k,l=1}^{d_\nu} {1\over m_\nu^{(D)}} {\mathrm{hook}(\lambda) \over \mathrm{hook}(\nu)} [\pi_\nu]_{a_\nu^{\lambda} k}^* [\pi_\nu]_{lb_\nu^\lambda}(E^{\lambda,d}_{ab})_{\mcI^n} \otimes (E^{\nu,D}_{kl})_{\mcP\mcO^n}.\label{eq:Sigma_yybasis}
\end{align}
\end{proof}

The parameters $x,y,z,v,w$ in Lemma \ref{lem:deconposition_of_fV} are related to the constraints and the figure of merit of each task considered in this work as follows:
\begin{itemize}
    \item Probabilistic exact isometry inversion: $x = p$ and $y = 0$.
    \item Deterministic isometry inversion: $x + y/d^2 = F_\mathrm{worst}$.
    \item Universal error detection: $v = 1$, $w = \alpha$.
    \item Isometry adjointation: $x+y = 1$, $\max\{1-x-y/d^2, z\} = \epsilon$.
\end{itemize}
Using this property, we derive the SDP to obtain the optimal transformations of isometry operations in the next section.

\subsection{Derivation of the SDP to obtain optimal transformation of isometry operations}
\label{appendix_sec:sdp_derivation}
\subsubsection{Probabilistic exact isometry inversion}
From Theorem \ref{thm:unitary_group_symmetry} and Appendix \ref{appendix_sec:unitary_group_symmetry_conditions}, the optimization problem of the success probability for isometry inversion is formulated as follows:
\begin{align}
\begin{split}
    &\max \Tr(C_S\Omega)\\
    \mathrm{s.t.}\;&0\leq C_S, C_F\in\mcL(\mcI^n \otimes \mcO^n \otimes \mcP \otimes \mcF),\\
    &C \coloneqq C_S+C_F\in\mcW^{(x)},\\
    &\Tr(C_S\Omega) = \Tr[C_S(\Xi\otimes \1_\mcF)],\\
    &[C_a, U^{\otimes n+1}_{\mcI^n\mcF} \otimes U^{\prime\otimes n+1}_{\mcP\mcO^n}] = 0 \quad \forall U\in\U(d), U'\in\U(D), a\in\{S,F\}.
\end{split}\label{eq:sdp_probabilistic_isometry_inversion_comp_basis}
\end{align}
Using the $\U(d)\times \U(D)$ symmetry of $\{C_S, C_F\}$, we write $\{C_S, C_F\}$ similarly to Eq.~(\ref{eq:c_yy_basis1}) as
\begin{align}
    C_S &= \sum_{\mu\in \young{d}{n+1}} \sum_{\nu\in\young{D}{n+1}}\sum_{i,j=1}^{d_\mu} \sum_{k,l=1}^{d_\nu} \frac{[S_{\mu\nu}]_{ik, jl}}{m_\mu^{(d)} m_\nu^{(D)}} (E_{ij}^{\mu, d})_{\mcI^n\mcF} \otimes (E_{kl}^{\nu, D})_{\mcP\mcO^n},\label{eq:cs_yybasis}\\
    C_F &= \sum_{\mu\in \young{d}{n+1}} \sum_{\nu\in\young{D}{n+1}}\sum_{i,j=1}^{d_\mu} \sum_{k,l=1}^{d_\nu} \frac{[F_{\mu\nu}]_{ik, jl}}{m_\mu^{(d)} m_\nu^{(D)}} (E_{ij}^{\mu, d})_{\mcI^n\mcF} \otimes (E_{kl}^{\nu, D})_{\mcP\mcO^n}.\label{eq:cf_yybasis}
\end{align}
Then, $\Tr(C_S\Omega)$ and $\Tr[C_S(\Xi\otimes \1_{\mcF})]$ are given by
\begin{align}
    \Tr(C_S \Omega)
    &= \sum_{\mu\in\young{d}{n+1}}\Tr(S_{\mu\mu}\Omega_\mu),\\
    \Tr[C_S(\Xi\otimes \1_{\mcF})]
    &= \Tr[\Tr_{\mcF}(C_S)\Xi]\\
    &= \sum_{\lambda\in \young{d}{n}} \sum_{\mu\in\lambda+_d\square} \sum_{\nu\in\young{D}{n+1}}\sum_{a,b=1}^{d_\lambda} \sum_{k,l=1}^{d_\nu} \frac{[S_{\mu\nu}]_{a^\lambda_\mu k, b^\lambda_\mu l}}{m_\mu^{(d)} m_\nu^{(D)}} \Tr[(E_{ab}^{\lambda, d})_{\mcI^n} \otimes (E_{kl}^{\nu, D})_{\mcP\mcO^n} \Xi]\\
    &= \sum_{\lambda\in \young{d}{n}} \sum_{\mu\in\lambda+_d\square} \sum_{\nu\in\young{d}{n+1}} \Tr[(X^\lambda_\mu \otimes \1_{d_\nu}) S_{\mu\nu} (X^\lambda_\mu \otimes \1_{d_\nu})^\dagger \Xi_{\lambda\nu}],
\end{align}
where $X^\lambda_\mu$ is defined in Eq.~(\ref{eq:def_X}) and $a^\lambda_\mu$ is the index of the standard tableau $s^\mu_{a^\lambda_\mu}$ obtained by adding a box \fbox{$n+1$} to the standard tableau $s^\lambda_a$.  Therefore, the SDP (\ref{eq:sdp_probabilistic_isometry_inversion_comp_basis}) is written as
\begin{align}
\begin{split}
    &\max \sum_{\mu\in\young{d}{n+1}}\Tr (S_{\mu\mu}\Omega_\mu)\\
    \text{s.t. } & 0\leq S_{\mu\nu}, F_{\mu\nu} \in \mcL(\CC^{d_\mu} \otimes \CC^{d_\nu}) \quad \forall \mu\in\young{d}{n+1}, \nu\in\young{D}{n+1},\\
    & \{C_{\mu\nu}\} \coloneqq \{S_{\mu\nu}+F_{\mu\nu}\} \in \mcW^{(x)}_\mathrm{sym},\\
    & \sum_{\mu\in\young{d}{n+1}}\Tr(S_{\mu\mu}\Omega_\mu) = \sum_{\lambda\in \young{d}{n}} \sum_{\mu\in\lambda+_d\square} \sum_{\nu\in\young{d}{n+1}} \Tr[(X^\lambda_\mu \otimes \1_{d_\nu}) S_{\mu\nu} (X^\lambda_\mu \otimes \1_{d_\nu})^\dagger \Xi_{\lambda\nu}].
\end{split}\label{eq:sdp_probabilistic_isometry_inversion}
\end{align}

\subsubsection{Deterministic isometry inversion}
From Theorem \ref{thm:unitary_group_symmetry} and Appendix \ref{appendix_sec:unitary_group_symmetry_conditions}, the optimization problem of the fidelity of deterministic isometry inversion is formulated as follows:
\begin{align}
\begin{split}
    &\max \Tr(C\Omega)\\
    \mathrm{s.t.}\;&0\leq C\in\mcL(\mcI^n \otimes \mcO^n \otimes \mcP \otimes \mcF),\\
    &C \in\mcW^{(x)},\\
    &[C, U^{\otimes n+1}_{\mcI^n\mcF} \otimes U^{\prime\otimes n+1}_{\mcP\mcO^n}] = 0 \quad \forall U\in\U(d), U'\in\U(D).
\end{split}\label{eq:sdp_deterministic_isometry_inversion_comp_basis}
\end{align}
Similarly to probabilistic isometry inversion, this SDP can be rewritten as follows:

\begin{align}
\begin{split}
    &\max \sum_{\mu\in\young{d}{n+1}}\Tr (C_{\mu\mu}\Omega_\mu)\\
    \text{s.t. } & 0\leq C_{\mu\nu} \in \mcL(\CC^{d_\mu} \otimes \CC^{d_\nu}) \quad \forall \mu\in\young{d}{n+1}, \nu\in\young{D}{n+1},\\
    & \{C_{\mu\nu}\} \in \mcW^{(x)}_\mathrm{sym}.
\end{split}\label{eq:sdp_deterministic_isometry_inversion}
\end{align}

\subsubsection{Universal error detection}
From Theorem \ref{thm:unitary_group_symmetry} and Appendix \ref{appendix_sec:unitary_group_symmetry_conditions}, the optimization problem of $\lambda$ of universal error detection is formulated as follows:

\begin{align}
\begin{split}
    &\min \Tr(C_I\Sigma)\\
    \mathrm{s.t.}\;&0\leq C_I, C_O\in\mcL(\mcI^n \otimes \mcO^n \otimes \mcP),\\
    &C \coloneqq C_I+C_O\in\mcW^{(x)},\\
    &\Tr(C_I\Xi) = 1,\\
    &[C_a, U^{\otimes n}_{\mcI^n} \otimes U^{\prime\otimes n+1}_{\mcP\mcO^n}] = 0 \quad \forall U\in\U(d), U'\in\U(D), a\in\{I,O\}.
\end{split}\label{eq:sdp_universal_error_detection_comp_basis}
\end{align}

Using the $\U(d)\times \U(D)$ symmetry of $\{C_I, C_O\}$, we write $\{C_I, C_O\}$ similarly to Eq.~(\ref{eq:c_yy_basis2}) as
\begin{align}
    C_I &= \sum_{\lambda\in \young{d}{n}} \sum_{\nu\in\young{D}{n+1}}\sum_{a,b=1}^{d_\mu} \sum_{k,l=1}^{d_\nu} \frac{[I_{\lambda\nu}]_{ak, bl}}{m_\lambda^{(d)} m_\nu^{(D)}} (E_{ab}^{\lambda, d})_{\mcI^n} \otimes (E_{kl}^{\nu, D})_{\mcP\mcO^n},\label{eq:ci_yybasis}\\
    C_O &= \sum_{\lambda\in \young{d}{n}} \sum_{\nu\in\young{D}{n+1}}\sum_{a,b=1}^{d_\mu} \sum_{k,l=1}^{d_\nu} \frac{[O_{\lambda\nu}]_{ak, bl}}{m_\lambda^{(d)} m_\nu^{(D)}} (E_{ab}^{\lambda, d})_{\mcI^n} \otimes (E_{kl}^{\nu, D})_{\mcP\mcO^n}.\label{eq:co_yybasis}
\end{align}
Then, $\Tr(C_I\Sigma)$ and $\Tr(C_I\Xi)$ are given by
\begin{align}
    \Tr(C_I\Sigma) &= \sum_{\lambda\in\young{d}{n}}\sum_{\nu\in\young{D}{n+1}}\Tr(I_{\lambda\nu}\Sigma_{\lambda\nu}),\\
    \Tr(C_I\Xi) &= \sum_{\lambda\in\young{d}{n}}\sum_{\nu\in\young{d}{n+1}}\Tr(I_{\lambda\nu}\Xi_{\lambda\nu}).
\end{align}
Therefore, the SDP (\ref{eq:sdp_universal_error_detection_comp_basis}) is written as

\begin{align}
\begin{split}
    &\min \sum_{\lambda\in\young{d}{n}}\sum_{\nu\in\young{D}{n+1}}\Tr(I_{\lambda\nu}\Sigma_{\lambda\nu})\\
    \text{s.t. } & 0\leq I_{\lambda\nu}, O_{\lambda\nu} \in \mcL(\CC^{d_\lambda} \otimes \CC^{d_\nu}) \quad \forall \lambda\in\young{d}{n}, \nu\in\young{D}{n+1},\\
    &\{C_{\lambda\nu}\}\coloneqq \{I_{\lambda\nu}+O_{\lambda\nu}\} \in {\mcW}^{(x)}_\mathrm{sym},\\
    &\sum_{\lambda\in\young{d}{n}}\sum_{\nu\in\young{d}{n+1}}\Tr(I_{\lambda\nu}\Xi_{\lambda\nu}) = 1.
\end{split}\label{eq:sdp_universal_error_detection}
\end{align}

\subsubsection{Isometry adjointation}
From Theorem \ref{thm:unitary_group_symmetry} and Appendix \ref{appendix_sec:unitary_group_symmetry_conditions}, the optimization problem of $\epsilon$ of isometry adjointation is formulated as follows:

\begin{align}
\begin{split}
    &\min \max\{1-\Tr(C_I\Omega), \Tr[C_I(\Sigma\otimes \1_\mcF)]\}\\
    \mathrm{s.t.}\;&0\leq C_I, C_O\in\mcL(\mcI^n \otimes \mcO^n \otimes \mcP \otimes \mcF),\\
    &C \coloneqq C_I+C_O\in\mcW^{(x)},\\
    &\Tr[C_I(\Xi \otimes \1_\mcF)] = 1,\\
    &[C_a, U^{\otimes n+1}_{\mcI^n\mcF} \otimes U^{\prime\otimes n+1}_{\mcP\mcO^n}] = 0 \quad \forall U\in\U(d), U'\in\U(D), a\in\{I,O\}.
\end{split}\label{eq:sdp_isometry_adjointation_comp_basis}
\end{align}

Using the $\U(d)\times \U(D)$ symmetry of $\{C_I, C_O\}$, we write $\{C_I, C_O\}$ similarly to Eq.~(\ref{eq:c_yy_basis1}) as
\begin{align}
    C_I &= \sum_{\mu\in \young{d}{n+1}} \sum_{\nu\in\young{D}{n+1}}\sum_{i,j=1}^{d_\mu} \sum_{k,l=1}^{d_\nu} \frac{[I_{\mu\nu}]_{ik, jl}}{m_\mu^{(d)} m_\nu^{(D)}} (E_{ij}^{\mu, d})_{\mcI^n\mcF} \otimes (E_{kl}^{\nu, D})_{\mcP\mcO^n},\label{eq:ci_yybasis2}\\
    C_O &= \sum_{\mu\in \young{d}{n+1}} \sum_{\nu\in\young{D}{n+1}}\sum_{i,j=1}^{d_\mu} \sum_{k,l=1}^{d_\nu} \frac{[O_{\mu\nu}]_{ik, jl}}{m_\mu^{(d)} m_\nu^{(D)}} (E_{ij}^{\mu, d})_{\mcI^n\mcF} \otimes (E_{kl}^{\nu, D})_{\mcP\mcO^n}.\label{eq:co_yybasis2}
\end{align}
Then, similarly for the case of probabilistic exact isometry inversion, the SDP (\ref{eq:sdp_isometry_adjointation_comp_basis}) as follows:

\begin{align}
\begin{split}
    &\min \max\Big\{1-\sum_{\mu\in\young{d}{n+1}}\Tr (I_{\mu\mu}\Omega_\mu), \sum_{\lambda\in \young{d}{n}} \sum_{\mu\in\lambda+_d\square} \sum_{\nu\in\young{D}{n+1}} \Tr[(X^\lambda_\mu \otimes \1_{d_\nu}) I_{\mu\nu} (X^\lambda_\mu \otimes \1_{d_\nu})^\dagger \Sigma_{\lambda\nu}] \Big\}\\
    \text{s.t. } & 0\leq I_{\mu\nu}, O_{\mu\nu} \in \mcL(\CC^{d_\mu} \otimes \CC^{d_\nu}) \quad \forall \mu\in\young{d}{n+1}, \nu\in\young{D}{n+1},\\
    & \{C_{\mu\nu}\} \coloneqq \{I_{\mu\nu}+O_{\mu\nu}\} \in \mcW^{(x)}_\mathrm{sym},\\
    & \sum_{\lambda\in \young{d}{n}} \sum_{\mu\in\lambda+_d\square} \sum_{\nu\in\young{d}{n+1}} \Tr[(X^\lambda_\mu \otimes \1_{d_\nu}) I_{\mu\nu} (X^\lambda_\mu \otimes \1_{d_\nu})^\dagger \Xi_{\lambda\nu}] = 1.
\end{split}\label{eq:sdp_isometry_adjointation}
\end{align}

\subsection{Derivation of the dual problems}
We derive the dual problems of the SDPs to obtain the optimal transformations of isometry operations.  Due to the strong duality, the dual problems gives the same optimal value as the corresponding primal problems.  To this end, we first introduce the dual set of the Choi operators of quantum superchannels.  Then, we derive the dual problems using the dual set. Finally, we simplify the derived dual problems using the $\U(d)\times \U(D)$ symmetry.

\subsubsection{Characterization of the dual processes}
We define the dual set $\overline{\mcW}^{(x)}$ of $\mcW^{(x)}$ for $x\in\{\mathrm{PAR},\mathrm{SEQ}, \mathrm{GEN}\}$ by
\begin{align}
    \overline{C}\in \overline{\mcW}^{(x)} \Longleftrightarrow \Tr(C\overline{C}) = 1 \quad \forall C\in\mcW^{(x)}.
\end{align}
Introducing the basis $\{\overline{C}_j\}$ of $\overline{\mcW}^{(x)}$, this relation leads to
\begin{align}
    \overline{C}\in \overline{\mcW}^{(x)} \Longleftrightarrow \Tr(C\overline{C}_j) = 1 \quad \forall j.
\end{align}
As shown in Ref.~\cite{bavaresco2021strict}, $\overline{\mcW}^{(x)}$ are given by
\begin{align}
    \overline{C}\in \overline{\mcW}^{(\mathrm{PAR})} &\Longleftrightarrow
    \begin{cases}
        \overline{C} = W \otimes \1_{\mcF}\\
    \Tr_{\mcO^n} W = \Tr_{\mcI^n \mcO^n} W \otimes \1_{\mcI^n}/\dim {\mcI^n}\\
    \Tr W = \dim \mcI^n
    \end{cases},\\
    \overline{C}\in \overline{\mcW}^{(\mathrm{SEQ})} &\Longleftrightarrow
    \begin{cases}
        \overline{C} = W \otimes \1_{\mcF}\\
        \Tr_{\mcO_i} W^{(i)} = \1_{\mcI_i} \otimes W^{(i-1)}, \quad \forall i \in \{1, \ldots, n\}\\
        \Tr W = \dim \mcI^n
    \end{cases},\\
    \overline{C}\in \overline{\mcW}^{(\mathrm{GEN})} &\Longleftrightarrow
    \begin{cases}
        \overline{C} = W \otimes \1_\mcF\\
        \Tr_{\mcO_i} W = \Tr_{\mcI_i \mcO_i} W \otimes {\1_{\mcI_i}}/{\dim{\mcI_i}} \quad \forall i\in \{1, \ldots, n\}\\
        \Tr W = \dim{\mcI^n}
    \end{cases},
\end{align}
where $W^{(i)}$ are defined by
\begin{align}
    W^{(i)}
    \coloneqq
    \begin{cases}
        W & (i=n)\\
        \Tr_{\mcI_{i+1} \mcO_{i+1}} W^{(i+1)}/\dim {\mcI_{i+1}} \quad \forall i \in \{0, \ldots, n\} & (i\in\{0,\cdots,n-1\})
    \end{cases}.
\end{align}
We also introduce the set $\cone[\overline{\mcW}^{(x)}]$ of the dual sets $\overline{\mcW}^{(x)}$ defined by
\begin{align}
    \cone[\overline{\mcW}^{(x)}] = \{\lambda \overline{C} | \lambda \in \CC, \overline{C}\in\overline{\mcW}^{(x)}\}.
\end{align}

\subsubsection{Probabilistic exact isometry inversion}
We write down the dual problem of the SDP (\ref{eq:sdp_probabilistic_isometry_inversion_comp_basis}).  To this end, we note that the $\U(d)\times \U(D)$ symmetry does not change the optimal value of the SDP (\ref{eq:sdp_probabilistic_isometry_inversion_comp_basis}), so we can remove it when we consider the dual problem.  Then, the SDP (\ref{eq:sdp_probabilistic_isometry_inversion_comp_basis}) corresponds to the following optimization problem:
\begin{align}
\begin{split}
    &\max \Tr(C_S\Omega)\\
    \text{s.t. } & 0\leq C_S, C_F \in \mcL(\mcI^n \otimes \mcO^n \otimes \mcP \otimes \mcF),\\
    & \Tr(C_S\Omega) = \Tr[C_S(\Xi\otimes \1_\mcF)],\\
    & \Tr[(C_S+C_F)\overline{C}_j] = 1, \quad \forall j.
\end{split}\label{eq:sdp_probabilistic_isometry_inversion_comp_basis_dualization}
\end{align}
By introducing the Lagrange multipliers $\omega, \lambda_j \in \RR$ and $0\leq \Gamma_S, \Gamma_F \in \mcL(\mcI^n \otimes \mcO^n \otimes \mcP \otimes \mcF)$, we can write down the corresponding Lagrangian
\begin{align}
    L =& \Tr(C_S\Omega) + [\Tr[C_S(\Xi\otimes \1_\mcF)] - \Tr(C_S \Omega)] \omega + \Tr(C_S\Gamma_S) + \Tr(C_F\Gamma_F) \nonumber\\
    &+ \sum_j [1-\Tr[(C_S+C_F)\overline{C}_j]]\lambda_j\\
    =& \sum_{j}\lambda_j + \Tr[C_S(\Omega - \omega\Omega + \omega\Xi\otimes \1_\mcF + \Gamma_S - \sum_j \lambda_j \overline{C}_j)] + \Tr[F(\Gamma_F - \sum_j \lambda_j \overline{C}_j)],
\end{align}
which gives the SDP (\ref{eq:sdp_probabilistic_isometry_inversion_comp_basis_dualization}) as the following optimization:
\begin{align}
    \max_{C_S, C_F\geq 0} \min_{\omega, \lambda_j\in \RR, \Gamma_S, \Gamma_F \geq 0} L.
\end{align}
This optimization problem corresponds to the following dual problem:
\begin{align}
\begin{split}
    &\min \sum_j \lambda_j\\
    \text{s.t. } & \omega, \lambda_j \in \RR, 0\leq \Gamma_S, \Gamma_F \in \mcL(\mcI^n \otimes \mcO^n \otimes \mcP \otimes \mcF),\\
    & (1-\omega) \Omega +\omega \Xi\otimes \1_\mcF + \Gamma_S - \sum_j \lambda_j \overline{C}_j = 0,\\
    & \Gamma_F - \sum_j \lambda_j \overline{C}_j = 0.
\end{split}
\end{align}
The variables $\Gamma_S, \Gamma_F$ can be removed, and the variables $\lambda_j$ and $\overline{C}_j$ can be replaced with $\overline{C}\coloneqq \sum_j \lambda_j \overline{C}_j$.  Since $\Tr \overline{C}_j = d^{n+1}$ holds for $\overline{C}_j\in\overline{\mcW}^{(x)}$, $\sum_j \lambda_j$ can be replaced with $\Tr \overline{C}/d^{n+1}$, which gives the following dual problem:
\begin{align}
\begin{split}
    &\min \Tr \overline{C}/d^{n+1}\\
    \text{s.t. } & \omega \in \RR, 0\leq \overline{C} \in \mcL(\mcI^n \otimes \mcO^n \otimes \mcP \otimes \mcF),\\
    &\overline{C} \in \cone[\overline{\mcW}^{(x)}],\\
    & \overline{C} \geq (1-\omega)\Omega +\omega \Xi\otimes \1_\mcF.
\end{split}\label{eq:sdp_probabilistic_isometry_inversion_comp_basis_dual}
\end{align}

\subsubsection{Deterministic isometry inversion}
Similarly to the case of probabilistic isometry inversion, the SDP (\ref{eq:sdp_deterministic_isometry_inversion_comp_basis}) can be rewritten as the following optimization problem:
\begin{align}
\begin{split}
    &\max \Tr(C\Omega)\\
    \text{s.t. } & 0\leq C \in \mcL(\mcI^n \otimes \mcO^n \otimes \mcP \otimes \mcF),\\
    & \Tr[C\overline{C}_j] = 1, \quad \forall j.
\end{split}\label{eq:sdp_deterministic_isometry_inversion_comp_basis_dualization}
\end{align}
By introducing the Lagrange multipliers $\lambda_j \in \RR$ and $0\leq \Gamma \in \mcL(\mcI^n \otimes \mcO^n \otimes \mcP \otimes \mcF)$, we can write down the corresponding Lagrangian
\begin{align}
    L &= \Tr(C\Omega) +\Tr(C\Gamma) + \sum_j [1-\Tr[C\overline{C}_j]]\lambda_j\\
    &= \sum_{j}\lambda_j + \Tr[C(\Omega +\Gamma - \sum_j \lambda_j \overline{C}_j)],
\end{align}
which gives the SDP (\ref{eq:sdp_deterministic_isometry_inversion_comp_basis}) as the following optimization:
\begin{align}
    \max_{C\geq 0} \min_{\lambda_j\in \RR, \Gamma \geq 0} L.
\end{align}
This optimization problem corresponds to the following dual problem:
\begin{align}
\begin{split}
    &\min \sum_j \lambda_j\\
    \text{s.t. } & \lambda_j \in \RR, 0\leq \Gamma \in \mcL(\mcI^n \otimes \mcO^n \otimes \mcP \otimes \mcF),\\
    & \Omega + \Gamma - \sum_j \lambda_j \overline{C}_j = 0.
\end{split}
\end{align}
The variable $\Gamma$ can be removed, and the variables $\lambda_j$ and $\overline{C}_j$ can be replaced with $\overline{C}\coloneqq \sum_j \lambda_j \overline{C}_j$ as
\begin{align}
\begin{split}
    &\min \Tr \overline{C}/d^{n+1}\\
    \text{s.t. } & \overline{C} \in \mcL(\mcI^n \otimes \mcO^n \otimes \mcP \otimes \mcF),\\
    &\overline{C} \in \cone[\overline{\mcW}^{(x)}],\\
    & \overline{C} \geq \Omega.
\end{split}\label{eq:sdp_deterministic_isometry_inversion_comp_basis_dual}
\end{align}

\subsubsection{Universal error detection}
Similarly to the case of probabilistic isometry inversion, the SDP (\ref{eq:sdp_universal_error_detection_comp_basis}) can be rewritten as the following optimization problem:
\begin{align}
\begin{split}
    &\min \Tr(C_I \Sigma)\\
    \text{s.t. } & 0\leq C_I, C_O \in \mcL(\mcI^n \otimes \mcO^n \otimes \mcP),\\
    & \Tr(C_I \Xi) = 1,\\
    & \Tr[(C_I+C_O)\overline{C}_j] = 1 \quad \forall j.\\
\end{split}
\end{align}
By introducing the Lagrange multipliers $\xi, \lambda_j \in \RR$ and $0\leq \Gamma_I, \Gamma_O \in \mcL(\mcI^n \otimes \mcO^n \otimes \mcP)$,  the corresponding Lagrangian is given by
\begin{align}
    L &= \Tr(C_I \Sigma) - \Tr(C_I\Gamma_I) - \Tr(C_O\Gamma_O) + [1-\Tr(C_I\Xi)]\xi + \sum_j [1-\Tr[(C_I+C_O)\overline{C}_j]] \lambda_j\\
    &= \sum_j \lambda_j + \xi + \Tr\Bigg[C_I\Big(\Sigma - \Gamma_I - \xi \Xi - \sum_j \lambda_j \overline{C}_j\Big)\Bigg] + \Tr[C_O(-\Gamma_O - \sum_j \lambda_j \overline{C}_j)],
\end{align}
which gives the SDP (\ref{eq:sdp_universal_error_detection_comp_basis}) as the following optimization:
\begin{align}
    \min_{C_I, C_O\geq 0} \max_{\xi, \lambda_j \in \RR, \Gamma_I, \Gamma_O\geq 0} L.
\end{align}
The corresponding dual problem is given by
\begin{align}
\begin{split}
    &\max \sum_j \lambda_j + \xi\\
    \text{s.t. }&\Sigma - \Gamma_I - \xi \Xi - \sum_j \lambda_j \overline{C}_j = 0,\\
    &-\Gamma_O - \sum_j \lambda_j \overline{C}_j = 0.
\end{split}
\end{align}
The variable $\Gamma$ can be removed, and the variables $\lambda_j$ and $\overline{C}_j$ can be replaced with $\overline{C}\coloneqq -\sum_j \lambda_j \overline{C}_j$ as
\begin{align}
\begin{split}
    &\max \xi - \Tr \overline{C}/d^n\\
    \text{s.t. }& \xi\in\RR, 0\leq \overline{C} \in \mcL(\mcI^n \otimes \mcO^n \otimes \mcP),\\
    &\overline{C} \in \cone[\overline{\mcW}^{(x)}],\\
    &\overline{C} \geq \xi \Xi-\Sigma.
\end{split}\label{eq:sdp_universal_error_detection_comp_basis_dual}
\end{align}

\subsubsection{Isometry adjointation}
Similarly to the case of probabilistic isometry inversion, the SDP (\ref{eq:sdp_isometry_adjointation_comp_basis}) can be rewritten as the following optimization problem:
\begin{align}
\begin{split}
    &\min p\\
    \text{s.t. } & 0\leq C_I, C_O \in \mcL(\mcI^n \otimes \mcO^n \otimes \mcP \otimes \mcF),\\
    &1-\Tr (C_I\Omega)\leq p,\\
    &\Tr[C_I(\Sigma\otimes\1_\mcF)]\leq p\\
    & \Tr[C_I (\Xi \otimes \1_\mcF)] = 1,\\
    & \Tr[(C_I+C_O) \overline{C}_j] = 1 \quad \forall j.
\end{split}
\end{align}
By introducing the Lagrange multipliers $\omega, \sigma\in\RR_{\geq 0}$, $\xi, \lambda_j \in \RR$ and $0\leq \Gamma_I, \Gamma_O \in \mcL(\mcI^n \otimes \mcO^n \otimes \mcP \otimes \mcF)$,  the corresponding Lagrangian is given by
\begin{align}
    L =& p -\Tr(C_I\Gamma_I)-\Tr(C_O\Gamma_O) + [1-\Tr (C_I\Omega)- p]\omega + [\Tr[C_I(\Sigma\otimes \1_\mcF)] - p]\sigma \nonumber\\
    &+ [1-\Tr[C_I(\Xi\otimes \1_\mcF)]]\xi + \sum_j[1-\Tr[(C_I+C_O) \overline{C}_j]]\lambda_j\\
    =& \sum_j \lambda_j +\omega+\xi + p(1-\omega-\sigma)+ \Tr[C_I(-\Gamma_I-\omega\Omega+\sigma \Sigma\otimes \1_\mcF -\xi\Xi\otimes \1_\mcF - \sum_j \lambda_j \overline{C}_j)]\nonumber\\
    &+\Tr[F(-\Gamma_O-\sum_j \lambda_j \overline{C}_j)].
\end{align}
The corresponding dual problem is given by
\begin{align}
\begin{split}
    &\max \sum_j \lambda_j +\omega+\xi\\
    \text{s.t. } & \omega + \sigma = 1,\\
    & -\Gamma_I-\omega\Omega+\sigma \Sigma\otimes \1_\mcF -\xi\Xi\otimes \1_\mcF-\sum_j \lambda_j \overline{C}_j = 0,\\
    & -\Gamma_O-\sum_j \lambda_j \overline{C}_j = 0.
\end{split}
\end{align}
The variables $\Gamma_I, \Gamma_O, \sigma$ can be removed, and the variables $\lambda_j$ and $\overline{C}_j$ can be replaced with $\overline{C}\coloneqq -\sum_j \lambda_j \overline{C}_j$ as
\begin{align}
\begin{split}
    &\max \omega+\xi-\Tr\overline{C}/d^{n+1}\\
    \text{s.t. } & \xi \in \RR, 0\leq \omega\leq 1, 0\leq \overline{C}\in\mcL(\mcI^n\otimes\mcO^n\otimes \mcP\otimes \mcF)\\
    &\overline{C}\in\cone[\overline{\mcW}^{(x)}],\\
    &\overline{C} \geq \omega\Omega+\xi\Xi\otimes \1_\mcF-(1-\omega)\Sigma\otimes \1_\mcF.
\end{split}\label{eq:sdp_isometry_adjointation_comp_basis_dual}
\end{align}

\subsubsection{Simplification of the dual problems using $\U(d)\times \U(D)$ and permutation symmetry}
In the dual SDPs (\ref{eq:sdp_probabilistic_isometry_inversion_comp_basis_dual}), (\ref{eq:sdp_deterministic_isometry_inversion_comp_basis_dual}), (\ref{eq:sdp_universal_error_detection_comp_basis_dual}) and (\ref{eq:sdp_isometry_adjointation_comp_basis_dual}), we can impose the $\U(d)\times\U(D)$ symmetry given by
\begin{align}
\begin{cases}
    [\overline{C}, U^{\otimes n+1}_{\mcI^n\mcF} \otimes U^{\prime\otimes n+1}_{\mcP\mcO^n}] = 0 & (\mathrm{isometry\;inversion}, \mathrm{universal\;error\;detection})\\
    [\overline{C}, U^{\otimes n}_{\mcI^n} \otimes U^{\prime\otimes n+1}_{\mcP\mcO^n}] = 0 & (\mathrm{isometry\;adjointation})
\end{cases}\label{eq:Cbar_unitary_group_symmetry}
\end{align}
for all $U\in\U(d)$ and $U'\in\U(D)$,
since for the optimal $\overline{C}_\mathrm{opt}$, the $\U(d)\times \U(D)$-twirled operator $\overline{C}'_\mathrm{opt}$ given by
\begin{align}
    \overline{C}'_\mathrm{opt}\coloneqq 
    \begin{cases}
    \int_{\U(d)} \dd U \map{U}^{\otimes n+1}_{\mcI^n\mcF} \otimes \map{U}^{\prime\otimes n+1}_{\mcP\mcO^n}(\overline{C}) & (\mathrm{isometry\;inversion}, \mathrm{universal\;error\;detection})\\
    \int_{\U(d)} \dd U \map{U}^{\otimes n}_{\mcI^n} \otimes \map{U}^{\prime\otimes n+1}_{\mcP\mcO^n}(\overline{C}) & (\mathrm{isometry\;adjointation})
\end{cases}
\end{align}
also gives the optimal values of the dual SDPs.  For the case of $x=\mathrm{GEN}$, we can also impose the permutation symmetry given by
\begin{align}
\begin{cases}
    [\overline{C}, (P_\pi)_{\mcI^n} \otimes (P_\pi)_{\mcO^n} \otimes \1_\mcP \otimes \1_\mcF] = 0 & (\mathrm{isometry\;inversion}, \mathrm{universal\;error\;detection})\\
    [\overline{C}, (P_\pi)_{\mcI^n} \otimes (P_\pi)_{\mcO^n} \otimes \1_\mcP] = 0 & (\mathrm{isometry\;adjointation})
\end{cases}\label{eq:Cbar_permutation_symmetry}
\end{align}
for all $\pi\in\mfS_n$ and $P_\pi$ is given in Eq.~(\ref{eq:def_pi_mu}) since the $\mfS_n$-twirled operator $\overline{C}''_\mathrm{opt}$ given by
\begin{align}
    \overline{C}''_\mathrm{opt}\coloneqq 
    \begin{cases}
    \sum_{\pi\in\mfS_n} (\map{P}_\pi)_{\mcI^n} \otimes (\map{P}_\pi)_{\mcO^n} \otimes \1_\mcP \otimes \1_\mcF (\overline{C}) & (\mathrm{isometry\;inversion}, \mathrm{universal\;error\;detection})\\
    \sum_{\pi\in\mfS_n} (\map{P}_\pi)_{\mcI^n} \otimes (\map{P}_\pi)_{\mcO^n} \otimes \1_\mcP (\overline{C})  & (\mathrm{isometry\;adjointation})
\end{cases}
\end{align}
also gives the optimal values of the dual SDPs.

We characterize the set $\overline{\mcW}^{(x)}$ under the $\U(d)\times \U(D)$ symmetry (\ref{eq:Cbar_unitary_group_symmetry}) [and the permutation symmetry (\ref{eq:Cbar_permutation_symmetry}) for $x=\mathrm{GEN}$].  Using the $\U(d)\times \U(D)$ symmetry (\ref{eq:Cbar_unitary_group_symmetry}), we write $\overline{C}$ using the operator $E^{\mu,d}_{ij}$ introduced in Eq.~(\ref{eq:def_E}) as
\begin{align}
    \overline{C} = \sum_{\mu\in \young{d}{n+1}} \sum_{\nu\in \young{D}{n+1}}\sum_{i,j=1}^{d_\mu} \sum_{k,l=1}^{d_\nu} [\overline{C}_{\mu\nu}]_{ik, jl}(E_{ij}^{\mu, d})_{\mcI^n\mcF} \otimes (E_{kl}^{\nu, D})_{\mcP\mcO^n}
\end{align}
for isometry inversion or isometry adjointation, and
\begin{align}
    \overline{C} = \sum_{\lambda\in \young{d}{n}} \sum_{\nu\in \young{D}{n+1}}\sum_{a,b=1}^{d_\lambda} \sum_{k,l=1}^{d_\nu} [\overline{C}_{\lambda\nu}]_{ak,bl}(E_{ab}^{\lambda, d})_{\mcI^n} \otimes (E_{kl}^{\nu, D})_{\mcP\mcO^n}
\end{align}
for universal error detection using a $d_\mu d_\nu$ ($d_\lambda d_\nu$)-dimensional square matrix $\overline{C}_{\mu\nu}$ ($\overline{C}_{\lambda\nu}$), where $ik$ ($ak$) and $jl$ ($bl$) are the indices for row and column numbers, respectively.  We also write $W$ and $W^{(i)}$ appearing in the characterization of $\overline{\mcW}^{(x)}$ as
\begin{align}
    W &= \sum_{\lambda\in\young{d}{n}} \sum_{\nu\in\young{D}{n+1}} \sum_{a,b=1}^{d_\lambda} \sum_{k,l=1}^{d_\nu} [W_{\lambda\nu}]_{ak,bl} (E^\lambda_{ab})_{\mcI^n} \otimes (E^{\nu}_{kl})_{\mcP\mcO^n},\\
    W^{(i)} &= \sum_{\lambda\in\young{d}{i}} \sum_{\nu\in\young{D}{i+1}} \sum_{a,b=1}^{d_\lambda} \sum_{k,l=1}^{d_\nu} [W^{(i)}_{\lambda\nu}]_{ak,bl} (E^\lambda_{ab})_{\mcI^i} \otimes (E^{\nu}_{kl})_{\mcP\mcO^i}.
\end{align}
Using Lemma \ref{lem:yy}, the condition $\overline{C}\in\overline{\mcW}^{(x)}$ for $x\in\{\mathrm{PAR}, \mathrm{SEQ}, \mathrm{GEN}\}$ are given by
\begin{align}
    \overline{C}\in \overline{\mcW}^{(x)} \Longleftrightarrow
    \begin{cases}
        \{\overline{C}_{\mu\nu}\} \in \overline{\mcW}_{\mathrm{sym}}^{(x)} & (\mathrm{isometry\;inversion}, \mathrm{isometry\;adjointation})\\
        \{\overline{C}_{\lambda\nu}\} \in \overline{\mcW}_{\mathrm{sym}}^{(x)} & (\mathrm{universal\;error\;detection})
    \end{cases},
\end{align}
where $\overline{\mcW}_{\mathrm{sym}}^{(x)}$ is given by
\begin{align}
    &\{\overline{C}_{\mu\nu}\} \in \overline{\mcW}_{\mathrm{sym}}^{(\mathrm{PAR})} \Longleftrightarrow \nonumber\\
    &
    \begin{cases}
        \begin{aligned}
            &\overline{C}_{\mu \nu} = \sum_{\lambda\in\mu-\square} (X^{\mu}_{\lambda} \otimes \1_{d_\nu}) W_{\lambda \nu}(X^{\mu}_{\lambda} \otimes \1_{d_\nu})^\dagger \quad \forall \mu\in\young{d}{n+1}, \nu\in\young{D}{n+1}\\
            &\sum_{\nu\in\young{D}{n+1}} m_\nu^{(D)}\Tr_{\nu}(W_{\lambda \nu}) = \1_{d_\lambda} \quad \forall \lambda\in\young{d}{n}
        \end{aligned}
    \end{cases},\\
    &\{\overline{C}_{\mu\nu}\} \in \overline{\mcW}_{\mathrm{sym}}^{(\mathrm{SEQ})} \Longleftrightarrow \nonumber\\
    &
    \begin{cases}
        \begin{aligned}
            &\overline{C}_{\mu \nu} = \sum_{\lambda\in\mu-\square} (X^{\mu}_{\lambda} \otimes \1_{d_\nu}) W_{\lambda \nu}(X^{\mu}_{\lambda} \otimes \1_{d_\nu})^\dagger \quad \forall \mu\in\young{d}{n+1}, \nu\in\young{D}{n+1}\\
            &\begin{aligned}
                \sum_{\nu\in\kappa+_D\square}m_\nu^{(D)}(\1_{d_\lambda} \otimes X^{\kappa}_{\nu})W^{(i)}_{\lambda \nu}(\1_{d_\lambda} \otimes X^{\kappa}_{\nu})^{\dagger} = \sum_{\gamma\in\lambda-\square} m_\kappa^{(D)}(X^{\lambda}_{\gamma} \otimes \1_{d_\kappa})W^{(i-1)}_{\gamma\kappa}(X^{\lambda}_{\gamma} \otimes \1_{d_\kappa})^\dagger\\
                \forall i\in\{1,\cdots,n\}, \lambda\in\young{d}{i}, \kappa\in\young{D}{i}
            \end{aligned}\\
            &W^{(0)}_{\emptyset \square} = {1\over D}
        \end{aligned}
    \end{cases},\\
    &\{\overline{C}_{\mu\nu}\} \in \overline{\mcW}_{\mathrm{sym}}^{(\mathrm{GEN})} \Longleftrightarrow \nonumber\\
    &
    \begin{cases}
        \begin{aligned}
            &\overline{C}_{\mu \nu} = \sum_{\lambda\in\mu-\square} (X^{\mu}_{\lambda} \otimes \1_{d_\nu}) W_{\lambda \nu}(X^{\mu}_{\lambda} \otimes \1_{d_\nu})^\dagger \quad \forall \mu\in\young{d}{n+1}, \nu\in\young{D}{n+1}\\
            &[W_{\lambda \nu}, \pi_\lambda\otimes \pi'_\nu] = 0 \quad \forall \pi\in\mfS_n,\\
            &\begin{aligned}
                \sum_{\nu\in\kappa+_D\square}m_\nu^{(D)}(\1_{d_\lambda} \otimes X^{\kappa}_{\nu})W_{\lambda \nu}(\1_{d_\lambda} \otimes X^{\kappa}_{\nu})^{\dagger} = \sum_{\gamma\in\lambda-\square} m_\kappa^{(D)}(X^{\lambda}_{\gamma} \otimes \1_{d_\kappa})W^{(n-1)}_{\gamma\kappa}(X^{\lambda}_{\gamma} \otimes \1_{d_\kappa})^\dagger \\
                \forall \lambda\in\young{d}{n}, \kappa\in\young{D}{n}
            \end{aligned}\\
            &\sum_{\lambda\in\young{d}{n}}\sum_{\nu\in\young{D}{n+1}} m_\lambda^{(d)}m_\nu^{(D)}\Tr(W_{\lambda\nu})=d^{n}
        \end{aligned}
    \end{cases}
\end{align}
for isometry inversion and isometry adjointation, where $W^{(i)}_{\gamma\kappa}$ for $\gamma\in\young{d}{i}$ and $\kappa\in\young{D}{i+1}$ are defined by
\begin{align}
     &W^{(i)}_{\gamma\kappa}\coloneqq
     \begin{cases}
         W_{\gamma\kappa} & (i=n)\\
         {1\over d}\sum_{\lambda\in\gamma+_d\square}\sum_{\nu\in\kappa+_D\square}{m_\lambda^{(d)}m_\nu^{(D)} \over m_\gamma^{(d)}m_\kappa^{(D)}}(X^{\gamma}_{\lambda} \otimes X^{\kappa}_{\nu})W_{\lambda\nu}^{(i+1)}(X^{\gamma}_{\lambda} \otimes X^{\kappa}_{\nu})^\dagger & (i\in\{0,\cdots,n-1\})
     \end{cases},
\end{align}
$\pi_\lambda$ is the irreducible representation of $\pi$ given in Eq.~(\ref{eq:def_pi_mu}), and $\pi'_\nu$ is the irreducible representation of $\pi'$ defined by $\pi'(1) = 1$ and $\pi'(i+1) = \pi(i)+1$ for $i\in\{1, \ldots, n\}$.  The set $\overline{\mcW}_{\mathrm{sym}}^{(x)}$ is given by
\begin{align}
    &\{\overline{C}_{\lambda\nu}\} \in \overline{\mcW}_{\mathrm{sym}}^{(\mathrm{PAR})} \Longleftrightarrow \nonumber\\
    &
    \begin{cases}
        \begin{aligned}
            &\overline{C}_{\lambda \nu} = W_{\lambda \nu} \quad \forall \lambda\in\young{d}{n}, \nu\in\young{D}{n+1}\\
            &\sum_{\nu\in\young{D}{n+1}} m_\nu^{(D)}\Tr_{\nu}(W_{\lambda \nu}) = \1_{d_\lambda} \quad \forall \lambda\in\young{d}{n}
        \end{aligned}
    \end{cases},\\
    &\{\overline{C}_{\lambda\nu}\} \in \overline{\mcW}_{\mathrm{sym}}^{(\mathrm{SEQ})} \Longleftrightarrow \nonumber\\
    &
    \begin{cases}
        \begin{aligned}
            &\overline{C}_{\lambda \nu} = W_{\lambda \nu} \quad \forall \lambda\in\young{d}{n}, \nu\in\young{D}{n+1}\\
            &\begin{aligned}
                \sum_{\nu\in\kappa+_D\square}m_\nu^{(D)}(\1_{d_\lambda} \otimes X^{\kappa}_{\nu})W^{(i)}_{\lambda \nu}(\1_{d_\lambda} \otimes X^{\kappa}_{\nu})^{\dagger} = \sum_{\gamma\in\lambda-\square} m_\kappa^{(D)}(X^{\lambda}_{\gamma} \otimes \1_{d_\kappa})W^{(i-1)}_{\gamma\kappa}(X^{\lambda}_{\gamma} \otimes \1_{d_\kappa})^\dagger\\
                \forall i\in\{1,\cdots,n\}, \lambda\in\young{d}{i}, \kappa\in\young{D}{i}
            \end{aligned}\\
            &W^{(0)}_{\emptyset \square} = {1\over D}
        \end{aligned}
    \end{cases},\\
    &\{\overline{C}_{\lambda\nu}\} \in \overline{\mcW}_{\mathrm{sym}}^{(\mathrm{GEN})} \Longleftrightarrow \nonumber\\
    &
    \begin{cases}
        \begin{aligned}
            &\overline{C}_{\lambda \nu} = W_{\lambda \nu} \quad \forall \lambda\in\young{d}{n}, \nu\in\young{D}{n+1}\\
            &[W_{\lambda \nu}, \pi_\lambda\otimes \pi'_\nu] = 0 \quad \forall \pi\in\mfS_n,\\
            &\begin{aligned}
                \sum_{\nu\in\kappa+_D\square}m_\nu^{(D)}(\1_{d_\lambda}\otimes X^{\kappa}_{\nu})W_{\lambda \nu}(\1_{d_\lambda} \otimes X^{\kappa}_{\nu})^{\dagger} = \sum_{\gamma\in\lambda-\square} m_\kappa^{(D)}(X^{\lambda}_{\gamma} \otimes \1_{d_\kappa})W^{(n-1)}_{\gamma\kappa}(X^{\lambda}_{\gamma} \otimes \1_{d_\kappa})^\dagger \\
                \forall \lambda\in\young{d}{n}, \kappa\in\young{D}{n}
            \end{aligned}\\
            &\sum_{\lambda\in\young{d}{n}}\sum_{\nu\in\young{D}{n+1}} m_\lambda^{(d)}m_\nu^{(D)}\Tr(W_{\lambda\nu})=d^{n}
        \end{aligned}
    \end{cases}
\end{align}
for universal error detection.  Using this characterization of $\overline{\mcW}^{(x)}$ with the $\U(d)\times \U(D)$ and permutation symmetry, the dual SDPs (\ref{eq:sdp_probabilistic_isometry_inversion_comp_basis_dual}), (\ref{eq:sdp_deterministic_isometry_inversion_comp_basis_dual}), (\ref{eq:sdp_universal_error_detection_comp_basis_dual}) and (\ref{eq:sdp_isometry_adjointation_comp_basis_dual}) can be simplified as follows:

\begin{itemize}
    \item Probabilistic exact isometry inversion
    \begin{align}
    \begin{split}
        &\min \sum_{\mu\in\young{d}{n+1}}\sum_{\nu\in\young{D}{n+1}} \Tr \overline{C}_{\mu\nu}/d^{n+1}\\
        \text{s.t. } & \omega \in \RR, 0\leq \overline{C}_{\mu\nu} \in \mcL(\CC^{d_\mu} \otimes \CC^{d_\nu}) \quad \forall \mu\in\young{d}{n+1}, \nu\in\young{D}{n+1},\\
        &\{\overline{C}_{\mu\nu}\} \in \cone[\overline{\mcW}^{(x)}_\mathrm{sym}],\\
        & \overline{C}_{\mu\nu} \geq
        \delta_{\mu\nu}(1-\omega)\Omega_\mu  +\omega \sum_{\lambda\in\mu-\square} (X^{\lambda}_\mu \otimes \1_{d_\mu})^\dagger \Xi_{\lambda\nu} (X^{\lambda}_\mu \otimes \1_{d_\mu}) \quad \forall \mu\in\young{d}{n+1}, \nu\in\young{D}{n+1}.
    \end{split}\label{eq:sdp_probabilistic_isometry_inversion_dual}
    \end{align}
    \item Deterministic isometry inversion
    \begin{align}
    \begin{split}
        &\min \sum_{\mu\in\young{d}{n+1}}\sum_{\nu\in\young{D}{n+1}} \Tr \overline{C}_{\mu\nu}/d^{n+1}\\
        \text{s.t. } & \omega \in \RR, 0\leq \overline{C}_{\mu\nu} \in \mcL(\CC^{d_\mu} \otimes \CC^{d_\nu}) \quad \forall \mu\in\young{d}{n+1}, \nu\in\young{D}{n+1},\\
        &\{\overline{C}_{\mu\nu}\} \in \cone[\overline{\mcW}^{(x)}_\mathrm{sym}],\\
        & \overline{C}_{\mu\mu} \geq \Omega_{\mu} \quad \forall \mu\in\young{d}{n+1}.
    \end{split}\label{eq:sdp_deterministic_isometry_inversion_dual}
    \end{align}
    \item Universal error detection
    \begin{align}
    \begin{split}
        &\max \xi - \sum_{\lambda\in\young{d}{n}} \sum_{\nu\in\young{D}{n+1}}\Tr \overline{C}_{\lambda\nu}/d^n\\
        \text{s.t. }& \xi\in\RR, 0\leq \overline{C}_{\lambda\nu} \in \mcL(\CC^{d_\lambda} \otimes \CC^{d_\nu}) \quad \forall \lambda\in\young{d}{n}, \nu\in\young{D}{n+1},\\
        &\{\overline{C}_{\lambda\nu}\} \in \cone[\overline{\mcW}^{(x)}_\mathrm{sym}],\\
        &\overline{C}_{\lambda\nu} \geq \xi \Xi_{\lambda\nu}-\Sigma_{\lambda\nu} \quad \forall \lambda\in\young{d}{n}, \nu\in\young{D}{n+1}.
    \end{split}\label{eq:sdp_universal_error_detection_dual}
    \end{align}
    \item Isometry adjointation
    \begin{align}
    \begin{split}
        &\max \omega+\xi-\sum_{\mu\in\young{d}{n+1}}\sum_{\nu\in\young{D}{n+1}}\Tr\overline{C}_{\mu\nu}/d^{n+1}\\
        \text{s.t. } & \xi \in \RR, 0\leq \omega\leq 1, 0\leq \overline{C}_{\mu\nu}\in\mcL(\CC^{d_\mu}\otimes \CC^{d_\nu}) \quad \forall \mu\in\young{d}{n+1}, \nu\in\young{D}{n+1},\\
        &\{\overline{C}_{\mu\nu}\}\in\cone[\overline{\mcW}^{(x)}_\mathrm{sym}],\\
        &\overline{C}_{\mu\nu} \geq \delta_{\mu\nu}\omega\Omega+\sum_{\lambda\in\mu-\square}(X^\lambda_\mu\otimes \1_{d_\nu})^\dagger[\xi\Xi_{\lambda\nu}-(1-\omega)\Sigma_{\lambda\nu}](X^\lambda_\mu\otimes \1_{d_\nu}) \quad \forall \mu\in\young{d}{n+1}, \nu\in\young{D}{n+1}.
    \end{split}\label{eq:sdp_isometry_adjointation_dual}
    \end{align}
\end{itemize}

\section{Proof of Corollary~\ref{cor:universal_programming}: Program cost of universal programming of two-outcome projective measurement}
\label{appendix_sec:universal_programming}
The proof of Theorem~\ref{thm:optimal_parallel_error_detection} shown in Appendix~\ref{appendix_sec:optimal_parallel_error_detection} constructs the parallel universal error detection protocol with the input state $\phi = {\1_{\mcU_\lambda^{(d)}} \over m_\lambda^{(d)}}\otimes \ketbra{\mathrm{arb}}_{\mcS_\lambda}$ for an arbitrary state $\ket{\mathrm{arb}}\in \mcS_\lambda$ and the POVM measurement $\mcP = \{\Pi^{(d)}, \1_D-\Pi^{(d)}\}$, where $\lambda$ is defined in Eq.~\eqref{eq:def_lambda}, and $\Pi^{(d)}$ is given by
\begin{align}
    \Pi^{(d)}\coloneqq \sum_{\lambda\in\young{d}{n}} \Pi_\lambda^{(D)}
\end{align}
using the Young projector $\Pi_\lambda^{(d)}$ defined in Eq.~\eqref{eq:young_projector}.
In the discussion below, we rewrite $\lambda$ as $\lambda^{(d)}$ to explicitly denote the dependency on $d$.
The parallel protocol achieves the approximation error $\alpha = O(d^2 n^{-1})$.
By using this protocol, we can implement the universal programming of rank-$d$ projective measurement $\mcM_\mathrm{PVM}^{(d)} = \{\Pi, \1_D-\Pi\}$ with the program state $\phi_{\mcM_\mathrm{PVM}^{(d)}} $ given by
\begin{align}
\label{eq:program_state_PVM}
    \phi_{\mcM_\mathrm{PVM}^{(d)}} \coloneqq \map{V}_\mathrm{in}^{\otimes n}(\phi)
\end{align}
using $V_\mathrm{in}\in\isometry{d}{D}$ such that $\Pi$ is the orthogonal projector onto $\Im V_\mathrm{in}$.
The approximation error is given by $\delta = \alpha = O(d^2 n^{-1})$, and the program cost is given by
\begin{align}
    C_{\mathbb{S}_\mathrm{PVM}^{(d,D)}}
    &= \log \dim \mcU_{\lambda^{(d)}}^{(D)}.
\end{align}
We evaluate the dimension of $\mcU_{\lambda^{(d)}}^{(D)}$ by using Eq.~(III.10) of Ref.~\cite{itzykson1966unitary} as follows:
\begin{align}
    &\dim \mcU_{\lambda^{(d)}}^{(D)}\nonumber\\
    &= {\prod_{1\leq i<j\leq D} (\lambda^{(d)}_i-\lambda^{(d)}_j-i+j) \over \prod_{i=1}^{D-1}i!}\\
    &= {\prod_{1\leq i<j\leq d} (\lambda^{(d)}_i-\lambda^{(d)}_j-i+j) \prod_{1\leq i \leq d, d<j\leq D} (\lambda^{(d)}_i-i+j) \prod_{d<i<j\leq D}(-i+j) \over \prod_{i=1}^{D-1}i!}\\
    &\leq {\prod_{1\leq i<j\leq d} (1-i+j) \prod_{1\leq i \leq d, d<j\leq D} (k+1-i+j) \prod_{d<i<j\leq D}(-i+j) \over \prod_{i=1}^{D-1}i!}\\
    &= {\prod_{1\leq i \leq d, d<j\leq D} (k+1-i+j) \prod_{i=1}^{D-d-1} i! \over \prod_{i=1}^{D-d-1}(i+d)!}\\
    &\leq (kD)^{d(D-d)}\\
    &\leq \left({Dn\over d}\right)^{d(D-d)}
\end{align}
Thus, we obtain
\begin{align}
    C_{\mathbb{S}_\mathrm{PVM}^{(d,D)}}
    &\leq d(D-d)\log \left({Dn\over d}\right)\\
    &\leq d(D-d)\log \left(\Theta(Dd)\over \delta\right).
\end{align}
We can also implement universal programming of two-outcome POVM measurement $\mcM = \{\Pi, \1_D-\Pi\}$ as follows.
Since $0\leq \Pi\leq \1_D$ holds, $\Pi$ has an eigendecomposition given by
\begin{align}
    \Pi = \sum_{j=1}^{D} a_j \ketbra{\psi_j},
\end{align}
where $0\leq a_1 \leq \cdots \leq a_D \leq 1$ and $\{\ket{\psi_j}\}_j$ is an orthonormal basis of $\CC^D$.
Defining rank-$d$ projective measurement $\mcM_\mathrm{PVM}^{(d)}\coloneqq \{\Pi^{(d)}, \1_D-\Pi^{(d)}$ by $\Pi^{(d)} \coloneqq \sum_{j=1}^{d} \ketbra{\psi_j}$ for $d\in\{1,\ldots, D\}$ and $\mcM_\mathrm{PVM}^{(0)}\coloneqq \{0, \1_D\}$, the POVM $\mcM$ is given by a probabilistic mixture of $\{\mcM^{(d)}\}_{d=0}^{D}$ as \cite{masanes2005extremal, davies1976quantum}
\begin{align}
    \mcM = \sum_{d=0}^{D} p_d \mcM_\mathrm{PVM}^{(d)},
\end{align}
where $p_d$ is given by $p_0 = 1-a_D$, $p_1=a_1$ and $p_d = a_{d}-a_{d-1}$ for $d\in\{2,\ldots,D\}$.
Then, we can implement the program state of $\mcM$ by
\begin{align}
    \phi_{\mcM}\coloneqq \sum_{d=0}^{D} p_j \ketbra{d} \otimes \phi_{\mcM_{\mathrm{PVM}}^{(d)}},
\end{align}
where $\phi_{\mcM_{\mathrm{PVM}}^{(d)}}$ is given in Eq.~\eqref{eq:program_state_PVM} for $d\in\{1,\ldots,D\}$ and $\phi_{\mcM_{\mathrm{PVM}}^{(0)}}$ can be taken as an arbitrary state.
The approximation error is given by
\begin{align}
    \delta \leq \sum_{d=0}^{D} p_d O(d^2 n^{-1}) = O(D^2 n^{-1})
\end{align}
and the program cost is evaluated as
\begin{align}
    C_{\mathbb{S}_\mathrm{PVM}^{(D)}}
    &= \log \left(1+\sum_{d=1}^{D} \dim \mcU_{\lambda^{(d)}}^{(D)}\right)\\
    &\leq \log (D+1) + \max_{1\leq d\leq D} C_{\mathbb{S}_\mathrm{PVM}^{(d,D)}}\\
    &\leq \log (D+1) + \lfloor D/2\rfloor (D-\lfloor D/2\rfloor) \log{\Theta(D^2)\over \delta}\\
    &\leq
    \begin{cases}
        {D^2\over 4} \log {\Theta(D^2)\over \delta} & (D \mathrm{\;is\;even})\\
        {D^2-1\over 4} \log {\Theta(D^2)\over \delta} & (D \mathrm{\;is\;odd})\\
    \end{cases}.
\end{align}

\bibliographystyle{apsrev4-2}
\bibliography{main}

\begin{thebibliography}{114}%
\makeatletter
\providecommand \@ifxundefined [1]{%
 \@ifx{#1\undefined}
}%
\providecommand \@ifnum [1]{%
 \ifnum #1\expandafter \@firstoftwo
 \else \expandafter \@secondoftwo
 \fi
}%
\providecommand \@ifx [1]{%
 \ifx #1\expandafter \@firstoftwo
 \else \expandafter \@secondoftwo
 \fi
}%
\providecommand \natexlab [1]{#1}%
\providecommand \enquote  [1]{``#1''}%
\providecommand \bibnamefont  [1]{#1}%
\providecommand \bibfnamefont [1]{#1}%
\providecommand \citenamefont [1]{#1}%
\providecommand \href@noop [0]{\@secondoftwo}%
\providecommand \href [0]{\begingroup \@sanitize@url \@href}%
\providecommand \@href[1]{\@@startlink{#1}\@@href}%
\providecommand \@@href[1]{\endgroup#1\@@endlink}%
\providecommand \@sanitize@url [0]{\catcode `\\12\catcode `\$12\catcode `\&12\catcode `\#12\catcode `\^12\catcode `\_12\catcode `\%12\relax}%
\providecommand \@@startlink[1]{}%
\providecommand \@@endlink[0]{}%
\providecommand \url  [0]{\begingroup\@sanitize@url \@url }%
\providecommand \@url [1]{\endgroup\@href {#1}{\urlprefix }}%
\providecommand \urlprefix  [0]{URL }%
\providecommand \Eprint [0]{\href }%
\providecommand \doibase [0]{https://doi.org/}%
\providecommand \selectlanguage [0]{\@gobble}%
\providecommand \bibinfo  [0]{\@secondoftwo}%
\providecommand \bibfield  [0]{\@secondoftwo}%
\providecommand \translation [1]{[#1]}%
\providecommand \BibitemOpen [0]{}%
\providecommand \bibitemStop [0]{}%
\providecommand \bibitemNoStop [0]{.\EOS\space}%
\providecommand \EOS [0]{\spacefactor3000\relax}%
\providecommand \BibitemShut  [1]{\csname bibitem#1\endcsname}%
\let\auto@bib@innerbib\@empty
\bibitem [{\citenamefont {Wootters}\ and\ \citenamefont {Zurek}(1982)}]{wootters1982single}%
  \BibitemOpen
  \bibfield  {author} {\bibinfo {author} {\bibfnamefont {W.~K.}\ \bibnamefont {Wootters}}\ and\ \bibinfo {author} {\bibfnamefont {W.~H.}\ \bibnamefont {Zurek}},\ }\href {https://doi.org/10.1038/299802a0} {\bibfield  {journal} {\bibinfo  {journal} {Nature}\ }\textbf {\bibinfo {volume} {299}},\ \bibinfo {pages} {802} (\bibinfo {year} {1982})}\BibitemShut {NoStop}%
\bibitem [{\citenamefont {Bu{\v{z}}ek}\ \emph {et~al.}(1999)\citenamefont {Bu{\v{z}}ek}, \citenamefont {Hillery},\ and\ \citenamefont {Werner}}]{buvzek1999optimal}%
  \BibitemOpen
  \bibfield  {author} {\bibinfo {author} {\bibfnamefont {V.}~\bibnamefont {Bu{\v{z}}ek}}, \bibinfo {author} {\bibfnamefont {M.}~\bibnamefont {Hillery}},\ and\ \bibinfo {author} {\bibfnamefont {R.~F.}\ \bibnamefont {Werner}},\ }\href {https://doi.org/10.1103/PhysRevA.60.R2626} {\bibfield  {journal} {\bibinfo  {journal} {Phys. Rev. A}\ }\textbf {\bibinfo {volume} {60}},\ \bibinfo {pages} {R2626} (\bibinfo {year} {1999})}\BibitemShut {NoStop}%
\bibitem [{\citenamefont {Buhrman}\ \emph {et~al.}(2001)\citenamefont {Buhrman}, \citenamefont {Cleve}, \citenamefont {Watrous},\ and\ \citenamefont {De~Wolf}}]{buhrman2001quantum}%
  \BibitemOpen
  \bibfield  {author} {\bibinfo {author} {\bibfnamefont {H.}~\bibnamefont {Buhrman}}, \bibinfo {author} {\bibfnamefont {R.}~\bibnamefont {Cleve}}, \bibinfo {author} {\bibfnamefont {J.}~\bibnamefont {Watrous}},\ and\ \bibinfo {author} {\bibfnamefont {R.}~\bibnamefont {De~Wolf}},\ }\href {https://doi.org/10.1103/PhysRevLett.87.167902} {\bibfield  {journal} {\bibinfo  {journal} {Phys. Rev. Lett.}\ }\textbf {\bibinfo {volume} {87}},\ \bibinfo {pages} {167902} (\bibinfo {year} {2001})}\BibitemShut {NoStop}%
\bibitem [{\citenamefont {Chabaud}\ \emph {et~al.}(2018)\citenamefont {Chabaud}, \citenamefont {Diamanti}, \citenamefont {Markham}, \citenamefont {Kashefi},\ and\ \citenamefont {Joux}}]{chabaud2018optimal}%
  \BibitemOpen
  \bibfield  {author} {\bibinfo {author} {\bibfnamefont {U.}~\bibnamefont {Chabaud}}, \bibinfo {author} {\bibfnamefont {E.}~\bibnamefont {Diamanti}}, \bibinfo {author} {\bibfnamefont {D.}~\bibnamefont {Markham}}, \bibinfo {author} {\bibfnamefont {E.}~\bibnamefont {Kashefi}},\ and\ \bibinfo {author} {\bibfnamefont {A.}~\bibnamefont {Joux}},\ }\href {https://doi.org/10.1103/PhysRevA.98.062318} {\bibfield  {journal} {\bibinfo  {journal} {Phys. Rev. A}\ }\textbf {\bibinfo {volume} {98}},\ \bibinfo {pages} {062318} (\bibinfo {year} {2018})}\BibitemShut {NoStop}%
\bibitem [{\citenamefont {Bisio}\ and\ \citenamefont {Perinotti}(2019)}]{bisio2019theoretical}%
  \BibitemOpen
  \bibfield  {author} {\bibinfo {author} {\bibfnamefont {A.}~\bibnamefont {Bisio}}\ and\ \bibinfo {author} {\bibfnamefont {P.}~\bibnamefont {Perinotti}},\ }\href {https://doi.org/10.1098/rspa.2018.0706} {\bibfield  {journal} {\bibinfo  {journal} {Proc. R. Soc. A}\ }\textbf {\bibinfo {volume} {475}},\ \bibinfo {pages} {20180706} (\bibinfo {year} {2019})}\BibitemShut {NoStop}%
\bibitem [{\citenamefont {Ying}(2016)}]{ying2016foundations}%
  \BibitemOpen
  \bibfield  {author} {\bibinfo {author} {\bibfnamefont {M.}~\bibnamefont {Ying}},\ }\href {https://doi.org/10.1016/C2022-0-02250-9} {\emph {\bibinfo {title} {Foundations of quantum programming}}}\ (\bibinfo  {publisher} {Morgan Kaufmann},\ \bibinfo {year} {2016})\BibitemShut {NoStop}%
\bibitem [{\citenamefont {Chitambar}\ and\ \citenamefont {Gour}(2019)}]{chitambar2019quantum}%
  \BibitemOpen
  \bibfield  {author} {\bibinfo {author} {\bibfnamefont {E.}~\bibnamefont {Chitambar}}\ and\ \bibinfo {author} {\bibfnamefont {G.}~\bibnamefont {Gour}},\ }\href {https://doi.org/10.1103/RevModPhys.91.025001} {\bibfield  {journal} {\bibinfo  {journal} {Rev. Mod. Phys.}\ }\textbf {\bibinfo {volume} {91}},\ \bibinfo {pages} {025001} (\bibinfo {year} {2019})}\BibitemShut {NoStop}%
\bibitem [{\citenamefont {Pollock}\ \emph {et~al.}(2018)\citenamefont {Pollock}, \citenamefont {Rodr{\'\i}guez-Rosario}, \citenamefont {Frauenheim}, \citenamefont {Paternostro},\ and\ \citenamefont {Modi}}]{pollock2018non}%
  \BibitemOpen
  \bibfield  {author} {\bibinfo {author} {\bibfnamefont {F.~A.}\ \bibnamefont {Pollock}}, \bibinfo {author} {\bibfnamefont {C.}~\bibnamefont {Rodr{\'\i}guez-Rosario}}, \bibinfo {author} {\bibfnamefont {T.}~\bibnamefont {Frauenheim}}, \bibinfo {author} {\bibfnamefont {M.}~\bibnamefont {Paternostro}},\ and\ \bibinfo {author} {\bibfnamefont {K.}~\bibnamefont {Modi}},\ }\href {https://doi.org/10.1103/PhysRevA.97.012127} {\bibfield  {journal} {\bibinfo  {journal} {Phys. Rev. A}\ }\textbf {\bibinfo {volume} {97}},\ \bibinfo {pages} {012127} (\bibinfo {year} {2018})}\BibitemShut {NoStop}%
\bibitem [{\citenamefont {Oreshkov}\ \emph {et~al.}(2012)\citenamefont {Oreshkov}, \citenamefont {Costa},\ and\ \citenamefont {Brukner}}]{oreshkov2012quantum}%
  \BibitemOpen
  \bibfield  {author} {\bibinfo {author} {\bibfnamefont {O.}~\bibnamefont {Oreshkov}}, \bibinfo {author} {\bibfnamefont {F.}~\bibnamefont {Costa}},\ and\ \bibinfo {author} {\bibfnamefont {{\v{C}}.}~\bibnamefont {Brukner}},\ }\href {https://doi.org/10.1038/ncomms2076} {\bibfield  {journal} {\bibinfo  {journal} {Nat. Commun.}\ }\textbf {\bibinfo {volume} {3}},\ \bibinfo {pages} {1092} (\bibinfo {year} {2012})}\BibitemShut {NoStop}%
\bibitem [{\citenamefont {Chiribella}\ and\ \citenamefont {Ebler}(2016)}]{chiribella2016optimal}%
  \BibitemOpen
  \bibfield  {author} {\bibinfo {author} {\bibfnamefont {G.}~\bibnamefont {Chiribella}}\ and\ \bibinfo {author} {\bibfnamefont {D.}~\bibnamefont {Ebler}},\ }\href {https://doi.org/10.1088/1367-2630/18/9/093053} {\bibfield  {journal} {\bibinfo  {journal} {New J. Phys.}\ }\textbf {\bibinfo {volume} {18}},\ \bibinfo {pages} {093053} (\bibinfo {year} {2016})}\BibitemShut {NoStop}%
\bibitem [{\citenamefont {Chiribella}\ \emph {et~al.}(2008{\natexlab{a}})\citenamefont {Chiribella}, \citenamefont {D’Ariano},\ and\ \citenamefont {Perinotti}}]{chiribella2008optimal}%
  \BibitemOpen
  \bibfield  {author} {\bibinfo {author} {\bibfnamefont {G.}~\bibnamefont {Chiribella}}, \bibinfo {author} {\bibfnamefont {G.~M.}\ \bibnamefont {D’Ariano}},\ and\ \bibinfo {author} {\bibfnamefont {P.}~\bibnamefont {Perinotti}},\ }\href {https://doi.org/10.1103/PhysRevLett.101.180504} {\bibfield  {journal} {\bibinfo  {journal} {Phys. Rev. Lett.}\ }\textbf {\bibinfo {volume} {101}},\ \bibinfo {pages} {180504} (\bibinfo {year} {2008}{\natexlab{a}})}\BibitemShut {NoStop}%
\bibitem [{\citenamefont {Bisio}\ \emph {et~al.}(2010)\citenamefont {Bisio}, \citenamefont {Chiribella}, \citenamefont {D’Ariano}, \citenamefont {Facchini},\ and\ \citenamefont {Perinotti}}]{bisio2010optimal}%
  \BibitemOpen
  \bibfield  {author} {\bibinfo {author} {\bibfnamefont {A.}~\bibnamefont {Bisio}}, \bibinfo {author} {\bibfnamefont {G.}~\bibnamefont {Chiribella}}, \bibinfo {author} {\bibfnamefont {G.~M.}\ \bibnamefont {D’Ariano}}, \bibinfo {author} {\bibfnamefont {S.}~\bibnamefont {Facchini}},\ and\ \bibinfo {author} {\bibfnamefont {P.}~\bibnamefont {Perinotti}},\ }\href {https://doi.org/10.1103/PhysRevA.81.032324} {\bibfield  {journal} {\bibinfo  {journal} {Phys. Rev. A}\ }\textbf {\bibinfo {volume} {81}},\ \bibinfo {pages} {032324} (\bibinfo {year} {2010})}\BibitemShut {NoStop}%
\bibitem [{\citenamefont {Sedl{\'a}k}\ \emph {et~al.}(2019)\citenamefont {Sedl{\'a}k}, \citenamefont {Bisio},\ and\ \citenamefont {Ziman}}]{sedlak2019optimal}%
  \BibitemOpen
  \bibfield  {author} {\bibinfo {author} {\bibfnamefont {M.}~\bibnamefont {Sedl{\'a}k}}, \bibinfo {author} {\bibfnamefont {A.}~\bibnamefont {Bisio}},\ and\ \bibinfo {author} {\bibfnamefont {M.}~\bibnamefont {Ziman}},\ }\href {https://doi.org/10.1103/PhysRevLett.122.170502} {\bibfield  {journal} {\bibinfo  {journal} {Phys. Rev. Lett.}\ }\textbf {\bibinfo {volume} {122}},\ \bibinfo {pages} {170502} (\bibinfo {year} {2019})}\BibitemShut {NoStop}%
\bibitem [{\citenamefont {Yang}\ \emph {et~al.}(2020)\citenamefont {Yang}, \citenamefont {Renner},\ and\ \citenamefont {Chiribella}}]{yang2020optimal}%
  \BibitemOpen
  \bibfield  {author} {\bibinfo {author} {\bibfnamefont {Y.}~\bibnamefont {Yang}}, \bibinfo {author} {\bibfnamefont {R.}~\bibnamefont {Renner}},\ and\ \bibinfo {author} {\bibfnamefont {G.}~\bibnamefont {Chiribella}},\ }\href {https://doi.org/10.1103/PhysRevLett.125.210501} {\bibfield  {journal} {\bibinfo  {journal} {Phys. Rev. Lett.}\ }\textbf {\bibinfo {volume} {125}},\ \bibinfo {pages} {210501} (\bibinfo {year} {2020})}\BibitemShut {NoStop}%
\bibitem [{\citenamefont {Sedl{\'a}k}\ and\ \citenamefont {Ziman}(2020)}]{sedlak2020probabilistic}%
  \BibitemOpen
  \bibfield  {author} {\bibinfo {author} {\bibfnamefont {M.}~\bibnamefont {Sedl{\'a}k}}\ and\ \bibinfo {author} {\bibfnamefont {M.}~\bibnamefont {Ziman}},\ }\href {https://doi.org/10.1103/PhysRevA.102.032618} {\bibfield  {journal} {\bibinfo  {journal} {Phys. Rev. A}\ }\textbf {\bibinfo {volume} {102}},\ \bibinfo {pages} {032618} (\bibinfo {year} {2020})}\BibitemShut {NoStop}%
\bibitem [{\citenamefont {Bisio}\ \emph {et~al.}(2014)\citenamefont {Bisio}, \citenamefont {D'Ariano}, \citenamefont {Perinotti},\ and\ \citenamefont {Sedl{\'a}k}}]{bisio2014optimal}%
  \BibitemOpen
  \bibfield  {author} {\bibinfo {author} {\bibfnamefont {A.}~\bibnamefont {Bisio}}, \bibinfo {author} {\bibfnamefont {G.~M.}\ \bibnamefont {D'Ariano}}, \bibinfo {author} {\bibfnamefont {P.}~\bibnamefont {Perinotti}},\ and\ \bibinfo {author} {\bibfnamefont {M.}~\bibnamefont {Sedl{\'a}k}},\ }\href {https://doi.org/10.1016/j.physleta.2014.04.042} {\bibfield  {journal} {\bibinfo  {journal} {Phys. Lett. A}\ }\textbf {\bibinfo {volume} {378}},\ \bibinfo {pages} {1797} (\bibinfo {year} {2014})}\BibitemShut {NoStop}%
\bibitem [{\citenamefont {D{\"u}r}\ \emph {et~al.}(2015)\citenamefont {D{\"u}r}, \citenamefont {Sekatski},\ and\ \citenamefont {Skotiniotis}}]{dur2015deterministic}%
  \BibitemOpen
  \bibfield  {author} {\bibinfo {author} {\bibfnamefont {W.}~\bibnamefont {D{\"u}r}}, \bibinfo {author} {\bibfnamefont {P.}~\bibnamefont {Sekatski}},\ and\ \bibinfo {author} {\bibfnamefont {M.}~\bibnamefont {Skotiniotis}},\ }\href {https://doi.org/10.1103/PhysRevLett.114.120503} {\bibfield  {journal} {\bibinfo  {journal} {Phys. Rev. Lett.}\ }\textbf {\bibinfo {volume} {114}},\ \bibinfo {pages} {120503} (\bibinfo {year} {2015})}\BibitemShut {NoStop}%
\bibitem [{\citenamefont {Chiribella}\ \emph {et~al.}(2015)\citenamefont {Chiribella}, \citenamefont {Yang},\ and\ \citenamefont {Huang}}]{chiribella2015universal}%
  \BibitemOpen
  \bibfield  {author} {\bibinfo {author} {\bibfnamefont {G.}~\bibnamefont {Chiribella}}, \bibinfo {author} {\bibfnamefont {Y.}~\bibnamefont {Yang}},\ and\ \bibinfo {author} {\bibfnamefont {C.}~\bibnamefont {Huang}},\ }\href {https://doi.org/10.1103/PhysRevLett.114.120504} {\bibfield  {journal} {\bibinfo  {journal} {Phys. Rev. Lett.}\ }\textbf {\bibinfo {volume} {114}},\ \bibinfo {pages} {120504} (\bibinfo {year} {2015})}\BibitemShut {NoStop}%
\bibitem [{\citenamefont {Soleimanifar}\ and\ \citenamefont {Karimipour}(2016)}]{soleimanifar2016no}%
  \BibitemOpen
  \bibfield  {author} {\bibinfo {author} {\bibfnamefont {M.}~\bibnamefont {Soleimanifar}}\ and\ \bibinfo {author} {\bibfnamefont {V.}~\bibnamefont {Karimipour}},\ }\href {https://doi.org/10.1103/PhysRevA.93.012344} {\bibfield  {journal} {\bibinfo  {journal} {Phys. Rev. A}\ }\textbf {\bibinfo {volume} {93}},\ \bibinfo {pages} {012344} (\bibinfo {year} {2016})}\BibitemShut {NoStop}%
\bibitem [{\citenamefont {Miyazaki}\ \emph {et~al.}(2019)\citenamefont {Miyazaki}, \citenamefont {Soeda},\ and\ \citenamefont {Murao}}]{miyazaki2019complex}%
  \BibitemOpen
  \bibfield  {author} {\bibinfo {author} {\bibfnamefont {J.}~\bibnamefont {Miyazaki}}, \bibinfo {author} {\bibfnamefont {A.}~\bibnamefont {Soeda}},\ and\ \bibinfo {author} {\bibfnamefont {M.}~\bibnamefont {Murao}},\ }\href {https://doi.org/10.1103/PhysRevResearch.1.013007} {\bibfield  {journal} {\bibinfo  {journal} {Phys. Rev. Res.}\ }\textbf {\bibinfo {volume} {1}},\ \bibinfo {pages} {013007} (\bibinfo {year} {2019})}\BibitemShut {NoStop}%
\bibitem [{\citenamefont {Ebler}\ \emph {et~al.}(2023)\citenamefont {Ebler}, \citenamefont {Horodecki}, \citenamefont {Marciniak}, \citenamefont {M{\l}ynik}, \citenamefont {Quintino},\ and\ \citenamefont {Studzi{\'n}ski}}]{ebler2023optimal}%
  \BibitemOpen
  \bibfield  {author} {\bibinfo {author} {\bibfnamefont {D.}~\bibnamefont {Ebler}}, \bibinfo {author} {\bibfnamefont {M.}~\bibnamefont {Horodecki}}, \bibinfo {author} {\bibfnamefont {M.}~\bibnamefont {Marciniak}}, \bibinfo {author} {\bibfnamefont {T.}~\bibnamefont {M{\l}ynik}}, \bibinfo {author} {\bibfnamefont {M.~T.}\ \bibnamefont {Quintino}},\ and\ \bibinfo {author} {\bibfnamefont {M.}~\bibnamefont {Studzi{\'n}ski}},\ }\href {https://doi.org/10.1109/TIT.2023.3263771} {\bibfield  {journal} {\bibinfo  {journal} {IEEE Trans. Inf. Theory}\ }\textbf {\bibinfo {volume} {69}},\ \bibinfo {pages} {5069} (\bibinfo {year} {2023})}\BibitemShut {NoStop}%
\bibitem [{\citenamefont {Ara{\'u}jo}\ \emph {et~al.}(2014)\citenamefont {Ara{\'u}jo}, \citenamefont {Feix}, \citenamefont {Costa},\ and\ \citenamefont {Brukner}}]{araujo2014quantum}%
  \BibitemOpen
  \bibfield  {author} {\bibinfo {author} {\bibfnamefont {M.}~\bibnamefont {Ara{\'u}jo}}, \bibinfo {author} {\bibfnamefont {A.}~\bibnamefont {Feix}}, \bibinfo {author} {\bibfnamefont {F.}~\bibnamefont {Costa}},\ and\ \bibinfo {author} {\bibfnamefont {{\v{C}}.}~\bibnamefont {Brukner}},\ }\href {https://doi.org/10.1088/1367-2630/16/9/093026} {\bibfield  {journal} {\bibinfo  {journal} {New J. Phys.}\ }\textbf {\bibinfo {volume} {16}},\ \bibinfo {pages} {093026} (\bibinfo {year} {2014})}\BibitemShut {NoStop}%
\bibitem [{\citenamefont {Bisio}\ \emph {et~al.}(2016)\citenamefont {Bisio}, \citenamefont {Dall'Arno},\ and\ \citenamefont {Perinotti}}]{bisio2016quantum}%
  \BibitemOpen
  \bibfield  {author} {\bibinfo {author} {\bibfnamefont {A.}~\bibnamefont {Bisio}}, \bibinfo {author} {\bibfnamefont {M.}~\bibnamefont {Dall'Arno}},\ and\ \bibinfo {author} {\bibfnamefont {P.}~\bibnamefont {Perinotti}},\ }\href {https://doi.org/10.1103/PhysRevA.94.022340} {\bibfield  {journal} {\bibinfo  {journal} {Phys. Rev. A}\ }\textbf {\bibinfo {volume} {94}},\ \bibinfo {pages} {022340} (\bibinfo {year} {2016})}\BibitemShut {NoStop}%
\bibitem [{\citenamefont {Dong}\ \emph {et~al.}(2019)\citenamefont {Dong}, \citenamefont {Nakayama}, \citenamefont {Soeda},\ and\ \citenamefont {Murao}}]{dong2019controlled}%
  \BibitemOpen
  \bibfield  {author} {\bibinfo {author} {\bibfnamefont {Q.}~\bibnamefont {Dong}}, \bibinfo {author} {\bibfnamefont {S.}~\bibnamefont {Nakayama}}, \bibinfo {author} {\bibfnamefont {A.}~\bibnamefont {Soeda}},\ and\ \bibinfo {author} {\bibfnamefont {M.}~\bibnamefont {Murao}},\ }\Eprint {https://arxiv.org/abs/1911.01645} {arXiv:1911.01645}  (\bibinfo {year} {2019})\BibitemShut {NoStop}%
\bibitem [{\citenamefont {Dong}\ \emph {et~al.}(2021)\citenamefont {Dong}, \citenamefont {Quintino}, \citenamefont {Soeda},\ and\ \citenamefont {Murao}}]{dong2021success}%
  \BibitemOpen
  \bibfield  {author} {\bibinfo {author} {\bibfnamefont {Q.}~\bibnamefont {Dong}}, \bibinfo {author} {\bibfnamefont {M.~T.}\ \bibnamefont {Quintino}}, \bibinfo {author} {\bibfnamefont {A.}~\bibnamefont {Soeda}},\ and\ \bibinfo {author} {\bibfnamefont {M.}~\bibnamefont {Murao}},\ }\href {https://doi.org/10.1103/PhysRevLett.126.150504} {\bibfield  {journal} {\bibinfo  {journal} {Phys. Rev. Lett.}\ }\textbf {\bibinfo {volume} {126}},\ \bibinfo {pages} {150504} (\bibinfo {year} {2021})}\BibitemShut {NoStop}%
\bibitem [{\citenamefont {Sardharwalla}\ \emph {et~al.}(2016)\citenamefont {Sardharwalla}, \citenamefont {Cubitt}, \citenamefont {Harrow},\ and\ \citenamefont {Linden}}]{sardharwalla2016universal}%
  \BibitemOpen
  \bibfield  {author} {\bibinfo {author} {\bibfnamefont {I.~S.}\ \bibnamefont {Sardharwalla}}, \bibinfo {author} {\bibfnamefont {T.~S.}\ \bibnamefont {Cubitt}}, \bibinfo {author} {\bibfnamefont {A.~W.}\ \bibnamefont {Harrow}},\ and\ \bibinfo {author} {\bibfnamefont {N.}~\bibnamefont {Linden}},\ }\Eprint {https://arxiv.org/abs/1602.07963} {arXiv:1602.07963}  (\bibinfo {year} {2016})\BibitemShut {NoStop}%
\bibitem [{\citenamefont {Quintino}\ \emph {et~al.}(2019{\natexlab{a}})\citenamefont {Quintino}, \citenamefont {Dong}, \citenamefont {Shimbo}, \citenamefont {Soeda},\ and\ \citenamefont {Murao}}]{quintino2019probabilistic}%
  \BibitemOpen
  \bibfield  {author} {\bibinfo {author} {\bibfnamefont {M.~T.}\ \bibnamefont {Quintino}}, \bibinfo {author} {\bibfnamefont {Q.}~\bibnamefont {Dong}}, \bibinfo {author} {\bibfnamefont {A.}~\bibnamefont {Shimbo}}, \bibinfo {author} {\bibfnamefont {A.}~\bibnamefont {Soeda}},\ and\ \bibinfo {author} {\bibfnamefont {M.}~\bibnamefont {Murao}},\ }\href {https://doi.org/10.1103/PhysRevA.100.062339} {\bibfield  {journal} {\bibinfo  {journal} {Phys. Rev. A}\ }\textbf {\bibinfo {volume} {100}},\ \bibinfo {pages} {062339} (\bibinfo {year} {2019}{\natexlab{a}})}\BibitemShut {NoStop}%
\bibitem [{\citenamefont {Quintino}\ \emph {et~al.}(2019{\natexlab{b}})\citenamefont {Quintino}, \citenamefont {Dong}, \citenamefont {Shimbo}, \citenamefont {Soeda},\ and\ \citenamefont {Murao}}]{quintino2019reversing}%
  \BibitemOpen
  \bibfield  {author} {\bibinfo {author} {\bibfnamefont {M.~T.}\ \bibnamefont {Quintino}}, \bibinfo {author} {\bibfnamefont {Q.}~\bibnamefont {Dong}}, \bibinfo {author} {\bibfnamefont {A.}~\bibnamefont {Shimbo}}, \bibinfo {author} {\bibfnamefont {A.}~\bibnamefont {Soeda}},\ and\ \bibinfo {author} {\bibfnamefont {M.}~\bibnamefont {Murao}},\ }\href {https://doi.org/10.1103/PhysRevLett.123.210502} {\bibfield  {journal} {\bibinfo  {journal} {Phys. Rev. Lett.}\ }\textbf {\bibinfo {volume} {123}},\ \bibinfo {pages} {210502} (\bibinfo {year} {2019}{\natexlab{b}})}\BibitemShut {NoStop}%
\bibitem [{\citenamefont {Quintino}\ and\ \citenamefont {Ebler}(2022)}]{quintino2022deterministic}%
  \BibitemOpen
  \bibfield  {author} {\bibinfo {author} {\bibfnamefont {M.~T.}\ \bibnamefont {Quintino}}\ and\ \bibinfo {author} {\bibfnamefont {D.}~\bibnamefont {Ebler}},\ }\href {https://doi.org/10.22331/q-2022-03-31-679} {\bibfield  {journal} {\bibinfo  {journal} {Quantum}\ }\textbf {\bibinfo {volume} {6}},\ \bibinfo {pages} {679} (\bibinfo {year} {2022})}\BibitemShut {NoStop}%
\bibitem [{\citenamefont {Yoshida}\ \emph {et~al.}(2023{\natexlab{a}})\citenamefont {Yoshida}, \citenamefont {Soeda},\ and\ \citenamefont {Murao}}]{yoshida2023reversing}%
  \BibitemOpen
  \bibfield  {author} {\bibinfo {author} {\bibfnamefont {S.}~\bibnamefont {Yoshida}}, \bibinfo {author} {\bibfnamefont {A.}~\bibnamefont {Soeda}},\ and\ \bibinfo {author} {\bibfnamefont {M.}~\bibnamefont {Murao}},\ }\href {https://doi.org/10.1103/PhysRevLett.131.120602} {\bibfield  {journal} {\bibinfo  {journal} {Phys. Rev. Lett.}\ }\textbf {\bibinfo {volume} {131}},\ \bibinfo {pages} {120602} (\bibinfo {year} {2023}{\natexlab{a}})}\BibitemShut {NoStop}%
\bibitem [{\citenamefont {Navascu{\'e}s}(2018)}]{navascues2018resetting}%
  \BibitemOpen
  \bibfield  {author} {\bibinfo {author} {\bibfnamefont {M.}~\bibnamefont {Navascu{\'e}s}},\ }\href {https://doi.org/10.1103/PhysRevX.8.031008} {\bibfield  {journal} {\bibinfo  {journal} {Phys. Rev. X}\ }\textbf {\bibinfo {volume} {8}},\ \bibinfo {pages} {031008} (\bibinfo {year} {2018})}\BibitemShut {NoStop}%
\bibitem [{\citenamefont {Trillo}\ \emph {et~al.}(2020)\citenamefont {Trillo}, \citenamefont {Dive},\ and\ \citenamefont {Navascu{\'e}s}}]{trillo2020translating}%
  \BibitemOpen
  \bibfield  {author} {\bibinfo {author} {\bibfnamefont {D.}~\bibnamefont {Trillo}}, \bibinfo {author} {\bibfnamefont {B.}~\bibnamefont {Dive}},\ and\ \bibinfo {author} {\bibfnamefont {M.}~\bibnamefont {Navascu{\'e}s}},\ }\href {https://doi.org/10.22331/q-2020-12-15-374} {\bibfield  {journal} {\bibinfo  {journal} {Quantum}\ }\textbf {\bibinfo {volume} {4}},\ \bibinfo {pages} {374} (\bibinfo {year} {2020})}\BibitemShut {NoStop}%
\bibitem [{\citenamefont {Trillo}\ \emph {et~al.}(2023)\citenamefont {Trillo}, \citenamefont {Dive},\ and\ \citenamefont {Navascu{\'e}s}}]{trillo2023universal}%
  \BibitemOpen
  \bibfield  {author} {\bibinfo {author} {\bibfnamefont {D.}~\bibnamefont {Trillo}}, \bibinfo {author} {\bibfnamefont {B.}~\bibnamefont {Dive}},\ and\ \bibinfo {author} {\bibfnamefont {M.}~\bibnamefont {Navascu{\'e}s}},\ }\href {https://doi.org/10.1103/PhysRevLett.130.110201} {\bibfield  {journal} {\bibinfo  {journal} {Phys. Rev. Lett.}\ }\textbf {\bibinfo {volume} {130}},\ \bibinfo {pages} {110201} (\bibinfo {year} {2023})}\BibitemShut {NoStop}%
\bibitem [{\citenamefont {Chen}\ \emph {et~al.}(2024)\citenamefont {Chen}, \citenamefont {Mo}, \citenamefont {Liu}, \citenamefont {Zhang},\ and\ \citenamefont {Wang}}]{chen2024quantum}%
  \BibitemOpen
  \bibfield  {author} {\bibinfo {author} {\bibfnamefont {Y.-A.}\ \bibnamefont {Chen}}, \bibinfo {author} {\bibfnamefont {Y.}~\bibnamefont {Mo}}, \bibinfo {author} {\bibfnamefont {Y.}~\bibnamefont {Liu}}, \bibinfo {author} {\bibfnamefont {L.}~\bibnamefont {Zhang}},\ and\ \bibinfo {author} {\bibfnamefont {X.}~\bibnamefont {Wang}},\ }\Eprint {https://arxiv.org/abs/2403.04704} {arXiv:2403.04704}  (\bibinfo {year} {2024})\BibitemShut {NoStop}%
\bibitem [{\citenamefont {Bisio}\ \emph {et~al.}(2011)\citenamefont {Bisio}, \citenamefont {D’Ariano}, \citenamefont {Perinotti},\ and\ \citenamefont {Sedl{\'a}k}}]{bisio2011cloning}%
  \BibitemOpen
  \bibfield  {author} {\bibinfo {author} {\bibfnamefont {A.}~\bibnamefont {Bisio}}, \bibinfo {author} {\bibfnamefont {G.~M.}\ \bibnamefont {D’Ariano}}, \bibinfo {author} {\bibfnamefont {P.}~\bibnamefont {Perinotti}},\ and\ \bibinfo {author} {\bibfnamefont {M.}~\bibnamefont {Sedl{\'a}k}},\ }\href {https://doi.org/10.1103/PhysRevA.84.042330} {\bibfield  {journal} {\bibinfo  {journal} {Phys. Rev. A}\ }\textbf {\bibinfo {volume} {84}},\ \bibinfo {pages} {042330} (\bibinfo {year} {2011})}\BibitemShut {NoStop}%
\bibitem [{\citenamefont {Yoshida}\ \emph {et~al.}(2023{\natexlab{b}})\citenamefont {Yoshida}, \citenamefont {Soeda},\ and\ \citenamefont {Murao}}]{yoshida2023universal}%
  \BibitemOpen
  \bibfield  {author} {\bibinfo {author} {\bibfnamefont {S.}~\bibnamefont {Yoshida}}, \bibinfo {author} {\bibfnamefont {A.}~\bibnamefont {Soeda}},\ and\ \bibinfo {author} {\bibfnamefont {M.}~\bibnamefont {Murao}},\ }\href {https://doi.org/10.22331/q-2023-03-20-957} {\bibfield  {journal} {\bibinfo  {journal} {Quantum}\ }\textbf {\bibinfo {volume} {7}},\ \bibinfo {pages} {957} (\bibinfo {year} {2023}{\natexlab{b}})}\BibitemShut {NoStop}%
\bibitem [{\citenamefont {Nielsen}\ and\ \citenamefont {Chuang}(2010)}]{nielsen2002quantum}%
  \BibitemOpen
  \bibfield  {author} {\bibinfo {author} {\bibfnamefont {M.~A.}\ \bibnamefont {Nielsen}}\ and\ \bibinfo {author} {\bibfnamefont {I.~L.}\ \bibnamefont {Chuang}},\ }\href {https://doi.org/10.1017/CBO9780511976667} {\emph {\bibinfo {title} {Quantum Computation and Quantum Information, 10th Anniversary Edition}}}\ (\bibinfo  {publisher} {Cambridge University Press, Cambridge, England},\ \bibinfo {year} {2010})\BibitemShut {NoStop}%
\bibitem [{\citenamefont {Stinespring}(1955)}]{stinespring1955positive}%
  \BibitemOpen
  \bibfield  {author} {\bibinfo {author} {\bibfnamefont {W.~F.}\ \bibnamefont {Stinespring}},\ }\href {https://doi.org/10.2307/2032342} {\bibfield  {journal} {\bibinfo  {journal} {Proc. Amer. Math. Soc.}\ }\textbf {\bibinfo {volume} {6}},\ \bibinfo {pages} {211} (\bibinfo {year} {1955})}\BibitemShut {NoStop}%
\bibitem [{\citenamefont {Hardy}(2007)}]{hardy2007towards}%
  \BibitemOpen
  \bibfield  {author} {\bibinfo {author} {\bibfnamefont {L.}~\bibnamefont {Hardy}},\ }\href {https://doi.org/10.1088/1751-8113/40/12/S12} {\bibfield  {journal} {\bibinfo  {journal} {J. Phys. A}\ }\textbf {\bibinfo {volume} {40}},\ \bibinfo {pages} {3081} (\bibinfo {year} {2007})}\BibitemShut {NoStop}%
\bibitem [{\citenamefont {Chiribella}\ \emph {et~al.}(2013{\natexlab{a}})\citenamefont {Chiribella}, \citenamefont {D’Ariano}, \citenamefont {Perinotti},\ and\ \citenamefont {Valiron}}]{chiribella2013quantum}%
  \BibitemOpen
  \bibfield  {author} {\bibinfo {author} {\bibfnamefont {G.}~\bibnamefont {Chiribella}}, \bibinfo {author} {\bibfnamefont {G.~M.}\ \bibnamefont {D’Ariano}}, \bibinfo {author} {\bibfnamefont {P.}~\bibnamefont {Perinotti}},\ and\ \bibinfo {author} {\bibfnamefont {B.}~\bibnamefont {Valiron}},\ }\href {https://doi.org/10.1103/PhysRevA.88.022318} {\bibfield  {journal} {\bibinfo  {journal} {Phys. Rev. A}\ }\textbf {\bibinfo {volume} {88}},\ \bibinfo {pages} {022318} (\bibinfo {year} {2013}{\natexlab{a}})}\BibitemShut {NoStop}%
\bibitem [{\citenamefont {Wilde}(2013)}]{wilde2013quantum}%
  \BibitemOpen
  \bibfield  {author} {\bibinfo {author} {\bibfnamefont {M.~M.}\ \bibnamefont {Wilde}},\ }\href {https://doi.org/10.1017/CBO9781139525343} {\emph {\bibinfo {title} {Quantum Information Theory}}}\ (\bibinfo  {publisher} {Cambridge University Press, Cambridge, England},\ \bibinfo {year} {2013})\BibitemShut {NoStop}%
\bibitem [{\citenamefont {Chiribella}\ \emph {et~al.}(2008{\natexlab{b}})\citenamefont {Chiribella}, \citenamefont {D'Ariano},\ and\ \citenamefont {Perinotti}}]{chiribella2008transforming}%
  \BibitemOpen
  \bibfield  {author} {\bibinfo {author} {\bibfnamefont {G.}~\bibnamefont {Chiribella}}, \bibinfo {author} {\bibfnamefont {G.~M.}\ \bibnamefont {D'Ariano}},\ and\ \bibinfo {author} {\bibfnamefont {P.}~\bibnamefont {Perinotti}},\ }\href {https://doi.org/10.1209/0295-5075/83/30004} {\bibfield  {journal} {\bibinfo  {journal} {Europhys. Lett.}\ }\textbf {\bibinfo {volume} {83}},\ \bibinfo {pages} {30004} (\bibinfo {year} {2008}{\natexlab{b}})}\BibitemShut {NoStop}%
\bibitem [{\citenamefont {Raginsky}(2001)}]{raginsky2001fidelity}%
  \BibitemOpen
  \bibfield  {author} {\bibinfo {author} {\bibfnamefont {M.}~\bibnamefont {Raginsky}},\ }\href {https://doi.org/10.1016/S0375-9601(01)00640-5} {\bibfield  {journal} {\bibinfo  {journal} {Phys. Lett. A}\ }\textbf {\bibinfo {volume} {290}},\ \bibinfo {pages} {11} (\bibinfo {year} {2001})}\BibitemShut {NoStop}%
\bibitem [{\citenamefont {Choi}(1975)}]{choi1975completely}%
  \BibitemOpen
  \bibfield  {author} {\bibinfo {author} {\bibfnamefont {M.-D.}\ \bibnamefont {Choi}},\ }\href {https://doi.org/10.1016/0024-3795(75)90075-0} {\bibfield  {journal} {\bibinfo  {journal} {Linear Algebra Appl.}\ }\textbf {\bibinfo {volume} {10}},\ \bibinfo {pages} {285} (\bibinfo {year} {1975})}\BibitemShut {NoStop}%
\bibitem [{\citenamefont {Jamio{\l}kowski}(1972)}]{jamiolkowski1972linear}%
  \BibitemOpen
  \bibfield  {author} {\bibinfo {author} {\bibfnamefont {A.}~\bibnamefont {Jamio{\l}kowski}},\ }\href {https://doi.org/10.1016/0034-4877(72)90011-0} {\bibfield  {journal} {\bibinfo  {journal} {Rep. Math. Phys.}\ }\textbf {\bibinfo {volume} {3}},\ \bibinfo {pages} {275} (\bibinfo {year} {1972})}\BibitemShut {NoStop}%
\bibitem [{\citenamefont {Navascu{\'e}s}\ and\ \citenamefont {Popescu}(2014)}]{navascues2014energy}%
  \BibitemOpen
  \bibfield  {author} {\bibinfo {author} {\bibfnamefont {M.}~\bibnamefont {Navascu{\'e}s}}\ and\ \bibinfo {author} {\bibfnamefont {S.}~\bibnamefont {Popescu}},\ }\href {https://doi.org/10.1103/PhysRevLett.112.140502} {\bibfield  {journal} {\bibinfo  {journal} {Phys. Rev. Lett.}\ }\textbf {\bibinfo {volume} {112}},\ \bibinfo {pages} {140502} (\bibinfo {year} {2014})}\BibitemShut {NoStop}%
\bibitem [{\citenamefont {Pucha{\l}a}\ \emph {et~al.}(2018)\citenamefont {Pucha{\l}a}, \citenamefont {Pawela}, \citenamefont {Krawiec},\ and\ \citenamefont {Kukulski}}]{puchala2018strategies}%
  \BibitemOpen
  \bibfield  {author} {\bibinfo {author} {\bibfnamefont {Z.}~\bibnamefont {Pucha{\l}a}}, \bibinfo {author} {\bibfnamefont {{\L}.}~\bibnamefont {Pawela}}, \bibinfo {author} {\bibfnamefont {A.}~\bibnamefont {Krawiec}},\ and\ \bibinfo {author} {\bibfnamefont {R.}~\bibnamefont {Kukulski}},\ }\href {https://doi.org/10.1103/PhysRevA.98.042103} {\bibfield  {journal} {\bibinfo  {journal} {Phys. Rev. A}\ }\textbf {\bibinfo {volume} {98}},\ \bibinfo {pages} {042103} (\bibinfo {year} {2018})}\BibitemShut {NoStop}%
\bibitem [{\citenamefont {Pucha{\l}a}\ \emph {et~al.}(2021)\citenamefont {Pucha{\l}a}, \citenamefont {Pawela}, \citenamefont {Krawiec}, \citenamefont {Kukulski},\ and\ \citenamefont {Oszmaniec}}]{puchala2021multiple}%
  \BibitemOpen
  \bibfield  {author} {\bibinfo {author} {\bibfnamefont {Z.}~\bibnamefont {Pucha{\l}a}}, \bibinfo {author} {\bibfnamefont {{\L}.}~\bibnamefont {Pawela}}, \bibinfo {author} {\bibfnamefont {A.}~\bibnamefont {Krawiec}}, \bibinfo {author} {\bibfnamefont {R.}~\bibnamefont {Kukulski}},\ and\ \bibinfo {author} {\bibfnamefont {M.}~\bibnamefont {Oszmaniec}},\ }\href {https://doi.org/10.22331/q-2021-04-06-425} {\bibfield  {journal} {\bibinfo  {journal} {Quantum}\ }\textbf {\bibinfo {volume} {5}},\ \bibinfo {pages} {425} (\bibinfo {year} {2021})}\BibitemShut {NoStop}%
\bibitem [{\citenamefont {Kitaev}(1997)}]{kitaev1997quantum}%
  \BibitemOpen
  \bibfield  {author} {\bibinfo {author} {\bibfnamefont {A.~Y.}\ \bibnamefont {Kitaev}},\ }\href {https://doi.org/10.1070/RM1997v052n06ABEH002155} {\bibfield  {journal} {\bibinfo  {journal} {Russ. Math. Surv.}\ }\textbf {\bibinfo {volume} {52}},\ \bibinfo {pages} {1191} (\bibinfo {year} {1997})}\BibitemShut {NoStop}%
\bibitem [{\citenamefont {Chiribella}\ \emph {et~al.}(2009)\citenamefont {Chiribella}, \citenamefont {D’Ariano},\ and\ \citenamefont {Perinotti}}]{chiribella2009theoretical}%
  \BibitemOpen
  \bibfield  {author} {\bibinfo {author} {\bibfnamefont {G.}~\bibnamefont {Chiribella}}, \bibinfo {author} {\bibfnamefont {G.~M.}\ \bibnamefont {D’Ariano}},\ and\ \bibinfo {author} {\bibfnamefont {P.}~\bibnamefont {Perinotti}},\ }\href {https://doi.org/10.1103/PhysRevA.80.022339} {\bibfield  {journal} {\bibinfo  {journal} {Phys. Rev. A}\ }\textbf {\bibinfo {volume} {80}},\ \bibinfo {pages} {022339} (\bibinfo {year} {2009})}\BibitemShut {NoStop}%
\bibitem [{\citenamefont {Chiribella}\ \emph {et~al.}(2008{\natexlab{c}})\citenamefont {Chiribella}, \citenamefont {D’Ariano},\ and\ \citenamefont {Perinotti}}]{chiribella2008quantum}%
  \BibitemOpen
  \bibfield  {author} {\bibinfo {author} {\bibfnamefont {G.}~\bibnamefont {Chiribella}}, \bibinfo {author} {\bibfnamefont {G.~M.}\ \bibnamefont {D’Ariano}},\ and\ \bibinfo {author} {\bibfnamefont {P.}~\bibnamefont {Perinotti}},\ }\href {https://doi.org/10.1103/PhysRevLett.101.060401} {\bibfield  {journal} {\bibinfo  {journal} {Phys. Rev. Lett.}\ }\textbf {\bibinfo {volume} {101}},\ \bibinfo {pages} {060401} (\bibinfo {year} {2008}{\natexlab{c}})}\BibitemShut {NoStop}%
\bibitem [{\citenamefont {Ceccherini-Silberstein}\ \emph {et~al.}(2010)\citenamefont {Ceccherini-Silberstein}, \citenamefont {Scarabotti},\ and\ \citenamefont {Tolli}}]{ceccherini2010representation}%
  \BibitemOpen
  \bibfield  {author} {\bibinfo {author} {\bibfnamefont {T.}~\bibnamefont {Ceccherini-Silberstein}}, \bibinfo {author} {\bibfnamefont {F.}~\bibnamefont {Scarabotti}},\ and\ \bibinfo {author} {\bibfnamefont {F.}~\bibnamefont {Tolli}},\ }\href {https://doi.org/10.1017/CBO9781139192361} {\emph {\bibinfo {title} {Representation theory of the symmetric groups: the Okounkov-Vershik approach, character formulas, and partition algebras}}},\ Vol.\ \bibinfo {volume} {121}\ (\bibinfo  {publisher} {Cambridge University Press, Cambridge, England},\ \bibinfo {year} {2010})\BibitemShut {NoStop}%
\bibitem [{\citenamefont {Bacon}\ \emph {et~al.}(2006)\citenamefont {Bacon}, \citenamefont {Chuang},\ and\ \citenamefont {Harrow}}]{bacon2006efficient}%
  \BibitemOpen
  \bibfield  {author} {\bibinfo {author} {\bibfnamefont {D.}~\bibnamefont {Bacon}}, \bibinfo {author} {\bibfnamefont {I.~L.}\ \bibnamefont {Chuang}},\ and\ \bibinfo {author} {\bibfnamefont {A.~W.}\ \bibnamefont {Harrow}},\ }\href {https://doi.org/10.1103/PhysRevLett.97.170502} {\bibfield  {journal} {\bibinfo  {journal} {Phys. Rev. Lett.}\ }\textbf {\bibinfo {volume} {97}},\ \bibinfo {pages} {170502} (\bibinfo {year} {2006})}\BibitemShut {NoStop}%
\bibitem [{\citenamefont {Bacon}\ \emph {et~al.}(2007)\citenamefont {Bacon}, \citenamefont {Chuang},\ and\ \citenamefont {Harrow}}]{bacon2007quantum}%
  \BibitemOpen
  \bibfield  {author} {\bibinfo {author} {\bibfnamefont {D.}~\bibnamefont {Bacon}}, \bibinfo {author} {\bibfnamefont {I.~L.}\ \bibnamefont {Chuang}},\ and\ \bibinfo {author} {\bibfnamefont {A.~W.}\ \bibnamefont {Harrow}},\ }in\ \href {https://dl.acm.org/doi/10.5555/1283383.1283516} {\emph {\bibinfo {booktitle} {Proceedings of the Eighteenth Annual ACM-SIAM Symposium on Discrete Algorithms}}},\ \bibinfo {series and number} {SODA '07}\ (\bibinfo  {publisher} {Society for Industrial and Applied Mathematics},\ \bibinfo {address} {USA},\ \bibinfo {year} {2007})\ p.\ \bibinfo {pages} {1235–1244}\BibitemShut {NoStop}%
\bibitem [{\citenamefont {Krovi}(2019)}]{krovi2019efficient}%
  \BibitemOpen
  \bibfield  {author} {\bibinfo {author} {\bibfnamefont {H.}~\bibnamefont {Krovi}},\ }\href {https://doi.org/10.22331/q-2019-02-14-122} {\bibfield  {journal} {\bibinfo  {journal} {Quantum}\ }\textbf {\bibinfo {volume} {3}},\ \bibinfo {pages} {122} (\bibinfo {year} {2019})}\BibitemShut {NoStop}%
\bibitem [{\citenamefont {Kirby}\ and\ \citenamefont {Strauch}(2018)}]{kirby2018practical}%
  \BibitemOpen
  \bibfield  {author} {\bibinfo {author} {\bibfnamefont {W.~M.}\ \bibnamefont {Kirby}}\ and\ \bibinfo {author} {\bibfnamefont {F.~W.}\ \bibnamefont {Strauch}},\ }\href {https://doi.org/10.26421/QIC18.9-10-1} {\bibfield  {journal} {\bibinfo  {journal} {Quantum Inf. Comput.}\ }\textbf {\bibinfo {volume} {18}},\ \bibinfo {pages} {0721} (\bibinfo {year} {2018})}\BibitemShut {NoStop}%
\bibitem [{\citenamefont {Pearce-Crump}(2022)}]{pearce2022multigraph}%
  \BibitemOpen
  \bibfield  {author} {\bibinfo {author} {\bibfnamefont {E.}~\bibnamefont {Pearce-Crump}},\ }\Eprint {https://arxiv.org/abs/2204.10694} {arXiv:2204.10694}  (\bibinfo {year} {2022})\BibitemShut {NoStop}%
\bibitem [{\citenamefont {Wills}\ and\ \citenamefont {Strelchuk}(2023)}]{wills2023generalised}%
  \BibitemOpen
  \bibfield  {author} {\bibinfo {author} {\bibfnamefont {A.}~\bibnamefont {Wills}}\ and\ \bibinfo {author} {\bibfnamefont {S.}~\bibnamefont {Strelchuk}},\ }\Eprint {https://arxiv.org/abs/2305.04069} {arXiv:2305.04069}  (\bibinfo {year} {2023})\BibitemShut {NoStop}%
\bibitem [{\citenamefont {Ram}\ and\ \citenamefont {Wenzl}(1992)}]{ram1992matrix}%
  \BibitemOpen
  \bibfield  {author} {\bibinfo {author} {\bibfnamefont {A.}~\bibnamefont {Ram}}\ and\ \bibinfo {author} {\bibfnamefont {H.}~\bibnamefont {Wenzl}},\ }\href {https://doi.org/10.1016/0021-8693(92)90109-Y} {\bibfield  {journal} {\bibinfo  {journal} {Journal of Algebra}\ }\textbf {\bibinfo {volume} {145}},\ \bibinfo {pages} {378} (\bibinfo {year} {1992})}\BibitemShut {NoStop}%
\bibitem [{\citenamefont {Mozrzymas}\ \emph {et~al.}(2018)\citenamefont {Mozrzymas}, \citenamefont {Studzi{\'n}ski},\ and\ \citenamefont {Horodecki}}]{mozrzymas2018simplified}%
  \BibitemOpen
  \bibfield  {author} {\bibinfo {author} {\bibfnamefont {M.}~\bibnamefont {Mozrzymas}}, \bibinfo {author} {\bibfnamefont {M.}~\bibnamefont {Studzi{\'n}ski}},\ and\ \bibinfo {author} {\bibfnamefont {M.}~\bibnamefont {Horodecki}},\ }\href {https://doi.org/10.1088/1751-8121/aaad15} {\bibfield  {journal} {\bibinfo  {journal} {J. Phys. A}\ }\textbf {\bibinfo {volume} {51}},\ \bibinfo {pages} {125202} (\bibinfo {year} {2018})}\BibitemShut {NoStop}%
\bibitem [{\citenamefont {Studzi{\'n}ski}\ \emph {et~al.}(2022)\citenamefont {Studzi{\'n}ski}, \citenamefont {Mozrzymas}, \citenamefont {Kopszak},\ and\ \citenamefont {Horodecki}}]{studzinski2022efficient}%
  \BibitemOpen
  \bibfield  {author} {\bibinfo {author} {\bibfnamefont {M.}~\bibnamefont {Studzi{\'n}ski}}, \bibinfo {author} {\bibfnamefont {M.}~\bibnamefont {Mozrzymas}}, \bibinfo {author} {\bibfnamefont {P.}~\bibnamefont {Kopszak}},\ and\ \bibinfo {author} {\bibfnamefont {M.}~\bibnamefont {Horodecki}},\ }\href {https://doi.org/10.1109/TIT.2022.3187852} {\bibfield  {journal} {\bibinfo  {journal} {IEEE Trans. Inf. Theory}\ }\textbf {\bibinfo {volume} {68}},\ \bibinfo {pages} {7892} (\bibinfo {year} {2022})}\BibitemShut {NoStop}%
\bibitem [{\citenamefont {Nguyen}(2023)}]{nguyen2023mixed}%
  \BibitemOpen
  \bibfield  {author} {\bibinfo {author} {\bibfnamefont {Q.~T.}\ \bibnamefont {Nguyen}},\ }\Eprint {https://arxiv.org/abs/2310.01613} {arXiv:2310.01613}  (\bibinfo {year} {2023})\BibitemShut {NoStop}%
\bibitem [{\citenamefont {Grinko}\ \emph {et~al.}(2023)\citenamefont {Grinko}, \citenamefont {Burchardt},\ and\ \citenamefont {Ozols}}]{grinko2023gelfand}%
  \BibitemOpen
  \bibfield  {author} {\bibinfo {author} {\bibfnamefont {D.}~\bibnamefont {Grinko}}, \bibinfo {author} {\bibfnamefont {A.}~\bibnamefont {Burchardt}},\ and\ \bibinfo {author} {\bibfnamefont {M.}~\bibnamefont {Ozols}},\ }\Eprint {https://arxiv.org/abs/2310.02252} {arXiv:2310.02252}  (\bibinfo {year} {2023})\BibitemShut {NoStop}%
\bibitem [{\citenamefont {Fei}\ \emph {et~al.}(2023)\citenamefont {Fei}, \citenamefont {Timmerman},\ and\ \citenamefont {Hayden}}]{fei2023efficient}%
  \BibitemOpen
  \bibfield  {author} {\bibinfo {author} {\bibfnamefont {J.}~\bibnamefont {Fei}}, \bibinfo {author} {\bibfnamefont {S.}~\bibnamefont {Timmerman}},\ and\ \bibinfo {author} {\bibfnamefont {P.}~\bibnamefont {Hayden}},\ }\Eprint {https://arxiv.org/abs/2310.01637} {arXiv:2310.01637}  (\bibinfo {year} {2023})\BibitemShut {NoStop}%
\bibitem [{\citenamefont {Chiribella}\ \emph {et~al.}(2005)\citenamefont {Chiribella}, \citenamefont {D'Ariano},\ and\ \citenamefont {Sacchi}}]{chiribella2005optimal}%
  \BibitemOpen
  \bibfield  {author} {\bibinfo {author} {\bibfnamefont {G.}~\bibnamefont {Chiribella}}, \bibinfo {author} {\bibfnamefont {G.}~\bibnamefont {D'Ariano}},\ and\ \bibinfo {author} {\bibfnamefont {M.}~\bibnamefont {Sacchi}},\ }\href {https://doi.org/10.1103/PhysRevA.72.042338} {\bibfield  {journal} {\bibinfo  {journal} {Phys. Rev. A}\ }\textbf {\bibinfo {volume} {72}},\ \bibinfo {pages} {042338} (\bibinfo {year} {2005})}\BibitemShut {NoStop}%
\bibitem [{qip()}]{qiptalk}%
  \BibitemOpen
  \href@noop {} {}\bibinfo {note} {{The error is pointed out in the talk at QIP2024, J.~Fei, S.~Timmerman and P.~Hayden, ``Quantum Algorithm for Reducing Induced Representations with Applications to Port-based Teleportation'' (\url{https://youtu.be/PhoEYpTXHqI?t=1377}).}}\BibitemShut {Stop}%
\bibitem [{\citenamefont {Grinko}()}]{grinkoprivate}%
  \BibitemOpen
  \bibfield  {author} {\bibinfo {author} {\bibfnamefont {D.}~\bibnamefont {Grinko}},\ }\href@noop {} {\bibinfo {title} {{Private communication}}}\BibitemShut {NoStop}%
\bibitem [{\citenamefont {Bagan}\ \emph {et~al.}(2004)\citenamefont {Bagan}, \citenamefont {Baig},\ and\ \citenamefont {Munoz-Tapia}}]{bagan2004entanglement}%
  \BibitemOpen
  \bibfield  {author} {\bibinfo {author} {\bibfnamefont {E.}~\bibnamefont {Bagan}}, \bibinfo {author} {\bibfnamefont {M.}~\bibnamefont {Baig}},\ and\ \bibinfo {author} {\bibfnamefont {R.}~\bibnamefont {Munoz-Tapia}},\ }\href {https://doi.org/10.1103/PhysRevA.69.050303} {\bibfield  {journal} {\bibinfo  {journal} {Phys. Rev. A}\ }\textbf {\bibinfo {volume} {69}},\ \bibinfo {pages} {050303} (\bibinfo {year} {2004})}\BibitemShut {NoStop}%
\bibitem [{\citenamefont {Chiribella}\ \emph {et~al.}(2004)\citenamefont {Chiribella}, \citenamefont {D’Ariano}, \citenamefont {Perinotti},\ and\ \citenamefont {Sacchi}}]{chiribella2004efficient}%
  \BibitemOpen
  \bibfield  {author} {\bibinfo {author} {\bibfnamefont {G.}~\bibnamefont {Chiribella}}, \bibinfo {author} {\bibfnamefont {G.}~\bibnamefont {D’Ariano}}, \bibinfo {author} {\bibfnamefont {P.}~\bibnamefont {Perinotti}},\ and\ \bibinfo {author} {\bibfnamefont {M.~F.}\ \bibnamefont {Sacchi}},\ }\href {https://doi.org/10.1103/PhysRevLett.93.180503} {\bibfield  {journal} {\bibinfo  {journal} {Phys. Rev. Lett.}\ }\textbf {\bibinfo {volume} {93}},\ \bibinfo {pages} {180503} (\bibinfo {year} {2004})}\BibitemShut {NoStop}%
\bibitem [{\citenamefont {Gu{\'e}rin}\ \emph {et~al.}(2019)\citenamefont {Gu{\'e}rin}, \citenamefont {Krumm}, \citenamefont {Budroni},\ and\ \citenamefont {Brukner}}]{guerin2019composition}%
  \BibitemOpen
  \bibfield  {author} {\bibinfo {author} {\bibfnamefont {P.~A.}\ \bibnamefont {Gu{\'e}rin}}, \bibinfo {author} {\bibfnamefont {M.}~\bibnamefont {Krumm}}, \bibinfo {author} {\bibfnamefont {C.}~\bibnamefont {Budroni}},\ and\ \bibinfo {author} {\bibfnamefont {{\v{C}}.}~\bibnamefont {Brukner}},\ }\href {https://doi.org/10.1088/1367-2630/aafef7} {\bibfield  {journal} {\bibinfo  {journal} {New J. Phys,}\ }\textbf {\bibinfo {volume} {21}},\ \bibinfo {pages} {012001} (\bibinfo {year} {2019})}\BibitemShut {NoStop}%
\bibitem [{\citenamefont {Grinko}\ and\ \citenamefont {Ozols}(2024)}]{grinko2024linear}%
  \BibitemOpen
  \bibfield  {author} {\bibinfo {author} {\bibfnamefont {D.}~\bibnamefont {Grinko}}\ and\ \bibinfo {author} {\bibfnamefont {M.}~\bibnamefont {Ozols}},\ }\href {https://doi.org/10.1007/s00220-024-05108-1} {\bibfield  {journal} {\bibinfo  {journal} {Commun. Math. Phys.}\ }\textbf {\bibinfo {volume} {405}},\ \bibinfo {pages} {278} (\bibinfo {year} {2024})}\BibitemShut {NoStop}%
\bibitem [{\citenamefont {MATLAB}(2021)}]{matlab}%
  \BibitemOpen
  \bibfield  {author} {\bibinfo {author} {\bibnamefont {MATLAB}},\ }\href@noop {} {\emph {\bibinfo {title} {version 9.11.0 (R2021b)}}}\ (\bibinfo  {publisher} {The MathWorks Inc.},\ \bibinfo {address} {Natick, Massachusetts},\ \bibinfo {year} {2021})\BibitemShut {NoStop}%
\bibitem [{\citenamefont {Grant}\ and\ \citenamefont {Boyd}(2020)}]{cvx}%
  \BibitemOpen
  \bibfield  {author} {\bibinfo {author} {\bibfnamefont {M.}~\bibnamefont {Grant}}\ and\ \bibinfo {author} {\bibfnamefont {S.}~\bibnamefont {Boyd}},\ }\href@noop {} {\bibinfo {title} {{CVX}: Matlab software for disciplined convex programming, version 2.2}},\ \bibinfo {howpublished} {\url{http://cvxr.com/cvx}} (\bibinfo {year} {2020})\BibitemShut {NoStop}%
\bibitem [{\citenamefont {Grant}\ and\ \citenamefont {Boyd}(2008)}]{gb08}%
  \BibitemOpen
  \bibfield  {author} {\bibinfo {author} {\bibfnamefont {M.}~\bibnamefont {Grant}}\ and\ \bibinfo {author} {\bibfnamefont {S.}~\bibnamefont {Boyd}},\ }in\ \href@noop {} {\emph {\bibinfo {booktitle} {Recent Advances in Learning and Control}}},\ \bibinfo {series and number} {Lecture Notes in Control and Information Sciences},\ \bibinfo {editor} {edited by\ \bibinfo {editor} {\bibfnamefont {V.}~\bibnamefont {Blondel}}, \bibinfo {editor} {\bibfnamefont {S.}~\bibnamefont {Boyd}},\ and\ \bibinfo {editor} {\bibfnamefont {H.}~\bibnamefont {Kimura}}}\ (\bibinfo  {publisher} {Springer-Verlag Limited},\ \bibinfo {year} {2008})\ pp.\ \bibinfo {pages} {95--110},\ \bibinfo {note} {\url{http://stanford.edu/~boyd/graph_dcp.html}}\BibitemShut {NoStop}%
\bibitem [{sdp()}]{sdpt3}%
  \BibitemOpen
  \href@noop {} {}\bibinfo {howpublished} {\url{http://www.math.nus.edu.sg/.mattohkc/sdpt3.html}}\BibitemShut {NoStop}%
\bibitem [{\citenamefont {Toh}\ \emph {et~al.}(1999)\citenamefont {Toh}, \citenamefont {Todd},\ and\ \citenamefont {T{\"u}t{\"u}nc{\"u}}}]{toh1999sdpt3}%
  \BibitemOpen
  \bibfield  {author} {\bibinfo {author} {\bibfnamefont {K.-C.}\ \bibnamefont {Toh}}, \bibinfo {author} {\bibfnamefont {M.~J.}\ \bibnamefont {Todd}},\ and\ \bibinfo {author} {\bibfnamefont {R.~H.}\ \bibnamefont {T{\"u}t{\"u}nc{\"u}}},\ }\href {https://doi.org/10.1080/10556789908805762} {\bibfield  {journal} {\bibinfo  {journal} {Optim. Methods Softw.}\ }\textbf {\bibinfo {volume} {11}},\ \bibinfo {pages} {545} (\bibinfo {year} {1999})}\BibitemShut {NoStop}%
\bibitem [{\citenamefont {T{\"u}t{\"u}nc{\"u}}\ \emph {et~al.}(2003)\citenamefont {T{\"u}t{\"u}nc{\"u}}, \citenamefont {Toh},\ and\ \citenamefont {Todd}}]{tutuncu2003solving}%
  \BibitemOpen
  \bibfield  {author} {\bibinfo {author} {\bibfnamefont {R.~H.}\ \bibnamefont {T{\"u}t{\"u}nc{\"u}}}, \bibinfo {author} {\bibfnamefont {K.-C.}\ \bibnamefont {Toh}},\ and\ \bibinfo {author} {\bibfnamefont {M.~J.}\ \bibnamefont {Todd}},\ }\href {https://doi.org/10.1007/s10107-002-0347-5} {\bibfield  {journal} {\bibinfo  {journal} {Math. Program.}\ }\textbf {\bibinfo {volume} {95}},\ \bibinfo {pages} {189} (\bibinfo {year} {2003})}\BibitemShut {NoStop}%
\bibitem [{\citenamefont {Sturm}(1999)}]{sedumi}%
  \BibitemOpen
  \bibfield  {author} {\bibinfo {author} {\bibfnamefont {J.~F.}\ \bibnamefont {Sturm}},\ }\href {https://doi.org/10.1080/10556789908805766} {\bibfield  {journal} {\bibinfo  {journal} {Optim. Methods Softw.}\ }\textbf {\bibinfo {volume} {11}},\ \bibinfo {pages} {625} (\bibinfo {year} {1999})}\BibitemShut {NoStop}%
\bibitem [{\citenamefont {{MOSEK ApS}}(2022)}]{mosek}%
  \BibitemOpen
  \bibfield  {author} {\bibinfo {author} {\bibnamefont {{MOSEK ApS}}},\ }\href {http://docs.mosek.com/9.0/toolbox/index.html} {\emph {\bibinfo {title} {The MOSEK optimization toolbox for MATLAB manual. Version 10.0.}}} (\bibinfo {year} {2022})\BibitemShut {NoStop}%
\bibitem [{\citenamefont {{The Sage Developers}}(2022)}]{sagemath}%
  \BibitemOpen
  \bibfield  {author} {\bibinfo {author} {\bibnamefont {{The Sage Developers}}},\ }\href@noop {} {\emph {\bibinfo {title} {{S}ageMath, the {S}age {M}athematics {S}oftware {S}ystem ({V}ersion 9.7)}}} (\bibinfo {year} {2022}),\ \bibinfo {note} {\url{https://www.sagemath.org}}\BibitemShut {NoStop}%
\bibitem [{git()}]{github}%
  \BibitemOpen
  \href@noop {} {}\bibinfo {howpublished} {\url{https://github.com/sy3104/isometry_adjointation}}\BibitemShut {NoStop}%
\bibitem [{mit()}]{mit_license}%
  \BibitemOpen
  \href@noop {} {}\bibinfo {howpublished} {\url{https://opensource.org/licenses/MIT}}\BibitemShut {NoStop}%
\bibitem [{\citenamefont {Barnum}\ and\ \citenamefont {Knill}(2002)}]{barnum2002reversing}%
  \BibitemOpen
  \bibfield  {author} {\bibinfo {author} {\bibfnamefont {H.}~\bibnamefont {Barnum}}\ and\ \bibinfo {author} {\bibfnamefont {E.}~\bibnamefont {Knill}},\ }\href {https://doi.org/10.1063/1.1459754} {\bibfield  {journal} {\bibinfo  {journal} {J. Math. Phys.}\ }\textbf {\bibinfo {volume} {43}},\ \bibinfo {pages} {2097} (\bibinfo {year} {2002})}\BibitemShut {NoStop}%
\bibitem [{\citenamefont {Ng}\ and\ \citenamefont {Mandayam}(2010)}]{ng2010simple}%
  \BibitemOpen
  \bibfield  {author} {\bibinfo {author} {\bibfnamefont {H.~K.}\ \bibnamefont {Ng}}\ and\ \bibinfo {author} {\bibfnamefont {P.}~\bibnamefont {Mandayam}},\ }\href {https://doi.org/10.1103/PhysRevA.81.062342} {\bibfield  {journal} {\bibinfo  {journal} {Phys. Rev. A}\ }\textbf {\bibinfo {volume} {81}},\ \bibinfo {pages} {062342} (\bibinfo {year} {2010})}\BibitemShut {NoStop}%
\bibitem [{\citenamefont {Beigi}\ \emph {et~al.}(2016)\citenamefont {Beigi}, \citenamefont {Datta},\ and\ \citenamefont {Leditzky}}]{beigi2016decoding}%
  \BibitemOpen
  \bibfield  {author} {\bibinfo {author} {\bibfnamefont {S.}~\bibnamefont {Beigi}}, \bibinfo {author} {\bibfnamefont {N.}~\bibnamefont {Datta}},\ and\ \bibinfo {author} {\bibfnamefont {F.}~\bibnamefont {Leditzky}},\ }\href {https://doi.org/10.1063/1.4961515} {\bibfield  {journal} {\bibinfo  {journal} {J. Math. Phys.}\ }\textbf {\bibinfo {volume} {57}} (\bibinfo {year} {2016})}\BibitemShut {NoStop}%
\bibitem [{\citenamefont {Hayden}\ and\ \citenamefont {Preskill}(2007)}]{hayden2007black}%
  \BibitemOpen
  \bibfield  {author} {\bibinfo {author} {\bibfnamefont {P.}~\bibnamefont {Hayden}}\ and\ \bibinfo {author} {\bibfnamefont {J.}~\bibnamefont {Preskill}},\ }\href {https://doi.org/10.1088/1126-6708/2007/09/120} {\bibfield  {journal} {\bibinfo  {journal} {J. High Energy Phys.}\ }\textbf {\bibinfo {volume} {2007}}\bibinfo  {number} { (09)},\ \bibinfo {pages} {120}}\BibitemShut {NoStop}%
\bibitem [{\citenamefont {Nakayama}\ \emph {et~al.}(2023)\citenamefont {Nakayama}, \citenamefont {Miyata},\ and\ \citenamefont {Ugajin}}]{nakayama2023petz}%
  \BibitemOpen
\bibfield  {number} {  }\bibfield  {author} {\bibinfo {author} {\bibfnamefont {Y.}~\bibnamefont {Nakayama}}, \bibinfo {author} {\bibfnamefont {A.}~\bibnamefont {Miyata}},\ and\ \bibinfo {author} {\bibfnamefont {T.}~\bibnamefont {Ugajin}},\ }\href {https://doi.org/10.1093/ptep/ptad147} {\bibfield  {journal} {\bibinfo  {journal} {Prog. Theor. Exp. Phys.}\ }\textbf {\bibinfo {volume} {2023}},\ \bibinfo {pages} {123B04} (\bibinfo {year} {2023})}\BibitemShut {NoStop}%
\bibitem [{\citenamefont {Utsumi}\ and\ \citenamefont {Nakata}(2024)}]{utsumi2024explicit}%
  \BibitemOpen
  \bibfield  {author} {\bibinfo {author} {\bibfnamefont {T.}~\bibnamefont {Utsumi}}\ and\ \bibinfo {author} {\bibfnamefont {Y.}~\bibnamefont {Nakata}},\ }\Eprint {https://arxiv.org/abs/2405.06051} {arXiv:2405.06051}  (\bibinfo {year} {2024})\BibitemShut {NoStop}%
\bibitem [{\citenamefont {Hayden}\ \emph {et~al.}(2004)\citenamefont {Hayden}, \citenamefont {Jozsa}, \citenamefont {Petz},\ and\ \citenamefont {Winter}}]{hayden2004structure}%
  \BibitemOpen
  \bibfield  {author} {\bibinfo {author} {\bibfnamefont {P.}~\bibnamefont {Hayden}}, \bibinfo {author} {\bibfnamefont {R.}~\bibnamefont {Jozsa}}, \bibinfo {author} {\bibfnamefont {D.}~\bibnamefont {Petz}},\ and\ \bibinfo {author} {\bibfnamefont {A.}~\bibnamefont {Winter}},\ }\href {https://doi.org/10.1007/s00220-004-1049-z} {\bibfield  {journal} {\bibinfo  {journal} {Commun. Math. Phys.}\ }\textbf {\bibinfo {volume} {246}},\ \bibinfo {pages} {359} (\bibinfo {year} {2004})}\BibitemShut {NoStop}%
\bibitem [{\citenamefont {Gily{\'e}n}\ \emph {et~al.}(2022)\citenamefont {Gily{\'e}n}, \citenamefont {Lloyd}, \citenamefont {Marvian}, \citenamefont {Quek},\ and\ \citenamefont {Wilde}}]{gilyen2022quantum}%
  \BibitemOpen
  \bibfield  {author} {\bibinfo {author} {\bibfnamefont {A.}~\bibnamefont {Gily{\'e}n}}, \bibinfo {author} {\bibfnamefont {S.}~\bibnamefont {Lloyd}}, \bibinfo {author} {\bibfnamefont {I.}~\bibnamefont {Marvian}}, \bibinfo {author} {\bibfnamefont {Y.}~\bibnamefont {Quek}},\ and\ \bibinfo {author} {\bibfnamefont {M.~M.}\ \bibnamefont {Wilde}},\ }\href {https://doi.org/10.1103/PhysRevLett.128.220502} {\bibfield  {journal} {\bibinfo  {journal} {Phys. Rev. Lett.}\ }\textbf {\bibinfo {volume} {128}},\ \bibinfo {pages} {220502} (\bibinfo {year} {2022})}\BibitemShut {NoStop}%
\bibitem [{\citenamefont {G{\"a}rttner}\ \emph {et~al.}(2017)\citenamefont {G{\"a}rttner}, \citenamefont {Bohnet}, \citenamefont {Safavi-Naini}, \citenamefont {Wall}, \citenamefont {Bollinger},\ and\ \citenamefont {Rey}}]{garttner2017measuring}%
  \BibitemOpen
  \bibfield  {author} {\bibinfo {author} {\bibfnamefont {M.}~\bibnamefont {G{\"a}rttner}}, \bibinfo {author} {\bibfnamefont {J.~G.}\ \bibnamefont {Bohnet}}, \bibinfo {author} {\bibfnamefont {A.}~\bibnamefont {Safavi-Naini}}, \bibinfo {author} {\bibfnamefont {M.~L.}\ \bibnamefont {Wall}}, \bibinfo {author} {\bibfnamefont {J.~J.}\ \bibnamefont {Bollinger}},\ and\ \bibinfo {author} {\bibfnamefont {A.~M.}\ \bibnamefont {Rey}},\ }\href {https://doi.org/10.1038/nphys4119} {\bibfield  {journal} {\bibinfo  {journal} {Nat. Phys.}\ }\textbf {\bibinfo {volume} {13}},\ \bibinfo {pages} {781} (\bibinfo {year} {2017})}\BibitemShut {NoStop}%
\bibitem [{\citenamefont {Li}\ \emph {et~al.}(2017)\citenamefont {Li}, \citenamefont {Fan}, \citenamefont {Wang}, \citenamefont {Ye}, \citenamefont {Zeng}, \citenamefont {Zhai}, \citenamefont {Peng},\ and\ \citenamefont {Du}}]{li2017measuring}%
  \BibitemOpen
  \bibfield  {author} {\bibinfo {author} {\bibfnamefont {J.}~\bibnamefont {Li}}, \bibinfo {author} {\bibfnamefont {R.}~\bibnamefont {Fan}}, \bibinfo {author} {\bibfnamefont {H.}~\bibnamefont {Wang}}, \bibinfo {author} {\bibfnamefont {B.}~\bibnamefont {Ye}}, \bibinfo {author} {\bibfnamefont {B.}~\bibnamefont {Zeng}}, \bibinfo {author} {\bibfnamefont {H.}~\bibnamefont {Zhai}}, \bibinfo {author} {\bibfnamefont {X.}~\bibnamefont {Peng}},\ and\ \bibinfo {author} {\bibfnamefont {J.}~\bibnamefont {Du}},\ }\href {https://doi.org/10.1103/PhysRevX.7.031011} {\bibfield  {journal} {\bibinfo  {journal} {Phys. Rev. X}\ }\textbf {\bibinfo {volume} {7}},\ \bibinfo {pages} {031011} (\bibinfo {year} {2017})}\BibitemShut {NoStop}%
\bibitem [{\citenamefont {Joshi}\ \emph {et~al.}(2020)\citenamefont {Joshi}, \citenamefont {Elben}, \citenamefont {Vermersch}, \citenamefont {Brydges}, \citenamefont {Maier}, \citenamefont {Zoller}, \citenamefont {Blatt},\ and\ \citenamefont {Roos}}]{joshi2020quantum}%
  \BibitemOpen
  \bibfield  {author} {\bibinfo {author} {\bibfnamefont {M.~K.}\ \bibnamefont {Joshi}}, \bibinfo {author} {\bibfnamefont {A.}~\bibnamefont {Elben}}, \bibinfo {author} {\bibfnamefont {B.}~\bibnamefont {Vermersch}}, \bibinfo {author} {\bibfnamefont {T.}~\bibnamefont {Brydges}}, \bibinfo {author} {\bibfnamefont {C.}~\bibnamefont {Maier}}, \bibinfo {author} {\bibfnamefont {P.}~\bibnamefont {Zoller}}, \bibinfo {author} {\bibfnamefont {R.}~\bibnamefont {Blatt}},\ and\ \bibinfo {author} {\bibfnamefont {C.~F.}\ \bibnamefont {Roos}},\ }\href {https://doi.org/10.1103/PhysRevLett.124.240505} {\bibfield  {journal} {\bibinfo  {journal} {Phys. Rev. Lett.}\ }\textbf {\bibinfo {volume} {124}},\ \bibinfo {pages} {240505} (\bibinfo {year} {2020})}\BibitemShut {NoStop}%
\bibitem [{\citenamefont {Blok}\ \emph {et~al.}(2021)\citenamefont {Blok}, \citenamefont {Ramasesh}, \citenamefont {Schuster}, \citenamefont {O’Brien}, \citenamefont {Kreikebaum}, \citenamefont {Dahlen}, \citenamefont {Morvan}, \citenamefont {Yoshida}, \citenamefont {Yao},\ and\ \citenamefont {Siddiqi}}]{blok2021quantum}%
  \BibitemOpen
  \bibfield  {author} {\bibinfo {author} {\bibfnamefont {M.~S.}\ \bibnamefont {Blok}}, \bibinfo {author} {\bibfnamefont {V.~V.}\ \bibnamefont {Ramasesh}}, \bibinfo {author} {\bibfnamefont {T.}~\bibnamefont {Schuster}}, \bibinfo {author} {\bibfnamefont {K.}~\bibnamefont {O’Brien}}, \bibinfo {author} {\bibfnamefont {J.-M.}\ \bibnamefont {Kreikebaum}}, \bibinfo {author} {\bibfnamefont {D.}~\bibnamefont {Dahlen}}, \bibinfo {author} {\bibfnamefont {A.}~\bibnamefont {Morvan}}, \bibinfo {author} {\bibfnamefont {B.}~\bibnamefont {Yoshida}}, \bibinfo {author} {\bibfnamefont {N.~Y.}\ \bibnamefont {Yao}},\ and\ \bibinfo {author} {\bibfnamefont {I.}~\bibnamefont {Siddiqi}},\ }\href {https://doi.org/10.1103/PhysRevX.11.021010} {\bibfield  {journal} {\bibinfo  {journal} {Phys. Rev. X}\ }\textbf {\bibinfo {volume} {11}},\ \bibinfo {pages} {021010} (\bibinfo {year} {2021})}\BibitemShut {NoStop}%
\bibitem [{\citenamefont {Odake}\ \emph {et~al.}(2024)\citenamefont {Odake}, \citenamefont {Yoshida},\ and\ \citenamefont {Murao}}]{odake2024analytical}%
  \BibitemOpen
  \bibfield  {author} {\bibinfo {author} {\bibfnamefont {T.}~\bibnamefont {Odake}}, \bibinfo {author} {\bibfnamefont {S.}~\bibnamefont {Yoshida}},\ and\ \bibinfo {author} {\bibfnamefont {M.}~\bibnamefont {Murao}},\ }\Eprint {https://arxiv.org/abs/2405.07625} {arXiv:2405.07625}  (\bibinfo {year} {2024})\BibitemShut {NoStop}%
\bibitem [{\citenamefont {Nielsen}\ and\ \citenamefont {Chuang}(1997)}]{nielsen1997programmable}%
  \BibitemOpen
  \bibfield  {author} {\bibinfo {author} {\bibfnamefont {M.~A.}\ \bibnamefont {Nielsen}}\ and\ \bibinfo {author} {\bibfnamefont {I.~L.}\ \bibnamefont {Chuang}},\ }\href {https://doi.org/10.1103/PhysRevLett.79.321} {\bibfield  {journal} {\bibinfo  {journal} {Phys. Rev. Lett.}\ }\textbf {\bibinfo {volume} {79}},\ \bibinfo {pages} {321} (\bibinfo {year} {1997})}\BibitemShut {NoStop}%
\bibitem [{\citenamefont {Du\ifmmode~\check{s}\else \v{s}\fi{}ek}\ and\ \citenamefont {Bu\ifmmode~\check{z}\else \v{z}\fi{}ek}(2002)}]{duvsek2002quantum}%
  \BibitemOpen
  \bibfield  {author} {\bibinfo {author} {\bibfnamefont {M.}~\bibnamefont {Du\ifmmode~\check{s}\else \v{s}\fi{}ek}}\ and\ \bibinfo {author} {\bibfnamefont {V.}~\bibnamefont {Bu\ifmmode~\check{z}\else \v{z}\fi{}ek}},\ }\href {https://doi.org/10.1103/PhysRevA.66.022112} {\bibfield  {journal} {\bibinfo  {journal} {Phys. Rev. A}\ }\textbf {\bibinfo {volume} {66}},\ \bibinfo {pages} {022112} (\bibinfo {year} {2002})}\BibitemShut {NoStop}%
\bibitem [{\citenamefont {Fiur\'a\ifmmode~\check{s}\else \v{s}\fi{}ek}\ \emph {et~al.}(2002)\citenamefont {Fiur\'a\ifmmode~\check{s}\else \v{s}\fi{}ek}, \citenamefont {Du\ifmmode~\check{s}\else \v{s}\fi{}ek},\ and\ \citenamefont {Filip}}]{fiuravsek2002universal}%
  \BibitemOpen
  \bibfield  {author} {\bibinfo {author} {\bibfnamefont {J.}~\bibnamefont {Fiur\'a\ifmmode~\check{s}\else \v{s}\fi{}ek}}, \bibinfo {author} {\bibfnamefont {M.}~\bibnamefont {Du\ifmmode~\check{s}\else \v{s}\fi{}ek}},\ and\ \bibinfo {author} {\bibfnamefont {R.}~\bibnamefont {Filip}},\ }\href {https://doi.org/10.1103/PhysRevLett.89.190401} {\bibfield  {journal} {\bibinfo  {journal} {Phys. Rev. Lett.}\ }\textbf {\bibinfo {volume} {89}},\ \bibinfo {pages} {190401} (\bibinfo {year} {2002})}\BibitemShut {NoStop}%
\bibitem [{\citenamefont {Fiur\'a\ifmmode~\check{s}\else \v{s}\fi{}ek}\ and\ \citenamefont {Du\ifmmode~\check{s}\else \v{s}\fi{}ek}(2004)}]{fiuravsek2004probabilistic}%
  \BibitemOpen
  \bibfield  {author} {\bibinfo {author} {\bibfnamefont {J.}~\bibnamefont {Fiur\'a\ifmmode~\check{s}\else \v{s}\fi{}ek}}\ and\ \bibinfo {author} {\bibfnamefont {M.}~\bibnamefont {Du\ifmmode~\check{s}\else \v{s}\fi{}ek}},\ }\href {https://doi.org/10.1103/PhysRevA.69.032302} {\bibfield  {journal} {\bibinfo  {journal} {Phys. Rev. A}\ }\textbf {\bibinfo {volume} {69}},\ \bibinfo {pages} {032302} (\bibinfo {year} {2004})}\BibitemShut {NoStop}%
\bibitem [{\citenamefont {D'Ariano}\ and\ \citenamefont {Perinotti}(2005)}]{dariano2005efficient}%
  \BibitemOpen
  \bibfield  {author} {\bibinfo {author} {\bibfnamefont {G.~M.}\ \bibnamefont {D'Ariano}}\ and\ \bibinfo {author} {\bibfnamefont {P.}~\bibnamefont {Perinotti}},\ }\href {https://doi.org/10.1103/PhysRevLett.94.090401} {\bibfield  {journal} {\bibinfo  {journal} {Phys. Rev. Lett.}\ }\textbf {\bibinfo {volume} {94}},\ \bibinfo {pages} {090401} (\bibinfo {year} {2005})}\BibitemShut {NoStop}%
\bibitem [{\citenamefont {Lewandowska}\ \emph {et~al.}(2022)\citenamefont {Lewandowska}, \citenamefont {Kukulski}, \citenamefont {Pawela},\ and\ \citenamefont {Pucha\l{}a}}]{lewandowska2022storage}%
  \BibitemOpen
  \bibfield  {author} {\bibinfo {author} {\bibfnamefont {P.}~\bibnamefont {Lewandowska}}, \bibinfo {author} {\bibfnamefont {R.}~\bibnamefont {Kukulski}}, \bibinfo {author} {\bibfnamefont {L.}~\bibnamefont {Pawela}},\ and\ \bibinfo {author} {\bibfnamefont {Z.}~\bibnamefont {Pucha\l{}a}},\ }\href {https://doi.org/10.1103/PhysRevA.106.052423} {\bibfield  {journal} {\bibinfo  {journal} {Phys. Rev. A}\ }\textbf {\bibinfo {volume} {106}},\ \bibinfo {pages} {052423} (\bibinfo {year} {2022})}\BibitemShut {NoStop}%
\bibitem [{\citenamefont {P\'erez-Garc\'{\i}a}(2006)}]{perez2006optimality}%
  \BibitemOpen
  \bibfield  {author} {\bibinfo {author} {\bibfnamefont {D.}~\bibnamefont {P\'erez-Garc\'{\i}a}},\ }\href {https://doi.org/10.1103/PhysRevA.73.052315} {\bibfield  {journal} {\bibinfo  {journal} {Phys. Rev. A}\ }\textbf {\bibinfo {volume} {73}},\ \bibinfo {pages} {052315} (\bibinfo {year} {2006})}\BibitemShut {NoStop}%
\bibitem [{\citenamefont {Davies}(1976)}]{davies1976quantum}%
  \BibitemOpen
  \bibfield  {author} {\bibinfo {author} {\bibfnamefont {E.~B.}\ \bibnamefont {Davies}},\ }\href@noop {} {\emph {\bibinfo {title} {Quantum theory of open systems}}}\ (\bibinfo  {publisher} {Academic Press, London, England},\ \bibinfo {year} {1976})\BibitemShut {NoStop}%
\bibitem [{\citenamefont {Masanes}(2005)}]{masanes2005extremal}%
  \BibitemOpen
  \bibfield  {author} {\bibinfo {author} {\bibfnamefont {L.}~\bibnamefont {Masanes}},\ }\Eprint {https://arxiv.org/abs/quant-ph/0512100} {arXiv:quant-ph/0512100}  (\bibinfo {year} {2005})\BibitemShut {NoStop}%
\bibitem [{\citenamefont {Acin}(2001)}]{acin2001statistical}%
  \BibitemOpen
  \bibfield  {author} {\bibinfo {author} {\bibfnamefont {A.}~\bibnamefont {Acin}},\ }\href {https://doi.org/10.1103/PhysRevLett.87.177901} {\bibfield  {journal} {\bibinfo  {journal} {Phys. Rev. Lett.}\ }\textbf {\bibinfo {volume} {87}},\ \bibinfo {pages} {177901} (\bibinfo {year} {2001})}\BibitemShut {NoStop}%
\bibitem [{\citenamefont {Chiribella}\ \emph {et~al.}(2013{\natexlab{b}})\citenamefont {Chiribella}, \citenamefont {Yang},\ and\ \citenamefont {Yao}}]{chiribella2013quantum2}%
  \BibitemOpen
  \bibfield  {author} {\bibinfo {author} {\bibfnamefont {G.}~\bibnamefont {Chiribella}}, \bibinfo {author} {\bibfnamefont {Y.}~\bibnamefont {Yang}},\ and\ \bibinfo {author} {\bibfnamefont {A.~C.-C.}\ \bibnamefont {Yao}},\ }\href {https://doi.org/10.1038/ncomms3915} {\bibfield  {journal} {\bibinfo  {journal} {Nat. Commun.}\ }\textbf {\bibinfo {volume} {4}},\ \bibinfo {pages} {2915} (\bibinfo {year} {2013}{\natexlab{b}})}\BibitemShut {NoStop}%
\bibitem [{\citenamefont {Matsumoto}(2012)}]{matsumoto2012input}%
  \BibitemOpen
  \bibfield  {author} {\bibinfo {author} {\bibfnamefont {K.}~\bibnamefont {Matsumoto}},\ }\Eprint {https://arxiv.org/abs/1209.2392} {arXiv:1209.2392}  (\bibinfo {year} {2012})\BibitemShut {NoStop}%
\bibitem [{\citenamefont {Grinko}(2025)}]{grinko2025thesis}%
  \BibitemOpen
  \bibfield  {author} {\bibinfo {author} {\bibfnamefont {D.}~\bibnamefont {Grinko}},\ }\emph {\bibinfo {title} {{Mixed Schur–Weyl duality in quantum information}}},\ \href {https://hdl.handle.net/11245.1/d9e16c26-ef20-40c5-b847-53c13d1a8a1a} {Ph.D. thesis},\ \bibinfo  {school} {{University of Amsterdam}} (\bibinfo {year} {2025})\BibitemShut {NoStop}%
\bibitem [{\citenamefont {Klimyk}\ and\ \citenamefont {Vilenkin}(1995)}]{klimyk1995representations}%
  \BibitemOpen
  \bibfield  {author} {\bibinfo {author} {\bibfnamefont {A.}~\bibnamefont {Klimyk}}\ and\ \bibinfo {author} {\bibfnamefont {N.~Y.}\ \bibnamefont {Vilenkin}},\ }in\ \href {https://doi.org/10.1007/978-94-017-2883-6} {\emph {\bibinfo {booktitle} {{Representation Theory and Noncommutative Harmonic Analysis II: Homogeneous Spaces, Representations and Special Functions}}}}\ (\bibinfo  {publisher} {Springer},\ \bibinfo {year} {1995})\ pp.\ \bibinfo {pages} {137--259}\BibitemShut {NoStop}%
\bibitem [{\citenamefont {Horodecki}\ and\ \citenamefont {Horodecki}(1999)}]{horodecki1999reduction}%
  \BibitemOpen
  \bibfield  {author} {\bibinfo {author} {\bibfnamefont {M.}~\bibnamefont {Horodecki}}\ and\ \bibinfo {author} {\bibfnamefont {P.}~\bibnamefont {Horodecki}},\ }\href {https://doi.org/10.1103/PhysRevA.59.4206} {\bibfield  {journal} {\bibinfo  {journal} {Phys. Rev. A}\ }\textbf {\bibinfo {volume} {59}},\ \bibinfo {pages} {4206} (\bibinfo {year} {1999})}\BibitemShut {NoStop}%
\bibitem [{\citenamefont {Mathematica}(2023)}]{mathematica}%
  \BibitemOpen
  \bibfield  {author} {\bibinfo {author} {\bibnamefont {Mathematica}},\ }\href {https://www.wolfram.com/mathematica} {\bibinfo {title} {{V}ersion 13.3}} (\bibinfo {year} {2023})\BibitemShut {NoStop}%
\bibitem [{\citenamefont {Ara{\'u}jo}\ \emph {et~al.}(2015)\citenamefont {Ara{\'u}jo}, \citenamefont {Branciard}, \citenamefont {Costa}, \citenamefont {Feix}, \citenamefont {Giarmatzi},\ and\ \citenamefont {Brukner}}]{araujo2015witnessing}%
  \BibitemOpen
  \bibfield  {author} {\bibinfo {author} {\bibfnamefont {M.}~\bibnamefont {Ara{\'u}jo}}, \bibinfo {author} {\bibfnamefont {C.}~\bibnamefont {Branciard}}, \bibinfo {author} {\bibfnamefont {F.}~\bibnamefont {Costa}}, \bibinfo {author} {\bibfnamefont {A.}~\bibnamefont {Feix}}, \bibinfo {author} {\bibfnamefont {C.}~\bibnamefont {Giarmatzi}},\ and\ \bibinfo {author} {\bibfnamefont {{\v{C}}.}~\bibnamefont {Brukner}},\ }\href {https://doi.org/10.1088/1367-2630/17/10/102001} {\bibfield  {journal} {\bibinfo  {journal} {New J. Phys.}\ }\textbf {\bibinfo {volume} {17}},\ \bibinfo {pages} {102001} (\bibinfo {year} {2015})}\BibitemShut {NoStop}%
\bibitem [{\citenamefont {Bavaresco}\ \emph {et~al.}(2021)\citenamefont {Bavaresco}, \citenamefont {Murao},\ and\ \citenamefont {Quintino}}]{bavaresco2021strict}%
  \BibitemOpen
  \bibfield  {author} {\bibinfo {author} {\bibfnamefont {J.}~\bibnamefont {Bavaresco}}, \bibinfo {author} {\bibfnamefont {M.}~\bibnamefont {Murao}},\ and\ \bibinfo {author} {\bibfnamefont {M.~T.}\ \bibnamefont {Quintino}},\ }\href {https://doi.org/10.1103/PhysRevLett.127.200504} {\bibfield  {journal} {\bibinfo  {journal} {Phys. Rev. Lett.}\ }\textbf {\bibinfo {volume} {127}},\ \bibinfo {pages} {200504} (\bibinfo {year} {2021})}\BibitemShut {NoStop}%
\bibitem [{\citenamefont {Itzykson}\ and\ \citenamefont {Nauenberg}(1966)}]{itzykson1966unitary}%
  \BibitemOpen
  \bibfield  {author} {\bibinfo {author} {\bibfnamefont {C.}~\bibnamefont {Itzykson}}\ and\ \bibinfo {author} {\bibfnamefont {M.}~\bibnamefont {Nauenberg}},\ }\href {https://doi.org/10.1103/RevModPhys.38.95} {\bibfield  {journal} {\bibinfo  {journal} {Rev. Mod. Phys.}\ }\textbf {\bibinfo {volume} {38}},\ \bibinfo {pages} {95} (\bibinfo {year} {1966})}\BibitemShut {NoStop}%
\end{thebibliography}%

\end{document}